\newcommand{\CAMERA}[1]{}
\newcommand{\REPORT}[1]{#1}
\newcommand{\CHANGED}[1]{#1}
\newcommand{\REMOVED}[1]{}

\documentclass[acmsmall,screen]{acmart}
\REPORT{
\acmJournal{PACMPL}
\acmVolume{1}
\acmNumber{POPL} 
\acmArticle{1}
\acmYear{2019}
\acmMonth{1}
\acmDOI{} 
\startPage{1}
\setcopyright{none}
}

\newcommand{\qslemb}[1]{\ensuremath{\mathtt{qsl}\llbracket #1 \rrbracket}}

\newcommand{\SLsingleton}[2]{#1 \mapsto #2}
\newcommand{\SLvalidpointer}[1]{#1 \mapsto \,{-}\,}
\newcommand{\SLwp}[2]{\underline{\wp{#1}{#2}}}

\newcommand{\emp}{\iverson{\sfsymbol{\textbf{emp}}}}

\newcommand{\singleton}[2]{\iverson{#1 \mapsto #2}}
\newcommand{\containsPointer}[2]{\iverson{#1 \hookrightarrow #2}}
\newcommand{\validpointer}[1]{\iverson{#1 \mapsto \,{-}\,}}
\newcommand{\heapSize}{\sfsymbol{\textbf{size}}}
\newcommand{\sepcon}{\mathbin{{\star}}}

\newcommand{\sepimp}{\mathbin{\text{\raisebox{-0.1ex}{$\boldsymbol{{-}\hspace{-.55ex}{-}}$}}\hspace{-1ex}\text{\raisebox{0.13ex}{\rotatebox{-17}{$\star$}}}}}
\newcommand{\nil}{0}

\newcommand{\bbigsepcon}[2]{\ensuremath{\overset{#2}{\underset{#1}{\bigstar}}}}

\newcommand{\QSL}{\sfsymbol{QSL}\xspace}
\newcommand{\SL}{\sfsymbol{SL}\xspace}

\newcommand{\hoare}[3]{\left\langle \,{#1}\vphantom{#3}\, \right\rangle \mathrel{#2} \left\langle \, {#3}\vphantom{#1} \, \right\rangle}

\newcommand{\sfsymbol}[1]{\textsf{\upshape {#1}}}
\newcommand{\ttsymbol}[1]{\texttt{\upshape {#1}}}
\newcommand{\Assert}[1]{\ensuremath{\text{\footnotesize ~\texttt{//}~$\displaystyle #1$} }}

\newcommand{\wpsymbol}{\sfsymbol{wp}}

\renewcommand{\wp}[2]{\wpsymbol\llbracket#1\rrbracket\left(#2\right)}
\newcommand{\extwp}[2]{\Ext{\wpsymbol}\llbracket#1\rrbracket\left(#2\right)}
\newcommand{\wpC}[1]{\wpsymbol\llbracket#1\rrbracket}

\newcommand{\awpsymbol}{\sfsymbol{awp}}

\newcommand{\awp}[2]{\awpsymbol\llbracket#1\rrbracket\left(#2\right)}

\newcommand{\wlpsymbol}{\sfsymbol{wlp}}

\newcommand{\wlp}[2]{\wlpsymbol\llbracket#1\rrbracket\left(#2\right)}

\newcommand{\esepcon}{\mathbin{{\bullet}}}
\newcommand{\esepimp}{\mathbin{\text{\raisebox{-0.1ex}{$\boldsymbol{{-}\hspace{-.55ex}{-}}$}}\hspace{-1ex}\text{\raisebox{0.13ex}{\rotatebox{-17}{$\bullet$}}}}}
\newcommand{\wepsymbol}{\sfsymbol{wep}}

\newcommand{\wep}[2]{\wepsymbol\llbracket#1\rrbracket\left(#2\right)}

\newcommand{\awlpsymbol}{\sfsymbol{awlp}}

\newcommand{\awlp}[2]{\awlpsymbol\llbracket#1\rrbracket\left(#2\right)}

\newcommand{\awepsymbol}{\sfsymbol{awep}}

\newcommand{\awep}[2]{\awepsymbol\llbracket#1\rrbracket\left(#2\right)}

\newcommand{\awlepsymbol}{\sfsymbol{awlep}}

\newcommand{\awlep}[2]{\awlepsymbol\llbracket#1\rrbracket\left(#2\right)}

\newcommand{\wlepsymbol}{\sfsymbol{wlep}}

\newcommand{\wlep}[2]{\wlepsymbol\llbracket#1\rrbracket\left(#2\right)}

\newcommand{\conditionalPair}[2]{{\let\oldarraystretch\arraystretch}\renewcommand{\arraystretch}{1}~\holter{~\raisebox{.5ex}{${#1}$}~}{~\raisebox{.125ex}{${#2}$}~}~\renewcommand{\arraystretch}{\oldarraystretch}}

\newcommand{\cc}{\ensuremath{c}} 
\newcommand{\guard}{\ensuremath{b}} 
\newcommand{\ee}{\ensuremath{e}} 

\newcommand{\preda}{\ensuremath{\varphi}} 
\newcommand{\predb}{\ensuremath{\psi}} 
\newcommand{\predc}{\ensuremath{\vartheta}}

\newcommand{\hh}{\ensuremath{h}} 
\newcommand{\sk}{\ensuremath{s}} 

\newcommand{\pp}{\ensuremath{p}} 

\newcommand{\ff}{\ensuremath{X}} 
\newcommand{\fg}{\ensuremath{Y}}
\newcommand{\fh}{\ensuremath{Z}}
\newcommand{\fk}{\ensuremath{R}}
\newcommand{\inv}{\ensuremath{I}} 

\newcommand{\oa}{\ensuremath{\alpha}} 
\newcommand{\ob}{\ensuremath{\beta}} 
\newcommand{\oc}{\ensuremath{\delta}} 

\newcommand{\za}{\ensuremath{\alpha}} 
\newcommand{\zb}{\ensuremath{\beta}}
\newcommand{\zc}{\ensuremath{\gamma}}

\newcommand{\SKIP}{\ttsymbol{skip}}

\newcommand{\AssignSymbol}{\mathrel{\textnormal{\texttt{:=}}}}
\newcommand{\ASSIGN}[2]{\ensuremath{#1 \AssignSymbol #2}}
\newcommand{\UNIFORM}[2]{\ensuremath{\mathtt{uniform}\left(#1,#2\right)}}
\newcommand{\ASSIGNUNIFORM}[3]{\ensuremath{#1 \AssignSymbol \UNIFORM{#2}{#3}}}
\newcommand{\ALLOC}[2]{\ensuremath{{#1} \AssignSymbol \mathtt{new}\left( #2 \right)}}

\newcommand{\AVAILLOC}[1]{\PosNats}
\usepackage{actuarialangle}

\newcommand{\dereference}[1]{\texttt{<}\,#1\,\texttt{>}}

\newcommand{\HASSIGN}[2]{\ensuremath{\dereference{#1} \AssignSymbol #2}}
\newcommand{\ASSIGNH}[2]{\ensuremath{#1 \AssignSymbol \dereference{#2}}}
\newcommand{\FREE}[1]{\ensuremath{\mathtt{free}(#1)}}

\newcommand{\SEMI}{\ensuremath{\,;\,}}
\newcommand{\COMPOSE}[2]{\ensuremath{{#1}{\,;}~ {#2}}}
\newcommand{\PCHOICE}[3]{\ensuremath{\left\{\, {#1} \,\right\}\mathrel{\left[\,#2\,\right]}\left\{\, {#3} \,\right\}}}

\newcommand{\IFSYMBOL}{\ensuremath{\textnormal{\texttt{if}}}}
\newcommand{\IF}[1]{\ensuremath{\IFSYMBOL\,\left(\, {#1} \,\right)\,\{}}
\newcommand{\ELSESYMBOL}{\ensuremath{\textnormal{\texttt{else}}}}
\newcommand{\ELSE}{\ensuremath{\}\,\ELSESYMBOL\,\{}}
\newcommand{\ITE}[3]{\ensuremath{\IFSYMBOL\,\left(\, {#1} \,\right)\,\left\{\, {#2} \,\right\}\,\ELSESYMBOL\,\left\{\, {#3} \,\right\}}}

\newcommand{\WHILESYMBOL}{\ensuremath{\textnormal{\texttt{while}}}}
\newcommand{\WHILE}[1]{\ensuremath{\WHILESYMBOL \left(\, {#1} \,\right)\left\{\right.}}

\newcommand{\WHILEDO}[2]{\ensuremath{\WHILESYMBOL \left(\, {#1} \,\right)\left\{\, {#2} \,\right\}}}

\newcommand{\hpgcl}{\textnormal{\sfsymbol{hpGCL}}\xspace}   
\newcommand{\boldhpgcl}{\textnormal{\textbf{\sfsymbol{hpGCL}}}\xspace}   
\newcommand{\Vars}{\ensuremath{\mathsf{Vars}}\xspace}   

\newcommand{\Nats}{\ensuremath{\mathbb{N}}\xspace}
\newcommand{\PosNats}{\ensuremath{\mathbb{N}_{>0}}\xspace}
\newcommand{\Ints}{\ensuremath{\mathbb{Z}}\xspace}
\newcommand{\Rats}{\ensuremath{\mathbb{Q}}\xspace}
\newcommand{\Reals}{\mathbb{R}}
\newcommand{\PosReals}{\mathbb{R}_{\geq 0}}
\newcommand{\PosRealsInf}{\mathbb{R}_{\geq 0}^\infty}

\newcommand{\E}{\mathbb{E}}

\newcommand{\Eone}{\mathbb{E}_{\leq 1}}

\newcommand{\Mod}[1]{\ensuremath{\textit{Mod}\left(#1\right)}}

\newcommand{\Permutations}{\ensuremath{\Pi}}
\newcommand{\FinPermutations}[1]{\ensuremath{\Pi_{#1}}}
\newcommand{\dom}[1]{\sfsymbol{dom}\left({#1}\right)}
\newcommand{\iverson}[1]{\left[ {#1} \right]}

\newcommand{\Max}[2]{\max\left\{\,{#1},\: {#2}\,\right\}}

\newcommand{\subst}[2]{\left[ {#1} \middle/ {#2}\right]}
\newcommand{\statesubst}[2]{\left[ {#1} \middle/ {#2}\right]}

\newcommand{\charwpsym}{\Phi}
\newcommand{\charwp}[3]{\charwpsym\llbracket#1,#2,#3\rrbracket} 
\newcommand{\charwpn}[4]{\charwpsym^{#4}\llbracket#1,#2,#3\rrbracket} 

\newcommand{\pot}[1]{\mathcal{P}\left({#1}\right)}

\newcommand{\Stacks}{\mathcal{S}}
\newcommand{\Heaps}{\mathcal{H}}
\newcommand{\emptyheap}{h_\emptyset}
\newcommand{\disjoint}{\mathrel{\bot}}
\newcommand{\States}{\Sigma}
\newcommand{\To}{\rightarrow}
\newcommand{\true}{\mathsf{true}}
\newcommand{\false}{\mathsf{false}}
\newcommand{\mydot}{\text{{\Large\textbf{.}}~}}

\newcommand{\qiff}{\quad\textnormal{iff}\quad}
\newcommand{\qqiff}{\qquad\textnormal{iff}\qquad}

\newcommand{\qand}{\quad\textnormal{and}\quad}

\newcommand{\qor}{\quad\textnormal{or}\quad}

\newcommand{\qimplies}{\quad\textnormal{implies}\quad}

\newcommand{\ppreceq}{~{}\preceq{}~}
\newcommand{\ssucceq}{~{}\succeq{}~}
\newcommand{\eeq}{~{}={}~}
\newcommand{\nneq}{~{}\neq{}~}
\newcommand{\qeq}{\quad{}={}\quad}

\newcommand{\lleq}{~{}\leq{}~}

\newcommand{\ggeq}{~{}\geq{}~}

\newcommand{\iimplies}{~{}\implies{}~}
\newcommand{\mmodels}{~{}\models{}~}

\newcommand{\iin}{~{}\in{}~}

\newcommand{\setcomp}[2]{\left\{\, {#1} ~\middle|~ {#2} \,\right\}}

\newcommand{\underdot}[1]{%
    \tikz[baseline=(todotted.base)]{
        \node[inner sep=1pt,outer sep=0pt] (todotted) {#1};
        \draw[densely dotted] (todotted.south west) -- (todotted.south east);
    }%
}%
\newcommand{\cloze}[1]{\underdot{\phantom{#1}}}

\newcommand{\nicepar}[1]{\vcenter{\hbox{$\displaystyle #1$}}}

\newcounter{computationarrowsone}
\newcounter{computationarrowstwo}

\newcounter{sarrow}

\newcommand{\lfp}{\ensuremath{\textnormal{\sfsymbol{lfp}}~}}
\newcommand{\gfp}{\ensuremath{\textnormal{\sfsymbol{gfp}}~}}

\newcommand{\Tree}[1]{\ensuremath{\iverson{\mathsf{tree}\left(#1\right)}}}

\newcommand{\Sll}[2]{\ensuremath{\mathsf{sll}\left(#1,#2\right)}}
\newcommand{\Path}[1]{\ensuremath{\mathsf{path}\left(#1\right)}}

\newcommand{\ExecSymbol}{\ensuremath{\rightarrow}}

\newcommand{\Exec}[8]{\ensuremath{#1,#2,#3 \,\xrightarrow{#4, #5}\, #6, #7, #8 }}

\newcommand{\ExecSimple}[3]{\ensuremath{#1 \,\xrightarrow{#2}\, #3}}
\newcommand{\Fault}{\ensuremath{\text{\Lightning}}}
\newcommand{\Term}{\ensuremath{\Downarrow}}

\newcommand{\Update}[2]{\ensuremath{\subst{#1}{#2}}} 
\newcommand{\HeapSet}[1]{\ensuremath{\{ #1 \}}}

\newcommand{\Fsymbol}{\ensuremath{\Phi}}
\newcommand{\F}[2]{\ensuremath{\Fsymbol\llbracket#1\rrbracket\left(#2\right)}}

\newcommand{\Ext}[1]{\ensuremath{\tilde{#1}}}
\newcommand{\EF}[2]{\ensuremath{\Ext{\Fsymbol}\llbracket#1\rrbracket\left(#2\right)}}

\newcommand{\Opsymbol}{\ensuremath{\mathsf{op}}}
\newcommand{\Opf}[2]{\ensuremath{\mathsf{op}\llbracket#1\rrbracket\left(#2\right)}}
\newcommand{\EOpf}[2]{\ensuremath{\Ext{\mathsf{op}}\llbracket#1\rrbracket\left(#2\right)}}

\newcommand{\OpStates}{\ensuremath{\mathsf{Conf}}}
\newcommand{\OpActions}{\ensuremath{\mathsf{Act}}}
\newcommand{\OpAct}[1]{\ensuremath{\mathsf{Act}(#1)}}
\newcommand{\OpRew}{\ensuremath{\mathsf{rew}}}
\newcommand{\OpInit}{\ensuremath{\mathsf{init}}}

\newcommand{\ExpRewC}[2]{\ensuremath{\mathsf{ExpRew}\llbracket #1 \rrbracket\left(#2\right)}}

\newcommand{\ProbSymbol}{\ensuremath{\mathsf{Prob}}}
\newcommand{\Prob}[1]{\ensuremath{\ProbSymbol(#1)}}

\newcommand{\Target}{\ensuremath{\mathcal{G}}}
\newcommand{\PathsFromTo}[2]{\ensuremath{\Pi[#1]\left(#2\right)}}
\newcommand{\Scheduler}{\ensuremath{\rho}}

\newcommand{\Ord}{\textit{Ord}}

\newcommand{\itag}[1]{\left\llbracket~\textnormal{#1}~\right\rrbracket\notag}
\newcommand{\eeqtag}[1]{\eeq & \itag{#1} \\}
\newcommand{\lleqtag}[1]{\lleq & \itag{#1} \\}
\newcommand{\ggeqtag}[1]{\ggeq & \itag{#1} \\}
\newcommand{\ppreceqtag}[1]{\ppreceq & \itag{#1} \\}
\newcommand{\ssucceqtag}[1]{\ssucceq & \itag{#1} \\}
\newcommand{\iimpliestag}[1]{\iimplies& \itag{#1} \\}
\newcommand{\leftrighttag}[1]{~\Longleftrightarrow~& \itag{#1} \\}

\newcommand{\Lssymbol}{\ensuremath{\mathsf{ls}}}
\newcommand{\Ls}[2]{\ensuremath{\iverson{\Lssymbol\left(#1,#2\right)}}}
\newcommand{\Lensymbol}{\ensuremath{\mathsf{len}}}
\newcommand{\Len}[2]{\ensuremath{\Lensymbol\left(#1,#2\right)}}
\newcommand{\npath}[2]{\ensuremath{\mathsf{path}\llbracket #1 \rrbracket\left(#2\right)}}

\newcommand{\rhpgcl}{\textnormal{\sfsymbol{rhpGCL}}\xspace}

\newcommand{\ProcNames}{\ensuremath{\mathsf{ProcNames}}}
\newcommand{\ProcName}[1]{\ensuremath{\mathtt{#1}}}

\newcommand{\ProcDeclIs}[4]{\ensuremath{\texttt{procedure}~\ProcName{#1}\left(#2\right) \left\{ \, #4 \, \right\}}}
\newcommand{\ProcDecl}[3]{\ensuremath{\texttt{procedure}~\ProcName{#1}\left(#2\right) \{ }}
\newcommand{\ProcCall}[3]{\ensuremath{\texttt{call}~\ProcName{#1}\left(#2\right)}}

\newcommand{\ProcPsi}[2]{\ensuremath{\Psi_{#1,#2}}}

\newcommand{\DeclPSym}{\ensuremath{\mathsf{DProc}}}
\newcommand{\DeclP}[3]{\ensuremath{\DeclPSym\llbracket #1, #2 \rrbracket\left(#3\right)}}

\newcommand{\Stores}{\ensuremath{\mathsf{Stores}}}
\newcommand{\Next}{\ensuremath{\mathsf{next}}}
\newcommand{\VarEnv}{\ensuremath{\mathsf{VarEnv}}}

\newcommand{\ProcEnv}{\ensuremath{\mathsf{ProcEnv}}}

\newcommand{\store}{\ensuremath{\tau}}
\newcommand{\varenv}{\ensuremath{\nu}}
\newcommand{\procenv}{\ensuremath{\rho}}
\newcommand{\procenvP}[3]{\ensuremath{\rho}\llbracket \ProcName{#1}, #2 \rrbracket}
\newcommand{\procenvPrime}[3]{\ensuremath{\rho'}\llbracket \ProcName{#1}, #2 \rrbracket}
\newcommand{\lookup}{\ensuremath{(\store \circ \varenv)}}
\newcommand{\stack}{\ensuremath{\mathsf{stack}}}
\newcommand{\ev}[2]{\ensuremath{#1 \circ \stack(#2)}} 
\newcommand{\evv}[1]{\ensuremath{\ev{#1}{\varenv}}}
\newcommand{\Estore}{\ensuremath{\E^{\store}}}
\newcommand{\Eadm}{\ensuremath{\E^{\store}_{\varenv}}}

\newcommand{\csPLossyReversal}{lossyReversal}

\newcommand{\HeapPartitions}[3]{\textit{Partitions}(#1,#2,#3)}

\newcommand{\lsHead}{\textit{hd}}
\newcommand{\lsRev}{r}
\newcommand{\lsTmp}{t}

\newcommand{\perm}[2]{\text{Perm}\left( #1 , #2 \right) }
\newcommand{\aarray}{\texttt{array}}

\newcommand{\crand}{\cc_{\text{loop}}}
\newcommand{\cbody}{\cc_{\text{body}}}

\usepackage{subcaption} 
                        
\usepackage{adjustbox}
\usepackage{amsmath,amssymb,amsfonts}
\usepackage{marvosym}
\usepackage{hyperref}
\usepackage{graphicx}
\usepackage{nameref}
\usepackage{stmaryrd}
\usepackage{xcolor}
\usepackage{xfrac}
\usepackage{xspace}
\usepackage{cancel}

\usepackage{listings,multicol}
\usepackage{tensor}
\usepackage{tikz}
\usetikzlibrary{patterns,decorations.pathreplacing,arrows,arrows.meta,decorations.pathmorphing,positioning,fit,trees,shapes,shadows,automata,calc}
\usepackage{pgfplots}

\usepackage{caption}
\usepackage{proof}
\usepackage{mathtools}

\usepackage[textwidth=19mm]{todonotes}

\allowdisplaybreaks

\begin{document}

\addtocontents{toc}{\protect\setcounter{tocdepth}{-1}}

\title[Quantitative Separation Logic]{Quantitative Separation Logic} 
\REPORT{\titlenote{This technical report supplements a paper of the same title published at POPL 2019.}}             
\subtitle{A Logic for Reasoning about Probabilistic Pointer Programs}                     


\author[Batz]{Kevin Batz}
\affiliation{
  \institution{RWTH Aachen University, Germany}            
}
\email{kevin.batz@rwth-aachen.de}          

\author[Kaminski]{Benjamin Lucien Kaminski}
\affiliation{
  \institution{RWTH Aachen University, Germany}            
}
\email{benjamin.kaminski@cs.rwth-aachen.de}          

\author[Katoen]{Joost-Pieter Katoen}
\affiliation{
  \institution{RWTH Aachen University, Germany}            
}
\email{katoen@cs.rwth-aachen.de}          

\author[Matheja]{Christoph Matheja}
\affiliation{
  \institution{RWTH Aachen University, Germany}            
}
\email{matheja@cs.rwth-aachen.de}          

\author[Noll]{Thomas Noll}
\affiliation{
  \institution{RWTH Aachen University, Germany}            
}
\email{noll@cs.rwth-aachen.de}          

\begin{abstract}
	We present \emph{quantitative separation logic} (\QSL).
	In contrast to classical separation logic, \QSL employs quantities which evaluate to real numbers instead of predicates which evaluate to Boolean values.
    The connectives of classical separation logic, separating conjunction 
    and separating implication, 
    are lifted from predicates to quantities.
    This extension is conservative: Both connectives 
    are backward compatible to their classical analogs and 
    obey the same laws,
    e.g.\ modus ponens, adjointness, 
    etc.

    Furthermore, we develop a weakest precondition calculus for quantitative reasoning about \emph{probabilistic pointer programs} in \textsf{QSL}.
    This calculus is a conservative extension of both \CHANGED{Ishtiaq's, O'Hearn's and Reynolds'} separation logic for heap-manipulating programs and Kozen's / McIver and Morgan's weakest preexpectations for probabilistic programs.
    Soundness is proven with respect to an operational semantics based on Markov decision processes.
    Our calculus preserves O'Hearn's \emph{frame rule}, which enables local reasoning. 
    We demonstrate that our calculus enables reasoning about quantities such as 
    the probability of terminating with an empty heap, the probability of reaching a certain array permutation, or the expected length of a list.
%
\end{abstract}



\maketitle


\section{Introduction}\label{sec:intro}

Randomization plays an important role in the construction of 
algorithms.
It typically improves average-case performance at the cost of a worse best-case performance or at the cost of incorrect results occurring with low probability.
The former is observed when, e.g., randomly picking the pivot in quicksort~\cite{DBLP:journals/cj/Hoare62}.
A prime example of the latter is Freivalds' matrix multiplication \mbox{verification algorithm~\cite{freivalds1977probabilistic}}.

Sophisticated algorithms often make use of \emph{randomized data structures}.
For instance, Pugh states that randomized skip lists enjoy ``the same asymptotic expected time bounds as balanced trees and are faster and use less space''~\cite{DBLP:journals/cacm/Pugh90}.
Other examples of randomized data structures include randomized splay trees~\cite{DBLP:journals/ipl/AlbersK02}, treaps~\cite{DBLP:conf/spaa/BlellochR98} and randomized search trees~\cite{DBLP:conf/focs/AragonS89,DBLP:journals/jacm/MartinezR98}.

Randomized algorithms are conveniently described by probabilistic programs, i.e. programs with the ability to sample from a probability distribution, e.g. by flipping coins. 
While randomized algorithms have desirable properties, their verification often requires reasoning about programs that mutate dynamic data structures \emph{and} behave probabilistically.
Both tasks are challenging on their own and have been the subject of intensive research, see e.g.~\cite{DBLP:conf/focs/Kozen79,DBLP:journals/pacmpl/McIverMKK18,DBLP:conf/cav/ChakarovS13,DBLP:series/natosec/OHearn12,DBLP:conf/popl/ChatterjeeFNH16,bartheprogram,DBLP:conf/popl/KrebbersTB17,DBLP:conf/pldi/NgoC018}.
However, to the best of our knowledge, work on formal verification of programs that are \emph{both} randomized \emph{and} heap-manipulating is scarce.
To highlight the need for \emph{quantitative properties} and their formal verification in this setting let us consider three examples.
\begin{figure}[t]
\begin{subfigure}[b]{0.45\textwidth}
\centering
\begin{align*}
        & \ProcDecl{randomize}{\aarray, n}{} \\
        & \qquad \ASSIGN{i}{0}\SEMI \\
        & \qquad \WHILE{0 \leq i < n} \\
        & \qquad \qquad \ASSIGNUNIFORM{j}{i}{n-1}\SEMI \\
        & \qquad \qquad \ProcCall{swap}{\aarray, i , j}{} \SEMI \\
        & \qquad \qquad \ASSIGN{i}{i+1} \\
        & \} ~ \}
\end{align*}
\caption{Procedure to randomize an array of length $n$}
\label{fig:randomize-array}
\end{subfigure}
\hfill
\begin{subfigure}[b]{0.45\textwidth}
\centering
\begin{align*}
        & \ProcDecl{\csPLossyReversal}{\lsHead}{} \\
        & \qquad \ASSIGN{\lsRev}{0}\SEMI \\
        & \qquad \WHILE{\lsHead \neq 0} \\
        & \qquad \qquad \ASSIGNH{\lsTmp}{\lsHead}\SEMI \\
        & \qquad \qquad \left\{\,
        \begin{aligned}
           & \HASSIGN{\lsHead}{\lsRev}\SEMI \\
           & \ASSIGN{\lsRev}{\lsHead}
        \end{aligned}
        \,\right\}\,\left[\sfrac{1}{2}\right]\,\left\{\,
        \begin{aligned}
           & \FREE{\lsHead} \\
           &
        \end{aligned}
        \,\right\} \\
        & \qquad \qquad \ASSIGN{\lsHead}{\lsTmp} \\
        & \} \, \}
\end{align*}
\caption{Lossy reversal of a list with head $\lsHead$}
\label{fig:lossy-list-reversal}
\end{subfigure}
\caption{Examples of probabilistic programs. We write $\dereference{e}$ to access the value stored at address $e$.}
\end{figure}

\paragraph{Example 1: Array randomization}
A common approach to design randomized algorithms is to randomize the input and process it in a deterministic manner. 
For instance, the only randomization involved in algorithms solving the famous Secretary Problem (cf.~\cite[Chapter 5.1]{DBLP:books/daglib/0023376}) is computing a random permutation of its input array.
A textbook implementation (cf.~\cite[Chapter 5.3]{DBLP:books/daglib/0023376}) of such a procedure $\ProcName{randomize}$ for an array of length $n$ is depicted in
Figure~\ref{fig:randomize-array}.
For each position in the array, the procedure uniformly samples a random number $j$ in the remaining array between the current position $i$ and the last position $n - 1$. 
After that, the elements at position $i$ and $j$ are swapped.
The procedure $\ProcName{randomize}$ is correct precisely if all outputs are equally likely.
Thus, to verify correctness of this procedure, we inevitably have to reason about a probability, hence a \emph{quantity}.
In fact, each of the $n!$ possible permutations of the input array is computed by procedure $\ProcName{randomize}$ with probability at most $\sfrac{1}{n!}$. 
%

\paragraph{Beyond randomized algorithms}

Probabilistic programs are a powerful modeling tool that is not limited to randomized algorithms.
Consider, for instance, approximate computing:
Programs running on unreliable hardware, where instructions may occasionally return incorrect results, are naturally captured by probabilistic programs~\cite{DBLP:journals/cacm/CarbinMR16}.
Since incorrect results are unavoidable in such a scenario, the notion of a program's correctness becomes blurred:
That is, quantifying (and minimizing) the probability of encountering a failure or the expected error of a program becomes crucial.
The need for quantitative reasoning is also stressed by~\cite{DBLP:journals/ife/Henzinger13} who argues that 
``the Boolean partition of software into correct and incorrect programs falls short of the practical need to assess the behavior of software in a more nuanced fashion [\,\ldots].''

\paragraph{Example 2: Faulty garbage collector}

Consider a procedure $\ProcName{delete}(x)$ that takes a tree with root $x$ and recursively deletes all of its elements.
This is a classical example due to~\cite{DBLP:conf/lics/Reynolds02,DBLP:series/natosec/OHearn12}.
However, our procedure fails with some probability $p \in [0,1]$ to continue deleting subtrees, i.e. 
running $\ProcName{delete}(x)$ on a tree with root $x$ does not necessarily result in the empty heap.
If failures of $\ProcName{delete}(x)$ are caused by unreliable hardware, they are unavoidable.
Instead of proving a Boolean correctness property, we are thus interested in evaluating the reliability of the procedure by \emph{quantifying the probability of collecting all garbage}.
%
In fact, the probability of completely deleting a tree with root $x$ containing $n$ nodes is at least $(1-p)^{n}$. 
Thus, to guarantee that a tree containing $100$ elements is deleted at least with probability $0.90$, the probability $p$ must \mbox{be below $0.00105305$}.

\paragraph{Example 3: Lossy list reversal}

A prominent benchmark when analyzing heap-manipulating programs
is in-place list-reversal (cf.~\cite{DBLP:conf/popl/KrebbersTB17, magill2006inferring, DBLP:journals/corr/abs-1104-1998}).
Figure~\ref{fig:lossy-list-reversal} depicts a \emph{lossy} list reversal: The procedure $\ProcName{\csPLossyReversal}$ traverses a list with head $\lsHead$ and attempts to move each element to the front of an initially empty list with head $\lsRev$. However, during each iteration, the current element is dropped with probability $\sfrac{1}{2}$. This is modeled by a \emph{probabilistic choice}, which either updates the value at address $\lsHead$ or disposes that address:
\begin{align*}
        \PCHOICE{\HASSIGN{\lsHead}{\lsRev}\SEMI\ASSIGN{\lsRev}{\lsHead}}{\sfrac{1}{2}}{\FREE{\lsHead}}
\end{align*}
%
The procedure $\ProcName{\csPLossyReversal}$ is not functionally correct in the sense that, upon termination, $\lsRev$ is the head of the reversed initial list:
Although the program never crashes due to a memory fault and indeed produces a singly-linked list, the length of this list varies between zero and the length of the initial list.
%
%
A more sensible quantity of interest is the \emph{expected}, i.e. average, \emph{length of the reversed list}.
In fact, the expected list length is \emph{at most half} of the length of the original list. 

\paragraph{Our approach}

We develop a \emph{quantitative separation logic} (\QSL) for
quantitative reasoning about heap-manipulating \emph{and} probabilistic programs at source code level.
Its distinguished features are:
\begin{itemize}
\item
\emph{\QSL is quantitative}: It evaluates to a real number instead of a Boolean value. 
It is capable of specifying values of program variables, heap sizes, list lengths, etc. 
%
\item
\emph{\QSL is probabilistic}: It enables reasoning about probabilistic programs, in particular about the \emph{probability of terminating with a correct result}. 
It allows to express \emph{expected values} of quantities, such as expected heap size or expected list length in a natural way. 
\item 
\emph{\QSL is a separation logic}: It conservatively extends separation logic (\SL)~\cite{DBLP:conf/popl/IshtiaqO01,DBLP:conf/lics/Reynolds02,DBLP:conf/fossacs/YangO02}.
Our quantitative analogs of \SL's key operators, i.e. separating conjunction $\sepcon$ and separating implication $\sepimp$, 
preserve virtually all properties of their Boolean versions. 
%
\end{itemize}%
For program verification, separation logic is often used in a (forward) Floyd-Hoare style.
For probabilistic programs, however, backward reasoning is more common.
In fact, certain forward-directed predicate transformers do not exist when reasoning about probabilistic programs~\cite[p. 135]{DBLP:phd/ethos/Jones90}.
%
We develop a (backward) weakest-precondition style calculus that uses \QSL to verify probabilistic heap-manipulating programs.
This calculus is a marriage of the weakest preexpectation calculus by~\cite{DBLP:series/mcs/McIverM05} and separation logic \`{a} la~\cite{DBLP:conf/popl/IshtiaqO01,DBLP:conf/lics/Reynolds02}.
In particular:

%
\begin{itemize}

    \item Our calculus is a \emph{conservative extension of two approaches}: 
          For programs that never access the heap, we obtain the calculus of McIver and Morgan.
          Conversely, for Boolean properties of ordinary programs, we recover exactly the $\wpsymbol$-rules of \CHANGED{Ishtiaq, O'Hearn, and Reynolds}.
          \QSL preserves virtually all properties of classical separation logic---including the \emph{frame rule}. 
    
    \item Our calculus is \emph{sound} with respect to an operational semantics based on Markov decision processes.
            While this has been shown before for simple probabilistic languages (cf.~\cite{DBLP:journals/pe/GretzKM14}), heap-manipulating statements introduce new technical challenges.
          In particular, allocating fresh memory yields \emph{countably infinite nondeterminism}, which breaks continuity and rules out 
          standard constructions for loops. 

    \item We apply our calculus to analyze all aforementioned examples.
\end{itemize}

\paragraph{Outline.} 

		In Section~\ref{sec:hpgcl}, we 
        present a probabilistic programming language with pointers 
        together with an operational semantics.
		Section~\ref{sec:qsl} introduces \QSL as an assertion language.
		In Section~\ref{sec:wp}, we develop a $\wpsymbol$-style calculus
        for the quantitative verification of (probabilistic) programs with \QSL.
		Furthermore, we prove soundness of our calculus and develop a \emph{frame rule} for \QSL.
        Section~\ref{sec:wp:landscape} discusses alternative design choices for $\wpsymbol$-style calculi and Section~\ref{sec:extensions} briefly addresses how recursive procedures are incorporated.
		In Section~\ref{sec:case-studies}, we apply \QSL to four case studies, including the three introductory examples. 
        Finally, we discuss related work in Section~\ref{sec:related-work} and conclude in Section~\ref{sec:conclusion}.

        \REPORT{\emph{Detailed proofs of all theorems are found in the appendix for the reader's convenience.}}%
        \CAMERA{\emph{Detailed proofs of all theorems are found in a separate technical report~\cite{DBLP:journals/corr/abs-1802-10467}.}}%

%
%
%
%
%
%

\section{Probabilistic Pointer Programs}
\label{sec:hpgcl}

We use 
a simple, imperative language \`{a} la Dijkstra's guarded command language with two distinguished features:
First, we endow our programs with a probabilistic choice instruction.
Second, we allow for statements that allocate, mutate, access, and dispose memory.

\subsection{Syntax}
The set of programs in \emph{heap-manipulating probabilistic guarded command language}, denoted $\hpgcl$, is given by the grammar
\begin{align*}
    \begin{aligned}
    \cc  ~~\longrightarrow~~ &\SKIP & \text{(effectless program)} \\
    & |~~ \ASSIGN{x}{\ee} & \text{(assignment)} \\
    & |~~ \COMPOSE{\cc}{\cc} & \text{(seq. composition)} \\
    & |~~ \ITE{\guard}{\cc}{\cc} & \text{(conditional choice)} \\
    & |~~ \WHILEDO{\guard}{\cc} & \text{(loop)} \\
    \end{aligned}
    & \quad
    \begin{aligned}
            & \quad |~~ \PCHOICE{\cc}{\pp}{\cc} & \text{(prob. choice)}\\
            & \quad |~~ \ALLOC{x}{\ee_1,\, \ldots,\, \ee_n} & \text{(allocation)} \\
            & \quad |~~ \HASSIGN{\ee}{\ee'} & \text{(mutation)} \\
            & \quad |~~ \ASSIGNH{x}{\ee} & \text{(lookup)} \\
            & \quad |~~ \FREE{\ee}, & \text{(deallocation)}
    \end{aligned}
\end{align*}
%
%
where $x$ is a variable in the set $\Vars$, $\ee, \ee',  \ee_1, \ldots, \ee_n$ are arithmetic expressions, 
$\guard$ is a predicate, i.e. an expression over variables evaluating to either true or false, 
and $\pp \in [0,\, 1] \cap \Rats$ is a probability.

\subsection{Program states}
A \emph{program state} $(s,\, h)$ consists of a \emph{stack} $s$, i.e.\ a valuation of variables by integers, and a \emph{heap} $h$ modeling dynamically allocated memory.
Formally, the set of \emph{stacks} is given by
$
	{\Stacks} = \setcomp{\sk}{\sk \colon \Vars \To \Ints}
$. 
Like in a standard RAM model, a heap consists of memory addresses that each store a value and is thus a \emph{finite} mapping from addresses (i.e.\ natural numbers) to values (which may themselves be allocated addresses in the heap).
Formally, the set of \emph{heaps} is given by
\begin{align*}
	{\Heaps} \eeq \setcomp{\hh}{\hh \colon N \To \Ints,~ N \subseteq \Nats_{>0},~ |N| < \infty}.
\end{align*}
%
The $0$ is excluded as a valid address in order to model e.g.\ null-pointer terminated lists.
The set of \emph{program states} is given by
$
	{\States} = \setcomp{(\sk, \hh)}{\sk \in \Stacks,~ \hh \in \Heaps}
$. 
Notice that expressions $\ee$ and guards $\guard$ 
may depend on variables only (i.e.\ they may \emph{not} depend upon the heap) and thus their evaluation never causes any side effects.
Side effects such as dereferencing unallocated memory can only occur \emph{after} evaluating an expression and trying to access the memory at \mbox{the evaluated address}. 

Given a program state $(\sk, \hh)$, we denote by \textbf{$\sk(\ee)$} the evaluation of expression $\ee$ in $\sk$, i.e.\ the value that is obtained by evaluating $\ee$ after replacing any occurrence of any variable $x$ in $\ee$ by the value $\sk(x)$.
\CHANGED{By slight abuse of notation, we also denote the evaluation of a Boolean expression $\guard$ by $\sk(\guard)$.}
Furthermore, we write $\sk\subst{x}{v}$ to indicate that we set variable $x$ to value $v \in \Ints$ in stack $\sk$, i.e.\footnote{We use $\lambda$-expressions to denote functions: Function $\lambda X \mydot f$ applied to an argument $\alpha$ evaluates to $f$ in which every occurrence of $X$ is replaced by $\alpha$.}	
\begin{align*}
	\sk\statesubst{x}{v} \eeq \lambda\, y\mydot \begin{cases}
		v, & \textnormal{if } y = x\\
		\sk(y), & \textnormal{if } y \neq x.
	\end{cases}
\end{align*}
For heap $\hh$, $\hh\Update{u}{v}$ is defined analogously. 
For a given heap $\hh \colon N \To \Ints$, we denote by $\dom{\hh}$ its \emph{domain} $N$. 
Furthermore, we write $\HeapSet{u \mapsto v_1,\ldots,v_n}$ as a shorthand for the heap \mbox{$\hh$ given by}
\begin{align*}
\dom{\hh} \eeq \{ u, u+1,\ldots,u+n-1\}, \quad \forall k \in \{0,\ldots,n-1\} \colon \hh(u+k) \eeq v_{k+1}.
\end{align*}
%
%
Two heaps $\hh_1$, $\hh_2$ are \emph{disjoint}, denoted $\hh_1 \disjoint \hh_2$, if their domains do not overlap, i.e.\  $\dom{\hh_1} \cap \dom{\hh_2} = \emptyset$.
%
%
The \emph{disjoint union} of two disjoint heaps $\hh_1 \colon N_1 \To \Ints$ and $\hh_2 \colon N_2 \To \Ints$ \mbox{is given by}
\begin{align*}
    \hh_1 \sepcon \hh_2\colon \dom{\hh_1} \mathrel{\dot{\cup}} \dom{\hh_2} \To \Ints, \quad 
	\bigl(\hh_1 \sepcon \hh_2 \bigr) (n) \eeq \begin{cases}\hh_1(n), & \textnormal{if } n \in \dom{\hh_1} \\ \hh_2(n), & \textnormal{if } n \in \dom{\hh_2}. \end{cases}
\end{align*}
We denote by $\emptyheap$ the \emph{empty heap} with $\dom{\emptyheap} = \emptyset$. 
%
%
%
Note that $\hh \sepcon \emptyheap = \emptyheap \sepcon \hh = \hh$ for any heap $\hh$.
%
We define \emph{heap inclusion} as 
%
	$\hh_1 \subseteq \hh_2$ iff $\exists\, \hh_1' \disjoint \hh_1 \colon~ \hh_1 \sepcon \hh_1' = \hh_2$.
%
%
Finally, we use the \emph{Iverson bracket}~\cite{knuth1992two} notation $\iverson{\preda}$
to associate with predicate $\preda$ its indicator function.
Formally, 
\begin{align*}
	\iverson{\preda} \colon\quad \States \To \{0,\, 1\},\quad \iverson{\preda}(\sk, \hh) \eeq \begin{cases}
		1, & \textnormal{if $(\sk,\, \hh) \models \preda$}\\
		0, & \textnormal{if $(\sk,\, \hh) \not\models \preda$},
	\end{cases}
\end{align*}
where $(\sk,\hh) \models \preda$ denotes that $\preda$ evaluates to true in $(\sk,\hh)$.
Notice that while predicates may generally speak about stack-heap pairs, guards in \hpgcl-programs may only refer to the stack.
\subsection{Semantics}
We assign meaning to \hpgcl-statements in terms of a small-step operational semantics, i.e. an execution relation $\ExecSymbol$ between \emph{program configurations}, which consist of a program state and either a program that is still to be executed, a symbol $\Term$ indicating successful termination, or a symbol $\Fault$ indicating a memory fault.
Formally, the set of program configurations is given by
\begin{align*}
        \OpStates \eeq \left(\hpgcl \cup \{\,\Term,\,\Fault\,\}\right) \,\times\, \States~.
\end{align*}
Since our programming language admits memory allocation and probabilistic choice,
our semantics has to account for both \emph{nondeterminism} (due to the fact that memory is allocated at nondeterministically chosen addresses) 
and \emph{execution probabilities}.
Our execution relation is hence \mbox{of the form}
\begin{align*}
        \ExecSymbol ~\subseteq~ \OpStates \,\times\, \Nats \,\times\, ([0,1] \cap \Rats) \,\times\, \OpStates~,
\end{align*}
where the second component is an \emph{action} labeling the nondeterministic choice taken in the execution step
and the third component is the execution step's probability.%
\footnote{For simplicity, we tacitly distinguish between the probabilities $0.5$ and $1-0.5$ to deal with the corner case of two identical executions between the same configurations.} %
We usually write $\Exec{\cc}{\sk}{\hh}{n}{\pp}{\cc'}{\sk'}{\hh'}$ instead of $((\cc,(\sk,\hh)),n,\pp,(\cc',(\sk',\hh'))) \in\,\ExecSymbol$.
The operational semantics of \hpgcl-programs, i.e.\ the execution relation \ExecSymbol, is determined by the rules in Figure~\ref{table:op}.
\begin{figure}[t]
\begin{align*}
  &
  \infer{
    \Exec{\SKIP}{\sk}{\hh}{0}{1}{\Term}{\sk}{\hh}
  }{
  }
  \quad
  \infer{
    \Exec{\ASSIGN{x}{\ee}}{\sk}{\hh}{0}{1}{\Term}{\sk\subst{x}{v}}{\hh}
  }{
    \sk(\ee) = v
  }
  \\[1ex]
  &
  \infer{
    \Exec{\COMPOSE{\cc_1}{\cc_2}}{\sk}{\hh}{a}{\pp}{\Fault}{\sk}{\hh}
  }{
    \Exec{\cc_1}{\sk}{\hh}{a}{p}{\Fault}{\sk}{\hh}
  }
  \quad
  \infer{
    \Exec{\COMPOSE{\cc_1}{\cc_2}}{\sk}{\hh}{a}{\pp}{\cc_2}{\sk'}{\hh'}
  }{
    \Exec{\cc_1}{\sk}{\hh}{a}{p}{\Term}{\sk'}{\hh'}
  }
  \quad
  \infer{
    \Exec{\COMPOSE{\cc_1}{\cc_2}}{\sk}{\hh}{a}{\pp}{\COMPOSE{\cc_1'}{\cc_2}}{\sk'}{\hh'}
  }{
    \Exec{\cc_1}{\sk}{\hh}{a}{\pp}{\cc_1'}{\sk'}{\hh'}
  }
  \\[1ex]
  &
  \infer{
    \Exec{\ITE{\guard}{\cc_1}{\cc_2}}{\sk}{\hh}{0}{1}{\cc_1}{\sk}{\hh}
  }{
    \sk(\guard) = \true
  }
  \quad
  \infer{
    \Exec{\ITE{\guard}{\cc_1}{\cc_2}}{\sk}{\hh}{0}{1}{\cc_2}{\sk}{\hh}
  }{
    \sk(\guard) = \false
  }
  \\[1ex]
  &
  \infer{
    \Exec{\WHILEDO{\guard}{\cc}}{\sk}{\hh}{0}{1}{\Term}{\sk}{\hh}
  }{
    \sk(\guard) = \false
  }
  \quad
  \infer{
  \Exec{\WHILEDO{\guard}{\cc}}{\sk}{\hh}{0}{1}{\COMPOSE{\cc}{\WHILEDO{\guard}{\cc}}}{\sk}{\hh}
  }{
    \sk(\guard) = \true
  }
  \\[1ex]
  &
  \infer{
    \Exec{\PCHOICE{\cc_1}{\pp}{\cc_2}}{\sk}{\hh}{0}{\pp}{\cc_1}{\sk}{\hh}
    }{
    }
  \quad
  \infer{
    \Exec{\PCHOICE{\cc_1}{\pp}{\cc_2}}{\sk}{\hh}{0}{1-\pp}{\cc_2}{\sk}{\hh}
    }{
    }
  \\[1ex]
  &
  \infer{
          \Exec{\ALLOC{x}{\ee_1,\ldots,\ee_n}}{\sk}{\hh}{u}{1}{\Term}{\sk\Update{x}{u}}{\hh \sepcon \HeapSet{u \mapsto v_1, \ldots, v_n}}
  }{
  u,u+1,\ldots,u+n-1 \in \Nats_{> 0} \setminus \dom{\hh} \quad \sk(\ee_1) = v_1, \ldots, \sk(\ee_n) = v_n
  }
  \\[1ex]
  &
  \infer{
          \Exec{\HASSIGN{\ee}{\ee'}}{\sk}{\hh}{0}{1}{\Term}{\sk}{\hh\Update{u}{v}}
  }{
  \sk(\ee) = u \in \dom{\hh} \quad \sk(\ee') = v
  }
  \quad
  \infer{
    \Exec{\HASSIGN{\ee}{\ee'}}{\sk}{\hh}{0}{1}{\Fault}{\sk}{\hh}
  }{
    \sk(\ee) \notin \dom{\hh}
  }
  \\[1ex]
  &
  \infer{
    \Exec{\ASSIGNH{x}{\ee}}{\sk}{\hh}{0}{1}{\Term}{\sk\Update{x}{v}}{\hh}
  }{
    \sk(\ee) = u \in \dom{\hh} \quad \hh(u) = v
  }
  \quad
  \infer{
    \Exec{\ASSIGNH{x}{\ee}}{\sk}{\hh}{0}{1}{\Fault}{\sk}{\hh}
  }{
    \sk(\ee) \notin \dom{\hh}
  }
  \\[1ex]
  &
  \infer{
          \Exec{\FREE{x}}{\sk}{\hh \sepcon \HeapSet{u \mapsto v}}{0}{1}{\Term}{\sk}{\hh}
  }{
    \sk(x) = u
  }
  \quad
  \infer{
    \Exec{\FREE{x}}{\sk}{\hh}{0}{1}{\Fault}{\sk}{\hh}
  }{
    \sk(x) \notin \dom{\hh}
  }
\end{align*}
\caption{Inference rules determining the execution relation $\ExecSymbol$.}
\label{table:op}
\end{figure}

Let us briefly go over those rules.
The rules for $\SKIP$, assignments, conditionals, and loops are standard.
In each case, the execution proceeds deterministically, hence all actions are labeled $0$ and the execution probability is $1$.
For a probabilistic choice $\PCHOICE{\cc_1}{\pp}{\cc_2}$ there are two possible executions: With probability $\pp$ we execute $\cc_1$ and with probability $1-\pp$, we execute $\cc_2$.

The remaining statements access or manipulate memory.
$\ALLOC{x}{\ee_1,\ldots,\ee_n}$ allocates a block of $n$ memory addresses and stores the first allocated address in variable $x$.
Since allocated addresses are chosen nondeterministically by the memory allocator, there are countably infinitely many possible executions, which are each labeled by an action corresponding to the first allocated address.
Under the assumption that an infinite amount of memory is available, \emph{memory allocation cannot fail}.
$\HASSIGN{\ee}{\ee'}$ attempts to write the value of $\ee'$ to address $\ee$. If address $\ee$ has not been allocated before, we encounter a memory fault, i.e. move to a configuration marked by $\Fault$.
Conversely, $\ASSIGNH{x}{\ee}$ assigns the value at address $\ee$ to variable $x$. Again, failing to find address $\ee$ on the heap leads to an error.
Finally, $\FREE{\ee}$ disposes the memory cell at address $\ee$ if it is present and fails otherwise.

Notice that no statement other than memory allocation introduces nondeterminism, i.e.\ entails an action label different from $0$
Moreover, for every action $n \in \Nats$, we have
\begin{align*}
  \sum_{\Exec{\cc}{\sk}{\hh}{n}{\pp}{\cc'}{\sk'}{\hh'}} \pp ~\in~ \{0,1\},
\end{align*}
where we set $\sum_{\emptyset} = 0$.
Our execution relation thus describes a Markov Decision Process, which is an established model for probabilistic systems (cf.~\cite{DBLP:books/daglib/0020348, puterman2005markov}).




\section{Quantitative Separation Logic}\label{sec:qsl}

The term \emph{separation logic} refers to both a logical assertion language as well as a Floyd-Hoare-style proof system for reasoning about pointer programs (cf.~\cite{DBLP:conf/popl/IshtiaqO01,DBLP:conf/lics/Reynolds02}).
In this section, we develop \QSL in the sense of an assertion language.
A 
proof system for reasoning about $\hpgcl$ programs is introduced in Section~\ref{sec:wp}.
The rationale of \QSL is to combine concepts from two worlds:
\begin{enumerate}
	\item From separation logic (\SL): \emph{separating conjunction}~($\sepcon$) and \emph{separating implication} ($\sepimp$).
	
	\item From probabilistic program verification: \emph{expectations}.
\end{enumerate}
Separating conjunction and implication
are the two distinguished logical connectives featured in \SL~\cite{DBLP:conf/popl/IshtiaqO01,DBLP:conf/lics/Reynolds02}.
Expectations~\cite{DBLP:series/mcs/McIverM05} on the other hand take over the role of logical formulae when doing \emph{quantitative reasoning} about probabilistic programs.
In what follows, we gradually develop both a \emph{quantitative separating conjunction} and a \emph{quantitative separating implication} which each connect expectations instead of formulae (as in the classical setting).

\subsection{Expectations}
\label{sec:expectations}

Floyd-Hoare logic~\cite{DBLP:journals/cacm/Hoare69} as well as Dijkstra's weakest preconditions~\cite{DBLP:books/ph/Dijkstra76} employ first-order logic for reasoning about the correctness of programs.
For probabilistic programs, Kozen in his PPDL~\cite{DBLP:conf/stoc/Kozen83} was the first to generalize from predicates to measurable functions (or random variables). 
Later,~\cite{DBLP:series/mcs/McIverM05} coined the term \emph{expectation} for such functions. 
Here, we define the set $\E$ of expectations and the set $\Eone$ of one-bounded expectations as
\begin{align*}
	\E \eeq \setcomp{\ff}{\ff\colon \States \To \PosRealsInf} \qand
	\Eone \eeq \setcomp{\fg}{\fg\colon \States \To [0,\, 1]}.
\end{align*}
An expectation $\ff$ maps every program state to a non-negative real number or $\infty$. 
$\Eone$ allows for reasoning about probabilities of events whereas $\E$ allows for reasoning about expected values of more general random variables such as the expected value of a variable $x$, the expected height of a tree (in the heap), etc. 
Notice that a predicate is a particular expectation, namely its Iverson bracket, that maps only to $\{0,\ 1\}$.
In contrast to~\cite{DBLP:series/mcs/McIverM05}, our expectations are not necessarily bounded. Hence, $(\E,\, {\preceq})$ and $(\Eone,\,{\preceq})$, where 
%
	$\ff \preceq \fg$ iff  $\forall (\sk,\, \hh) \in \States\colon~ \ff(\sk,\, \hh) \leq \fg(\sk,\, \hh)$
%
each form a complete lattice with least element $0$ and greatest element $\infty$ \mbox{and 1, respectively}.\footnote{By slight abuse of notation, for any constant $k \in \PosRealsInf$, we write $k$ for $\lambda (\sk,\, \hh) \mydot k$.
}
%
\CHANGED{%
    We present most of our results with respect to the domain $(\E,\,{\preceq})$, i.e. we develop a logic for reasoning about expected values.
    A logic for reasoning about probabilities of events can be constructed analogously by using the complete lattice $(\Eone,{\preceq})$ instead.
}%

Analogously to~\cite{DBLP:conf/lics/Reynolds02}, we call an expectation $\ff \in \E$ \emph{domain-exact} iff for all stacks $\sk \in \Stacks$ and heaps $\hh,\hh' \in \Heaps$,
$\ff(\sk,\, \hh) > 0$ and $\ff(\sk,\, \hh') > 0$ together implies that $\dom{\hh} = \dom{\hh'}$,
%
i.e.\ for a fixed stack, the domain of all heaps such that the quantity $\ff$ does not \mbox{vanish is constant}.

We next lift the atomic formulas of \SL to a quantitative setting:
The \emph{empty-heap predicate} $\emp$, which evaluates to $1$ iff the heap is empty, is defined as
\begin{align*}
	\emp \eeq \lambda (\sk,\hh)\mydot 
	\begin{cases}
		1, &\textnormal{if } \dom{\hh} = \emptyset, \\
		0, &\textnormal{otherwise.}
	\end{cases}
\end{align*}
%
The \emph{points-to predicate} $\singleton{\ee}{\ee'}$, evaluating to $1$ iff the heap consists of exactly one cell with address $\ee$ and content $\ee'$, \mbox{is defined as}
\begin{align*}
	\singleton{\ee}{\ee'} \eeq \lambda (\sk,\hh)\mydot 
	\begin{cases}
		1, &\textnormal{if } \dom{\hh} = \{\sk(\ee)\} \text{ and } \hh(\sk(\ee)) = \sk(\ee') \\
		0, &\textnormal{otherwise.}
	\end{cases}
\end{align*}
%
Notice that if $\sk(\ee) \not\in \Nats_{>0}$ then automatically $\dom{\hh} \neq \{\sk(\ee)\}$.
As a shorthand, we denote by $\singleton{\ee}{ \ee_1',\ldots, \ee_n'}$
the predicate 
that evaluates to $1$ on $(\sk,\hh)$ iff the heap $\hh$ contains exactly $n$ cells with addresses $\sk(\ee),\, \ldots,\, \ss(\ee) + n-1$ and \mbox{respective contents $\sk(\ee_1'),\, \ldots,\, \sk(\ee_n')$}.
		
The \emph{allocated pointer predicate} $\validpointer{\ee}$, which evaluates to $1$ iff the heap consists of a single cell with address $\ee$ (but arbitrary content), \mbox{is defined as}
\begin{align*}
	\validpointer{\ee} \eeq \lambda (\sk,\hh)\mydot 
	\begin{cases}
		1, &\textnormal{if } \dom{\hh} =  \{\sk(\ee) \}, \\
		0, &\textnormal{otherwise.}
	\end{cases}
\end{align*}
All of the above predicates are domain-exact expectations evaluating to either zero or one.

As an example of a truly \emph{quantitative} expectation 
consider the \emph{heap size quantity}
\begin{align*}
    \heapSize \eeq \lambda (\sk,\hh)\mydot |\dom{\hh}|,
\end{align*}
where $|\dom{\hh}|$ denotes the cardinality of $\dom{\hh}$, which measures the number of allocated cells in a heap $h$.
In contrast to the standard \SL predicates, $\heapSize$ is neither domain-exact nor a predicate.

\subsection{Separating Connectives between Expectations}

We now develop quantitative versions of \SL's connectives.
\emph{Standard conjunction} ($\wedge$) is modeled by pointwise multiplication. 
This is backward compatible as for any two predicates $\preda$ and $\predb$ we have
%
	$\iverson{\preda \wedge \predb} = \iverson{\preda} \cdot \iverson{\predb} = \lambda (\sk,\, \hh) \mydot\, \iverson{\preda}(\sk,\, \hh) \cdot \iverson{\predb}(\sk,\, \hh)$.
%
Towards a \emph{quantitative separating conjunction}, 
let us first examine the classical case, 
which is defined for two predicates $\preda$ and $\predb$ as
\begin{align*}
	(\sk, \hh) \mmodels \preda \sepcon \predb \qqiff
	\exists \, \hh_1, \hh_2 \colon\quad \hh \eeq \hh_1 \sepcon \hh_2 ~\text{ and }~ (\sk, \hh_1) \mmodels \preda ~\text{ and }~ (\sk, \hh_2) \mmodels \predb.
\end{align*}
In words, a state $(\sk, \hh)$ satisfies $\preda \sepcon \predb$ iff there exists a \emph{partition} of the heap $\hh$ into two heaps $\hh_1$ and $\hh_2$ such that 
the stack $\sk$ together with heap $\hh_1$ satisfies $\preda$, and $\sk$ together with $\hh_2$ satisfies $\predb$.

How should we connect two expectations $\ff$ and $\fg$ in a similar fashion? 
As logical ``and'' corresponds to a multiplication, we need to find a partition of the heap $\hh$ into $\hh_1 \sepcon \hh_2$, measure $\ff$ in $\hh_1$, measure $\fg$ in $\hh_2$, and finally multiply these two measured quantities.
The naive approach,
\begin{align*}
	& \bigl( \ff \sepcon \fg \bigr)(\sk,\, \hh) \qeq \exists \, \hh_1, \hh_2\colon~ \iverson{\hh = \hh_1 \sepcon \hh_2} \cdot \ff(\sk, \hh_1) \cdot \fg(\sk, \hh_2),
\end{align*}
is not meaningful.
At the very least, it is ill-typed.
Moreover, what precisely determined quantity would the above express?
After all, the existentially quantified partition of $\hh$ need not be unique.

Our key redemptive insight here is that $\exists$ should correspond to $\max$.
\CHANGED{%
        From an algebraic perspective, this corresponds to the usual interpretation of existential quantifiers in a complete Heyting algebra or Boolean algebra as a disjunction (cf.~\cite{scott2008algebraic} for an overview), which we will interpret as a maximum in the realm of expectations.
}%
In first-order logic, the effect of the quantified predicate $\exists v\colon \preda(v)$ is so-to-speak to ``maximize the truth of $\preda(v)$'' by a suitable choice of $v$.
In \QSL, instead of truth, we maximize a quantity:
Out of all partitions $\hh = \hh_1 \sepcon \hh_2$, we choose the one---out of finitely many for any given $\hh$---that maximizes the product $\ff(\sk, \hh_1) \cdot \fg(\sk, \hh_2)$.
We thus define the quantitative $\sepcon$ as follows:
\begin{definition}[Quantitative Separating Conjunction]	
    The \emph{quantitative separating conjunction} $\ff \sepcon \fg$ of two expectations $\ff,\fg \in \E$ is defined as 
	%
	\begin{align*}
		\ff \sepcon \fg \eeq \lambda (\sk, \hh)\mydot \max_{\hh_1, \hh_2} \setcomp{\vphantom{\big(}\ff(\sk, \hh_1) \cdot \fg(\sk, \hh_2)}{\hh = \hh_1 \sepcon \hh_2}. \tag*{$\triangle$}
	\end{align*}
\end{definition}
\noindent%
As a first sanity check, notice that this definition is backward compatible to the qualitative setting: For predicates $\preda$ and $\predb$, we have $\left(\iverson{\preda} \sepcon \iverson{\predb}\right)(\sk,\hh) \in \{0,1\}$ and moreover
$
   \left( \iverson{\preda} \sepcon \iverson{\predb} \right)(\sk,\hh) = 1
$ 
holds in \QSL if and only if 
$
   (\sk,\hh) \mmodels \preda \sepcon \predb
$ 
holds in \SL\REPORT{ (a proof is found in Appendix~\ref{app:back-comp-sep-con}, p.~\pageref{app:back-comp-sep-con})}. 
%

%
Next, we turn to \emph{separating implication}.
For \SL, this is defined for predicates $\preda$ and $\predb$ as
\begin{align*}
	(\sk, \hh) \mmodels \preda \sepimp \predb \qqiff
    \forall \, \hh'\colon \quad \hh' \disjoint \hh ~\text{ and }~ (\sk, \hh') \models \preda ~~\text{ implies }~~ (\sk, \hh \sepcon \hh') \models \predb.
\end{align*}
So $(\sk, \hh)$ satisfies $\preda \sepimp \predb$ iff the following holds:
Whenever we can find a heap $\hh'$ disjoint from $\hh$ such that stack $\sk$ together with heap $\hh'$ satisfies $\preda$, 
then $\sk$ together with the \emph{conjoined} heap $\hh \sepcon \hh'$ must satisfy $\predb$.
In other words: 
We measure the truth of $\predb$ in \emph{extended} heaps $\hh \sepcon \hh'$, where all admissible extensions $\hh'$ \mbox{must satisfy $\preda$}.

How should we connect expectations $\fg$ and $\ff$ in a similar fashion?
Intuitively, $\fg \sepimp \ff$ intends to measure $\ff$ in extended heaps, subject to the fact that the extensions satisfy $\fg$.
\REMOVED{
\emph{Satisfying an expectation $\fg$}, however, is not a meaningful notion in general, since $\fg$ could represent any quantity.
We thus restrict $\fg$ to predicates.
}
\CHANGED{
    Since the least element of our complete lattice, i.e. $0$, corresponds to $\false$ when evaluating a predicate, we interpret \emph{satisfying an expectation $\fg$} as measuring some positive quantity, i.e. $\fg(\sk,\hh) > 0$.
}

As for the universal quantifier, our key insight is now that---dually to $\exists$ corresponding to $\max$---$\forall$ should correspond to $\min$:
Whereas in first-order logic the predicate $\forall v\colon \preda(v)$ ``minimizes the truth of $\preda(v)$'' by requiring that $\preda(v)$ must be true for all choices of $v$, in \QSL we minimize a quantity:
Out of all heap extensions $\hh'$ disjoint from $\hh$ that satisfy a given \REMOVED{predicate $\preda$}\CHANGED{expectation $\fg$}, we choose an extension that minimizes the quantity $\ff(\sk, \hh \sepcon \hh')$.
Intuitively speaking, we pick the smallest possible\footnote{In terms of measuring $\ff(\sk, \hh \sepcon \hh')$.} extension $\hh'$ that barely satisfies \REMOVED{$\preda$}\CHANGED{$\fg$}.
Since for given \REMOVED{$\preda$}\CHANGED{$\fg$} and $\hh$, there may be infinitely many (or no) admissible choices for $\hh'$, we define the quantitative $\sepimp$ by an infimum: 
\CHANGED{
        \begin{align*}
		    \fg \sepimp \ff
            \eeq
            \lambda (\sk, \hh)\mydot \inf_{\hh'}~ \setcomp{\frac{\ff(\sk, \hh \sepcon \hh')}{\fg(\sk,\hh')} }{\hh' \disjoint \hh ~\textnormal{ and }~ \fg(\sk, \hh') > 0}~.
        \end{align*}
        This definition is well-behaved with $(\Eone,{\preceq})$ as the underlying lattice.\footnote{In particular, quantitative separating implication and quantitative separating conjunction are adjoint.}
        However, for the domain $(\E,{\preceq})$ the above definition of $\sepimp$ is not well-defined if $\fg(\sk,\hh') = \infty$ holds.
        We thus restrict $\fg$ to predicates. The above definition then simplifies as follows:
}
\begin{definition}[Quantitative Separating Implication]	
    The \emph{quantitative separating implication} $\iverson{\preda} \sepimp \ff$ of predicate $\preda$ and expectation $\ff \in \E$ is defined as
	\begin{align*}
		\iverson{\preda} \sepimp \ff
        \eeq
        \lambda (\sk, \hh)\mydot \inf_{\hh'}~ \setcomp{\ff(\sk, \hh \sepcon \hh')}{\hh' \disjoint \hh ~\textnormal{ and }~ (\sk, \hh') \models \preda}. \tag*{$\triangle$}
	\end{align*}
    %
	%
\end{definition}
\noindent%
Unfortunately, backward compatibility for quantitative separating implication comes with certain reservations:
Suppose for a particular state $(\sk,\, \hh)$ there exists no heap extension $\hh'$ such that $(\sk,\, \hh') \models \preda$.
Then $\setcomp{\ff(\sk, \hh \sepcon \hh')}{\hh' \disjoint \hh \textnormal{ and } (\sk, \hh') \models \preda}$ is empty, and the greatest lower bound (within our domain $\PosRealsInf$) of the empty set is $\infty$ and not 1.
In particular, 
$
	\false \sepimp \predb \equiv \true
$ 
holds in $\SL$, but
$
	0 \sepimp \iverson{\predb} = \infty
$ 
holds in \QSL. 
Since $0 = \iverson{\false}$ but $\infty \neq \iverson{\true}$, backward compatibility of quantitative separating implication breaks here.
As a silver lining, however, we notice that $\true$ is the greatest element in the complete lattice of predicates and correspondingly $\infty$ is the greatest element in $\E$. 
In this light, the above appears not at all surprising.
\REMOVED{In fact, if we restrict the codomain of expectations to $[0,1]$ instead of $\PosRealsInf$ (which is sufficient to reason about probabilities of events; see also $\Eone$ in Section~\ref{sec:expectations}) we achieve full backward compatibility (cf.\ Appendix~\ref{app:back-comp-sep-imp}, p.~\pageref{app:back-comp-sep-imp}).}
\CHANGED{
In fact, if we restrict ourselves to the domain $(\Eone,{\preceq})$ to reason about probabilities, we achieve full backward compatibility.
To be precise, let us explicitly embed classical separation logic (\SL) into \QSL.
\begin{definition}[Embedding of \SL into \QSL]\label{def:embedding-sl-qsl}
Formulas in classical separation logic (\SL) are embedded into quantitative separation logic by a function
$\qslemb{.}\colon \SL \to \Eone$ mapping formulas in \SL to expectations in $\Eone$. 
This function is defined inductively as follows:
\begin{align*}
        & \qslemb{\preda} \eeq \iverson{\preda}~\text{for any atomic formula $\varphi \in \SL$}
        \qquad \qslemb{\neg \preda} \eeq 1 - \qslemb{\preda} \\
        & \qslemb{\preda_1 \sepcon \preda_2} \eeq \qslemb{\preda_1} \sepcon \qslemb{\preda_2} 
        \qquad \qslemb{\preda_1 \sepimp \preda_2} \eeq \qslemb{\preda_1} \sepimp \qslemb{\preda_2} \\
        & \qslemb{\exists x\colon \preda} \eeq \sup_{v \in \Ints}~\qslemb{\preda}\subst{x}{v} 
        \qquad \qslemb{\preda_1 \wedge \preda_2} \eeq \qslemb{\preda_1} \cdot \qslemb{\preda_2}
        \tag*{$\triangle$}
\end{align*}
\end{definition}
Every atomic separation logic formula is thus interpreted as its Iverson bracket in \QSL.
Furthermore, every connective is replaced by its quantitative variant. 
We then obtain that \QSL---as an assertion language---is a conservative extension of classical separation logic.
\begin{theorem}[Conservativity of \QSL as an assertion language]\label{thm:qsl:conservativity:language}
    For all classical separation logic formulas $\preda \in \SL$ and all states $(\sk,\hh) \in \States$, we have
    \begin{enumerate}
        \item $\qslemb{\preda}(\sk,\hh) \in \{0,1\}$, and 
        \label{thm:qsl:conservativity:language:0-1}
        \item $(\sk,\hh) \models \varphi$~~if and only if~~$\qslemb{\preda}(\sk,\hh) \eeq 1$.
        \label{thm:qsl:conservativity:language:equivalence}
    \end{enumerate}
\end{theorem}
\REPORT{
\begin{proof}
    See Appendix~\ref{app:qsl:conservativity:language}, p.~\pageref{app:qsl:conservativity:language}.
\end{proof}
}
The same result is achieved for the expectation domain $(\E,{\preceq})$ if we define the embedding of separating implication as $\qslemb{\preda_1 \sepimp \preda_2} = \min\{1,\qslemb{\preda_1} \sepimp \qslemb{\preda_2}\}$.
}

\subsection{Properties of Quantitative Separating Connectives}

Besides backward compatibility, the separating connectives of \QSL are well-behaved in the sense that they satisfy most properties of their counterparts in \SL.
To justify this claim, we now present a collection of quantitative analogs of properties of classical separating conjunction and implication.
Most of those properties originate from the seminal papers on classical separation logic~\cite{DBLP:conf/popl/IshtiaqO01,DBLP:conf/lics/Reynolds02}.
We start with algebraic laws for quantitative separating conjunction:
%
%
\begin{theorem}
\label{thm:sep-con-monoid}
	$(\E,\, {\sepcon},\, \emp)$ is a commutative monoid, i.e.\ for all $\ff,\fg,\fh \in \E$ the following holds:
	\begin{enumerate}
		\item 
		\label{thm:sep-con-monoid:ass}
			\emph{Associativity:} \quad $\ff \sepcon (\fg \sepcon \fh) \eeq (\ff \sepcon \fg) \sepcon \fh$
		
		\item 
		\label{thm:sep-con-monoid:neut}
			\emph{Neutrality of ${\emp}$:} \quad $\ff \sepcon \emp \eeq \emp \sepcon \ff \eeq \ff$
			
		\item 
		\label{thm:sep-con-monoid:comm}
			\emph{Commutativity:} \quad $\ff \sepcon \fg \eeq \fg \sepcon \ff$
	\end{enumerate}
\end{theorem}%
\REPORT{%
\begin{proof} %
    See Appendix~\ref{app:sep-con-monoid}, p.~\pageref{app:sep-con-monoid}.%
\end{proof}%
}%
\begin{theorem}[(Sub)distributivity Laws]
\label{thm:sep-con-distrib}
	Let $\ff,\fg,\fh \in \E$ and let $\preda$ be a predicate. Then:
	%
	\begin{enumerate}
		\item 
		\label{thm:sep-con-distrib:sepcon-over-max}
			$\ff \sepcon \Max{\fg}{\fh} \eeq \Max{\ff \sepcon \fg}{\ff \sepcon \fh}$

		\item 
		\label{thm:sep-con-distrib:sepcon-over-plus}
			$\ff \sepcon (\fg + \fh) \ppreceq \ff \sepcon \fg + \ff \sepcon \fh$
		
		\item 
		\label{thm:sep-con-distrib:sepcon-over-times}
			$\iverson{\preda} \sepcon (\fg \cdot \fh) \ppreceq \big(\iverson{\preda} \sepcon \fg \big) \cdot \big(\iverson{\preda} \sepcon \fh \big)$

	\end{enumerate}
	Furthermore, if $\ff$ and $\iverson{\preda}$ are domain-exact, we obtain full distributivity laws: 
	\begin{enumerate}
	\setcounter{enumi}{3}
		\item 
		\label{thm:sep-con-distrib:sepcon-over-plus-full}
			$\ff \sepcon (\fg + \fh) \eeq \ff \sepcon \fg + \ff \sepcon \fh$
		
		\item 
		\label{thm:sep-con-distrib:sepcon-over-times-full}
			$\iverson{\preda} \sepcon (\fg \cdot \fh) \eeq \big(\iverson{\preda} \sepcon \fg \big) \cdot \big(\iverson{\preda} \sepcon \fh \big)$
	\end{enumerate}
\end{theorem}
\REPORT{%
\begin{proof}%
    See Appendix~\ref{app:sep-con-distrib}, p.~\pageref{app:sep-con-distrib}.%
\end{proof}%
}%
\noindent%
The $\max$ in Theorem~\ref{thm:sep-con-distrib}.\ref{thm:sep-con-distrib:sepcon-over-max} corresponds to a disjunction ($\vee$) in the classical setting as for any two predicates $\preda$ and $\predb$ we have $\iverson{\preda \vee \predb} = \Max{\iverson{\preda}}{\iverson{\predb}}$, where the $\max$ is taken pointwise.
\CHANGED{Moreover, if $\preda$ and $\predb$ are mutually exclusive, i.e. $\iverson{\preda} \cdot \iverson{\predb} = 0$, their maximum coincides with their sum. That is, we have $\max \{ \iverson{\preda}, \iverson{\predb} \} = \iverson{\preda} + \iverson{\predb}$.}
Theorem~\ref{thm:sep-con-distrib}.\ref{thm:sep-con-distrib:sepcon-over-max} shows that $\sepcon$ distributes over $\max$.
Unfortunately, for $+$ we only have sub-distributivity (Theorem~\ref{thm:sep-con-distrib}.\ref{thm:sep-con-distrib:sepcon-over-plus}).
We recover full distributivity in case that $\ff$ is \mbox{domain-exact (Theorem~\ref{thm:sep-con-distrib}.\ref{thm:sep-con-distrib:sepcon-over-plus-full})}.

A further important analogy to $\SL$ is that quantitative separating conjunction is monotonic:
\begin{theorem}[Monotonicity of $\sepcon$] 
\label{thm:sepcon-monotonic}
	%
	$
    \ff \preceq \ff' ~\text{and}~ \fg \preceq \fg'  ~\text{implies}~ \ff \sepcon \fg \preceq \ff' \sepcon \fg'
	$. 
\end{theorem}
\REPORT{%
\begin{proof}%
    See Appendix~\ref{app:sepcon-monotonic}, p.~\pageref{app:sepcon-monotonic}.%
\end{proof}%
}%
\noindent%
Next, we look at a quantitative analog to \emph{modus ponens}.
The classical modus ponens rule states that $\preda \sepcon (\preda \sepimp \predb)$ implies $\predb$.
In a quantitative setting, implication generalizes to $\preceq$, i.e.\ the partial order we defined in Section~\ref{sec:expectations}~(see also \cite{DBLP:series/mcs/McIverM05}).
%
%
\begin{theorem}[Quantitative Modus Ponens]
	\label{thm:modus-ponens}
	$ 
		\iverson{\preda} \sepcon \bigl( \iverson{\preda} \sepimp \ff \bigr) \ppreceq \ff
	$. 
\end{theorem}
\REPORT{%
\begin{proof}%
    See Appendix~\ref{app:modus-ponens}, p.~\pageref{app:modus-ponens}.%
\end{proof}%
}%
\noindent
Analogously to the qualitative setting, quantitative $\sepcon$ and $\sepimp$ are adjoint operators: 
%
%
\begin{theorem}[Adjointness of $\sepcon$ and $\sepimp$]
\label{thm:adjointness}
	$
		\ff \sepcon \iverson{\preda} \ppreceq \fg \qiff \ff \ppreceq  \iverson{\preda} \sepimp \fg
	$. 
\end{theorem}
\REPORT{%
\begin{proof}%
    See Appendix~\ref{app:adjointness}, p.~\pageref{app:adjointness}.%
\end{proof}%
}%
\noindent
Intuitively, a separating conjunction $\underdot{\phantom{\ff}}\sepcon \iverson{\preda}$ carves out a portion of the heap, since $\ff \sepcon \iverson{\preda}$ splits of a part of the heap satisfying $\preda$ and measures $\ff$ in the remaining heap.
Conversely, $\iverson{\preda} \sepimp \underdot{\phantom{\fg}}$ extends the heap by a portion satisfying $\preda$.
Adjointness now tells us that instead of carving out something on the left-hand side of an inequality, we can extend something on the right-hand side and vice versa.
This is analogous to
$
	a - \epsilon \leq b$ iff 
    $a \leq \epsilon + b
$ 
\mbox{in standard calculus}.

\begin{example}
    \label{ex:qsl:formulas}
    Let us consider a few examples to gain more intuition on quantitative separating connectives.
    For that, let $\sk$ be any stack 
    and let heap $\hh = \{ 1 \mapsto 2, 2 \mapsto 3, 4 \mapsto 5 \}$. Then: 
    \begin{align*}
            \left( \singleton{1}{2} \sepcon \heapSize \right)(\sk,\hh) \eeq 2 \eeq & \heapSize(\sk,\hh) - 1 \\
            \left( \singleton{3}{4} \sepimp \heapSize \right)(\sk,\hh) \eeq 4 \eeq & \heapSize(\sk,\hh) + 1 \\
            \left( \singleton{3}{4} \sepcon \heapSize \right)(\sk,\hh) \eeq 0 \eeq & \left( \singleton{1}{2} \sepcon \singleton{1}{2} \sepcon \heapSize \right)(\sk,\hh) \\
            \left( \singleton{1}{2} \sepcon (\singleton{1}{2} \sepimp \heapSize) \right)(\sk,\hh) \eeq 3 \eeq & \heapSize(\sk,\hh) 
            \eeq \left( \singleton{3}{4} \sepimp (\singleton{3}{4} \sepcon \heapSize) \right)(\sk,\hh) \\
            \left( \singleton{1}{2} \sepimp \heapSize \right)(s,h) \eeq \infty \eeq & \left( \singleton{3}{4} \sepimp (\singleton{3}{4} \sepimp \heapSize) \right)(s,h) 
        %
        \tag*{$\triangle$}
    \end{align*}
    %
    %
    %
\end{example}

\subsection{Pure Expectations}

In $\SL$, a predicate is called \emph{pure} iff its truth does not depend on the heap but only on the stack.
Analogously, in $\QSL$ we call an expectation $\ff$ pure iff
\begin{align*}
	\forall\, \sk, \hh_1, \hh_2\colon\quad  \ff(\sk, \hh_1) \eeq \ff(\sk, \hh_2)~.
\end{align*}
For pure expectations, several of~\cite{DBLP:conf/lics/Reynolds02} laws for \SL hold as well:
\begin{theorem}[Algebraic Laws for $\sepcon$ under Purity]
\label{thm:sep-con-algebra-pure}
	Let $\ff,\fg,\fh \in \E$ and let $\ff$ be pure. Then
	\begin{enumerate}
		\item $\ff \cdot \fg \preceq \ff \sepcon \fg$,
		\item $\ff \cdot \fg \eeq \ff \sepcon \fg$, if additionally $\fg$ is also pure, and
		\item $(\ff \cdot \fg) \sepcon \fh = \ff \cdot (\fg \sepcon \fh)$.
	\end{enumerate}
\end{theorem}
\REPORT{%
\begin{proof}%
    See Appendix~\ref{app:sep-con-algebra-pure}, p.~\pageref{app:sep-con-algebra-pure}.%
\end{proof}%
}%

\subsection{Intuitionistic Expectations}
\label{sec:qsl:intuitionistic}

In \SL, a predicate $\preda$ is called \emph{intuitionistic}, iff for all stacks $\sk$ and heaps $\hh, \hh'$ with $\hh \subseteq \hh'$,
$%
	(\sk,\, \hh) \models \preda$ implies $(\sk,\, \hh') \models \preda
$. %
So as we extend the heap from $\hh$ to $\hh'$, an intuitionistic predicate can only get \emph{``more true''}.
Analogously, in \QSL, as we extend the heap from $\hh$ to $\hh'$, the quantity measured by an \emph{intuitionistic expectation} can only \emph{increase}.
Formally, an expectation $\ff$ is called \emph{intuitionistic} iff
\begin{align*}
	\forall\,\sk, \hh \subseteq \hh' \colon \quad \ff(\sk, \hh) \lleq \ff(\sk, \hh')~.
\end{align*}
A natural example of an intuitionistic expectation is the heap size quantity 
\begin{align*}
    \heapSize \eeq \lambda (\sk,\hh)\mydot |\dom{\hh}|~.
\end{align*}
%

\cite{DBLP:conf/lics/Reynolds02} describes a systematic way to construct intuitionistic predicates from possibly non-intuitionistic ones: 
For any predicate $\preda$, 
%
		$\preda \sepcon \true$ is the strongest intuitionistic predicate weaker than $\preda$, and 
		$\true \sepimp \preda$ is the weakest intuitionistic predicate stronger than $\preda$.
%
In \QSL: 
\begin{theorem}[Tightest Intuitionistic Expectations]
\label{thm:intuitionistification}
	Let $\ff \in \E$. Then:
	\begin{enumerate}
		\item
			$\ff \sepcon 1$ is the smallest intuitionistic expectation that is greater than $\ff$. Formally, $\ff \sepcon 1$ is intuitionistic, $\ff \preceq \ff \sepcon 1$, and 
            for all intuitionistic $\ff'$ satisfying $\ff \preceq \ff'$, we have $\ff \sepcon 1 \preceq \ff'$.
			%
			%
		\item
			$1 \sepimp \ff$ is the greatest intuitionistic expectation that is smaller than $\ff$. Formally,  $1 \sepimp \ff$ is intuitionistic,  $1 \sepimp \ff \preceq \ff$, and 
            for all intuitionistic $\ff'$ satisfying $\ff' \preceq \ff$, we have $\ff' \preceq 1 \sepimp \ff$.
			%
			%
	\end{enumerate}
\end{theorem}
\REPORT{%
\begin{proof}
    See Appendix~\ref{app:intuitionistification}, p.~\pageref{app:intuitionistification}.%
\end{proof}%
}%
\noindent
For example, the \emph{contains-pointer predicate} $\containsPointer{\ee}{\ee'}$ defined by
\begin{align*}
	\containsPointer{\ee}{\ee'} \eeq \singleton{\ee}{\ee'} \sepcon 1
\end{align*}
is an intuitionistic version of the points-to predicate $\singleton{\ee}{\ee'}$: 
Whereas $\singleton{\ee}{\ee'}$ evaluates to $1$ iff the heap consists of \emph{exactly} one cell with value $\ee'$ at address $\ee$ and no other cells, $\containsPointer{\ee}{\ee'}$ evaluates to $1$ iff the heap \emph{contains} a cell with value $\ee'$ at address $\ee$ but possibly also other allocated memory.

Analogously, the fact that some cell with address $\ee$ exists on the heap is formalized by
\begin{align*}
        \containsPointer{\ee}{-} \eeq \validpointer{\ee} \sepcon 1~.
\end{align*}
With intuitionistic versions of points-to predicates at hand, we can derive specialized laws when dealing with the heap size quantity, which we already observed for a concrete heap in Example~\ref{ex:qsl:formulas}.
\begin{theorem}[Heap Size Laws]\label{thm:qsl:heap-size}
Let $\ff,\fg \in \E$ and $\ee,\ee'$ be arithmetic expressions. Then:
\begin{enumerate}
   \item $\singleton{\ee}{\ee'} \sepcon \heapSize \eeq \containsPointer{\ee}{\ee'} \cdot (\heapSize - 1)$ \label{thm:qsl:heap-size:sepcon}
   \item $\singleton{\ee}{\ee'} \sepimp \heapSize \eeq 1 + \heapSize + \containsPointer{\ee}{-} \cdot \infty$ \label{thm:qsl:heap-size:sepimp}
   \item $\left(\ff \sepcon \fg\right) \cdot \heapSize \ppreceq \left(\ff \cdot \heapSize\right) \sepcon \fg + \ff \sepcon \left(\fg \cdot \heapSize\right)$ \label{thm:qsl:heap-size:dist}
    \item $\left(\ff \sepcon \fg\right) \cdot \heapSize \eeq \left(\ff \cdot \heapSize\right) \sepcon \fg + \ff \sepcon \left(\fg \cdot \heapSize\right)$, if $X$ or $Y$ is domain-exact. \label{thm:qsl:heap-size:dist-full}
\end{enumerate}
\end{theorem}
\REPORT{%
\begin{proof}%
  See Appendix~\ref{app:qsl:heap-size}, p.~\pageref{app:qsl:heap-size}.%
\end{proof}%
}%
The first two rules illustrate the role of $\sepcon$ and $\sepimp$:
$\sepcon$ removes a part of the heap that is measured and consequently decreases the size of the remaining heap.
Dually, $\sepimp$ extends the heap and hence increases its size.
If the heap cannot be extended appropriately, the infimum in the definition of $\sepimp$ yields $\infty$.
The third and fourth rule intuitively state that the size of the heap captured by $\ff \sepcon \fg$ is the sum of the sizes of the heap captured by $\ff$, i.e. $\ff \cdot \heapSize$, and of the heap captured by $\fg$, i.e. $\fg \cdot \heapSize$.
However, in both cases we have to account for parts of the heap whose size is not measured, i.e. $\fg$ if we measure the size of $\ff$ and vice versa.
These parts are ``absorbed'' by an additional separating conjunction with $\fg$ and $\ff$, respectively.

\subsection{Recursive Expectation Definitions}
\label{sec:qsl:recursive}

To reason about unbounded data structures such as lists, trees, etc., separation logic relies on inductive predicate definitions (cf.~\cite{DBLP:conf/lics/Reynolds02,DBLP:conf/sas/Brotherston07}).
In \QSL, quantitative properties of unbounded data structures are specified similarly using recursive equations of the form
\begin{align}
  P(\vec{\za}) \eeq \ff_{P}(\vec{\za}), \label{eq:recursive-definition} 
\end{align}
where 
$\vec{\za} \in \Ints^n$,
$P\colon \Ints^n \To \E$, and $\ff_{\:\cdot\:}({\:\cdot\:}) \colon (\Ints^n \To \E) \to (\Ints^n \To \E)$ is a monotone function.
\begin{example}\label{ex:qsl:ls}
Consider a recursive predicate definition from standard separation logic: A singly-linked list segment with head $\za$ and tail $\zb$ is given by the equation
\begin{align*}
        \Ls{\za}{\zb} \eeq \underbrace{ \iverson{\za = \zb} \cdot \emp + \iverson{\za \neq \zb} \cdot {\textstyle \sup_{\zc}}~ \singleton{\za}{\zc} \sepcon \Ls{\zc}{\zb} }_{ {} \eqqcolon \ff_{\Lssymbol}(\za,\zb) }.
\end{align*}
Clearly, 
$\ff_{P}(\za,\zb)$ is monotone, i.e. $P \preceq P'$ implies $\ff_{P}(\za,\zb) \preceq \ff_{P'}(\za,\zb)$.
Hence, all list segments between $\alpha$ and $\beta$ are given by the least fixed point of the above equation.
\hfill$\triangle$
\end{example}%
The semantics of (\ref{eq:recursive-definition}) is defined as the least fixed point of a monotone expectation transformer
\begin{align*}
        \Psi_{P}\colon \quad \left(\Ints^n \To \E\right) \to \left(\Ints^n \To \E\right),\quad Q \mapsto \lambda \vec{\za} \mydot \ff_{Q}(\vec{\za}).
\end{align*}
Thus, we define the expectation given by recursive equation (\ref{eq:recursive-definition}) as
        $P(\vec{\za}) = \bigl( \lfp Q \mydot \Psi_{P}(Q) \bigr)(\vec{\za})$,
where $\lfp Q\mydot \Psi(Q)$ denotes the least fixed point of $\Psi$.
Existence of the least fixed point is guaranteed due to Tarski and Knaster's fixed point theorem (cf.~\cite{cousot1979constructive}).

This notion of recursive definitions coincides with the semantics of inductive predicates in \SL~\cite{DBLP:conf/sas/Brotherston07} if expectations are restricted to predicates. 
For instance, $\Ls{\za}{\zb}(\sk,\hh) = 1$ iff $\hh$ consists \emph{exactly} of a singly-linked list with head $\za$ and tail $\zb$. 

Recursive expectation definitions in $\QSL$ are, however, not limited to predicates.
For example, 
the \emph{length} of a singly-linked list segment can be defined as follows:
\begin{align*}
        \textstyle \Len{\za}{\zb} \eeq \iverson{\za \neq \zb} \cdot \sup_{\zc}~ \singleton{\za}{\zc} \sepcon \left( \Ls{\zc}{\zb} + \Len{\zc}{\zb} \right)
\end{align*}
If the heap exclusively consists of a singly-linked list from $\za$ to $\zb$, then the expectation $\Len{\za}{\zb}$ evaluates to the length of that list, and to zero otherwise.
We next collect a few properties of the two closely related expectations $\Lensymbol$ and $\Lssymbol$ that simplify reasoning about programs.
\begin{lemma}[Properties of List Segments and Lengths of List Segments] We have:
\label{thm:ls-props}
	\begin{enumerate}
		\item 
		\label{thm:ls-props:char}
			$\Len{\za}{\zb} \eeq \Ls{\za}{\zb} \cdot \heapSize$ 
			
		\item
		\label{thm:ls-props:lsls}
			$\Ls{\za}{\zb} \eeq \sup_{\zc} \Ls{\za}{\zc} \sepcon \Ls{\zc}{\zb}$
	\end{enumerate}
\end{lemma}
\REPORT{%
\begin{proof}%
  See Appendix~\ref{app:ls-props}, p.~\pageref{app:ls-props}.%
\end{proof}%
}%
\noindent%
The first property gives an alternative characterization of list lengths which exploits the fact that $\iverson{\Lssymbol}$ ensures that nothing but a list is contained in the heap. 
Consequently, the length of that list is given by the size of the specified heap. 
The second property shows that lists can be split into multiple lists or merged into a single list at any address in between.

The list-length quantity $\Lensymbol$ actually serves two purposes: It ensures that the heap is a list and if so determines the longest path through the heap. 
The latter part can be generalized to other data structures. 
To this end, assume the heap is organized into fixed-size, successive blocks of memory representing \emph{records}, for example the left and right pointer of a binary tree.
If the size of records is a constant $n \in \Nats$, then the longest path through these records starting in $\za$ is given by
\begin{align*}
        \npath{n}{\za} \eeq \textstyle \sup_{\zb \in \Nats}~ \Bigl(\left(\max_{0 \leq k < n} \singleton{\za + k}{\zb}\right) 
                               \sepcon \left(1+\npath{n}{\zb}\right)\Bigr)~.
\end{align*}
Intuitively, $\npath{n}{\za}$ always selects the successor address $\zb$ among the possible pointers in the record belonging to $\za$ which is the source of the longest path through the remaining heap. 
\CHANGED{
Notice that no explicit base case is needed, because the length of empty paths is zero.
Moreover, the use of the separating conjunction prevents selecting the same pointer twice. 
}
The quantity $\textsf{path}$ is more liberal than $\Lensymbol$ in the sense that heaps may contain pointers that do not lie on the specified path.
The $\textsf{path}$ quantity can then be easily combined with stricter \mbox{data structure specifications}.
\begin{example}
Consider a classical recursive \SL predicate specifying binary trees with root $\za$:
\begin{align*}
  \Tree{\za} \eeq \textstyle \iverson{\za = \nil} \cdot \emp ~+~ \sup_{\zb,\zc \in \Nats}~ \singleton{\za}{\zb,\zc} \sepcon \Tree{\zb} \sepcon \Tree{\zc}~.
\end{align*}
Combining $\Tree{\za}$ with $\npath{2}{\za}$, we can measure the height of binary trees \mbox{with root $\za$}:
\begin{align*}
  \mathsf{treeHeight}(\za) \eeq \Tree{\za} \cdot \npath{2}{\za}~.
\end{align*}
This is illustrated in Figure~\ref{fig:path-equation}, where two heaps are graphically depicted as directed graphs.
The left graph contains a cycle and thus does \emph{not} constitute a binary tree. Consequently, $\Tree{\za} = 0$. 
The longest path through this heap is $\za \zb_1 \ldots \zb_4$, i.e. $\npath{2}{\za} = 5$.
In contrast, the right graph \emph{is} a binary tree with root $\za$, i.e. $\Tree{\za} = 1$. The longest path through this heap is of length two, e.g. $\za \zb_1 \zb_2$.
Hence, the height of the tree 
is given by $\mathsf{treeHeight}(\za) = \Tree{\za} \cdot \npath{2}{\za} = 2$.
\hfill$\triangle$
\end{example}
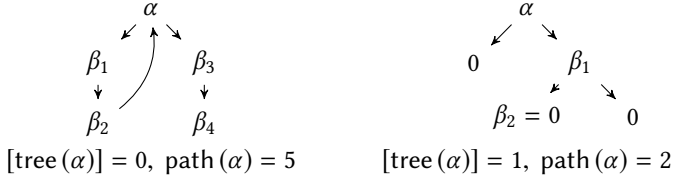
\begin{figure}[t]
\begin{tikzpicture}[->,>=stealth',shorten >=1pt,auto,node distance=0.8cm]

  \begin{scope}
  \node (a) {$\alpha$};
  \node (b1) [node distance=1cm, below left of=a] {$\beta_1$};
  \node (b2) [below of=b1] {$\beta_2$};
  \node (b3) [node distance=1cm, below right of=a] {$\beta_3$};
  \node (b4) [below of=b3] {$\beta_4$};

  \node (text) [below of = a, node distance = 2cm] {$\Tree{\alpha} = 0,~\Path{\alpha} = 5$};

  \path[->]
    (a) edge (b1)
    (b1) edge (b2)
    (a) edge (b3)
    (b3) edge (b4)
    (b2) edge[bend right] (a)
    ;
  \end{scope}
  \begin{scope}[shift={(5,0)}]
  \node (a) {$\alpha$};
  \node (c1) [node distance=1cm, below left of=a] {$0$};
  \node (b1) [node distance=1cm, below right of=a] {$\beta_1$};
  \node (c2) [node distance=1cm, below right of=b1] {$0$};
  \node (b2) [node distance=1cm, below left of=b1] {$\beta_2=0$};

  \node (text) [below of = a, node distance = 2cm] {$\Tree{\alpha} = 1,~\Path{\alpha} = 2$};

  \path[->]
    (a) edge (b1)
    (a) edge (c1)
    (b1) edge (c2)
    (b1) edge (b2)
    ;
  \end{scope}
\end{tikzpicture}
\caption{Evaluation of $\Tree{\alpha}$ and $\Path{\alpha}$ for two heaps depicted as graphs. Here, an edge $x \to y$ denotes $h(s(x)) = s(y)$ or $h(s(x+1))=s(y)$.}
\label{fig:path-equation}
\end{figure}
%
%
%

%


\section{Reasoning about Programs} \label{sec:wp}

We now turn from \QSL as an assertion language to program verification.
Classical separation logic is commonly applied as a basis for Floyd-Hoare-style correctness proofs.
The main concept in Floyd-Hoare logic are \emph{Hoare triples}.
A Hoare triple $\hoare{\preda}{\cc}{\predb}$ consists of a precondition $\preda$, a non--probabilistic program $\cc$, and a postcondition $\predb$.

One approach to proving a triple $\hoare{\preda}{\cc}{\predb}$ valid is to 
determine whether precondition $\preda$ is covered by all initial states that --- executed on $\cc$ --- reach a final state satisfying postcondition $\predb$.
This kind of \emph{backward reasoning} corresponds to Dijkstra's weakest preconditions.
More precisely, the \emph{weakest precondition of $\cc$ with respect to postcondition $\predb$} is the weakest predicate $\wp{\cc}{\predb}$, such that the triple $\hoare{\wp{\cc}{\predb}}{\cc}{\predb}$ is valid, i.e.~$\wp{\cc}{\predb}$ is the predicate such that
\begin{align*}
	\forall~ \preda\colon \qquad \preda \implies \wp{\cc}{\predb} \qiff \hoare{\preda}{\cc}{\predb} \text{ is valid}~.
\end{align*}
For \SL, validity of Hoare triples usually includes that ``correct programs do not fail''~\cite{DBLP:conf/fossacs/YangO02,DBLP:conf/lics/Reynolds02}, i.e. no execution satisfying the precondition may lead to a memory fault.

Reasoning about probabilistic programs is more subtle.
Running a probabilistic program on an initial state does not yield one or more final states, but a \emph{sub}distribution of final states.
The missing probability mass corresponds to the \emph{probability of nontermination or encountering a memory fault}.
Furthermore, when performing \emph{quantitative} reasoning, the notion of correctness becomes blurred.
For instance, it might be acceptable that a program fails with some small probability.

In order to account for probabilistic behavior,~\cite{DBLP:conf/stoc/Kozen83} generalized weakest precondition reasoning from predicates to measurable functions and later~\cite{DBLP:series/mcs/McIverM05} (re)introduced nondeterminism and coined the term \emph{weakest preexpectation}.
To incorporate dynamic memory, we extend their approach by lifting the backward reasoning rules of~\cite{DBLP:conf/popl/IshtiaqO01,DBLP:conf/lics/Reynolds02} to a quantitative setting.
To be precise, our calculus is designed for \emph{total correctness}, asserts that \emph{no memory faults} happen during any execution (with positive probability), and assumes a \emph{demonic} interpretation of nondeterminism.
Alternative design choices are discussed in Section~\ref{sec:wp:landscape}.

Notice that forward reasoning in the sense of strongest postexpectations is not an option as in general strongest postexpectations do not exist for probabilistic programs~\cite{DBLP:phd/ethos/Jones90}.
This also justifies our need for the separating implication in \QSL which --- in classical approaches based on separation logic --- is not needed when applying forward reasoning.

\begin{table*}[t]
\caption{Rules for the weakest preexpectation transformer. Here $\ff \in \E$ is a (post)expectation, 
$\ff\subst{x}{v} =  \lambda (\sk,\hh)\mydot \ff(\sk\subst{x}{\sk(v)}, \hh)$ is the ``syntactic replacement'' of $x$ by $v$ in $\ff$, and
$\vec{\ee} = (\ee_1,\ldots,\ee_n)$ is a tuple of expressions. 
\CHANGED{Moreover, $\lfp \fg\mydot \Phi(\fg)$ is the least fixed point of $\Phi$.}
\REMOVED{Moreover, $\AVAILLOC{\vec{\ee}} = \lambda(\sk,\hh)\mydot \{ v \in \Nats ~|~ v,v+1,\ldots,v+|\vec{\ee}|-1 \notin \dom{\hh} \}$ collects all suitable memory locations for allocation of $\vec{\ee}$ in heap $\hh$ and 
$\lfp \fg\mydot \Phi(\fg)$ is the least fixed point of $\Phi$.}
}
\label{table:wp}
\renewcommand{\arraystretch}{1.5}
\begin{tabular}{@{\hspace{1em}}l@{\hspace{2em}}l}
	\hline
	$\boldsymbol{\cc}$			& $\boldsymbol{\textbf{\textsf{wp}}\,\left \llbracket \cc\right\rrbracket  \left(\ff \right)}$ \\
	\hline
	$\SKIP$					& $\ff$ 																					\\
	$\ASSIGN{x}{\ee}$			& $\ff\subst{x}{\ee}$ \\
	$\COMPOSE{\cc_1}{\cc_2}$		& $\wp{\cc_1}{\vphantom{\big(}\wp{\cc_2}{\ff}}$ \\
	$\ITE{\guard}{\cc_1}{\cc_2}$		& $\iverson{\guard} \cdot \wp{\cc_1}{\ff} + \iverson{\neg \guard} \cdot \wp{\cc_2}{\ff}$ \\
	$\WHILEDO{\guard}{\cc'}$		& $\lfp \fg\mydot \iverson{\neg \guard} \cdot \ff + \iverson{\guard} \cdot \wp{\cc'}{\fg}$ \\
	$\PCHOICE{\cc_1}{\pp}{\cc_2}$		& $\pp \cdot \wp{\cc_1}{\ff} + (1- \pp) \cdot \wp{\cc_2}{\ff}$ \\
    $\ALLOC{x}{\vec{\ee}}$	& $\displaystyle\inf_{\CHANGED{v \in \AVAILLOC{\vec{\ee}}}} \singleton{v}{\vec{\ee}} \sepimp \ff\subst{x}{v}$ \\
	$\ASSIGNH{x}{\ee}$			& $\displaystyle\sup_{v \in \Ints} \singleton{\ee}{v} \sepcon \bigl( \singleton{\ee}{v} \sepimp \ff\subst{x}{v} \bigr)$ \\
	$\HASSIGN{\ee}{\ee'}$			& $\validpointer{\ee} \sepcon \bigl(\singleton{\ee}{\ee'} \sepimp \ff \bigr)$ \\
	$\FREE{\ee}$				& $\validpointer{\ee} \sepcon \ff$ \\
	\hline
\end{tabular}
\end{table*}
\subsection{Weakest Preexpectations}
The \emph{weakest preexpectation} 
of program $c$ with respect to \emph{postexpectation} $\ff \in \E$
is an expectation $\wp{\cc}{\ff} \in \E$, such that $\wp{\cc}{\ff}(\sk,\, \hh)$ is the \emph{least expected value} of $\ff$ (measured in the final states) \emph{after successful termination}, i.e.\ \emph{no memory faults} during execution, of $\cc$ on initial state $(\sk,\, \hh)$.
In particular, if $\ff$ is a predicate 
then $\wp{\cc}{\ff}(\sk,\hh)$ is the least probability that $\cc$ executed on initial state $(\sk,\, \hh)$ does not cause a memory fault and terminates successfully in a final state satisfying $\ff$.
In the following, we extend the weakest preexpectation calculus of~\cite{DBLP:series/mcs/McIverM05} to heap-manipulating programs, i.e.\ $\hpgcl$ as presented \mbox{in Section \ref{sec:hpgcl}}.
\begin{definition}[Weakest Preexpectation Transformer]
    The \emph{weakest preexpectation $\wp{\cc}{\ff}$ of $\cc \in \hpgcl$ with respect to postexpectation 
    $\ff \in \E$} 
    is defined according to the rules in \autoref{table:wp}.
\hfill $\triangle$
\end{definition}%
\noindent%
Let us go over the individual rules for $\wpsymbol$ stated in \autoref{table:wp}.
We start with briefly considering the non-heap-manipulating constructs.
$\wpC{\SKIP}$ behaves as the identity since $\SKIP$ does not modify the program state. 
For $\wp{\ASSIGN{x}{\ee}}{\ff}$ we return $\ff\subst{x}{\ee}$ which is obtained from $\ff$ by ``syntactically replacing'' $x$ with $\ee$. 
More formally, $\ff\subst{x}{\ee} = \lambda (\sk,\hh)\mydot \ff(\sk\subst{x}{\sk(\ee)},\hh)$.
For sequential composition, $\wp{\COMPOSE{\cc_1}{\cc_2}}{\ff}$ obtains a preexpectation of the program $\COMPOSE{\cc_1}{\cc_2}$ by applying $\wpC{\cc_1}$ to
the intermediate expectation obtained from $\wp{\cc_2}{\ff}$.
For conditional choice, $\wp{\ITE{\guard}{\cc_1}{\cc_2}}{\ff}$ selects either $\wp{\cc_1}{\ff}$ or $\wp{\cc_2}{\ff}$ by multiplying them accordingly with the indicator function of $\guard$ or the indicator function of $\neg \guard$ and adding those two products.
For the probabilistic choice, $\wp{\PCHOICE{\cc_1}{\pp}{\cc_2}}{\ff}$ is a convex sum that weighs $\wp{\cc_1}{\ff}$ and $\wp{\cc_2}{\ff}$ by probabilities $\pp$ and $(1-\pp)$, respectively.
For loops, $\wp{\WHILEDO{\guard}{\cc'}}{\ff}$ is characterized as a least fixed point of loop unrollings. We discuss loops and corresponding proof rules in Section~\ref{sec:wp:loops}.
For a detailed treatment of weakest preexpectations for these standard constructs, confer~\cite{DBLP:series/mcs/McIverM05}.
Before we consider the remaining statements, let us collect a few basic properties of $\wpsymbol$:
\begin{theorem}[Basic Properties of \textnormal{$\wpsymbol$}]\label{thm:wp:basic}
  For all \hpgcl-programs $\cc$, expectations $\ff,\fg \in \E$, predicates $\preda$ and constants $k \in \PosReals$, we have:
  \begin{enumerate}
        \CHANGED{
        \item Monotonicity:\quad $\ff \ppreceq \fg \qimplies \wp{\cc}{\ff} \ppreceq \wp{\cc}{\fg}$
              \label{thm:wp:basic:monotonicity}
        \item Super--linearity:\quad $\wp{\cc}{k \cdot \ff + \fg} \ppreceq k \cdot \wp{\cc}{\ff} + \wp{\cc}{\fg}$
              \label{thm:wp:basic:super-linearity}
        }
        \item Strictness:\quad $\wp{\cc}{0} \eeq 0$ 
              \label{thm:wp:basic:preservation-of-0}
        \item $1$--Boundedness of Predicates:\quad $\wp{\cc}{\iverson{\preda}} \ppreceq 1$.
              \label{thm:wp:basic:1-boundedness}
  \end{enumerate}
  Additionally, if $\cc$ does not contain an allocation statement $\ALLOC{x}{\vec{\ee}}$, we have:
  \begin{enumerate}
    \setcounter{enumi}{4}
    \item $\omega$-continuity: For every increasing $\omega$-chain $\ff_1 \preceq \ff_2 \ppreceq \ldots$ in $\E$, we have
    \belowdisplayskip=0pt
    \begin{align*}
    	\textstyle \sup_{n} \wp{\cc}{\ff_n} \eeq \wp{\cc}{\sup_{n} \ff_n}~.
    \end{align*}
          \label{thm:wp:basic:continuity}
    \item Linearity: $\wp{\cc}{k \cdot \ff + \fg} \eeq k \cdot \wp{\cc}{\ff} + \wp{\cc}{\fg}$
          \label{thm:wp:basic:linearity}
  \end{enumerate}
\end{theorem}
\REPORT{%
\begin{proof}%
    By induction on the program structure.
    See Appendix~\ref{app:wp:basic}, p.~\pageref{app:wp:basic}.
\end{proof}%
}%

\subsection{Deallocation, Heap Mutation, and Lookup} 

We now go over the definitions for deterministic heap--accessing language constructs in \autoref{table:wp}.

\paragraph{Memory deallocation.}
A memory cell is deleted from the current heap using the $\FREE{\ee}$ construct as illustrated in \autoref{fig:wp-free}.
%
%
\begin{figure}[t]
	\begin{adjustbox}{max width=0.6\linewidth}
		\includegraphics{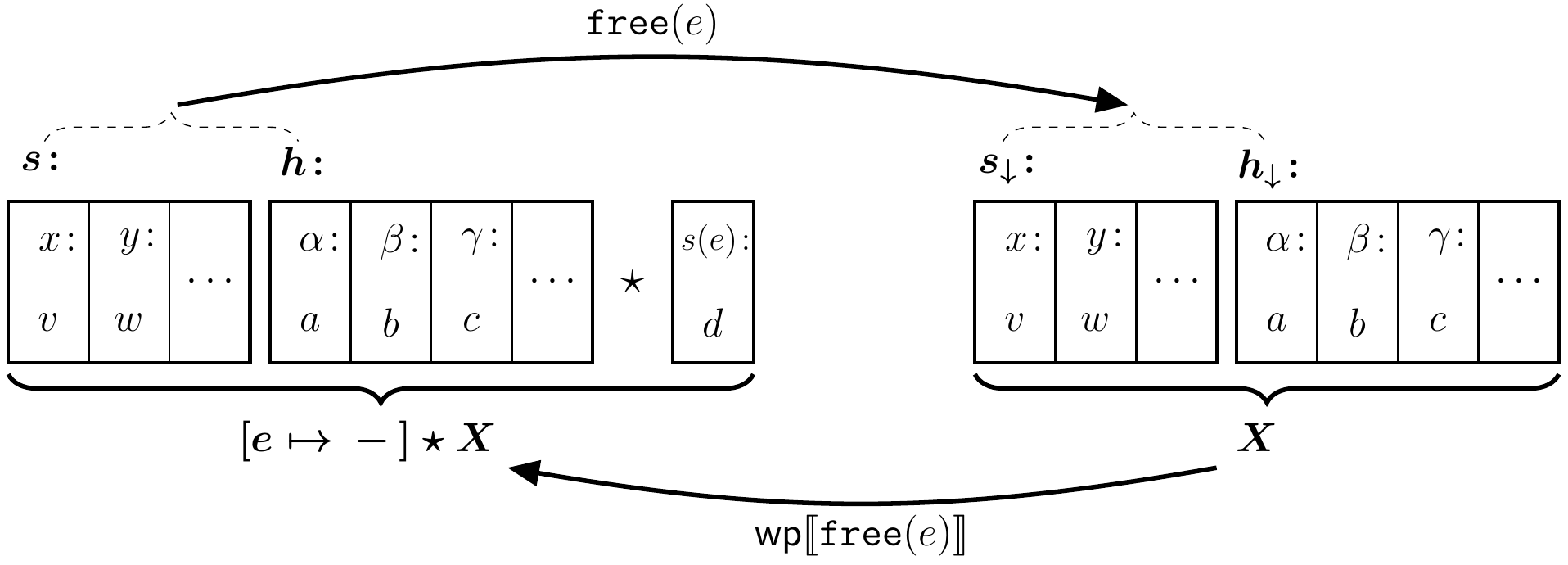}
	\end{adjustbox}
	\caption{Weakest preexpectation of memory deallocation.}
	\label{fig:wp-free}
\end{figure}
$\FREE{\ee}$ starts on some initial state $(\sk,\, \hh)$ shown on the left-hand side and tries to deallocate the memory cell with address $\sk(\ee)$.
In case that $\sk(\ee)$ is a valid address (as depicted in \autoref{fig:wp-free}), i.e.\ $\sk(\ee) \in \dom{\hh}$, $\FREE{\ee}$ removes the corresponding cell from the heap and terminates in a final state $(\sk_\downarrow,\, \hh_\downarrow)$ shown on the right-hand side.
In case that $\sk(\ee)$ is not a valid address (\emph{not} depicted in \autoref{fig:wp-free}), i.e.\ \mbox{$\sk(\ee) \not\in \dom{\hh}$, $\FREE{\ee}$ crashes}.

What is the weakest preexpectation of $\FREE{\ee}$ with respect to a postexpectation $\ff$?
For answering that, we need to construct an expectation $\wp{\FREE{\ee}}{\ff}$, such that the quantity $\wp{\FREE{\ee}}{\ff}$ measured in the initial state coincides with quantity $\ff$ measured in the final state.
The way we will construct $\wp{\FREE{\ee}}{\ff}$ is to \emph{measure $\ff$ in the initial state} and successively rectify the difference to \emph{measuring $\ff$ in \mbox{the final state}}.
%
So what is that difference? We need to dispose the allocated memory cell with address $\sk(\ee)$ in the initial state.
We can rectify this through (a) ensuring that this memory cell actually exists and (b) notionally separating it from the rest of the heap and measuring $\ff$ only in that rest.
Both (a) and (b) are achieved by separatingly conjoining $\ff$ with $\validpointer{\ee}$, thus obtaining $\validpointer{\ee} \sepcon \ff$.
Notice that only heaps consisting of a single cell with address $\sk(\ee)$ make $\validpointer{e}$ evaluate to 1 and are hence the only  
possible choices such that $\validpointer{\ee} \sepcon \ff$ is evaluated to some quantity possibly \mbox{larger than 0}.
This also means that if $\FREE{\ee}$ crashes because address $\sk(\ee)$ is not allocated, then $\wp{\FREE{\ee}}{\ff}(\sk,\, \hh)$ correctly yields $0$.

\CHANGED{
\paragraph{Memory allocation.}
The memory allocation statement $\ALLOC{x}{\ee}$ deserves special attention as it is the only statement that exhibits nondeterministic behavior.
For simplicity, let us consider $\ALLOC{x}{\ee}$ instead of $\ALLOC{x}{\vec{\ee}}$, i.e.~we only allocate a \emph{single} memory cell.
The situation is illustrated in \autoref{fig:wp-new}.
\begin{figure}[t]
	\begin{adjustbox}{max width=0.6\linewidth}
		\includegraphics{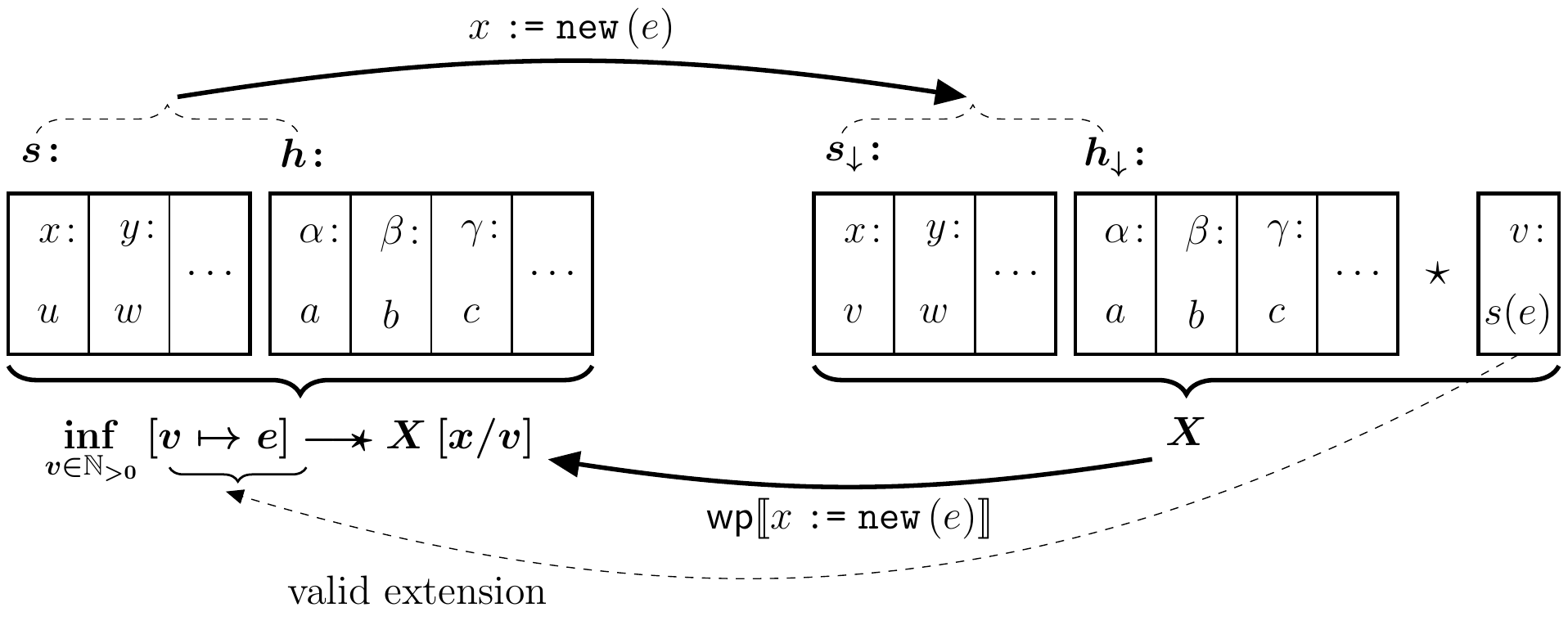}
	\end{adjustbox}
	\caption{Weakest preexpectation of memory allocation.}
	\label{fig:wp-new}
\end{figure}
Operationally, the instruction $\ALLOC{x}{\ee}$ starts on some initial state $(\sk,\, \hh)$ shown on the left-hand side, adds (allocates) to the domain of heap $\hh$ a single \emph{fresh} 
address $v$, and stores at this address content $\sk(\ee)$.
After allocating memory at address $v$, the address $v$ is stored in variable $x$.
The statement $\ALLOC{x}{\ee}$ then terminates in a final state $(\sk_\downarrow,\, \hh_\downarrow)$ shown on the right-hand side.
Since $v$ is chosen \emph{nondeterministically}, we cannot give any a-priori guarantees on $v$ except for $v \not\in \dom{\hh}$.
Furthermore, notice that in our memory model there are at any point infinitely many free addresses available for allocation.
Allocation thus never causes a memory fault. 

What is now the weakest preexpectation of $\ALLOC{x}{\ee}$ with respect to a postexpectation $\ff$?
Again, we construct $\wp{\ALLOC{x}{\ee}}{\ff}$ by measuring $\ff$ in the initial state and rectifying the differences to measuring $\ff$ in the final state.
So what are those differences?
The first difference is that we are missing in the initial state the newly allocated memory cell with address $v$ and content $\sk(\ee)$ which is present in the final state. 
We can rectify this through notionally extending the heap of the initial state by measuring $\singleton{v}{\ee} \sepimp \ff$ instead of $\ff$.
Notice that a heap consisting of a single cell with address $v$ and content $\sk(\ee)$ is the \emph{only} valid extension \mbox{that satisfies $\singleton{v}{\ee}$}.
The next difference is that in the final state variable $x$ has value $v$.
We can mimic this by a syntactic replacement of $x$ by $v$ in $\ff$, thus obtaining
$\singleton{v}{\ee} \sepimp \ff\subst{x}{v}$.
Finally, we have to account for the fact that the newly allocated address $v$ is chosen nondeterministically.
Following McIver and Morgan's demonic nondeterminism school of thought, we select by $\inf_{v \in \AVAILLOC{\ee}}$ any address that \emph{minimizes} the sought-after quantity.
\REMOVED{, where $\AVAILLOC{\ee}$ represents the set of all suitable free addresses.}
We thus obtain $\wp{\ALLOC{x}{\ee}}{\ff} = \inf_{v \in \AVAILLOC{\ee}} \singleton{v}{\ee} \sepimp \ff\subst{x}{v}$.
}

\paragraph{Heap mutation.}

\autoref{fig:wp-mutate} illustrates how the heap is mutated by a statement $\HASSIGN{\ee}{\ee'}$. 
%
%
%
%
Operationally, we can dissect this instruction 
into two parts:
Starting in some initial state $(\sk,\, \hh)$ shown on the left-hand side, we first deallocate the memory at address $\sk(\ee)$ by $\FREE{\ee}$ and thereby obtain an intermediate state $(\sk',\, \hh')$.
Second, we allocate a new memory cell with content $\sk(\ee')$.
In contrast to the statement $\ALLOC{x}{\ee'}$, which is addressed in the next section, the address of that cell is fixed to $\sk(\ee)$.
This is achieved by the instruction $\texttt{new}(\ee')\texttt{@}\ee$, which we introduce here ad-hoc just for illustration purposes.
Consequently, the weakest preexpectation of $\texttt{new}(\ee')\texttt{@}\ee$ coincides with the weakest preexpectation of $\ALLOC{x}{\ee'}$ except that (a) the allocated address $v$ is fixed to $\sk(\ee)$ and (b) we do not perform an assignment to $x$.
Thus, $\wp{\texttt{new}(\ee')\texttt{@}\ee}{\ff} = \singleton{\ee}{\ee'} \sepimp \ff$.

Since $\HASSIGN{\ee}{\ee'}$ has the same effect as $\COMPOSE{\FREE{\ee}}{\texttt{new}(\ee')\texttt{@}\ee}$, its weakest preexpectation is given by
\begin{align*}
	\wp{\HASSIGN{\ee}{\ee'}}{\ff} & \eeq \wp{\COMPOSE{\FREE{\ee}}{\texttt{new}(\ee')\texttt{@}\ee}}{\ff} \tag{see above}\\
	%
	& \eeq \wp{\FREE{\ee}}{\wp{\texttt{new}(\ee')\texttt{@}\ee}{\ff}} \tag{see \autoref{table:wp}}\\
	& \eeq \wp{\FREE{\ee}}{\singleton{\ee}{\ee'} \sepimp \ff} \tag{see above}\\
	& \eeq \validpointer{\ee} \sepcon \bigl( \singleton{\ee}{\ee'} \sepimp \ff \bigr). \tag{see \autoref{table:wp}}
\end{align*}
Another explanation of $\validpointer{\ee} \sepcon \bigl( \singleton{\ee}{\ee'} \sepimp \ff \bigr)$ from a syntactic point of view is as follows:
By $\validpointer{\ee} \sepcon \cloze{\ff}$, we ensure that the heap contains a cell with address $\ee$ and carve it out from the heap.
Thereafter, by $\singleton{\ee}{\ee'} \sepimp \cloze{\ff}$, we extend the heap by a single cell with address $\ee$ and content $\ee'$.
After performing the aforementioned two operations, \mbox{we measure $\ff$}.
\begin{figure*}[t]
	\begin{adjustbox}{max width=\linewidth}
		\includegraphics{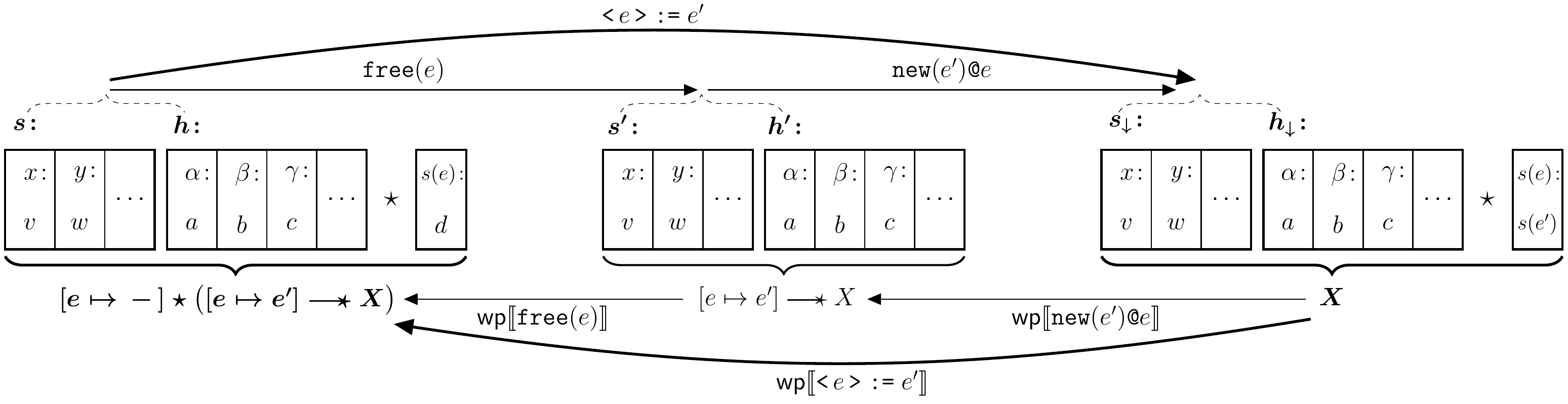}
	\end{adjustbox}
	\caption{Weakest preexpectation of heap mutation.}
	\label{fig:wp-mutate}
\end{figure*}

\paragraph{Heap lookup.}
The statement $\ASSIGNH{x}{\ee}$ determines the value at address $\ee$ and stores it in variable $x$.
Its weakest preexpectation is defined as $\sup_{v \in \Ints} \singleton{\ee}{v} \sepcon \bigl( \singleton{\ee}{v} \sepimp \ff\subst{x}{v} \bigr)$.
We give an intuition on this preexpectation on a syntactic level.
By $\singleton{\ee}{v} \sepcon \cloze{\ff}$, we ensure that the heap contains a cell with address $\ee$ and content $v$, and carve it out from the heap.
It is noteworthy that the value $v$ at address $\ee$ is really \emph{selected} (rather than maximized) by $\sup_{v \in \Ints}$.
This is because either address $e$ is not allocated at all (i.e.\ $\singleton{\ee}{v}$ becomes $0$ for all choices of $v$), or there is a \emph{unique} value $v$ at address $\ee$ which is selected by $\sup_{v \in \Ints}$ (i.e.\ $\singleton{\ee}{v}$ becomes $1$).
We can thus think of the $\sup$ here as taking the role of a $\exists!$-quantifier.
After carving out the cell with address $\ee$ and content $v$, this very cell is put back into the \mbox{heap by $\singleton{\ee}{v} \sepimp \cloze{\ff}$}.
The aforementioned two operations serve only as a mechanism for selecting $v$ at address $\ee$ as we can now measure $\ff$ in a state where variable $x$ has value $v$ through \mbox{finally measuring $\ff\subst{x}{v}$}.
Notice that $\sup_{v \in \Ints} \containsPointer{\ee}{v} \cdot \ff\subst{x}{v}$ is equivalent to $\wp{\ASSIGNH{x}{\ee}}{\ff}$ 
(cf. \REPORT{Appendix~\ref{app:wand-reynolds}}\CAMERA{\cite{DBLP:journals/corr/abs-1802-10467}}).

\subsection{On continuity of $\wpsymbol$}\label{sec:wp:alloc}
For an initially empty heap, the allocation instruction $\ALLOC{x}{\ee}$ nondeterministically assigns a positive natural number to variable $x$. 
It is thus a \emph{countably infinitely branching} nondeterministic assignment. 
The presence of countably infinite nondeterminism in our semantics has dire consequences: Our $\wpsymbol$-calculus is \emph{not} continuous.
Consider, for instance, an $\omega$-chain of expectations $\ff_n = \iverson{1 \leq x \leq n}$. 
Moreover, let $\emptyheap$ be the empty heap. 
Then, for an arbitrary stack $\sk$, 
\begin{align*}
    \wp{\ALLOC{x}{0}}{\textstyle \sup_{n} \ff_n}(\sk,\emptyheap) \eeq 1 \neq 0 \eeq \sup_{n \in \Nats} \wp{\ALLOC{x}{0}}{\ff_n}(\sk,\emptyheap).
\end{align*}
\REPORT{Detailed calculations are found in Appendix~\ref{app:wp:continuity-counterexample}.}
Why do we not attempt to find an alternative semantics of $\ALLOC{x}{\ee}$ that restores continuity?
There are two main reasons:

First, \cite{DBLP:conf/fossacs/YangO02} argue that nondeterministic allocation in \SL is \emph{essential} to enable local reasoning in the presence of address arithmetic.
Alternative approaches for allocation, such as always picking the smallest available memory cell, would invalidate the frame rule (cf. Section~\ref{sec:wp:frame-rule}).

Second,~\cite{apt1986countable} show that it is \emph{impossible} to define a (fully abstract) continuous least fixed point semantics, such as our $\wpsymbol$-style calculus, that exhibits countably infinite nondeterministic assignments. 
Without further restrictions, e.g. limiting ourselves to a finite total amount of available memory, there is thus no hope for a continuous weakest preexpectation transformer.


%

\subsection{Weakest Preexpectations of Loops}
\label{sec:wp:loops}

As is standard in denotational semantics, the weakest preexpectation of a loop $\WHILEDO{\guard}{\cc}$ is characterized as a least fixed point of the loop's unrollings.
That is, the weakest preexpectation of program $\WHILEDO{\guard}{\cc}$ with respect to postexpectation $\ff$ is given by the least fixed point of 
\begin{align*}
        \charwp{\guard}{\cc}{\ff}(\fg) \eeq \iverson{\neg \guard} \cdot \ff + \iverson{\guard} \cdot \wp{\cc}{\fg}~.
\end{align*}
Unfortunately, since our $\wpsymbol$ transformer is \emph{not continuous} in general (see Section~\ref{sec:wp:alloc}),  we cannot rely on 
Kleene's fixed point theorem.
However, due to \autoref{thm:wp:basic}, both $\wpsymbol$ and $\charwp{\guard}{\cc}{\ff}$ are monotone.
%
We may thus resort to a constructive version of the more general fixed point theorem 
due to Tarski and Knaster (cf.~\cite{cousot1979constructive}) for (countable) ordinals:
\begin{theorem}
For every loop $\WHILEDO{\guard}{\cc}$ and $\ff \in \E$, there exists an ordinal $\oa$ such that
\[ \wp{\WHILEDO{\guard}{\cc}}{\ff} \eeq \lfp \fg\mydot \charwp{\guard}{\cc}{\ff}(\fg) \eeq \charwpn{\guard}{\cc}{\ff}{\oa}(0)~.\]
\end{theorem}%
\noindent%
Hence, weakest preexpectations of loops are well-defined.
Reasoning about the exact least fixed point of a loop may, however, require transfinite arguments.
%
Fortunately, we have an invariant-based rule for reasoning about \emph{upper bounds} on preexpectations of loops, which is easier to discharge.

\begin{theorem}\label{thm:upper-invariant}
For loop $\WHILEDO{\guard}{\cc}$ and expectations $\ff, \inv \in \E$, we have
\begin{align*}
	\charwp{\guard}{\cc}{\ff}(\inv) \ppreceq \inv \qimplies \wp{\WHILEDO{\guard}{\cc}}{\ff} \ppreceq \inv~. 
\end{align*}
In this case, we call $\inv$ an \emph{invariant} with respect to program $\WHILEDO{\guard}{\cc}$ and expectation $\ff$.
\end{theorem}

\begin{proof}
        By the Tarski and Knaster fixed point theorem, 
        $\lfp \fg\mydot \charwp{\guard}{\cc}{\ff}(\fg)$ is the smallest pre-fixed point of $\charwp{\guard}{\cc}{\ff}$ (cf.~\cite{cousot1979constructive}).
        It is thus the smallest $\inv$ satisfying $\charwp{\guard}{\cc}{\ff}(\inv) \preceq \inv$.
        Consequently, by \autoref{table:wp}, $\wp{\WHILEDO{\guard}{\cc}}{\ff} = \lfp \fg\mydot \charwp{\guard}{\cc}{\ff}(X) \preceq \inv$.
\end{proof}

\subsection{Soundness of Weakest Preexpectations}
\label{sec:wp:soundness}

We prove the soundness of our weakest preexpectation semantics with respect to the operational semantics introduced in Section~\ref{sec:hpgcl}.
To capture the expected value expectation $\ff \in \E$, we assign a \emph{reward} to every program configuration.
Our operational model is a special case of Markov decision process with rewards (cf.~\cite{DBLP:books/daglib/0020348,puterman2005markov}).
Let 
$
    \Target \eeq \{ (\Term, \sk, \hh) ~|~ (\sk,\hh) \in \States \}
$ 
be the collection of all (goal) configurations indicating successful program termination.
Given $\ff \in \E$, goal configuration $(\Term, \sk, \hh)$ is assigned reward $\ff(\sk,\hh)$.
All other configurations are assigned zero reward.
Formally, the \emph{reward function} for expectation $\ff$ is given by
\begin{align*}
        \OpRew \colon \OpStates \to \PosRealsInf, \quad (\cc, \sk, \hh) \mapsto \iverson{\cc =\, \Term} \cdot \ff(\sk,\hh)~.
\end{align*}
%
We are interested in the minimal (due to demonic nondeterminism) \emph{expected reward} of reaching a goal configuration in $\Target$ (and thus successfully terminating) from an initial configuration $\OpInit \in \OpStates$.
Intuitively, the expected reward is given by the minimal (for all resolutions of nondeterminism) sum over all finite paths $\pi$ from $\OpInit$ to a configuration in $\Target$ weighted by the probability of path $\pi$ and the reward of the reached goal configuration.

Formally, nondeterminism is resolved by a \emph{scheduler} $\Scheduler \colon \OpStates^{+} \to \Nats$ mapping finite sequences of visited configurations to the next action.
Moreover, let $\ProbSymbol$ be a function collecting the total probability mass of execution steps ($\ExecSymbol$) between two configurations for a given action:
\begin{align*}
    \ProbSymbol \colon \OpStates \times \Nats \times \OpStates \to [0,1] \cap \Rats, \quad (t, n, t') \mapsto \sum_{\ExecSimple{t}{n,p}{t'}} p~.
\end{align*}
The set of \emph{finite paths} from $t \in \OpStates$ to some goal configuration using scheduler $\Scheduler$ is given by
\begin{align*}
        \PathsFromTo{t}{\Scheduler} \eeq \{ t_1 \ldots t_m ~|~ & m \in \Nats,\, t_1 =  t,\, t_m \in \Target, \\
                                                                   & \forall k \in \{1,\ldots,m-1\} \,\colon\, \Prob{t_k,\,\Scheduler(t_1 \ldots t_{k}),\,t_{k+1}} > 0 \}~.
\end{align*}
The \emph{probability} of a path $t_1 \ldots t_m \in \PathsFromTo{t}{\Scheduler}$ is the product of its transition probabilities, i.e.
\begin{align*}
        \Prob{t_1 \ldots t_m} \eeq \prod_{1 \leq k < m} \Prob{t_k,\,\Scheduler(t_1 \ldots t_{k}),\,t_{k+1}}~.
\end{align*}
With these notions at hand, the \emph{expected reward of successful termination} 
with respect to expectation $\ff \in \E$
when starting execution in configuration $t \in \OpStates$ \mbox{is defined as}
\begin{align*}
    \ExpRewC{\ff}{t} \eeq \inf_{\Scheduler} \, \sum_{t_1 \ldots t_m \in \PathsFromTo{t}{\Scheduler}} \Prob{t_1 \ldots t_m} \cdot \OpRew(t_m)~.
\end{align*}
The main result of this subsection asserts that our weakest preexpectation calculus for \hpgcl programs is sound with respect to our operational model.
\begin{theorem}[Soundness of Weakest Preexpectation Semantics]
\label{thm:wp:soundness}
For all \hpgcl-programs $\cc$, expectations $\ff \in \E$, and initial states $(\sk,\hh) \in \Sigma$, we have
    $
    \wp{\cc}{\ff}(\sk,\hh) = \ExpRewC{\ff}{\cc, \sk, \hh}
    $.
\end{theorem}
\REPORT{%
\begin{proof}%
        See Appendix~\ref{app:wp:soundness}, p.~\pageref{app:wp:soundness}.%
\end{proof}%
}%
%
%

\CHANGED{
\subsection{Conservativity}
\label{sec:wp:conservativity}

\QSL is a conservative extension of both the weakest preexpectation calculus of~\cite{DBLP:series/mcs/McIverM05} and classical separation logic as developed in~\cite{DBLP:conf/popl/IshtiaqO01,DBLP:conf/lics/Reynolds02}.
Since, for programs that never access the heap, we use the same expectation transformer as~\cite{DBLP:series/mcs/McIverM05}, it is immediate that \QSL conservatively extends weakest preexpectations.

To show that \QSL is also a conservative extension of separation logic, recall from Definition~\ref{def:embedding-sl-qsl}, p.~\pageref{def:embedding-sl-qsl}, the embedding $\qslemb{.}\colon \SL \to \QSL$ of \SL formulas into \QSL.
We then obtain conservativity with respect to separation logic in the following sense:
\begin{theorem}[Conservativity of \QSL as a verification system]\label{thm:qsl:conservativity:wp}
    Let $\cc \in \hpgcl$ be a non-probabilistic program.
    Then, for all classical separation logic formulas $\preda,\predb \in \SL$,
    \begin{align*}
            \text{the Hoare triple}~\{ \,\preda\, \}\,\cc\,\{\,\predb\,\}~\text{is valid for total correctness} \quad\text{iff}\quad \qslemb{\preda} \ppreceq \wp{\cc}{\qslemb{\predb}}.
    \end{align*}
\end{theorem}
\REPORT{
\begin{proof}
See Appendix~\ref{app:qsl:conservativity:wp}, p.~\pageref{app:qsl:conservativity:wp}.
\end{proof}
}

A key principle underlying separation logic is that correct programs must be memory safe (cf.~\cite{DBLP:conf/lics/Reynolds02}), i.e. all executions of a program do not lead to a memory error.
By the above theorem, the same holds for our $\wpsymbol$ calculus when considering non-probabilistic programs.
For probabilistic programs, however, we get a more fine-grained view as we can quantify the probability of encountering a memory error. 
This allows to evaluate programs if failures are unavoidable, for example due to unreliable hardware.
In particular, the weakest preexpecation $\wp{\cc}{1}$ measures the probability that program $\cc$ terminates without a memory fault.
Does this mean that---for probabilistic programs---our calculus can only prove memory safety with probability one, but is unable to prove that a program is \emph{certainly} memory safe?
After all, there might exist an execution of program $\cc$ that encounters a memory error with probability zero.
The answer to this question is \emph{no}:
Assume there is some execution of a program $\cc$ that encounters a memory error.
By the correspondence between $\wpsymbol$ and our operational semantics (cf.\ Theorem~\ref{thm:wp:soundness}),
there is a path from some initial state to an error state $(\Fault,\sk,\hh)$.
Since such a path must be finite and thus has a positive probability, the probability of encountering a memory error must be positive.
In other words,
\begin{corollary}
    An \hpgcl program is memory safe with probability one iff it is memory safe.
\end{corollary}
}

\subsection{The Quantitative Frame Rule}
\label{sec:wp:frame-rule}

In classical \SL (in the sense of a proof system), the \emph{frame rule} is a distinguished feature that allows for local reasoning~\cite{DBLP:conf/fossacs/YangO02}.
Intuitively, it states that a part of the heap that is not explicitly modified by a program is unaffected by that program.
Consequently, it suffices to reason locally only on the subheap that is actually mutated.
The frame rule 
reads as follows:
\begin{align*}
        \infer[~\text{if}~\Mod{\cc} \cap \Vars(\predc) = \emptyset.]
        {
            \hoare{\preda \sepcon \predc}{\cc}{\predb \sepcon \predc}
        }{
            \hoare{\preda}{\cc}{\predb}
        }
\end{align*}
Here, $\Mod{\cc}$ is the set of variables updated by a program $\cc$, i.e.\ all variables appearing on a left-hand side of an assignment in $\cc$.%
\footnote{More formally, $\Mod{\cc} = \{x\}$ if $\cc$ is of the form $\ASSIGN{x}{\ee}$, $\ALLOC{x}{\vec{\ee}}$, or $\ASSIGNH{x}{\ee}$, and $\Mod{\cc} = \emptyset$ if $\cc$ is $\SKIP$, 
$\FREE{\ee}$, or $\HASSIGN{\ee}{\ee'}$.
For the composed programs, we have $\Mod{\cc} = \Mod{\cc_1} \cup \Mod{\cc_2}$ if $\cc$ is either $\ITE{\guard}{\cc_1}{\cc_2}$, $\PCHOICE{\cc_1}{\pp}{\cc_2}$, or $\COMPOSE{\cc_1}{\cc_2}$.
For loops, \mbox{we have $\Mod{\WHILEDO{\guard}{\cc}} = \Mod{\cc}$}.
}
Moreover, $\Vars(\predc)$ collects all variables that ``occur'' in $\predc$.%
\footnote{Formally, $x \in \Vars(\predc)$ iff $\exists\, (\sk,\hh) \in \States ~ \exists\, v,v' \in \Ints\colon \predc(\sk\subst{x}{v},\hh) \nneq \predc(\sk\subst{x}{v'},\hh)$.}

\paragraph{Towards a \underline{quantitative} frame rule}
Let us first translate the above Hoare-style rule into an equivalent version for weakest preconditions.
To this end, we use the well-established fact that
\begin{align*}
    \hoare{\preda}{\cc}{\predb} ~\text{is valid} \quad\text{iff}\quad \preda \Rightarrow \wp{\cc}{\predb}~.
\end{align*}
Notice that this fact remains valid for memory-fault avoiding interpretations of Hoare triples as used by~\cite{DBLP:conf/fossacs/YangO02}. 
Based on this fact, we obtain a suitable formulation of the frame rule in the setting of weakest preconditions:
Assume that $\Mod{\cc} \cap \Vars(\predc) = \emptyset$. Then
\begin{align*}
        & \infer
        {
            \hoare{\preda \sepcon \predc}{\cc}{\predb \sepcon \predc}
        }{
            \hoare{\preda}{\cc}{\predb}
        } \\
        \qqiff & 
        \left(\preda \Rightarrow \wp{\cc}{\predb}\right) \quad\Rightarrow\quad \left(\preda \sepcon \predc \Rightarrow \wp{\cc}{\predb \sepcon \predc }\right) \tag{$\clubsuit$}\\
        \qqiff & 
        \wp{\cc}{\predb} \sepcon \predc \quad\Rightarrow\quad \wp{\cc}{\predb \sepcon \predc} \tag{$\spadesuit$}~.
\end{align*}
To understand the last equivalence, assume that $(\clubsuit)$ holds and choose $\preda = \wp{\cc}{\predb}$. Then replacing $\preda$ by $ \wp{\cc}{\predb}$ in the conclusion of $(\clubsuit)$ immediately yields the implication $(\spadesuit)$.
Conversely, assume $(\spadesuit)$ holds and let $\preda \Rightarrow \wp{\cc}{\predb}$. By monotonicity of $\sepcon$, we obtain
$
        \preda \sepcon \predc \Rightarrow \wp{\cc}{\predb} \sepcon \predc
$. 
Then $(\spadesuit)$ yields that $\wp{\cc}{\predb} \sepcon \predc$ implies $\wp{\cc}{\predb \sepcon \predc}$, i.e.\ $(\clubsuit)$ holds.

In a quantitative setting 
the analog to implication $\Rightarrow$ is 
$\preceq$.
Hence, the frame rule for \QSL is: 
\begin{theorem}[Quantitative Frame Rule]\label{thm:frame-rules}
  For every $\hpgcl$-program $\cc$ and expectations $\ff,\fg \in \E$ with $\Mod{\cc} \cap \Vars(\fg) = \emptyset$, we have
  $
          \wp{\cc}{\ff} \sepcon \fg \preceq  \wp{\cc}{\ff \sepcon \fg}
  $.
\end{theorem}

\begin{proof}
  By structural induction on $\hpgcl$ programs.
  For loops, we additionally have to perform a transfinite induction on the number of iterations.
  \REPORT{See Appendix~\ref{app:frame-rule}, p.~\pageref{app:frame-rule} for a full proof.}
\end{proof}

\paragraph{What about the converse direction?}
Can we also obtain a frame rule of the form $\wp{\cc}{\ff} \sepcon \fg \succeq \wp{\cc}{\ff \sepcon \fg}$?
In the quantitative case, a converse frame rule breaks for probabilistic choice due to the fact that $\sepcon$ and $+$ are only subdistributive in general (\autoref{thm:sep-con-distrib}).
This problem can partially be avoided by requiring $\fg$ to be domain-exact.
However, the ``converse frame rule'' also breaks in the qualitative case, i.e.\ if $\ff$ and $\fg$ are predicates and $\succeq$ corresponds to $\Leftarrow$:
For $\ff = \emp$, we have
\begin{align*}
        \wp{\HASSIGN{x}{0}}{\emp} \eeq & \validpointer{x} \sepcon \left( \singleton{x}{0} \sepimp \emp \right) \eeq 0~.
\end{align*}
If we additionally choose $\fg = \containsPointer{x}{0}$, we also obtain $\Mod{\cc} \cap \Vars(\fg) = \emptyset$ and
\begin{align*}
        \wp{\HASSIGN{x}{0}}{\emp \sepcon \containsPointer{x}{0}} \eeq & \validpointer{x} \sepcon \left( \singleton{x}{0} \sepimp (\emp \sepcon \containsPointer{x}{0}) \right) \\
        \eeq & \containsPointer{x}{-}.
\end{align*}
Put together, this yields a counterexample---even in the \emph{qualitative} case:
\begin{align*}
        \wp{\HASSIGN{x}{0}}{\emp} \sepcon \containsPointer{x}{0} ~\not\succeq~
        \wp{\HASSIGN{x}{0}}{\emp \sepcon \containsPointer{x}{0}}~.
\end{align*}
Hence, there is no converse version of the frame rule for a conservative extension of \SL. 

\section{A Landscape of Weakest Preexpectation Calculi}\label{sec:wp:landscape}

Our weakest preexpectation calculus for \QSL is for total correctness with intrinsic memory safety and demonic nondeterminism.
We now briefly discuss alternative possibilities. 

\paragraph{Angelic nondeterminism}
For a program $\cc$ and $\ff \in \E$, instead of the \emph{least} expected value $\wp{\cc}{\ff}$,
we are now interested in the \emph{largest} expected value $\awp{\cc}{\ff}$ (read: angelic weakest preexpectation) of $\ff$ after execution of $\cc$.
How does angelic nondeterminism affect the inductive definition of $\wpsymbol$ in \autoref{table:wp}?
For nondeterministic statements, we now have to maximize instead of minimize the expected value. 
As $\ALLOC{x}{\vec{\ee}}$ is the only statement that exhibits nondeterminism, we get 
\CHANGED{
\begin{align*}
        \awp{\ALLOC{x}{\vec{\ee}}}{\ff} \eeq 
        \lambda(\sk,\hh)\mydot
        \sup_{v \in \PosNats : v,v+1,\ldots,v+|\vec{\ee}|-1 \notin \dom{\hh}} 
        \left(\singleton{v}{\vec{\ee}} \sepimp \ff\subst{x}{v}\right)(\sk,\hh)~, 
\end{align*}
where, since allocation never fails, we only choose from locations that are not already allocated.
}
For all other statements, $\awpsymbol$ is defined just as $\wpsymbol$ in \autoref{table:wp} (except that $\wpsymbol$ is replaced by $\awpsymbol$).

Since a question like ``what is the expected value of $x$ after execution of program $\cc$'' does not make much sense if there is some positive probability such that
$\cc$ does not terminate or encounters a memory fault, the remainder of this section considers expectations in $\Eone$ only.

\paragraph{Partial correctness}
The weakest \emph{liberal} preexpectation $\wlp{\cc}{\ff}(\sk,\,\hh)$ of program $\cc$ and expectation $\ff \in \Eone$ for an initial state $(\sk,\,\hh)$ corresponds to the weakest preexpectation $\wp{\cc}{\ff}(\sk,\,\hh)$ 
plus the probability that $\cc$ does not terminate on state $(\sk,\,\hh)$.
%
How does shifting to partial correctness affect the inductive definition of $\wpsymbol$ in \autoref{table:wp}?
Following~\cite{DBLP:series/mcs/McIverM05}, we consider the \emph{greatest} fixed point for loops: 
\begin{align*}
  \wlp{\WHILEDO{\guard}{\cc'}}{\ff} \eeq & \gfp \fg\mydot \underbrace{\iverson{\neg \guard} \cdot \ff + \iverson{\guard} \cdot \wlp{\cc'}{\fg}}_{\eeq \charwp{\guard}{\cc}{\ff}(\fg)}~.
\end{align*}
For weakest liberal preexpectations, our quantitative frame rule also applies:
\begin{theorem}[Quantitative Frame Rule for \textnormal{$\wlpsymbol$}]\label{thm:wlp-frame-rules}
  For every $\hpgcl$-program $\cc$ and expectations $\ff,\fg \in \Eone$ with $\Mod{\cc} \cap \Vars(\fg) = \emptyset$, we have
          $\wlp{\cc}{\ff} \sepcon \fg \preceq  \wlp{\cc}{\ff \sepcon \fg}$.
\end{theorem}
\REPORT{%
\begin{proof}%
  See Appendix~\ref{app:wlp-frame-rule}, p.~\pageref{app:wlp-frame-rule}.
\end{proof}%
}%

\noindent%
Furthermore, a dual version of our proof rule for invariant-based reasoning about loops is available for weakest liberal preexpectations.
Its proof is analogous to the proof of Theorem~\ref{thm:upper-invariant}.

\begin{theorem}
For loop $\WHILEDO{\guard}{\cc}$, postexpectation $\ff \in \Eone$ and invariant $\inv \in \Eone$, we have
\begin{align*}
	\inv \ppreceq \charwp{\guard}{\cc}{\ff}(\inv) \qimplies \inv \ppreceq \wlp{\WHILEDO{\guard}{\cc}}{\ff}~.
\end{align*}
\end{theorem}

\paragraph{Extrinsic memory safety}
Finally, we assume terminating with a memory fault is acceptable.
This is analogous to weakest liberal preexpectations, where nontermination is considered acceptable. 
The weakest \emph{extrinsic memory safe} preexpectation $\wep{\cc}{\ff}(\sk,\,\hh)$ 
corresponds to the weakest preexpectation
$\wp{\cc}{\ff}$ plus the probability that $\cc$ terminates with a memory fault on initial state $(\sk,\,\hh)$.
%
How does extrinsic memory safety affect the inductive definition of $\wpsymbol$ in \autoref{table:wp}?
We have to modify the connectives $\sepcon$ and $\sepimp$ to add the probability of memory faults.
The resulting connectives, denoted $\ff \esepcon \fg$ and $\iverson{\preda} \esepimp \fg$, where $\ff,\fg \in \Eone$ and $\preda$ is a predicate, are defined below.
\begin{align*}
    \ff \esepcon \fg \eeq & \lambda (\sk,\hh) \mydot \min \setcomp{1-\ff(\sk,\hh_1) + \ff(\sk,\hh_1) \cdot \fg(\sk,\hh_2)}{\hh = \hh_1 \sepcon \hh_2} \\
    \iverson{\preda} \esepimp \fg \eeq & \lambda (\sk,\hh) \mydot \sup_{\hh'} \left\{ \fg(\sk,\hh \sepcon \hh') ~|~ \hh \disjoint \hh' ~\text{and}~ (\sk,\hh') \models \preda \right\}
\end{align*}
The rules of $\wepsymbol$ are then obtained from the rules for $\wpsymbol$ in \autoref{table:wp} by replacing 
every occurrence of $\sepcon$ by $\esepcon$ and $\sepimp$ by $\esepimp$, respectively, 
and changing the rule for heap lookups to
\begin{align*}
        \wep{\ASSIGNH{x}{\ee}}{\ff} \eeq \inf_{v \in \Ints} \singleton{\ee}{v} \esepcon \left( \singleton{\ee}{v} \esepimp \ff\subst{x}{v} \right)~.
\end{align*}
Thus, we replaced the 
supremum by an infimum as encountering a memory fault is acceptable. 

\paragraph{The weakest preexpectation landscape}
The individual changes to $\wpsymbol$ can easily be combined.
Thus, apart from $\wpsymbol$, $\awpsymbol$, $\wlpsymbol$, and $\wepsymbol$, we also have transformers $\awlepsymbol$, $\awepsymbol$, $\wlepsymbol$, and $\awlpsymbol$.
How are these transformers related?
As a first observation, we note that for every $\hpgcl$-program $\cc$, we have $\wp{\cc}{0} = 0$ and $\awlep{\cc}{1} = 1$.
%
%
%
For each given initial state $(\sk,\,\hh)$ there are at most four possible outcomes: $\cc$ \emph{diverges}, $\cc$ encounters a memory \emph{fault}, $\cc$ successfully terminates in a state ``captured by $\ff$'', or $\cc$ terminates in some other state, which we denote by $\neg \ff$.
The total probability of these four outcomes is one.
Hence, we can describe the probability of successful termination and measuring $\ff$, i.e.\ $\wp{\cc}{\ff}$, as one minus the probability of the other three events. 
Similar dualities are obtained for all of the possible calculi:
\begin{theorem}[Duality principle for the weakest preexpectation landscape]
\label{thm:wp-duality}
  Let $\cc \in \hpgcl$ be a program.
  Moreover, let $\ff \in \Eone$. Then
  \begin{align*}
          \wp{\cc}{\ff} \eeq & 1 - \awlep{\cc}{1-\ff}~, \tag{probability of $\ff$} \\
          \wlp{\cc}{\ff} \eeq & 1 - \phantom{\textsf{l}}\awep{\cc}{1-\ff}~, \tag{probability of $\ff$ + divergence}\\
          \wep{\cc}{\ff} \eeq & 1 - \phantom{\textsf{e}}\awlp{\cc}{1-\ff}~,~\text{and} \tag{probability of $\ff$ + memory fault} \\
          \wlep{\cc}{\ff} \eeq & 1 - \phantom{\textsf{el}}\awp{\cc}{1-\ff}~. \tag{probability of $\ff$ + divergence + memory fault}
  \end{align*}
\end{theorem}
\REPORT{%
\begin{proof}
        By induction on the structure of $\hpgcl$ programs. See Appendix~\ref{app:wp-duality} for details.
\end{proof}
}%

\section{Beyond \boldhpgcl Programs}
\label{sec:extensions}

We presented our results in terms a simple probabilistic programming language.
Some of the case studies presented in the next section, however, additionally use procedure calls \CHANGED{and sample from discrete uniform probability distributions.}
Let us thus briefly discuss how our $\wpsymbol$ calculus is extended accordingly.%
\footnote{Detailed formalizations and extensions of previous proofs are found in\CAMERA{~\cite{DBLP:journals/corr/abs-1802-10467}}\REPORT{~Appendix~\ref{app:sec:extensions}}.}

%
%
We allow programs $\cc$ to contain \emph{procedure calls} of the form
        %
$\ProcCall{P}{\vec{\ee}}{}$,
        %
where $\ProcName{P}$ is a procedure name and 
$\vec{\ee}$ is a tuple of arithmetic expressions representing the values passed to the procedure.
Since assume parameters are passed by value, no variables are modified by a procedure call, i.e.\ $\Mod{\ProcCall{P}{\vec{\ee}}{}} = \emptyset$.
The meaning of procedure calls is determined by \emph{procedure declarations} of the form
        %
        $\ProcDeclIs{P}{\vec{x}}{}{\textit{body}(P)}$,
        %
where $\textit{body}(P) \in \hpgcl$ is the procedure's body that may contain (recursive) procedure calls
and $\vec{x}$ is a tuple of variables that are never changed by program $\textit{body}(P)$.
All variables in $\textit{body}(P)$ except for its parameters are considered local variables.

For non-recursive procedures, the weakest preexpectation of a procedure call coincides with the weakest preexpectation of its body.
The semantics of recursive procedure calls is determined by a least fixed point of a transformer on procedure environments mapping procedure names and parameters to expectations.
In particular, our previous results, such as linearity of $\wpsymbol$, monotonicity, and the frame rule, remain valid in the presence of recursive procedure calls.

Furthermore, we employ a standard proof rule to deal with recursion (cf.~\cite{DBLP:journals/fac/Hesselink94}):
\begin{align*}
\infer[\text{[rec]}]{
        \forall \vec{\ee} \colon \wp{\ProcCall{P}{\vec{\ee}}{}}{\ff} \ppreceq \inv(\vec{\ee}),
  }{
  \forall \vec{\ee} \colon \wp{\ProcCall{P}{\vec{\ee}}{}}{\ff} \ppreceq \inv(\vec{\ee}) ~\Vdash~ \wp{\textit{body}(\ProcName{P})}{\ff} \ppreceq \inv(\vec{\ee})
  }
\end{align*}
where 
$\ff \in \E$ is a postexpectation, $\inv(\vec{\ee}) \in \E$ is an invariant, and $\vec{\ee}$ is a tuple of expression passed to the called procedure.
Intuitively, for proving that a procedure call satisfies a specification, 
it suffices to show that the procedures body satisfies the specification---assuming that all recursive calls in the procedure's body do so, too.
An analogous rule is obtained for weakest \emph{liberal} preexpectations by replacing all occurrences of $\preceq$ by $\succeq$.
%

\CHANGED{
Moreover, we support sampling from arbitrary discrete distributions instead of flipping coins.
While these sampling instructions, such as $\ASSIGNUNIFORM{x}{\ee}{\ee'}$, which samples an integer in the interval $[\ee,\ee']$ uniformly at random, can be simulated with coin flips, it is more convenient to directly derive their semantics. For example, the weakest preexpecation of the \emph{uniform random assignment} $\ASSIGNUNIFORM{x}{\ee}{\ee'}$ with respect to postexpectation $\ff \in \E$ is given by
\begin{align*}
        \wp{\ASSIGNUNIFORM{x}{\ee}{\ee'}}{\ff} \eeq \lambda(\sk,\hh)\mydot \frac{1}{\sk(\ee')-\sk(\ee)+1} \cdot \sum_{k=\sk(\ee)}^{\sk(\ee')} \ff\subst{x}{k}(\sk,\hh)~.
\end{align*}
}


\section{Case Studies} \label{sec:case-studies}

We examine a few examples---including the programs presented in Section~\ref{sec:intro}---to demonstrate 
\QSL's applicability to reason about probabilities and expected values of \hpgcl programs. 


\subsection{Array Randomization}\label{sec:case-studies:randomize-array}
For our first example, recall the procedure $\ProcName{randomize}(\aarray, n)$ in Section~\ref{sec:intro}, Figure~\ref{fig:randomize-array} that computes a random permutation of an array of size $n$.
%
To conveniently specify subarrays, we use iterated separating conjunctions (cf.~\cite{DBLP:conf/lics/Reynolds02}) given by
\begin{align*}
        \bbigsepcon{k=i}{n} \ff_k \eeq \lambda(\sk,\hh)\mydot
   \begin{cases}
           \left(\ff_{\sk(i)} \sepcon \ff_{\sk(i+1)} \sepcon \ldots \sepcon \ff_{\sk(n)}\right)(\sk,\hh) & ~\text{if}~\sk(i) \leq \sk(n) \\
           \emp(\sk,\hh) &~\text{otherwise.}
   \end{cases}
\end{align*}
Our goal is to show that no particular permutation of the input array has a higher probability than other ones.
Since there are $n!$ permutations of an array of length $n$, we prove that the probability of computing an arbitrary, but fixed, permutation is at most $\sfrac{1}{n!}$.
That is, we compute an upper bound of $\wp{\ProcCall{randomize}{\aarray, n}{}}{\singleton{\aarray}{\za_0,\ldots,\za_{n-1}}}$, 
\REMOVED{where $\za_0,\ldots,\za_{n-1}$ are arbitrary, but fixed, values.}
\CHANGED{where we use variables $\za_0,\ldots,\za_{n-1}$, which do not appear in the program, to keep track of the individual values in the array.}
To this end, we propose the invariant
\begin{align*}
   I \eeq & \iverson{0 \leq i < n} \cdot \frac{1}{(n-i)!} \cdot \bbigsepcon{k=0}{i-1} \singleton{\aarray+k}{\alpha_k}
             \sepcon \sum\limits_{\pi \in \perm{i}{n-1}} \bbigsepcon{k=i}{n-1} \singleton{\aarray+k}{\alpha_{\pi(k)}} \\
          &\quad + \iverson{\neg (0 \leq i < n)} \cdot \singleton{\aarray}{\alpha_0,\ldots, \alpha_{n-1}}
\end{align*}
for the loop $\crand$ in procedure $\ProcName{randomize}$, where $\perm{\ee}{\ee'}$ denotes the set of permutations over $\{\ee,\ee+1,\ldots,\ee'\}$.
Intuitively, $\inv$ describes the situation for $i$ \emph{remaining} loop iterations (since we reason backwards):
All but the first $i$ array elements are already known to be swapped consistently with our fixed permutation. 
In our preexpectation, the last $n-i$ elements are thus arbitrarily permuted and the probability of hitting the right permutation for these elements is $\sfrac{1}{(n-i)!}$.
The remaining $i$ iterations still have to be executed, i.e. the first $i$ array elements coincide with our postexpectation.
\REPORT{A detailed proof that $\inv$ is an invariant of $\crand$ in the sense of Theorem~\ref{thm:upper-invariant} is found in Appendix~\ref{app:case-studies:randomize-array}.}
For the whole procedure $\ProcName{randomize}$ we continue as follows:
\begin{align*}
  & \wp{\ProcCall{randomize}{\aarray, n}{}}{\singleton{\aarray}{\za_0,\ldots,\za_{n-1}}} \\
  \eeq & \wp{\ASSIGN{i}{0}}{\wp{\crand}{\singleton{\aarray}{\za_0,\ldots,\za_{n-1}}}} \tag{Definition of $\wpsymbol$ for procedure body} \\
  \ppreceq & \wp{\ASSIGN{i}{0}}{\inv} \eeq \inv\subst{i}{0} \tag{Theorem~\ref{thm:upper-invariant} for invariant $\inv$, Table~\ref{table:wp}} \\
  \eeq & \frac{1}{n!} \cdot \sum\limits_{\pi \in \perm{0}{n-1}} \bbigsepcon{k=0}{n-1} \singleton{\aarray+k}{\alpha_{\pi(k)}}. \tag{Algebra}
\end{align*}
The probability of computing exactly the permutation $\za_0, \ldots, \za_{n-1}$ is thus at most $\sfrac{1}{n!}$.
Moreover, if the initial heap is not some permutation of our fixed array, the probability becomes $0$.

%


\subsection{Faulty Garbage Collector}
\label{sec:case-studies:faulty-gc}

The next example is a garbage collector that is executed on cheap, but unreliable hardware (cf. Section~\ref{sec:intro}):
Procedure $\ProcName{delete}$ takes a binary tree with root $x$ and recursively deletes all elements in the tree.
However, with some probability $p \in [0,1]$, the condition $x \neq \nil$, which checks whether the tree is empty, is ignored although $x$ is the root of a non-empty tree.
This scenario is implemented by the probabilistic program in Figure~\ref{fig:faulty-gc},
where each node in a tree consists of two consecutive pointers:
$\dereference{\alpha}$ and $\dereference{\alpha+1}$ respectively represent the left and right child of $\alpha$.
%
%
%

Our goal is to establish a lower bound on the probability that the garbage collector successfully deletes the whole tree, i.e. 
$\wlp{\ProcCall{delete}{x}{}}{ \emp }$.
To this end we claim that
\begin{align*}
        \wlp{\ProcCall{delete}{x}{}}{\emp} \ssucceq \Tree{x} \cdot \left(1-p\right)^{\heapSize}~, \tag{$\dag$}
\end{align*}
where $\left(\nicepar{p^{\ff}}\right) (\sk,\, \hh) = p^{\ff(\sk,\hh)}$ for some rational $p$ and $\ff \in \Eone$.
%
%
%
%
The main steps of a proof of our claim are sketched in Figure~\ref{fig:faulty-gc}:
%
Starting with postexpectation $\emp$, step (1) results from applying the $\wlpsymbol$ rule for $\FREE{\cloze{\ee}}$ and the fact that
\begin{align*}
  \validpointer{x} \sepcon \validpointer{x+1} \sepcon \emp \eeq \singleton{x}{-,-} \sepcon \emp~.
\end{align*}
Step (2) deserves special attention.
We would like to apply rule $[rec]$ for recursive procedures (and $\wlpsymbol$) using the premise 
\begin{align*}
        \wlp{\ProcCall{delete}{r}{}}{\emp} \ssucceq t(r) \eeq \Tree{r} \cdot \left(1-p\right)^{\heapSize},
\end{align*}
but the postexpectation is $\singleton{x}{-,-} \sepcon \emp$ instead of $\emp$.
Here, the quantitative frame rule (Theorem~\ref{thm:frame-rules}) allows us to apply the rule $[rec]$ for recursive procedures to postexpectation $\emp$ and derive
$\singleton{x}{-,-} \sepcon t(r)$. 
Notice that the frame rule would not be applicable without the separating conjunction.
In particular, our proof would have to deal with \emph{aliasing}:  
It is not immediate that the heaps reachable from $l$ and $r$ do not share memory.

Step (3) first extends the postexpectation exploiting that $\fh = \fh \sepcon \emp$ for any $\fh \in \Eone$. We then proceed analogously to step (2).
Step (4) is an application of the lookup rule with minor simplifications to improve readability.
Steps (5) and (6) apply $\wlpsymbol$ to the probabilistic choice and the conditional.
Finally, we show that (6) is entailed by the expectation in step (7), i.e.\ $(6) \succeq (7)$.
\REPORT{A detailed proof 
is found in Appendix~\ref{app:tree-delete}.}

%
%
\begin{figure}[t]
\begin{align*}
& \ProcDecl{delete}{x}{} \Assert{ \Tree{x} \cdot \left(1-p\right)^{\heapSize} } \tag{7} \\
        & \quad \Assert{ \iverson{x \neq \nil} \cdot \left(p \cdot \emp + (1-p) \cdot f\right)+ \iverson{x = \nil} \cdot \emp } \tag{6} \\
        & \quad \IF{x \neq \nil} \Assert{p \cdot \emp + (1-p) \cdot f} \tag{5} \\
        & \qquad \{\, \SKIP \,\}\,\mathrel{\left[ p \right]}\,\{ \Assert{ \sup_{\alpha,\beta} \singleton{x}{\alpha,\beta} \sepcon \left( \singleton{x}{\alpha,\beta} \sepimp g\subst{l,r}{\alpha,\beta} \right) ~=: f } \tag{4} \\
            & \quad\qquad \ASSIGNH{l}{x}\SEMI\ASSIGNH{r}{x+1}\SEMI \Assert{ \singleton{x}{-,-} \sepcon t(l) \sepcon t(r) ~=: g } \tag{3} \\
            & \quad\qquad \ProcCall{delete}{l}{}\SEMI \Assert{ \singleton{x}{-,-} \sepcon t(r) } \tag{2} \\
            & \quad\qquad \ProcCall{delete}{r}{}\SEMI \Assert{ \singleton{x}{-,-} \sepcon \emp } \tag{1}\\
            & \quad\qquad \FREE{x}\SEMI\, \FREE{x+1}\,\} \Assert{\emp} \\
    & \quad \ELSE \SKIP \}~ \} \Assert{ \emp  } \\
    & \} \Assert{\emp}
\end{align*}
\caption{Faulty garbage collection procedure with a proof sketch, where $t(\alpha) \eeq \Tree{\alpha} \cdot \left(1-p\right)^{\heapSize}$.}
\label{fig:faulty-gc}
\end{figure}
%


\subsection{Lossy List Reversal}\label{sec:case-studies:lossy-reversal}

We analyze the lossy list reversal presented in Section~\ref{sec:intro}, Figure~\ref{fig:lossy-list-reversal}.
Our goal is to obtain an upper bound on the \emph{expected length of the reversed list} after successful termination, i.e. we compute an upper bound of
$\wp{\ProcCall{\csPLossyReversal}{\lsHead}{}}{\Len{\lsRev}{0}}$.
To this end, we propose the invariant
\begin{align*}
        \inv \eeq \Len{\lsRev}{0} \sepcon \Ls{\lsHead}{0} + \sfrac{1}{2} \cdot \iverson{\lsHead \neq 0} \cdot \left( \Len{\lsHead}{0} \sepcon \Ls{\lsRev}{0} \right).
\end{align*}
%
%
Intuitively, invariant $\inv$ states that during each loop iteration, the expected length of the list with head $\lsRev$ is its current length, i.e. $\Len{\lsRev}{0}$, plus half of the length of the remaining list with head $\lsHead$, i.e. $\Len{\lsHead}{0}$. 
To obtain a tight specification, i.e. describe the exact content of the heap, we additionally use predicates $\Ls{\lsHead}{0}$ and $\Ls{\lsRev}{0}$ to cover the remaining parts of the heap when measuring the length of a list.
\REPORT{A detailed proof that $\inv$ is an invariant in the sense of Theorem~\ref{thm:upper-invariant} is found in Appendix~\ref{app:case-studies:lossy-reversal}, p.~\pageref{app:case-studies:lossy-reversal}.}
We then continue as follows:
\begin{align*}
        & \wp{\ProcCall{\csPLossyReversal}{\lsHead}{}}{\Len{\lsRev}{0}} \\
        \eeq & \wp{\ASSIGN{\lsRev}{0}}{\wp{\WHILEDO{\lsHead \neq 0}{\ldots}}{\Len{\lsRev}{0}}}
        \tag{Definition of procedure body, Table~\ref{table:wp}} \\
        \ppreceq & \wp{\ASSIGN{\lsRev}{0}}{\inv} 
        \tag{Theorem~\ref{thm:upper-invariant}} \\
        \eeq & \underbrace{\Len{0}{0}}_{\eeq 0} \sepcon \Ls{\lsHead}{0} + \sfrac{1}{2} \cdot \iverson{\lsHead \neq 0} \cdot ( \Len{\lsHead}{0} \sepcon \underbrace{\Ls{0}{0}}_{\eeq 1} )
        \tag{Def. of $\inv$, Table~\ref{table:wp}} \\
        \eeq & \sfrac{1}{2} \cdot \iverson{\lsHead \neq 0} \cdot \Len{\lsHead}{0}.
\end{align*}
Hence, the expected length of the reversed list after successful termination is at most half of the length of the original list. 
%



\subsection{Randomized List Extension}\label{sec:list-length}

As a last example, we consider a program $\cc_{\textrm{list}}$ that inserts new elements at the beginning of a list with head $x$, but gradually loses interest in adding further elements:
%
%
\begin{align*}
        \cc_{\textrm{list}}\colon\quad
        \ASSIGN{c}{1}\SEMI\WHILEDO{c = 1}{\PCHOICE{\ASSIGN{c}{0}}{\sfrac{1}{2}}{\ASSIGN{c}{1}\SEMI\ALLOC{x}{x}}}
\end{align*}
%
Our goal is to compute an upper bound on the expected length of the list with head $x$ after termination of program $\cc_{\textrm{list}}$, i.e. we compute an upper bound of $\wp{\cc_{\textrm{list}}}{\Len{x}{\nil}}$.
To this end, we propose the loop invariant
%
$\inv = \Len{x}{\nil} + \iverson{c = 1}$,
which states that the length of the list is increased by one if variable $c$ equals one.
\REPORT{A detailed proof that $\inv$ is an invariant in the sense of Theorem~\ref{thm:upper-invariant} is found in Appendix~\ref{app:list-length}.}
For the full program we proceed as follows:
\begin{align*}
& \wp{\cc_{\textrm{list}}}{\Len{x}{\nil}} \\
\eeq & \wp{\ASSIGN{c}{1}}{\wp{\WHILEDO{c=1}{\ldots}}{\Len{x}{\nil}}} \tag{Definition of $\wpsymbol$} \\
\ppreceq & \wp{\ASSIGN{c}{1}}{\inv} \eeq \inv\subst{c}{1} \tag{Theorem~\ref{thm:upper-invariant} for invariant $\inv$} \\
\eeq & \Len{x}{\nil} + 1. \tag{Algebra}
\end{align*}
Hence, in expectation, program $\cc_{\textrm{lists}}$ increases the length of the initial list by at most one element.

\section{Related Work}\label{sec:related-work}

Although many algorithms rely on randomized data structures,
formal reasoning about probabilistic programs that mutate memory has received scarce attention.
To the best of our knowledge, there is little other work on formal verification of programs that are both probabilistic and heap manipulating. 
A notable exception is recent work by~\cite{tassarotti2017separation}
who combine concurrent separation logic with probabilistic relational Hoare logic (cf.~\cite{DBLP:conf/mpc/BartheGB12}).
Their focus is on program refinement. 
Verification is thus understood as establishing a relation between a program to be analyzed and a program which is known to be well-behaved.
In contrast to that, the goal of our logic is to directly measure quantitative program properties on source code level using a weakest-precondition style calculus.
\CHANGED{In particular, programs that do not \emph{certainly} terminate, e.g. the list extension example in Section~\ref{sec:list-length}, are outside the scope of their approach (cf.~\cite[Theorem 3.1]{tassarotti2017separation}). Furthermore, they do not consider unbounded expectations.}


\paragraph{Probabilistic program verification.}
Seminal work on semantics and verification of probabilistic programs is due to~\cite{DBLP:conf/focs/Kozen79,DBLP:conf/stoc/Kozen83}. 
\cite{DBLP:series/mcs/McIverM05,DBLP:journals/toplas/MorganMS96} developed the weakest preexpectation calculus to reason about a probabilistic variant of Dijkstra's guarded command language.
While variants of their calculus have been successfully applied to programs that access data structures, 
such as the coupon collector's problem~\cite{DBLP:conf/esop/KaminskiKMO16} and a probabilistic binary search~\cite{DBLP:conf/lics/OlmedoKKM16}, 
treatment of data structures is usually added in an ad--hoc manner.
In particular, proofs quickly get extremely complicated if programs do not only access but also mutate a data structure.
Our work extends the calculus of McIver and Morgan to formally reason about heap manipulating probabilistic programs.

\paragraph{Separation Logic.}
Apart from the backward reasoning rules in~\cite{DBLP:conf/popl/IshtiaqO01,DBLP:conf/lics/Reynolds02}, weakest preconditions are extensively used by~\cite{DBLP:conf/popl/KrebbersTB17}.
For ordinary 
programs, our calculus allows for reasoning about quantities of heaps, such as the length of lists. 
Such shape--numeric properties have been investigated before, see, e.g., \cite{DBLP:conf/popl/ChangR08, DBLP:journals/jar/BozgaIP10}.
\cite{DBLP:journals/scp/ChinDNQ12} use recursive predicate definitions together with fold/unfold reasoning to verify properties, such as balancedness of trees.
Furthermore,~\cite{DBLP:journals/corr/abs-1104-1998} developed a proof logic that combines separation logic with reasoning about consumable resources.
\CHANGED{%
His work supports reasoning about quantities by means of special predicates that that are evaluated by one or more resources in addition to the heap.
However, the amount of resources must be bounded. 
It is unclear how this approach can be extended to reason about expected values of probabilistic programs.
}
\REMOVED{In contrast to our approach, these logics evaluate to Boolean values, while our logic truly rests in a quantitative setting.}



\section{Conclusion}\label{sec:conclusion}

We presented \QSL{} --- a quantitative separation logic that evaluates to real numbers instead of truth values.
Our $\wpsymbol$ calculus built on top of \QSL is a conservative extension of both 
separation logic and Kozen's / McIver and Morgan's weakest preexpectations.
In particular, virtually all properties of separation logic remain valid.
We applied \QSL to reason about four examples, ranging from the success probability of a faulty garbage collector, over the expected list length of a list reversal algorithm to a textbook procedure to randomize arrays.

Our calculus provides a foundation for formal reasoning about randomized algorithms on source code level.
Future work includes developing \emph{proof systems for quantitative entailments} and analyzing more involved algorithms, e.g.~randomized skip lists or randomized splay trees.

\begin{acks}
We are grateful for the valuable and very constructive comments we received from the anonymous reviewers. This applies particularly to the formulation of Theorems~\ref{thm:qsl:conservativity:language} and~\ref{thm:qsl:conservativity:wp}.

Furthermore, we acknowledge the support of this work by DFG research training group 2236 UnRAVeL and by DFG grant NO 401/2-1.
\end{acks}


\bibliography{bibfile}

\appendix

\newpage
\section*{Appendix}
\label{app}
The appendix contains all missing proofs ordered by occurrence of the respective theorem in the main part of the paper.
More precisely,
\begin{itemize}
  \item Appendix~\ref{app:sec:qsl} contains all proofs regarding \QSL as a logical language, 
  \item Appendix~\ref{app:sec:wp} contains all proofs regarding weakest preexpectations with \QSL, 
  \item Appendix~\ref{app:sec:extensions} is concerned with extensions of our programming language \hpgcl by recursion and more general probabilistic assignments.
  \item Appendix~\ref{app:sec:case-studies} contains all proofs regarding the case studies in Section~\ref{sec:case-studies}.
  \item Appendix~\ref{app:sec:misc} contains additional simple inference rules for computing with expectations in \QSL.
\end{itemize}

\renewcommand{\contentsname}{Appendix}
\tableofcontents
\addtocontents{toc}{\protect\setcounter{tocdepth}{2}}

\newpage
\section{Appendix to Section~\ref{sec:qsl} (Quantitative Separation Logic)}
\label{app:sec:qsl}
%

\subsection{Backward Compatibility of Separating Conjunction}
\label{app:back-comp-sep-con}

\begin{theorem}\label{thm:back-comp-sep-con}
    For \SL predicates $\preda,\predb$, we have
    \begin{enumerate}
        \item $\left(\iverson{\preda} \sepcon \iverson{\predb}\right)(\sk,\hh) \in \{0,1\}$, and
        \item $\left( \iverson{\preda} \sepcon \iverson{\predb} \right)(\sk,\hh) \eeq 1$ holds in \QSL if and only if $(\sk,\hh) \models \preda \sepcon \predb$ holds in \SL.
    \end{enumerate}
\end{theorem}

\begin{proof}\label{proof:thm:back-comp-sep-con}
	For the first claim, consider the following:
	\begin{align}
		\bigl(\iverson{\preda} \sepcon \iverson{\predb}\bigr) (\sk, \hh) \eeq \underbrace{\max_{\hh_1, \hh_2} \underbrace{\bigl\{\underbrace{\underbrace{\iverson{\preda}(\sk, \hh_1)}_{{} \in \{0,\, 1\}} \cdot \underbrace{\iverson{\predb}(\sk, \hh_2)}_{{} \in \{0,\, 1\}}}_{{} \in \{0,\, 1\}} ~\big|~ \hh = \hh_1 \sepcon \hh_2 \bigr\}}_{{} \in \pot{\{0,\, 1\}} \setminus \emptyset}}_{{} \in \{0,\, 1\}} \iin \{0,\, 1\}~.
	\end{align}
    For the second claim, assume for all stack-heap pairs $(\sk,\hh)$ that 
    $\iverson{\preda}(\sk,\hh) = 1$ iff $(\sk,\hh) \models \preda$ 
    and
    $\iverson{\predb}(\sk,\hh) = 1$ iff $(\sk,\hh) \models \predb$.
    Then 
    \begin{align}
        & (\sk,\hh) \models \preda \sepcon \predb \\
        \leftrighttag{Definition of $\sepcon$ in \SL}
        & \exists \hh_1,\hh_2 \colon \hh = \hh_1 \sepcon \hh_2  \qand (\sk,\hh_1) \mmodels \preda \qand (\sk,\hh_2) \mmodels \predb \\
        \leftrighttag{assumption}
        & \exists \hh_1,\hh_2 \colon \hh = \hh_1 \sepcon \hh_2  \qand \iverson{\preda}(\sk,\hh_1) = 1 \qand \iverson{\predb}(\sk,\hh_2) = 1 \\
        \leftrighttag{$\preda,\predb$ are predicates}
        & \max_{\hh_1,\hh_2} \setcomp{\iverson{\preda}(\sk,\hh_1) \cdot \iverson{\predb}(\sk,\hh_2)}{\hh = \hh_1 \sepcon \hh_2} \eeq 1  \\
        \leftrighttag{Definition of $\sepcon$ in \QSL} 
        & \left(\iverson{\preda} \sepcon \iverson{\predb}\right)(\sk,\hh) \eeq 1. 
    \end{align}
\end{proof}

\subsection{Conservativity of \QSL as an assertion language}\label{app:qsl:conservativity:language}

\begin{proof}[Proof of Theorem~\ref{thm:qsl:conservativity:language}.\ref{thm:qsl:conservativity:language:0-1}]
Our goal is to show that for all classical separation logic formulas $\preda \in \SL$ and all states $(\sk,\hh) \in \States$, we have
\begin{align}
    \qslemb{\preda}(\sk,\hh) \in \{0,1\}~.
\end{align}
    By induction on the structure of the syntax of formulas in \SL.

    \emph{The case of atomic formulas $\preda \in \SL$:}
    By Definition of the embedding of \SL into \QSL, we have $\qslemb{\preda} = \iverson{\preda} \in \{0,1\}$.

%
%
%

    \emph{The case $\exists x\colon \preda$:}
    \begin{align}
            & \qslemb{\exists x\colon \preda} \\ 
            \eeqtag{applying embedding of \SL into \QSL}
            & \sup_{v \in \Ints}\, \underbrace{\qslemb{\preda}}_{~\in~ \{0,1\}~\text{by I.H.}}\subst{x}{v} \in \{0,1\}~.
    \end{align}

    \emph{The case $\preda_1 \wedge \preda_2$:}
    \begin{align}
            & \qslemb{\preda_1 \wedge \preda_2} \\ 
            \eeqtag{applying embedding of \SL into \QSL}
            & \underbrace{\qslemb{\preda_1}}_{\in \{0,1\}~\text{by I.H.}} ~\cdot~ \underbrace{\qslemb{\preda_2}}_{\in \{0,1\}~\text{by I.H.}} \in \{0,1\}~.
    \end{align}

    \emph{The case $\neg \preda$:}
    \begin{align}
            & \qslemb{\neg \preda} \\ 
            \eeqtag{applying embedding of \SL into \QSL}
            & 1 - \underbrace{\qslemb{\preda}}_{\in \{0,1\}~\text{by I.H.}} \in \{0,1\}~.
    \end{align}
    
    \emph{The case $\preda_1 \sepcon \preda_2$:}
    This is immediate by Theorem~\ref{thm:back-comp-sep-con}, p.~\pageref{app:back-comp-sep-con}.
    
    \emph{The case $\preda_1 \sepimp \preda_2$:}
        \begin{align}
                & \left(\iverson{\preda} \sepimp \iverson{\predb}\right)(\sk,\hh) \\
                \eeqtag{Definition of $\sepimp$}
                & \inf_{\hh'} \left\{ \iverson{\predb}(\sk,\hh \sepcon \hh') ~\middle|~ \hh \disjoint \hh', \sk,\hh' \models \preda \right\} \\
                ~\in~ & \itag{$\iverson{\predb}$ is a predicate and the domain is restricted to $\Eone$, i.e. $\inf \emptyset = 1$} \\
                & \{0,1\}.
        \end{align}
\end{proof}

\begin{proof}[Proof of Theorem~\ref{thm:qsl:conservativity:language}.\ref{thm:qsl:conservativity:language:equivalence}]
    Our goal is to show that for all classical separation logic formulas $\preda \in \SL$ and all states $(\sk,\hh) \in \States$, we have
    \begin{align}
        (\sk,\hh) \models \varphi \quad\text{iff}\quad \qslemb{\preda}(\sk,\hh) \eeq 1~.
    \end{align}
    By induction on the structure of the syntax of formulas in \SL.
    
    \emph{The case of atomic formulas $\preda \in \SL$:}
    By Definition of the embedding of \SL into \QSL, we have $\qslemb{\preda} = \iverson{\preda}$.
    Then $(\sk,\hh) \models \preda$ holds if and only if $\iverson{\preda}(\sk,\hh) = 1$, because
    \begin{align}
        \iverson{\preda}(\sk,\hh) \eeq \begin{cases} 1 & ~\text{if}~(\sk,\hh) \models \preda \\ 0 & ~\text{otherwise.} \end{cases}
    \end{align}

    \emph{The case $\exists x\colon \preda$:}
    \begin{align}
            & (\sk,\hh) \models \exists x\colon \preda \\
            \leftrighttag{SL semantics} 
            & \exists x \in \Ints\colon (\sk\subst{x}{v},\hh) \models \preda \\
            \leftrighttag{I.H.} 
            & \exists v \in \Ints\colon \qslemb{\preda}\subst{x}{v}(\sk,\hh) \eeq 1 \\
            \leftrighttag{$\qslemb{\preda} \in \{0,1\}$} 
            & \sup_{v \in \Ints} \qslemb{\preda}\subst{x}{v}(\sk,\hh) \eeq 1 \\
            \leftrighttag{applying embedding of \SL into \QSL}
            & \qslemb{\exists x\colon \preda}(\sk,\hh) \eeq 1~.
    \end{align}

    \emph{The case $\preda_1 \wedge \preda_2$:}
    \begin{align}
            & (\sk,\hh) \models \preda_1 \wedge \preda_2 \\
            \leftrighttag{SL semantics} 
            & (\sk,\hh) \models \preda_1 ~\text{and}~ (\sk,\hh) \models \preda_2 \\
            \leftrighttag{I.H.} 
            & \qslemb{\preda_1}(\sk,\hh) \eeq 1 ~\text{and}~ \qslemb{\preda_1}(\sk,\hh) \eeq 1 \\
            \leftrighttag{$\qslemb{\preda_1},\qslemb{\preda_2} \in \{0,1\}$} 
            & \qslemb{\preda_1}(\sk,\hh) \cdot \qslemb{\preda_1}(\sk,\hh) \eeq 1 \\
            \leftrighttag{applying embedding of \SL into \QSL}
            & \qslemb{\preda_1 \wedge \preda_2}(\sk,\hh) \eeq 1~.
    \end{align}

    \emph{The case $\neg \preda$:}
    \begin{align}
            & (\sk,\hh) \models \neg \preda \\
            \leftrighttag{SL semantics}
            & (\sk,\hh) \not \models \preda \\
            \leftrighttag{I.H.}
            & \preda(\sk,\hh) \neq 1 \\
            \leftrighttag{Theorem~\ref{thm:qsl:conservativity:language}.\ref{thm:qsl:conservativity:language:0-1}}
            & \preda(\sk,\hh) \eeq 0 \\
            \leftrighttag{algebra}
            & 1 - \preda(\sk,\hh) \eeq 1 \\
            \leftrighttag{applying embedding of \SL into \QSL}
            & \qslemb{\neg \preda}(\sk,\hh) \eeq 1~.
    \end{align}

    \emph{The case $\preda_1 \sepcon \preda_2$:}
    This is immediate by Theorem~\ref{thm:back-comp-sep-con}, p.~\pageref{app:back-comp-sep-con}.

    \emph{The case $\preda_1 \sepimp \preda_2$:}
        \begin{align}
                & \left(\iverson{\preda} \sepimp \iverson{\predb}\right)(\sk,\hh) \eeq 1 \\
                \leftrighttag{Definition of $\sepimp$ in \QSL}
                & \inf_{\hh'} \left\{ \iverson{\predb}(\sk,\hh \sepcon \hh') ~\middle|~ \hh \disjoint \hh', \sk,\hh' \models \preda \right\} \eeq 1\\
                \leftrighttag{$\iverson{\predb}(\sk,\hh \sepcon \hh') = 1$ iff $\sk, \hh \sepcon \hh' \models \predb$}
                & \underbrace{\inf_{\hh'} \left\{ \iverson{\sk, \hh \sepcon \hh' \models \predb} ~\middle|~ \hh \disjoint \hh', \sk,\hh \models \preda \right\}}_{\eeq \ff} \eeq 1 \\
                \leftrighttag{ $\ff = 0$ iff exists $\hh'$ s.t. $\hh \disjoint \hh'$ and $\sk,\hh \models \preda$ and $\sk,\hh \sepcon \hh' \not\models \predb$ }
                & \neg \exists \hh' : \hh \disjoint \hh' ~\text{and}~ \sk,\hh\models \preda ~\text{and}~ \sk,\hh \sepcon \hh' \not\models \predb \\
                \leftrighttag{pushing negation inside}
                & \forall \hh' : \neg \hh \disjoint \hh' ~\text{or}~ \sk,\hh\not\models \preda ~\text{or}~ \sk,\hh \sepcon \hh' \models \predb \\
                \leftrighttag{ first-order logic}
                & \forall \hh' : (\hh \disjoint \hh' ~\text{and}~ \sk,\hh\models \preda) ~\text{implies}~ \sk,\hh \sepcon \hh' \models \predb \\
                \leftrighttag{Definition of $\sepimp$ in $\SL$}
                & (\sk,\hh) \models \preda \sepimp \predb.
        \end{align}

\end{proof}


\subsection{Proof of Theorem~\ref{thm:sep-con-monoid} (Monoid Properties)}
\label{app:sep-con-monoid}

\begin{proof}\label{proof:thm:sep-con-monoid}
    (1). For \emph{associativity}, consider the following:
	\begin{align}
		& \bigl( \ff \sepcon ( \fg \sepcon \fk ) \bigr) (\sk, \hh)\\
        \eeqtag{Definition of $\sepcon$}
		& \max_{\hh_1, \hh_2}~ \setcomp{\ff(\sk, \hh_1) \cdot \max_{\hh_{21}, \hh_{22}}~ \setcomp{\fg(\sk, \hh_{21}) \cdot \fk(\sk, \hh_{22})}{\hh_2 = \hh_{21} \sepcon \hh_{22}}}{\hh = \hh_1 \sepcon \hh_2}\\
        \eeqtag{algebra}
		& \max_{\hh_1, \hh_2, \hh_3}~ \setcomp{\ff(\sk, \hh_1) \cdot \fg(\sk, \hh_2) \cdot \fk(\sk, \hh_3)}{\hh = \hh_1 \sepcon \hh_2 \sepcon \hh_3}\\
        \eeqtag{algebra}
		& \max_{\hh_1, \hh_2}~ \setcomp{\max_{\hh_{11}, \hh_{12}}~ \setcomp{\ff(\sk, \hh_{11}) \cdot \fg(\sk, \hh_{12})}{\hh_1 = \hh_{11} \sepcon \hh_{12}} \cdot \fk(\sk, \hh_2)}{\hh = \hh_1 \sepcon \hh_2}\\
        \eeqtag{Definition of $\sepcon$}
		& \eeq \bigl( ( \ff \sepcon \fg ) \sepcon \fk \bigr) (\sk, \hh).
	\end{align}
    (2). For \emph{neutrality of $\boldsymbol{\emp}$}, recall that $\hh \sepcon \emptyheap = \hh$ and consider the following:
	\begin{align}
		& \bigl( \ff \sepcon \emp \bigr) (\sk, \hh) \\
        \eeqtag{Definition of $\sepcon$}
		& \max_{\hh_1, \hh_2}~ \setcomp{\ff(\sk, \hh_1) \cdot \emp(\sk, \hh_2)}{\hh = \hh_1 \sepcon \hh_2}\\
        \eeqtag{by $\hh = \hh \sepcon \emptyheap$ and $\emp(\sk, \hh_2) = 0$ if $\hh_2 \neq \emptyheap$}
		& \ff(\sk, \hh) \cdot \emp(\sk, \emptyheap) \\
        \eeqtag{by commutativity, see below}
        & \emp(\sk, \emptyheap) \cdot \ff(\sk,\hh) \\
        \eeqtag{by $\emp(\sk, \emptyheap) = 1$}
		& 1 \cdot \ff(\sk, \hh) \\
        \eeqtag{algebra}
		& \ff(\sk, \hh).
	\end{align}
    (3). For \emph{commutativity}, consider the following:
	\begin{align}
        & \bigl( \ff \sepcon \fg \bigr) (\sk, \hh) \\
        \eeqtag{Definition of $\sepcon$}
        & \max_{\hh_1, \hh_2}~ \setcomp{\ff(\sk, \hh_1) \cdot \fg(\sk, \hh_2)}{\hh = \hh_1 \sepcon \hh_2}\\
        \eeqtag{algebra}
		& \max_{\hh_2, \hh_1}~ \setcomp{\fg(\sk, \hh_2) \cdot \ff(\sk, \hh_1)}{\hh = \hh_2 \sepcon \hh_1}\\
        \eeqtag{Definition of $\sepcon$}
		& \bigl( \fg \sepcon \ff \bigr) (\sk, \hh).
	\end{align}
\end{proof}
%

%
\subsection{Proof of Theorem~\ref{thm:sep-con-distrib} (Laws for Separating Conjunction)}
\label{app:sep-con-distrib}
%
%

\begin{proof}[Proof of Theorem~\ref{thm:sep-con-distrib}.\ref{thm:sep-con-distrib:sepcon-over-max}]

    For \emph{distributivity of $\sepcon$ over $\max$}, consider the following:
	\begin{align}
		& \bigl( \ff \sepcon \Max{\fg}{\fh} \bigr) (\sk, \hh) \\
        \eeqtag{Definition of $\sepcon$}
		& \max_{\hh_1, \hh_2}~ \setcomp{\ff(\sk, \hh_1) \cdot \bigl( \Max{\fg}{\fh} \bigr) (\sk, \hh_2)}{\hh = \hh_1 \sepcon \hh_2}\\
        \eeqtag{Definition of $\max$}
		& \max_{\hh_1, \hh_2}~ \setcomp{\ff(\sk, \hh_1) \cdot \Max{\fg(\sk, \hh_2)}{\fh(\sk, \hh_2)}}{\hh = \hh_1 \sepcon \hh_2}\\
        \eeqtag{algebra, $\ff(\sk,\hh_1) \in \Reals$}
		& \max_{\hh_1, \hh_2}~ \setcomp{\Max{\ff(\sk, \hh_1) \cdot \fg(\sk, \hh_2)}{\ff(\sk, \hh_1) \cdot \fh(\sk, \hh_2)}}{\hh = \hh_1 \sepcon \hh_2}\\
        \eeqtag{algebra}
		& \Max{\max_{\hh_1, \hh_2}~ \setcomp{\ff(\sk, \hh_1) \cdot \fg(\sk, \hh_2)}{\hh = \hh_1 \sepcon \hh_2}}{\max_{\hh_1', \hh_2'}~ \setcomp{\ff(\sk, \hh_1') \cdot \fh(\sk, \hh_2')}{\hh = \hh_1' \sepcon \hh_2'}}\\
        \eeqtag{Definition of $\sepcon$}
		& \Max{ \bigl( \ff \sepcon \fg \bigr) (\sk, \hh) }{ \bigl( \ff \sepcon \fh \bigr) (\sk, \hh) }\\
        \eeqtag{algebra}
		& \bigl( \Max{\ff \sepcon \fg }{ \ff \sepcon \fh } \bigr) (\sk, \hh).
	\end{align}
\end{proof}

\begin{proof}[Proof of Theorem~\ref{thm:sep-con-distrib}.\ref{thm:sep-con-distrib:sepcon-over-plus}]
    For \emph{sub-distributivity of $\sepcon$ over $+$}, consider the following:
	\begin{align}
		& \bigl( \ff \sepcon (\fg + \fh) \bigr) (\sk, \hh) \\
        \eeqtag{Definition of $\sepcon$}
		& \max_{\hh_1, \hh_2}~ \setcomp{\ff(\sk, \hh_1) \cdot \bigl( \fg + \fh \bigr) (\sk, \hh_2)}{\hh = \hh_1 \sepcon \hh_2}\\
        \eeqtag{distributivity of $\cdot$ and $+$}
		& \max_{\hh_1, \hh_2}~ \setcomp{\ff(\sk, \hh_1) \cdot \fg(\sk, \hh_2) + \ff(\sk, \hh_1) \cdot \fh(\sk, \hh_2)}{\hh = \hh_1 \sepcon \hh_2}\\
        \lleqtag{triangle inequality}
		& \max_{\hh_1, \hh_2}~ \setcomp{\ff(\sk, \hh_1) \cdot \fg(\sk, \hh_2)}{\hh = \hh_1 \sepcon \hh_2} + \max_{\hh_1', \hh_2'}~ \setcomp{\ff(\sk, \hh_1') \cdot  \fh(\sk, \hh_2')}{\hh = \hh_1' \sepcon \hh_2'}\\
        \eeqtag{Definition of $\sepcon$}
		& \bigl( \ff \sepcon \fg \bigr) (\sk, \hh) + \bigl( \ff \sepcon \fh \bigr) (\sk, \hh)\\
        \eeqtag{algebra}
		& \bigl( \ff \sepcon \fg +  \ff \sepcon \fh \bigr) (\sk, \hh).
	\end{align}
\end{proof}
\begin{proof}[Proof of Theorem~\ref{thm:sep-con-distrib}.\ref{thm:sep-con-distrib:sepcon-over-times}]
    For \emph{restricted sub-distributivity of $\sepcon$ over $\cdot$}, consider the following:
	\begin{align}
        & \bigl( \iverson{\varphi} \sepcon (\fg \cdot \fh) \bigr) (\sk, \hh) \\
        \eeqtag{Definition of $\sepcon$}
		& \max_{\hh_1, \hh_2}~ \setcomp{\iverson{\varphi}(\sk, \hh_1) \cdot \bigl( \fg \cdot \fh \bigr) (\sk, \hh_2)}{\hh = \hh_1 \sepcon \hh_2}\\
        \eeqtag{Definition of $\cdot$}
		& \max_{\hh_1, \hh_2}~ \setcomp{\iverson{\varphi}(\sk, \hh_1) \cdot \fg(\sk, \hh_2) \cdot \fh(\sk, \hh_2)}{\hh = \hh_1 \sepcon \hh_2}\\
        \eeqtag{$\iverson{\varphi}(\sk,\hh_1) \in \{0,1\}$}
		& \max_{\hh_1, \hh_2}~ \setcomp{\iverson{\varphi}(\sk, \hh_1) \cdot \iverson{\varphi}(\sk, \hh_1) \cdot \fg(\sk, \hh_2) \cdot \fh(\sk, \hh_2)}{\hh = \hh_1 \sepcon \hh_2}\\
        \lleqtag{triangle inequality}
		& \max_{\hh_1, \hh_2}~ \setcomp{\iverson{\varphi}(\sk, \hh_1) \cdot \fg(\sk, \hh_2)}{\hh = \hh_1 \sepcon \hh_2} \cdot \max_{\hh_1', \hh_2'}~ \setcomp{\iverson{\varphi}(\sk, \hh_1') \cdot  \fh(\sk, \hh_2')}{\hh = \hh_1' \sepcon \hh_2'} \\
        \eeqtag{Definition of $\sepcon$}
		& \bigl( \iverson{\varphi} \sepcon \fg \bigr) (\sk, \hh) \cdot \bigl( \iverson{\varphi} \sepcon \fh \bigr) (\sk, \hh)\\
        \eeqtag{algebra}
		& \Bigl( \bigl( \iverson{\varphi} \sepcon \fg \bigr) \cdot  \bigl( \iverson{\varphi} \sepcon \fh \bigr) \Bigr) (\sk, \hh).
	\end{align}
\end{proof}
\begin{lemma}\label{thm:domain-exact:unique-partitioning}
    Let $\ff \in \E$ domain-exact and $(\sk,\hh)$ be a stack-heap pair. Moreover, let
    \begin{align*}
            \HeapPartitions{\ff}{\sk}{\hh} \eeq \{ (\hh_1,\hh_2) ~|~ \hh = \hh_1 \sepcon \hh_2 ~\text{and}~ \ff(\sk,\hh_1) > 0 \}.
    \end{align*}
    Then $|\HeapPartitions{\ff}{\sk}{\hh}| \leq 1$.
\end{lemma}
\begin{proof}
        By definition, $(\hh_1,\hh_2) \in \HeapPartitions{\ff}{\sk}{\hh}$ implies $\hh_1,\hh_2 \subseteq \hh$.
        Moreover, for a fixed heap $\hh_1$, the corresponding heap $\hh_2$ is uniquely determined by $\hh = \hh_1 \sepcon \hh_2$.
        We distinguish two cases.

        First, assume there exists a heap $\hh' \subseteq \hh$ such that $\ff(\sk,\hh') > 0$, i.e. $|\HeapPartitions{\ff}{\sk}{\hh}| \geq 1$.
        This heap corresponds to heap $\hh$ restricted to $\dom{\hh'}$.
        Now assume there exists another heap $\hh'' \neq \hh'$ with $\hh'' \subseteq \hh$ such that $\ff(\sk,\hh'') > 0$.
        Since $\ff$ is domain-exact, we have $\dom{\hh''} = \dom{\hh'}$. 
        Then the restriction of heap $\hh$ to domain $\dom{\hh''} = \dom{\hh'}$ yields the heap $\hh'$, which contradicts our assumption. 
        Hence, $\HeapPartitions{\ff}{\sk}{\hh} = 1$.

        Second , assume there exists \emph{no} heap $\hh' \subseteq \hh$ such that $\ff(\sk,\hh') > 0$.
        Then $\HeapPartitions{\ff}{\sk}{\hh} = \emptyset$ and thus $|\HeapPartitions{\ff}{\sk}{\hh}| = 0$.
\end{proof}
\begin{proof}[Proof of Theorem~\ref{thm:sep-con-distrib}.\ref{thm:sep-con-distrib:sepcon-over-plus-full}]
    For \emph{domain-restricted distributivity of $\sepcon$ over $+$}, consider the following:
    \begin{align}
            & \ff \sepcon (\fg + \fh) \\
            \eeqtag{Definition of $\sepcon$} 
            & \lambda(\sk,\hh)\mydot \max_{\hh_1, \hh_2} \setcomp{\ff(\sk,\hh_1) \cdot (\fg(\sk,\hh_2) + \fh(\sk,\hh_2))}{\hh = \hh_1 \sepcon \hh_2} \\
            \eeqtag{algebra} 
            & \lambda(\sk,\hh)\mydot \max_{\hh_1,\hh_2} \setcomp{\ff(\sk,\hh_1) \cdot \fg(\sk,\hh_2) + \ff(\sk,\hh_1) \cdot \fh(\sk,\hh_2)}{\hh = \hh_1 \sepcon \hh_2} \\
            \eeqtag{Lemma~\ref{thm:domain-exact:unique-partitioning}, maximum over singleton} 
            & \lambda(\sk,\hh)\mydot \max_{\hh_1, \hh_2} \setcomp{\ff(\sk,\hh_1) \cdot \fg(\sk,\hh_2)}{\hh = \hh_1 \sepcon \hh_2} \\
            & \qquad + \max_{\hh_1, \hh_2} \setcomp{\ff(\sk,\hh_1) \cdot \fh(\sk,\hh_2)}{\hh = \hh_1 \sepcon \hh_2} \notag \\
            \eeqtag{Definition of $\sepcon$} 
            & \ff \sepcon \fg + \ff \sepcon \fh.
    \end{align}
\end{proof}
\begin{proof}[Proof of Theorem~\ref{thm:sep-con-distrib}.\ref{thm:sep-con-distrib:sepcon-over-times-full}]
    For \emph{domain-restricted distributivity of $\sepcon$ over $\cdot$}, consider the following:
    \begin{align}
            & \iverson{\varphi} \sepcon (\fg \cdot \fh) \\
            \eeqtag{Definition of $\sepcon$}
            & \lambda(\sk,\hh)\mydot \max_{\hh_1,\hh_2} \setcomp{\iverson{\varphi}(\sk,\hh_1) \cdot (\fg \cdot \fh)(\sk,\hh_2)}{\hh = \hh_1 \sepcon \hh_2} \\
            \eeqtag{algebra}
            & \lambda(\sk,\hh)\mydot \max_{\hh_1,\hh_2} \setcomp{\iverson{\varphi}(\sk,\hh_1) \cdot \fg(\sk,\hh_2) \cdot \fh(\sk,\hh_2)}{\hh = \hh_1 \sepcon \hh_2} \\
            \eeqtag{$\iverson{\varphi}(\sk,\hh_1) \in \{0,1\}$}
            & \lambda(\sk,\hh)\mydot \max_{\hh_1,\hh_2} \setcomp{\iverson{\varphi}(\sk,\hh_1) \cdot \iverson{\varphi}(\sk,\hh_1) \cdot \fg(\sk,\hh_2) \cdot \fh(\sk,\hh_2)}{\hh = \hh_1 \sepcon \hh_2} \\
            \eeqtag{algebra}
            & \lambda(\sk,\hh)\mydot \max_{\hh_1,\hh_2} \setcomp{(\iverson{\varphi}(\sk,\hh_1) \cdot \fg(\sk,\hh_2)) \cdot (\iverson{\varphi}(\sk,\hh_1) \cdot \fh(\sk,\hh_2))}{\hh = \hh_1 \sepcon \hh_2} \\
            %
            \eeqtag{Lemma~\ref{thm:domain-exact:unique-partitioning}, maximum over singleton} 
            & \lambda(\sk,\hh)\mydot \left(\max_{\hh_1, \hh_2} \setcomp{\iverson{\varphi}(\sk,\hh_1) \cdot \fg(\sk,\hh_2)}{\hh = \hh_1 \sepcon \hh_2}\right) \\
            & \qquad \cdot \left(\max_{\hh_1, \hh_2} \setcomp{\iverson{\varphi}(\sk,\hh_1) \cdot \fh(\sk,\hh_2)}{\hh = \hh_1 \sepcon \hh_2}\right) \notag \\
            \eeqtag{Definition of $\sepcon$}
            & \big(\iverson{\varphi} \sepcon \fg \big) \cdot \big(\iverson{\varphi} \sepcon \fh \big).
    \end{align}
\end{proof}
	%


\subsection{Proof of Theorem~\ref{thm:sepcon-monotonic} (Monotonicity of Separating Conjunction)}
\label{app:sepcon-monotonic}

\begin{proof}
\label{proof:thm:sepcon-monotonic}
	Consider the following:
	\begin{align}
        & \bigl( \ff \sepcon \fg \bigr) (\sk, \hh) \\
        \eeqtag{Definition of $\sepcon$}
        & \max_{\hh_1,\hh_2} \setcomp{\ff(\sk, \hh_1) \cdot \fg(\sk, \hh_2)}{\hh = \hh_1 \sepcon \hh_2} \\
        \lleqtag{by $\ff \preceq \ff'$, $\fg \preceq \fg'$, and monotonicity of ${}\cdot{}$~}
        & \max_{\hh_1,\hh_2} \setcomp{\ff'(\sk, \hh_1) \cdot \fg'(\sk, \hh_2)}{\hh = \hh_1 \sepcon \hh_2} \\
        \eeqtag{Definition of $\sepcon$}
		& \bigl( \ff' \sepcon \fg' \bigr) (\sk, \hh).
	\end{align}
\end{proof}
%


\subsection{Proof of Theorem~\ref{thm:modus-ponens} (Modus Ponens)}
\label{app:modus-ponens}

\begin{proof}\label{proof:thm:modus-ponens}
	Consider the following:
	\begin{align}
		& \bigl( \iverson{\varphi} \sepcon \big( \iverson{\varphi} \sepimp \ff) \bigr) (\sk, \hh) \\
        \eeqtag{Definition of $\sepcon$ and $\sepimp$}
		& \max_{\hh_1, \hh_2}~\setcomp{ \iverson{\varphi}(\sk, \hh_1) ~{}\cdot{}~ \inf_{\hh_2'}~ \setcomp{ \ff(\sk, \hh_2 \sepcon \hh_2') }{ \hh_2' \disjoint \hh_2 \text{ and } (\sk, \hh_2') \models \varphi } }{ \hh = \hh_1 \sepcon \hh_2 }
	\end{align}
	If there exists no partition $\hh_1 \sepcon \hh_2 = \hh$ such that $(\sk, \hh_1) \models \varphi$, then the above becomes 0 and trivially $0 \leq \ff(\sk, \hh)$.
	Otherwise, fix a partition $\hh_1 \sepcon \hh_2 = \hh$ such that $(\sk, \hh_1) \models \varphi$ and the above becomes maximal.
	In that case we are left with
	\begin{align}
		\eeq & \inf_{\hh_2'}~ \setcomp{ \ff(\sk, \hh_2 \sepcon \hh_2') }{ \hh_2' \disjoint \hh_2 \text{ and } (\sk, \hh_2') \models \varphi }~,
	\end{align}
	which is always smaller or equal than $\ff(\sk, \hh)$, since we can choose $\hh_2' = \hh_1$ because $(\sk, \hh_1) \models \varphi$, $\hh_1 \disjoint \hh_2$ and $\ff(\sk, \hh_2 \sepcon \hh_1) = \ff(\sk, \hh) \leq \ff(\sk, \hh)$.
\end{proof}
%


\subsection{Proof of Theorem~\ref{thm:adjointness} (Adjointness)}
\label{app:adjointness}
\begin{proof}
\label{proof:thm:adjointness}
	We first show that 
	\begin{align}
            \ff \sepcon \iverson{\preda} \ppreceq \fg \qimplies  \ff \ppreceq  \iverson{\preda} \sepimp \fg. \label{eq:proof:thm:adjointness:implication}
	\end{align}
	Assume $\ff \sepcon \iverson{\preda} \ppreceq \fg$. By commutativity of $\sepcon$, we have $\iverson{\preda} \sepcon \ff \ppreceq \fg$.
	By definition of $\sepcon$, this means that for \emph{any} state $(\hat{\sk}, \hat{\hh})$ it holds that
	\begin{align}
		 \max_{\hat{\hh}_1, \hat{\hh}_2} \setcomp{\iverson{\preda}(\hat{\sk}, \hat{\hh}_1) \cdot \ff(\hat{\sk}, \hat{\hh}_2)}{\hat{\hh} = \hat{\hh}_1 \sepcon \hat{\hh}_2} \lleq \fg(\hat{\sk}, \hat{\hh}).
	\end{align}
	Then, for \emph{any} partition of the heap $\hat{\hh}$ into $\hat{\hh} = \hat{\hh}_1' \sepcon \hat{\hh}_2'$ with $(\hat{\sk}, \hat{\hh}_1') \models \preda$, we have
	\begin{align}
		 \ff(\hat{\sk}, \hat{\hh}_2') \lleq \fg(\hat{\sk}, \hat{\hh}) ~. \label{eq:proof:thm:adjointness:1}
	\end{align}
	Consider now a state $(\sk, \hh)$.
	There are two cases:
	First, there exists no heap $\hh'$ with $\hh \disjoint \hh'$ and $(\sk, \hh') \models \preda$.
	Then
	\begin{align}
        & \bigl( \iverson{\preda} \sepimp \fg \bigr) (\sk, \hh) \\
        \eeqtag{Definition of $\sepimp$}
        & \inf_{\hh'}~ \setcomp{\fg(\sk, \hh \sepcon \hh')}{\hh' \disjoint \hh \textnormal{ and } (\sk, \hh') \models \preda} \\
        \eeqtag{by assumption} 
		& \inf~ \emptyset \\
        \eeqtag{algebra} 
		& \infty \\
        \ggeqtag{algebra} 
		& \ff(\sk, \hh).
	\end{align}
	The second case is that there does exist a heap $\hh'$ with $\hh \disjoint \hh'$ and $(\sk, \hh') \models \preda$.
	Let $\hh'$ be \emph{any} such heap.
	Then $\sk$, $\hh$, $\hh'$, and $\hh \sepcon \hh'$ satisfy all preconditions of Equation \ref{eq:proof:thm:adjointness:1} (choose $\hat{\sk} = \sk$, $\hat{\hh} = \hh \sepcon \hh'$, $\hat{\hh}_1' = \hh'$, and $\hat{\hh}_2' = \hh$).
    We then obtain
	\begin{align}
		\ff(\sk, \hh) \lleq \fg(\sk, \hh \sepcon \hh').
	\end{align}
	In particular, since the above is true for any heap $\hh'$ that satisfies $\hh \disjoint \hh'$ and $(\sk, \hh') \models \preda$, we have
	\begin{align}
            & \ff(\sk, \hh) \\
            \lleqtag{see above} 
            & \inf_{\hh'}~ \setcomp{\fg(\sk, \hh \sepcon \hh')}{\hh' \disjoint \hh \textnormal{ and } (\sk, \hh') \models \preda} \\
            \eeqtag{Definition of $\sepimp$}
            & \bigl( \iverson{\preda} \sepimp \fg \bigr) (\sk, \hh).
	\end{align}
    This proves one direction of the claim (see equation~(\ref{eq:proof:thm:adjointness:implication})).
	
	We next show the other direction, namely that
	\begin{align}
            \ff \ppreceq  \iverson{\preda} \sepimp \fg \qimplies \ff \sepcon \iverson{\preda} \ppreceq \fg. \label{eq:proof:thm:adjointness:converse}
	\end{align}
	Assume $\ff \ppreceq  \iverson{\preda} \sepimp \fg$.
	By definition of $\sepimp$, for any state $(\hat{\sk}, \hat{\hh})$ it then holds that
	\begin{align}
		\ff(\hat{\sk}, \hat{\hh}) \lleq \inf_{\hh'}~ \setcomp{\fg(\hat{\sk}, \hat{\hh} \sepcon \hat{\hh}')}{\hat{\hh}' \disjoint \hat{\hh} \textnormal{ and } (\hat{\sk}, \hat{\hh}') \models \preda}.
	\end{align}
    In particular, for \emph{any} disjoint extension $\hat{\hh}'$ of the heap $\hat{\hh}$ into $\hat{\hh} \sepcon \hat{\hh}'$ with $(\hat{\sk}, \hat{\hh}') \models \preda$ we have
	\begin{align}
		 \ff(\hat{\sk}, \hat{\hh}) \lleq \fg(\hat{\sk}, \hat{\hh} \sepcon \hat{\hh}') ~. \label{eq:proof:thm:adjointness:2}
	\end{align}
	Consider now a state $(\sk, \hh)$.
	There are two cases:
	First, there exists no partition of $\hh$ into $\hh = \hh_1 \sepcon \hh_2$ such that $(\sk, \hh_1) \models \preda$.
	Then
	\begin{align}
        & \bigl( \ff \sepcon \iverson{\preda} \bigr) (\sk, \hh) \\
        \eeqtag{Commutativity of $\sepcon$}
        & \bigl( \iverson{\preda} \sepcon \ff \bigr) (\sk, \hh) \\
        \eeqtag{Definition of $\sepcon$}
		& \max_{\hh_1, \hh_2} \setcomp{\iverson{\preda}(\sk, \hh_1) \cdot \ff(\sk, \hh_2)}{\hh = \hh_1 \sepcon \hh_2} \\
        \eeqtag{by assumption any partition leads to $\iverson{\preda}(\sk,\hh_1) = 0$}
        & \max \{ 0 \} \\
        \eeqtag{algebra}
		& 0 \\
        \lleqtag{algebra}
		& \fg(\sk, \hh).
	\end{align}
	The second case is that there does exist a partitioning of $\hh$ into $\hh = \hh_1 \sepcon \hh_2$ such that $(\sk, \hh_1) \models \preda$.
	Let $\hh_1 \sepcon \hh_2$ be \emph{any} such partitioning.
	Then $\sk$, $\hh_1$, and $\hh_2$ satisfy all preconditions of Equation \ref{eq:proof:thm:adjointness:2} (choose $\hat{\sk} = \sk$, $\hat{\hh} = \hh_2$, and $\hat{\hh}' = \hh_1$). 
    Then
	\begin{align}
		& \ff(\sk, \hh_2) \lleq \fg(\sk, \hh_2 \sepcon \hh_1) \\
        \leftrighttag{since $(\sk, \hh_1) \models \preda$ and $\hh_2 \sepcon \hh_1 =\hh$}
		& \iverson{\preda}(\sk, \hh_1) \cdot \ff(\sk, \hh_2) \lleq \fg(\sk, \hh).
	\end{align}
	Consequently, for any partitioning $\hh = \hh_1 \sepcon \hh_2$ that satisfies $(\sk, \hh_1) \models \preda$, we get 
	\begin{align}
            & \fg(\sk, \hh) \\
            \ggeqtag{see above for \emph{any} partitioning $\hh_1 \sepcon \hh_2 = \hh$}
            & \max_{\hh_1, \hh_2} \setcomp{\iverson{\preda}(\sk, \hh_1) \cdot \ff(\sk, \hh_2)}{\hh = \hh_1 \sepcon \hh_2} \\
            \eeqtag{Definition of $\sepcon$}
            & \bigl( \ff \sepcon \iverson{\preda} \bigr) (\sk, \hh).
	\end{align}
	This proves the second implication.
\end{proof}

%
%
\subsection{Proof of Theorem~\ref{thm:sep-con-algebra-pure} (Laws for Pure Expectations)}
\label{app:sep-con-algebra-pure}

\begin{proof}

(1). Let $\ff,\fg \in \E$, $\ff$ pure. Then
	\begin{align}
        & \bigl( \ff \cdot \fg \bigr) (\sk, \hh) \\
        \eeqtag{Definition of $\cdot$~}
        & \ff(\sk, \hh) \cdot \fg(\sk, \hh)\\
        \lleqtag{$\hh_2$ can be chosen as $\hh$}
		& \ff(\sk, \hh) \cdot \max_{\hh_1, \hh_2}~ \setcomp{\fg(\sk, \hh_2)}{\hh = \hh_1 \sepcon \hh_2}\\
        \eeqtag{$\ff(\sk,\hh) \in \Reals$ is a constant}
		& \max_{\hh_1, \hh_2}~ \setcomp{\ff(\sk, \hh) \cdot \fg(\sk, \hh_2)}{\hh = \hh_1 \sepcon \hh_2}\\
        \eeqtag{$\ff$ is pure}
		& \max_{\hh_1, \hh_2}~ \setcomp{\ff(\sk, \hh_1) \cdot \fg(\sk, \hh_2)}{\hh = \hh_1 \sepcon \hh_2} \\
        \eeqtag{Definition of $\sepcon$}
		& \bigl( \ff \sepcon \fg \bigr)(\sk, \hh).
	\end{align}
(2). Let $\ff,\fg \in \E$ and both $\ff$ and $\fg$ be pure. Then
\begin{align}
        & \bigl( \ff \sepcon \fg \bigr) (\sk, \hh) \\
        \eeqtag{Definition of $\sepcon$}
        & \max_{\hh_1, \hh_2}~ \setcomp{\ff(\sk, \hh_1) \cdot \fg(\sk, \hh_2)}{\hh = \hh_1 \sepcon \hh_2}\\
        \eeqtag{$\ff,\fg$ are pure}
		& \max_{\hh_1, \hh_2}~ \setcomp{\ff(\sk, \hh) \cdot \fg(\sk, \hh)}{\hh = \hh_1 \sepcon \hh_2} \\
        \eeqtag{algebra}
		& \ff(\sk, \hh) \cdot \fg(\sk, \hh) \\
        \eeqtag{Definition of $\cdot$~}
		& \bigl( \ff \cdot \fg \bigr) (\sk, \hh).
\end{align}
(3). Let $\ff$ be pure. Then
	\begin{align}
        & \bigl( ( \ff \cdot \fg ) \sepcon \fh \bigr) (\sk, \hh) \\
        \eeqtag{Definition of $\sepcon$}
        & \max_{\hh_1, \hh_2}~ \setcomp{\bigl( \ff \cdot \fg \bigr)(\sk,\, \hh_1) \cdot \fh(\sk,\, \hh_2) }{\hh = \hh_1 \sepcon \hh_2}\\
        \eeqtag{Definition of $\cdot$~}
		& \max_{\hh_1, \hh_2}~ \setcomp{\ff(\sk,\, \hh_1) \cdot \fg(\sk,\, \hh_1) \cdot \fh(\sk,\, \hh_2) }{\hh = \hh_1 \sepcon \hh_2}\\
        \eeqtag{$\ff$ is pure}
		& \max_{\hh_1, \hh_2}~ \setcomp{\ff(\sk,\, \hh) \cdot \fg(\sk,\, \hh_1) \cdot \fh(\sk,\, \hh_2) }{\hh = \hh_1 \sepcon \hh_2} \tag{purity of $\ff$}\\
        \eeqtag{$\ff(\sk,\hh) \in \Reals$ is a constant}
		& \ff(\sk,\, \hh) \cdot \max_{\hh_1, \hh_2}~ \setcomp{\fg(\sk,\, \hh_1) \cdot \fh(\sk,\, \hh_2) }{\hh = \hh_1 \sepcon \hh_2} \tag{purity of $\ff$}\\
        \eeqtag{Definition of $\sepcon$}
		& \ff(\sk,\, \hh) \cdot \bigl( \fg \sepcon \fh \bigr)(\sk,\, \hh) \\
        \eeqtag{algebra}
		& \bigl( \ff \cdot ( \fg \sepcon \fh ) \bigr)(\sk,\, \hh).
	\end{align}
\end{proof}
%
%

\subsection{Proof of Theorem~\ref{thm:intuitionistification} (Tightest Intuitionistic Expectations)}
\label{app:intuitionistification}

We have to show that 
\begin{enumerate}
        \item $\ff \sepcon 1$ is intuitionistic, i.e. for all $\hh \subseteq \hh'$, $(\ff \sepcon 1)(\sk,\hh) \leq (\ff \sepcon 1)(\sk,\hh')$. \label{app:intuitionistification:1}
        \item $\ff \preceq \ff \sepcon 1$. \label{app:intuitionistification:2}
        \item for all intuitionistic $\ff'$, $\ff \preceq \ff'$ implies $\ff \sepcon 1 \preceq \ff'$. \label{app:intuitionistification:3}
        \item $1 \sepimp \ff$ is intuitionistic. \label{app:intuitionistification:4}
        \item $1 \sepimp \ff \preceq \ff$. \label{app:intuitionistification:5}
        \item for all intuitionistic $\ff'$, $\ff' \preceq \ff$ implies $\ff' \preceq 1 \sepimp \ff$. \label{app:intuitionistification:6}
\end{enumerate}

\begin{proof}[Proof of Theorem~\ref{thm:intuitionistification}, (\ref{app:intuitionistification:1})]
\begin{align}
        & (\ff \sepcon 1)(\sk,\hh \sepcon \hh') \\
        \eeqtag{Definition of $\sepcon$}
        & \max_{\hh_1, \hh_2} \setcomp{\ff(\sk,\hh_1) \cdot 1(\sk,\hh_2)}{\hh \sepcon \hh' = \hh_1 \sepcon \hh_2} \\
        \ggeqtag{consider subset in which $\hh_1 = \hh$}
        & \max_{\hh_1,\hh_2} \setcomp{\ff(\sk,\hh_1) \cdot 1(\sk,\hh_2 \sepcon \hh')}{\hh = \hh_1 \sepcon \hh_2} \\
        \eeqtag{algebra}
        & \max_{\hh_1,\hh_2} \setcomp{\ff(\sk,\hh_1) \cdot 1(\sk,\hh_2)}{\hh = \hh_1 \sepcon \hh_2} \\
        \eeqtag{Definition of $\sepcon$}
        & (\ff \sepcon 1)(\sk,\hh).
\end{align}
\end{proof}

\begin{proof}[Proof of Theorem~\ref{thm:intuitionistification}, (\ref{app:intuitionistification:2})]
\begin{align}
        & (\ff \sepcon 1)(\sk,\hh) \\
        \eeqtag{Definition of $\sepcon$}
        & \max_{\hh_1,\hh_2} \setcomp{\ff(\sk,\hh_1) \cdot 1(\sk,\hh_2)}{\hh = \hh_1 \sepcon \hh_2} \\
        \eeqtag{algebra}
        & \max_{\hh_1,\hh_2} \setcomp{\ff(\sk,\hh_1)}{\hh = \hh_1 \sepcon \hh_2} \\
        \ggeqtag{consider subset in which $\hh_1 = \hh$}
        & \ff(\sk,\hh).
\end{align}
\end{proof}

\begin{proof}[Proof of Theorem~\ref{thm:intuitionistification}, (\ref{app:intuitionistification:3})]
Let $\ff'$ be an intuitionistic expectation with $\ff \preceq \ff'$. Then
\begin{align}
        & (\ff \sepcon 1)(\sk,\hh) \\
        \eeqtag{Definition of $\sepcon$}
        & \max_{\hh_1,\hh_2} \setcomp{\ff(\sk,\hh_1) \cdot 1(\sk,\hh_2)}{\hh = \hh_1 \sepcon \hh_2} \\
        \lleqtag{$\ff \preceq \ff'$}
        & \max_{\hh_1, \hh_2} \setcomp{\ff'(\sk,\hh_1) \cdot 1(\sk,\hh_2)}{\hh = \hh_1 \sepcon \hh_2} \\
        \eeqtag{algebra}
        & \max_{\hh_1, \hh_2} \setcomp{\ff'(\sk,\hh_1)}{\hh = \hh_1 \sepcon \hh_2} \\
        \eeqtag{$\ff'$ intuitionistic. Hence the maximum is attained for $\hh_1 = \hh$.}
        & \ff'(\sk,\hh).
\end{align}
\end{proof}

\begin{proof}[Proof of Theorem~\ref{thm:intuitionistification}, (\ref{app:intuitionistification:4})]
\begin{align}
        & (1 \sepimp \ff)(\sk,\hh \sepcon \hh') \\
        \eeqtag{Definition of $\sepimp$}
        & \inf_{\hh''} \left\{ \ff(\sk,\hh \sepcon \hh' \sepcon \hh'') ~|~ \hh \sepcon \hh' \disjoint \hh'', \sk,\hh'' \models 1 \right\} \\
        \eeqtag{$1$ is always satisfied}
        & \inf_{\hh''} \left\{ \ff(\sk,\hh \sepcon \hh' \sepcon \hh'') ~|~ \hh \sepcon \hh' \disjoint \hh'' \right\} \\
        \lleqtag{for the empty heap $\emptyheap$, we have $\hh \sepcon \hh' \disjoint \emptyheap$}
        & \ff(\sk,\hh \sepcon \hh' \sepcon \emptyheap) \\
        \eeqtag{Theorem~\ref{thm:sep-con-monoid}}
        & \ff(\sk,\hh \sepcon \hh').
\end{align}
\end{proof}

\begin{proof}[Proof of Theorem~\ref{thm:intuitionistification}, (\ref{app:intuitionistification:5})]
Follows directly from the proof of Theorem~\ref{thm:intuitionistification} (\ref{app:intuitionistification:4}) by setting $\hh' = \emptyheap$.
\end{proof}

\begin{proof}[Proof of Theorem~\ref{thm:intuitionistification}, (\ref{app:intuitionistification:6})]
Let $\ff'$ be an intuitionistic expectation with $\ff' \preceq \ff$. Then
\begin{align}
        & (1 \sepimp \ff)(\sk,\hh) \\
        \eeqtag{Definition of $\sepimp$}
        & \inf_{\hh'} \left\{ \ff(\sk,\hh \sepcon \hh') ~|~ \hh \disjoint \hh', \sk,\hh' \models 1 \right\} \\
        \eeqtag{$1$ is always satisfied}
        & \inf_{\hh'} \left\{ \ff(\sk,\hh \sepcon \hh') ~|~ \hh \disjoint \hh' \right\} \\
        \ggeqtag{$\ff' \preceq \ff$}
        & \inf_{\hh'} \left\{ \ff'(\sk,\hh \sepcon \hh') ~|~ \hh \disjoint \hh' \right\} \\
        \eeqtag{$\ff'$ intuitionistic. Hence the infimum is attained for $\hh' = \emptyheap$.}
        & \ff'(\sk,\hh \sepcon \emptyheap) \\
        \eeqtag{Theorem~\ref{thm:sep-con-monoid}}
        & \ff'(\sk,\hh).
\end{align}
\end{proof}

\subsection{Proof of Theorem~\ref{thm:qsl:heap-size} (Heap Size Laws)}
\label{app:qsl:heap-size}
  
\paragraph{Proof of Theorem~\ref{thm:qsl:heap-size}.~\ref{thm:qsl:heap-size:sepcon}}
We have to show that 
\begin{align}
   \singleton{\ee}{\ee'} \sepcon \heapSize \eeq \containsPointer{\ee}{\ee'} \cdot (\heapSize - 1).
\end{align}

\begin{proof}
\begin{align}
        & \singleton{\ee}{\ee'} \sepcon \heapSize \\
        \eeqtag{Definition of $\sepcon$}
        & \lambda(\sk,\hh)\mydot \max_{\hh_1, \hh_2} \setcomp{ \singleton{\ee}{\ee'}(\sk,\hh_1) \cdot \heapSize(\sk,\hh_2) }{ \hh = \hh_1 \sepcon \hh_2 } \\
        \eeqtag{Definition of $\heapSize$}
        & \lambda(\sk,\hh)\mydot \max_{\hh_1, \hh_2} \setcomp{ \singleton{\ee}{\ee'}(\sk,\hh_1) \cdot |\dom{\hh_2}| }{ \hh = \hh_1 \sepcon \hh_2 }  \\
        \eeqtag{$\dom{\hh} = \dom{\hh_1} + \dom{\hh_2}$}
        & \lambda(\sk,\hh)\mydot \max_{\hh_1,\hh_2} \setcomp{ \singleton{\ee}{\ee'}(\sk,\hh_1) \cdot (|\dom{\hh}| - |\dom{\hh_1}|) }{ \hh = \hh_1 \sepcon \hh_2 }  \\
        %
        %
        %
        \eeqtag{$\singleton{\ee}{\ee'}(\sk,\hh_1) \neq 0$ implies $\dom{\hh_1} = \{ \sk(\ee) \}$}
        & \lambda(\sk,\hh)\mydot \max_{\hh_1,\hh_2} \setcomp{ \singleton{\ee}{\ee'}(\sk,\hh_1) \cdot (|\dom{\hh}| - 1)}{ \hh = \hh_1 \sepcon \hh_2 }  \\
        \eeqtag{algebra}
        & \lambda(\sk,\hh)\mydot (|\dom{\hh}|-1) \cdot \max_{\hh_1,\hh_2} \setcomp{ \singleton{\ee}{\ee'}(\sk,\hh_1) }{ \hh = \hh_1 \sepcon \hh_2 }  \\
        \eeqtag{algebra}
        & \lambda(\sk,\hh)\mydot (|\dom{\hh}|-1) \cdot \max_{\hh_1,\hh_2} \setcomp{ \singleton{\ee}{\ee'}(\sk,\hh_1) \cdot 1(\sk,\hh_2) }{ \hh = \hh_1 \sepcon \hh_2 }  \\
        \eeqtag{Definition of $\sepcon$}
        & \lambda(\sk,\hh)\mydot (|\dom{\hh}|-1) \cdot (\singleton{\ee}{\ee'} \sepcon 1)(\sk,\hh) \\ 
        %
        \eeqtag{$\containsPointer{\ee}{\ee'} = \singleton{\ee}{\ee'} \sepcon 1$}
        & \lambda(\sk,\hh)\mydot (|\dom{\hh}|-1) \cdot \containsPointer{\ee}{\ee'}(\sk,\hh) \\
        \eeqtag{Definition of $\heapSize$}
        & \lambda(\sk,\hh)\mydot (\heapSize(\sk,\hh) - 1) \cdot \containsPointer{\ee}{\ee'}(\sk,\hh) \\
        \eeqtag{algebra}
        & \containsPointer{\ee}{\ee'} \cdot (\heapSize - 1).
\end{align}
\end{proof}

\paragraph{Proof of Theorem~\ref{thm:qsl:heap-size}.~\ref{thm:qsl:heap-size:sepimp}}
We have to show that 
\begin{align}
   \singleton{\ee}{\ee'} \sepimp \heapSize \eeq \heapSize + 1 + \containsPointer{\ee}{-} \cdot \infty.
\end{align}

\begin{proof}
\begin{align}
        & \singleton{\ee}{\ee'} \sepimp \heapSize \\
        \eeqtag{Definition of $\sepimp$}
        & \lambda(\sk,\hh)\mydot \inf_{\hh'} \setcomp{ \heapSize(\sk,\hh \sepcon \hh') }{ \hh \disjoint \hh' ~\text{and}~ (\sk,\hh) \models \singleton{\ee}{\ee'} } \\
        \eeqtag{Definition of $\heapSize$}
        & \lambda(\sk,\hh)\mydot \inf_{\hh'} \setcomp{ |\dom{\hh \sepcon \hh'}| }{ \hh \disjoint \hh' ~\text{and}~ (\sk,\hh') \models \singleton{\ee}{\ee'} } \\
        \eeqtag{for all $\hh'$ with $(\sk,\hh') \models \singleton{\ee}{\ee'}$, we have $|\dom{\hh'}| = 1$} 
        & \lambda(\sk,\hh)\mydot \inf_{\hh'} \left\{ |\dom{\hh}| + 1 ~|~ \hh \disjoint \hh', \sk,\hh' \models \singleton{\ee}{\ee'} \right\} \\
        \eeqtag{case distinction: $\sk(\ee) \in \dom{\hh}$ or $\sk(\ee) \notin \dom{\hh}$}
        & \lambda(\sk,\hh)\mydot 
        \containsPointer{\ee}{-}(\sk,\hh) \cdot \inf_{\hh'} \left\{ |\dom{\hh}| + 1 ~|~ \hh \disjoint \hh', \sk,\hh' \models \singleton{\ee}{\ee'} \right\} \\
        & \qquad + (1-\containsPointer{\ee}{-}(\sk,\hh)) \cdot \inf_{\hh'} \left\{ |\dom{\hh}| + 1 ~|~ \hh \disjoint \hh', \sk,\hh' \models \singleton{\ee}{\ee'} \right\} \notag \\
        \eeqtag{First case: $\containsPointer{\ee}{-}(\sk,\hh) = 1$ implies there is no $\hh' \disjoint \hh$}
        & \lambda(\sk,\hh)\mydot 
        \containsPointer{\ee}{-}(\sk,\hh) \cdot \inf_{\hh'} \emptyset \\
        & \qquad + (1-\containsPointer{\ee}{-}(\sk,\hh)) \cdot \inf_{\hh'} \left\{ |\dom{\hh}| + 1 ~|~ \hh \disjoint \hh', \sk,\hh' \models \singleton{\ee}{\ee'} \right\} \notag \\
        \eeqtag{Second case: $\containsPointer{\ee}{-}(\sk,\hh) = 0$ implies there is $\hh' \disjoint \hh$}
        & \lambda(\sk,\hh)\mydot \containsPointer{\ee}{-}(\sk,\hh) \cdot \inf_{\hh'} \emptyset + (1-\containsPointer{\ee}{-}(\sk,\hh)) \cdot (|\dom{\hh}|+1) \\
        \eeqtag{algebra}
        & \lambda(\sk,\hh)\mydot \containsPointer{\ee}{-}(\sk,\hh) \cdot \infty + (1-\containsPointer{\ee}{-}(\sk,\hh)) \cdot (|\dom{\hh}|+1) \\
        \eeqtag{Definition of $\heapSize$}
        & \lambda(\sk,\hh)\mydot \containsPointer{\ee}{-}(\sk,\hh) \cdot \infty + (1-\containsPointer{\ee}{-}(\sk,\hh)) \cdot (\heapSize(\sk,\hh)+1) \\
        \eeqtag{algebra}
        & \lambda(\sk,\hh)\mydot \containsPointer{\ee}{-}(\sk,\hh) \cdot (\infty+\heapSize(\sk,\hh)+1) + (1-\containsPointer{\ee}{-}(\sk,\hh)) \cdot (\heapSize(\sk,\hh)+1) \\
        %
        %
        %
        %
        \eeqtag{algebra}
        & \heapSize + 1 + \containsPointer{\ee}{-} \cdot \infty.
\end{align}
\end{proof}

\paragraph{Proof of Theorem~\ref{thm:qsl:heap-size}.~\ref{thm:qsl:heap-size:dist}}
We have to show that 
\begin{align}
    \left(\ff \sepcon \fg\right) \cdot \heapSize \ppreceq \left(\ff \cdot \heapSize\right) \sepcon \fg + \ff \sepcon \left(\fg \cdot \heapSize\right).
\end{align}

\begin{proof}
\begin{align}
        & \left( \ff \sepcon \fg \right) \cdot \heapSize \\
        \eeqtag{algebra}
        & \lambda(\sk,\hh)\mydot \left( \ff \sepcon \fg \right)(\sk,\hh) \cdot \heapSize(\sk,\hh) \\
        \eeqtag{Definition of $\sepcon$}
        & \lambda(\sk,\hh)\mydot \max \left\{ \ff(\sk,\hh_1) \cdot \fg(\sk,\hh_2) ~\middle|~ \hh = \hh_1 \sepcon \hh_2 \right\} \cdot \heapSize(\sk,\hh) \\
        \eeqtag{algebra}
        & \lambda(\sk,\hh)\mydot \max \left\{ \ff(\sk,\hh_1) \cdot \fg(\sk,\hh_2) \cdot \heapSize(\sk,\hh) ~\middle|~ \hh = \hh_1 \sepcon \hh_2 \right\}  \\
        \eeqtag{Definition of $\heapSize$}
        & \lambda(\sk,\hh)\mydot \max \left\{ \ff(\sk,\hh_1) \cdot \fg(\sk,\hh_2) \cdot |\dom{\hh}| ~\middle|~ \hh = \hh_1 \sepcon \hh_2 \right\}  \\
        \eeqtag{$|\dom{\hh}| = |\dom{\hh_1}| + |\dom{\hh_2}|$}
        & \lambda(\sk,\hh)\mydot \max \left\{ \ff(\sk,\hh_1) \cdot \fg(\sk,\hh_2) \cdot (|\dom{\hh_1}| + |\dom{\hh_2}|) ~\middle|~ \hh = \hh_1 \sepcon \hh_2 \right\}  \\
        \eeqtag{algebra}
        & \lambda(\sk,\hh)\mydot \max \{ \left( \ff(\sk,\hh_1) \cdot |\dom{\hh_1}| \right) \cdot \fg(\sk,\hh_2) \label{eq:qsl:heap-size:dist-before-triangle} \\
        & \qquad \qquad \qquad + \ff(\sk,\hh_1) \cdot \left(\fg(\sk,\hh_2) \cdot |\dom{\hh_2}|\right) ~|~ \hh = \hh_1 \sepcon \hh_2 \}  \notag \\
        \ppreceqtag{triangle inequality}
        & \lambda(\sk,\hh)\mydot \max \{ \left( \ff(\sk,\hh_1) \cdot |\dom{\hh_1}| \right) \cdot \fg(\sk,\hh_2) ~|~ \hh = \hh_1 \sepcon \hh_2 \} \\
        & \qquad \qquad \qquad + \max \{ \ff(\sk,\hh_1) \cdot \left(\fg(\sk,\hh_2) \cdot |\dom{\hh_2}|\right) ~|~ \hh = \hh_1 \sepcon \hh_2 \}  \notag \\
        \eeqtag{Definition of $\sepcon$, $\heapSize$}
        & \lambda(\sk,\hh)\mydot \left( (\ff \cdot \heapSize) \sepcon \fg\right)(\sk,\hh) + \left( \ff \sepcon (\fg \cdot \heapSize) \right)(\sk,\hh) \\
        \eeqtag{algebra}
        & (\ff \cdot \heapSize) \sepcon \fg + \ff \sepcon (\fg \cdot \heapSize).
\end{align}
\end{proof}

\paragraph{Proof of Theorem~\ref{thm:qsl:heap-size}.~\ref{thm:qsl:heap-size:dist-full}}
We have to show for domain-exact $\ff$ or $\fg$ that
\begin{align}
    \left(\ff \sepcon \fg\right) \cdot \heapSize \eeq \left(\ff \cdot \heapSize\right) \sepcon \fg + \ff \sepcon \left(\fg \cdot \heapSize\right).
\end{align}

\begin{proof}
The proof is analogous to the proof of Theorem~\ref{thm:qsl:heap-size}.~\ref{thm:qsl:heap-size:dist}.
However, instead of applying the triangle inequality to equation~\ref{eq:qsl:heap-size:dist-before-triangle}, we apply
apply Lemma~\ref{thm:domain-exact:unique-partitioning} to $\ff$ or $\fg$ (depending on whether $\ff$ or $\fg$ is domain-exact).
Since we then take a maximum over a singleton, we proceed as follows:
\begin{align}
        & \text{(continuing from equation~\ref{eq:qsl:heap-size:dist-before-triangle})} \\
        \eeqtag{Apply Lemma~\ref{thm:domain-exact:unique-partitioning} to $\ff$ or $\fg$, max over singleton} 
        & \lambda(\sk,\hh)\mydot \max \{ \left( \ff(\sk,\hh_1) \cdot |\dom{\hh_1}| \right) \cdot \fg(\sk,\hh_2) ~|~ \hh = \hh_1 \sepcon \hh_2 \} \\
        & \qquad \qquad \qquad + \max \{ \ff(\sk,\hh_1) \cdot \left(\fg(\sk,\hh_2) \cdot |\dom{\hh_2}|\right) ~|~ \hh = \hh_1 \sepcon \hh_2 \}  \notag \\
        \eeqtag{Definition of $\sepcon$}
        & \lambda(\sk,\hh)\mydot \left( (\ff \cdot \heapSize) \sepcon \fg\right)(\sk,\hh) + \left( \ff \sepcon (\fg \cdot \heapSize) \right)(\sk,\hh) \\
        \eeqtag{algebra}
        & (\ff \cdot \heapSize) \sepcon \fg + \ff \sepcon (\fg \cdot \heapSize).
\end{align}
\end{proof}

\subsection{Proof of Lemma~\ref{thm:ls-props} (Properties of List Segments)}
\label{app:ls-props}

Recall the definition of $\Lssymbol$ and $\Lensymbol$:

\begin{align}
        \Ls{\za}{\zb} \eeq & \underbrace{ \iverson{\za = \zb} \cdot \emp + \iverson{\za \neq \zb} \cdot \sup_{\zc \in \Ints} \singleton{\za}{\zc} \sepcon \Ls{\zc}{\zb} }_{ \eeq \ff_{\Lssymbol}(\za,\zb) } \\
        \Len{\za}{\zb} \eeq & \underbrace{\iverson{\za \neq \zb} \cdot \sup_{\zc \in \Ints} \singleton{\za}{\zc} \sepcon \left( \Ls{\zc}{\zb} + \Len{\zc}{\zb} \right)}_{\eeq \fg_{\Len{\za}{\zb}}}
\end{align}

By definition, we have 
\begin{align}
    \Ls{\za}{\zb} \eeq \lfp \fh . \lambda (\za,\zb)\mydot \ff_{\fh}(\za,\zb). 
    \label{eq:app:ls-props:ls}
\end{align}

\paragraph{Continuity of $\Ls{\za}{\zb}$ and $\Len{\za}{\zb}$}

We first note that the underlying functional is continuous.
\begin{lemma}
\label{thm:ls:continuous}
        For all sequences of $P_n \in \Rats^2 \to \E$, $n \in \Nats$, we have
        \begin{align*}
                \sup_{n \in \Nats} \lambda(\za,\zb)\mydot \ff_{P_n}(\za,\zb)
                \eeq
                \lambda(\za,\zb)\mydot \ff_{\sup_{n\in\Nats}P_n}(\za,\zb)
        \end{align*}
\end{lemma}

\begin{proof}
\begin{align}
        &  \sup_{n \in \Nats} \lambda(\za,\zb)\mydot \ff_{P_n}(\za,\zb) \\
        \eeqtag{algebra}
        & \lambda(\za,\zb)\mydot \sup_{n \in \Nats} \ff_{P_n}(\za,\zb) \\
        \eeqtag{Definition of $\ff_{P_n}$}
        & \lambda(\za,\zb)\mydot \sup_{n \in \Nats} \left(
            \iverson{\za = \zb} \cdot \emp 
            + \iverson{\za \neq \zb} \cdot \sup_{\zc \in \Ints} \singleton{\za}{\zc} \sepcon P_n(\zc,\zb)
          \right) \\
        \eeqtag{algebra}
        & \lambda(\za,\zb)\mydot \left(
            \iverson{\za = \zb} \cdot \emp 
            + \sup_{n \in \Nats} \iverson{\za \neq \zb} \cdot \sup_{\zc \in \Ints} \singleton{\za}{\zc} \sepcon P_n(\zc,\zb)
          \right) \\
        \eeqtag{algebra}
        & \lambda(\za,\zb)\mydot \left(
            \iverson{\za = \zb} \cdot \emp 
            + \iverson{\za \neq \zb} \cdot \sup_{\zc \in \Ints} 
            \sup_{n \in \Nats} \singleton{\za}{\zc} \sepcon P_n(\zc,\zb)
          \right) \\
        \eeqtag{Definition of $\sepcon$}
        & \lambda(\za,\zb)\mydot (
            \iverson{\za = \zb} \cdot \emp 
            + \iverson{\za \neq \zb} \cdot \sup_{\zc \in \Ints} \\
        & \qquad
            \sup_{n \in \Nats} 
            \lambda(\sk,\hh)\mydot \max_{\hh_1,\hh_2} \setcomp{ \singleton{\za}{\zc}(\sk,\hh_1) \cdot P_n(\zc,\zb)(\sk,\hh_2) }{ \hh = \hh_1 \sepcon \hh_2 }
          ) \notag \\
        \eeqtag{algebra}
        & \lambda(\za,\zb)\mydot (
            \iverson{\za = \zb} \cdot \emp 
            + \iverson{\za \neq \zb} \cdot \sup_{\zc \in \Ints} \\
        & \qquad
            \lambda(\sk,\hh)\mydot \max_{\hh_1,\hh_2} \setcomp{
            \sup_{n \in \Nats} 
            \singleton{\za}{\zc}(\sk,\hh_1) \cdot P_n(\zc,\zb)(\sk,\hh_2) 
            }{ \hh = \hh_1 \sepcon \hh_2}
          ) \notag \\
        \eeqtag{algebra}
        & \lambda(\za,\zb)\mydot \big(
            \iverson{\za = \zb} \cdot \emp 
            + \iverson{\za \neq \zb} \cdot \sup_{\zc \in \Ints} \\
        & \qquad
            \lambda(\sk,\hh)\mydot \max_{\hh_1,\hh_2} \setcomp{
                \singleton{\za}{\zc}(\sk,\hh_1) \cdot 
                \sup_{n \in \Nats} P_n(\zc,\zb)(\sk,\hh_2) 
            }{ \hh = \hh_1 \sepcon \hh_2 }
          \big) \notag \\
        \eeqtag{Definition of $\sepcon$}
        & \lambda(\za,\zb)\mydot \left(
            \iverson{\za = \zb} \cdot \emp 
            + \iverson{\za \neq \zb} \cdot \sup_{\zc \in \Ints} 
            \singleton{\za}{\zc} \sepcon \sup_{n \in \Nats} P_n(\zc,\zb)
          \right) \\
        \eeqtag{Definition of $\ff_{.}$}
        &  \lambda(\za,\zb)\mydot \ff_{\sup_{n \in \Nats} P_n}(\za,\zb).
\end{align}
\end{proof}

Similarly, we show that the functional underlying the list-length quantity is continuous.
\begin{lemma}
\label{thm:len:continuous}
        For all sequences of $P_n \in \Rats^2 \to \E$, $n \in \Nats$, we have
        \begin{align*}
                \sup_{n \in \Nats} \lambda(\za,\zb)\mydot \fg_{P_n}(\za,\zb)
                \eeq
                \lambda(\za,\zb)\mydot \fg_{\sup_{n\in\Nats}P_n}(\za,\zb)
        \end{align*}
\end{lemma}

\begin{proof}
\begin{align}
        &  \sup_{n \in \Nats} \lambda(\za,\zb)\mydot \fg_{P_n}(\za,\zb) \\
        \eeqtag{algebra}
        & \lambda(\za,\zb)\mydot \sup_{n \in \Nats} \fg_{P_n}(\za,\zb) \\
        \eeqtag{Definition of $\fg_{P_n}$}
        & \lambda(\za,\zb)\mydot \sup_{n \in \Nats} \left(
              \iverson{\za \neq \zb} \cdot \sup_{\zc \in \Ints} \singleton{\za}{\zc} \sepcon \left(\Ls{\zc}{\zb} + P_n(\zc,\zb)\right)
          \right) \\
        \eeqtag{algebra}
        & \lambda(\za,\zb)\mydot \iverson{\za \neq \zb} \cdot \sup_{n \in \Nats} 
              \sup_{\zc \in \Ints} \singleton{\za}{\zc} \sepcon \left(\Ls{\zc}{\zb} + P_n(\zc,\zb)\right) \\
        \eeqtag{algebra}
        & \lambda(\za,\zb)\mydot \iverson{\za \neq \zb} \cdot \sup_{\zc \in \Ints} \sup_{n \in \Nats} 
              \singleton{\za}{\zc} \sepcon \left(\Ls{\zc}{\zb} + P_n(\zc,\zb)\right) \\
        \eeqtag{Definition of $\sepcon$}
        & \lambda(\za,\zb)\mydot \iverson{\za \neq \zb} \cdot 
        \sup_{\zc \in \Ints} \sup_{n \in \Nats} \lambda(\sk,\hh)\mydot \\
        & \qquad \max_{\hh_1,\hh_2} \setcomp{ \singleton{\za}{\zc}(\sk,\hh1) \cdot \left(\Ls{\zc}{\zb}(\sk,\hh_2) + P_n(\zc,\zb)(\sk,\hh_2)\right) }{\hh = \hh_1 \sepcon \hh_2} \notag \\
        \eeqtag{algebra}
        & \lambda(\za,\zb)\mydot \iverson{\za \neq \zb} \cdot \sup_{\zc \in \Ints} \lambda(\sk,\hh)\mydot \\
        & \qquad \max_{\hh_1,\hh_2} \setcomp{ \singleton{\za}{\zc}(\sk,\hh1) \cdot \left(\Ls{\zc}{\zb}(\sk,\hh_2) + \sup_{n \in \Nats} P_n(\zc,\zb)(\sk,\hh_2)\right) }{ \hh = \hh_1 \sepcon \hh_2 } \notag \\
        \eeqtag{Definition of $\sepcon$}
        & \lambda(\za,\zb)\mydot \iverson{\za \neq \zb} \cdot \sup_{\zc \in \Ints} \singleton{\za}{\zc} \sepcon \left(\Ls{\zc}{\zb} + \sup_{n \in \Nats} P_n(\zc,\zb)\right) \\
        \eeqtag{Definition of $\fg_{.}$}
        &  \lambda(\za,\zb)\mydot \fg_{\sup_{n \in \Nats} P_n}(\za,\zb).
\end{align}
\end{proof}

\paragraph{Proof of Lemma~\ref{thm:ls-props}.\ref{thm:ls-props:char}}

\begin{proof}
We show for all stack-heap pairs $(\sk,\hh)$ that $\Len{\za}{\zb}(\sk,\hh) = \left(\Ls{\za}{\zb} \cdot \heapSize\right)(\sk,\hh)$ by induction on the size of the heap $n = |\dom{\hh}|$.
For $n = 0$, we have
\begin{align}
        & \left(\Ls{\za}{\zb}(\sk,\hh) \cdot \heapSize\right)(\sk,\hh) \\
        \eeqtag{Definition of $\Ls{\za}{\zb}$}
        & \left( \left(\iverson{\za = \zb} \cdot \emp + \iverson{\za \neq \zb} \cdot \sup_{\gamma \in \Ints} \singleton{\za}{\zc} \sepcon \Ls{\zc}{\zb}\right) \cdot \heapSize \right)(\sk,\hh) \\
        \eeqtag{algebra}
        & \left( \iverson{\za = \zb} \cdot \emp \cdot \heapSize + \iverson{\za \neq \zb} \cdot \heapSize \cdot \sup_{\gamma \in \Ints} \singleton{\za}{\zc} \sepcon \Ls{\zc}{\zb} \right)(\sk,\hh) \\
        \eeqtag{by assumption $n = |\dom{\hh}| = 0$ and thus the second summand is $0$}
        & \left( \iverson{\za = \zb} \cdot \emp \cdot \heapSize \right)(\sk,\hh) \\
        \eeqtag{$\emp \cdot \heapSize = 0$}
        & 0 \\
        \eeqtag{$(\sup_{\zc \in \Ints} \singleton{\za}{\zc})(\sk,\hh) = 0$}
        & \left(\iverson{\za \neq \zb} \cdot \sup_{\zc \in \Ints} \singleton{\za}{\zc} \sepcon \left( \Ls{\zc}{\zb} + \Len{\zc}{\zb} \right)\right)(\sk,\hh) \\
        \eeqtag{Definition of $\Len{\za}{\zb}$}
        & \Len{\za}{\zb}(\sk,\hh).
\end{align}
For $n > 0$, we have
\begin{align}
        & \left(\Ls{\za}{\zb}(\sk,\hh) \cdot \heapSize\right)(\sk,\hh) \\
        \eeqtag{Definition of $\Ls{\za}{\zb}$}
        & \left( \left(\iverson{\za = \zb} \cdot \emp + \iverson{\za \neq \zb} \cdot \sup_{\gamma \in \Ints} \singleton{\za}{\zc} \sepcon \Ls{\zc}{\zb}\right) \cdot \heapSize \right)(\sk,\hh) \\
        \eeqtag{algebra}
        & \left( \iverson{\za = \zb} \cdot \emp \cdot \heapSize + \iverson{\za \neq \zb} \cdot \heapSize \cdot \sup_{\gamma \in \Ints} \singleton{\za}{\zc} \sepcon \Ls{\zc}{\zb} \right)(\sk,\hh) \\
        \eeqtag{$\emp \cdot \heapSize = 0$}
        & \left( \iverson{\za \neq \zb} \cdot \heapSize \cdot \sup_{\gamma \in \Ints} \singleton{\za}{\zc} \sepcon \Ls{\zc}{\zb} \right)(\sk,\hh) \\
        \eeqtag{algebra}
        & \left( \iverson{\za \neq \zb} \cdot \sup_{\gamma \in \Ints} \heapSize (\singleton{\za}{\zc} \sepcon \Ls{\zc}{\zb}) \right)(\sk,\hh) \\
        \eeqtag{Theorem~\ref{thm:qsl:heap-size}.\ref{thm:qsl:heap-size:dist-full}}
        & \left( \iverson{\za \neq \zb} \cdot \sup_{\gamma \in \Ints} \left(\left(\singleton{\za}{\zc} \cdot \heapSize\right) \sepcon \Ls{\zc}{\zb} + \singleton{\za}{\zc} \sepcon \left(\Ls{\zc}{\zb} \cdot \heapSize\right) \right)\right)(\sk,\hh) \\
        \eeqtag{$\singleton{\za}{\zc} \cdot \heapSize = \singleton{\za}{\zc}$}
        & \left( \iverson{\za \neq \zb} \cdot \sup_{\gamma \in \Ints} \left(\singleton{\za}{\zc} \sepcon \Ls{\zc}{\zb} + \singleton{\za}{\zc} \sepcon \left(\Ls{\zc}{\zb} \cdot \heapSize\right) \right)\right)(\sk,\hh) \\
        \eeqtag{I.H. (notice that the heap size is reduced by one for $\Ls{\gamma}{\beta} \cdot \heapSize$)}
        & \left( \iverson{\za \neq \zb} \cdot \sup_{\gamma \in \Ints} \left(\singleton{\za}{\zc} \sepcon \Ls{\zc}{\zb} + \singleton{\za}{\zc} \sepcon \Len{\zc}{\zb} \right) \right)(\sk,\hh) \\
        \eeqtag{Theorem~\ref{thm:sep-con-distrib}.\ref{thm:sep-con-distrib:sepcon-over-plus-full}}
        & \left( \iverson{\za \neq \zb} \cdot \sup_{\gamma \in \Ints} \singleton{\za}{\zc} \sepcon \left(\Ls{\zc}{\zb} + \Len{\zc}{\zb} \right) \right)(\sk,\hh) \\
        \eeqtag{Definition of $\Len{\za}{\zb}$}
        & \Len{\za}{\zb}(\sk,\hh).
\end{align}
\end{proof}

\paragraph{Proof of Lemma~\ref{thm:ls-props}.\ref{thm:ls-props:lsls}}
\begin{proof}
   By equation~\ref{eq:app:ls-props:ls} and Lemma~\ref{thm:ls:continuous}, 
   we may apply the Kleene fixed point theorem to obtain
   \begin{align}
           \Ls{\za}{\zb}
           \eeq
           \lfp \fh\mydot \lambda(\za,\zb)\mydot \ff_{\fh}(\za,\zb)
           \eeq
           \lambda(\za,\zb)\mydot \sup_{n \in \Nats} \ff^{n}_{0}(\za,\zb).
   \end{align}
   To complete the proof, we show by induction on $n \geq 1$ that
   \begin{align}
           \sup_{\zc \in \Ints} \ff^{n}_{0}(\za,\zc) \eeq \Ls{\zc}{\zb} \ppreceq \Ls{\za}{\zb}
   \end{align}
   For $n = 1$, we have 
   \begin{align}
           & \sup_{\zc \Ints} \ff^{1}_{0}(\za,\zc) \sepcon \Ls{\zc}{\zb} \\
           \eeqtag{Definition of $\ff_{0}$}
           & \sup_{\zc} \left(
             \iverson{\za = \zc} \cdot \emp + \iverson{\za \neq \zc} \cdot \sup_{\delta \in \Ints} \singleton{\za}{\delta} \sepcon 0
             \right) \sepcon \Ls{\zc}{\zb} \\
           \eeqtag{algebra}
           & \sup_{\zc \in \Ints} \left( \iverson{\za = \zc} \cdot \emp \right) \sepcon \Ls{\zc}{\zb} \\
           \eeqtag{algebra, Theorem~\ref{thm:sep-con-monoid}}
           & \Ls{\za}{\zb}.
   \end{align}
   For the induction step, we have
   \begin{align}
           & \sup_{\zc \in \Ints} \ff^{n+1}_{0}(\za,\zc) \sepcon \Ls{\zc}{\zb} \\
           \eeqtag{Definition of $\ff_{0}$}
           & \sup_{\zc \in \Ints} \left(
           \iverson{\za = \zc} \cdot \emp + \iverson{\za \neq \zc} \cdot \sup_{\delta \in \Ints} \singleton{\za}{\delta} \sepcon \ff^{n}_{0}(\delta,\zc)
             \right) \sepcon \Ls{\zc}{\zb} \\
           \eeqtag{Lemma~\ref{thm:plus-to-max}}
           & \sup_{\zc \in \Ints} 
           \max \{ 
              \iverson{\za = \zc} \cdot \emp,
              \iverson{\za \neq \zc} \cdot \sup_{\delta \in \Ints} \singleton{\za}{\delta} \sepcon \ff^{n}_{0}(\delta,\zc) 
           \} \sepcon \Ls{\zc}{\zb} \\
           \eeqtag{Theorem~\ref{thm:sep-con-distrib}.\ref{thm:sep-con-distrib:sepcon-over-max}}
           & \sup_{\zc \in \Ints} 
           \max \{ 
              (\iverson{\za = \zc} \cdot \emp) \sepcon \Ls{\zc}{\zb},
              (\iverson{\za \neq \zc} \cdot \sup_{\delta \in \Ints} \singleton{\za}{\delta} \sepcon \ff^{n}_{0}(\delta,\zc)) \sepcon \Ls{\zc}{\zb}
           \} \\
           \eeqtag{Theorem~\ref{thm:sep-con-algebra-pure}, algebra}
           & \sup_{\zc \in \Ints} 
           \max \{ 
              (\iverson{\za = \zc} \cdot \emp) \sepcon \Ls{\zc}{\zb},
              (\iverson{\za \neq \zc} \cdot (\sup_{\delta \in \Ints} \singleton{\za}{\delta} \sepcon (\ff^{n}_{0}(\delta,\zc) \sepcon \Ls{\zc}{\zb}))
           \} \\
           \eeqtag{I.H.}
           & \sup_{\zc \in \Ints} 
           \max \{ 
              (\iverson{\za = \zc} \cdot \emp) \sepcon \Ls{\zc}{\zb},
              (\iverson{\za \neq \zc} \cdot (\sup_{\delta \in \Ints} \singleton{\za}{\delta} \sepcon (\Ls{\delta}{\zb}))
           \} \\
           \eeqtag{Theorem~\ref{thm:sep-con-monoid}.\ref{thm:sep-con-monoid:neut}, algebra}
           & \sup_{\zc \in \Ints} 
           \max \{ 
              \Ls{\za}{\zb},
              \underbrace{\iverson{\za \neq \zc} \cdot \sup_{\delta \in \Ints} \singleton{\za}{\delta} \sepcon \Ls{\delta}{\zb}}_{\ppreceq \Ls{\za}{\zb}}
           \} \\
           \eeqtag{algebra}
           & \Ls{\za}{\zb}.
   \end{align}
\end{proof}

\newpage
\section{Appendix to Section~\ref{sec:wp} (Weakest Preexpectations)}
\label{app:sec:wp}
\subsection{Proof of Theorem~\ref{thm:wp:basic} (Basic Properties of $\wpsymbol$)}
\label{app:wp:basic}

Each of the properties of Theorem~\ref{thm:wp:basic} is proven individually below:
\begin{itemize}
        \item Monotonicity, i.e. Theorem~\ref{thm:wp:basic}.\ref{thm:wp:basic:monotonicity}, is shown in Appendix~\ref{app:wp:basic:monotonicity}, p.~\pageref{app:wp:basic:monotonicity}.
        \item (Super)Linearity, i.e. Theorems~\ref{thm:wp:basic}.\ref{thm:wp:basic:super-linearity} and~\ref{thm:wp:basic}.\ref{thm:wp:basic:linearity}, are shown in Appendix~\ref{app:wp:basic:linearity}, p.~\pageref{app:wp:basic:linearity}.
        \item Preservation of $0$, i.e. Theorem~\ref{thm:wp:basic}.\ref{thm:wp:basic:preservation-of-0}, is shown in Appendix~\ref{app:wp:basic:preservation-of-0}, p.~\pageref{app:wp:basic:preservation-of-0}.
        \item $1$-Boundedness, i.e. Theorem~\ref{thm:wp:basic}.\ref{thm:wp:basic:1-boundedness}, is shown in Appendix~\ref{app:wp:basic:1-boundedness}, p.~\pageref{app:wp:basic:1-boundedness}.
        \item Continuity, i.e. Theorem~\ref{thm:wp:basic}.\ref{thm:wp:basic:continuity}, is shown in Appendix~\ref{app:wp:basic:continuity}, p.~\pageref{app:wp:basic:continuity}.
\end{itemize}
\subsection{Proof of Theorem~\ref{thm:wp:basic}.\ref{thm:wp:basic:monotonicity} (Monotonicity)}
\label{app:wp:basic:monotonicity}
\begin{proof}
We show by induction on the structure of $\hpgcl$-programs that for all $\cc \in \hpgcl$, $\wp{\cc}{\cdot}$ is a monotone function.
That is, for all $\ff,\fg \in \E$ it holds that
\begin{align}
   \ff \ppreceq \fg \qimplies \wp{\cc}{\ff} \ppreceq \wp{\cc}{\fg}.
\end{align}
The fact that
\begin{align*}
	\charwp{\guard}{\cc}{\ff}(\fh) \eeq \iverson{\neg \guard} \cdot \fh + \iverson{\guard} \cdot \wp{\cc}{\fk}
\end{align*}
is monotonic for all $\fh \in \E$ then follows from monotonicity of $\wp{\cc}{\cdot}$.
\emph{The case} $\SKIP$. 
\begin{align}
        & \wp{\SKIP}{\ff} \\
        \eeqtag{Table \ref{table:wp}}
        & \ff \\
        \ppreceqtag{by assumption: $\ff \ppreceq \fg$}
        & \fg \\
        \eeqtag{Table \ref{table:wp}}
        & 
        \wp{\SKIP}{\fg}. 
\end{align}
%
%
%
%
\emph{The case} $\ASSIGN{x}{\ee}$. 
\begin{align}
        & \wp{\ASSIGN{x}{\ee}}{\ff} \\
        \eeqtag{Table \ref{table:wp}} 
        & \ff\subst{x}{\ee} \\
        \eeqtag{algebra}
        & \lambda(\sk,\hh) \mydot \ff\left(\sk\subst{x}{\ee}, \hh \right) \\
        \ppreceqtag{by assumption: $\ff \ppreceq \fg$, monotonicity of substitution}
        & \lambda(\sk,\hh) \mydot \fg\left(\sk\subst{x}{\ee}, \hh \right) \\
        \eeqtag{algebra}
        & \fg\subst{x}{\ee} \\
        \eeqtag{Table \ref{table:wp}} 
        & \wp{\ASSIGN{x}{\ee}}{\fg}.
\end{align}
Before we continue with the next cases, we prove the following intermediate result:
Let $A, B \subseteq \PosRealsInf$ such that for all $b \in B$ there is an $a \in A$ with $a \leq b$. 
It then holds that 
\begin{align}
 \inf A \lleq \inf B.
 \label{lemma:inf}
\end{align}
We distinguish two cases: If $B = \emptyset$, then $\inf B = \infty$ and hence $\inf A \leq \infty = \inf B$.
For the remaining case $B \neq \emptyset$, it suffices to show that $\inf A$ is a lower bound of $B$ which is necessarily smaller or equal
to the \emph{greatest} lower bound of $B$. Thus we have to discharge that $\inf A \leq b$ for all $b \in B$. For that, let $b \in B$. 
Now, due to the assumption, there is an $a \in A$ such that
$a \leq b$. By definition of $\inf$ we have that $\inf A \leq a$ and therefore, by transitivity, $\inf A \leq b$. \par
We now continue with the remaining cases.
\emph{The case} $\ALLOC{x}{\vec{\ee}}$. Let $(\sk,\hh) \in \States$. Then 
\begin{align}
        & \wp{\ALLOC{x}{\vec{e}}}{\ff}(\sk,\hh) \\
        \eeqtag{Table \ref{table:wp}}
        & \displaystyle\inf_{v \in \AVAILLOC{\vec{\ee}}}  \big( \singleton{v}{\vec{\ee}} \sepimp \ff\subst{x}{v} \big)(\sk,\hh) \\
        \eeqtag{Definition of $\sepimp$}
        & \displaystyle\inf_{v \in \AVAILLOC{\vec{\ee}}} \displaystyle\inf_{\hh'} 
            \big\{ \ff\subst{x}{v}(\sk, \hh \sepcon \hh') ~\mid~ \hh \disjoint \hh' ~\text{and}~ (\sk,\hh') \models \singleton{v}{\vec{\ee}} \} \\
        \lleqtag{$\ff(\sk\subst{x}{v},\hh\sepcon \hh') \leq \fg(\sk\subst{x}{v},\hh \sepcon \hh')$, then apply (\ref{lemma:inf})}
        & \displaystyle\inf_{v \in \AVAILLOC{\vec{\ee}}} \displaystyle\inf_{\hh'} 
        \big\{\fg\subst{x}{v}(\sk, \hh \sepcon \hh') ~\mid~ \hh \disjoint \hh' ~\text{and}~ (\sk,\hh') \models \singleton{v}{\vec{\ee}} \} \\
        \eeqtag{Definition of $\sepimp$}
        & \displaystyle\inf_{v \in \AVAILLOC{\vec{\ee}}} \big( \singleton{v}{\vec{\ee}} \sepimp \fg\subst{x}{v} \big) (\sk,\hh) \\
        \eeqtag{Table \ref{table:wp}}
        & \wp{\ALLOC{x}{\vec{e}}}{\fg}(\sk,\hh). 
\end{align}
\emph{The case} $\ASSIGNH{x}{\ee}$. Let $(\sk,\hh) \in \States$. We distinguish two cases: $\sk(\ee) \in \dom{\hh}$ and $\sk(\ee) \not\in \dom{\hh}$.
First, assume $\sk(\ee) \in \dom{\hh}$.
Suppose without loss of generality that $\hh(\sk(\ee))=v'$.
Moreover, let $h_{\ee,v'}$ denote the heap with $\dom{h_{\ee,v'}} = \{\sk(\ee) \}$ and $h_{\ee,v'}(\sk(\ee)) = v'$. 
Furthermore, let $\hh'$ be the heap such that 
$\hh' \sepcon h_{\ee,v'} = \hh$. We then have
\begin{align}
        & \wp{\ASSIGNH{x}{\ee}}{\ff}(\sk,\hh) \\
        \eeqtag{by assumption, see above}
        & \wp{\ASSIGNH{x}{\ee}}{\ff}(\sk,\hh' \sepcon h_{\ee,v'}) \\
        \eeqtag{Table \ref{table:wp}} 
        & \displaystyle\sup_{v \in \Ints}  \Big( \singleton{\ee}{v} \sepcon \bigl( \singleton{\ee}{v} \sepimp \ff\subst{x}{v} \bigr) \Big)(\sk, \hh' \sepcon h_{\ee,v'} )  \\
        \eeqtag{$\hh(\sk(\ee)) = v'$}
        & \singleton{\ee}{v'} \sepcon \bigl( \singleton{\ee}{v'} \sepimp \ff\subst{x}{v'} \bigr)(\sk, \hh' \sepcon h_{\ee,v'} ) \\
        \eeqtag{$\singleton{\ee}{v'}(\sk,h_{\ee,v'}) = 1$}
        & \bigl( \singleton{\ee}{v'} \sepimp \ff\subst{x}{v'} \bigr)(\sk, \hh') \\
        \eeqtag{$\sk(\ee) \not\in \dom{\hh'}$}
        & \ff\subst{x}{v'}(\sk, \hh' \sepcon h_{\ee,v'}) \\
        \lleqtag{$\ff \lleq \fg$ by assumption}
        & \fg\subst{x}{v'}(\sk, \hh' \sepcon h_{\ee,v'}) \\
        \eeqtag{$\sk(\ee) \not\in \dom{\hh'}$}
        & \bigl( \singleton{\ee}{v'} \sepimp \fg\subst{x}{v'} \bigr)(\sk, \hh') \\
        \eeqtag{$\singleton{\ee}{v'}(\sk,h_{\ee,v'}) = 1$}
        & \singleton{\ee}{v'} \sepcon \bigl( \singleton{\ee}{v'} \sepimp \fg\subst{x}{v'} \bigr)(\sk, \hh' \sepcon h_{\ee,v'} ) \\
        \eeqtag{$\hh(\sk(\ee)) = v'$}
        & \displaystyle\sup_{v \in \Ints} \Big( \singleton{\ee}{v} \sepcon \bigl( \singleton{\ee}{v} \sepimp \fg\subst{x}{v} \bigr) \Big)(\sk, \hh' \sepcon h_{\ee,v'} ) \\
        \eeqtag{Table \ref{table:wp}}
        & \wp{\ASSIGNH{x}{\ee}}{\fg}(\sk,\hh' \sepcon h_{\ee,v'}) \\
        \eeqtag{by assumption, see above}
        & \wp{\ASSIGNH{x}{\ee}}{\fg}(\sk,\hh).
\end{align}
Second, assume $\sk(\ee) \not\in \dom{\hh}$. Then
\begin{align}
        & \wp{\ASSIGNH{x}{\ee}}{\ff}(\sk,\hh) \\
        \eeqtag{Table \ref{table:wp}}
        & \displaystyle\sup_{v \in \Ints} \Big( \singleton{\ee}{v} \sepcon \bigl( \singleton{\ee}{v} \sepimp \ff\subst{x}{v} \bigr) \Big)(\sk, \hh ) \\
        \eeqtag{by assumption $\singleton{\ee}{v}(\sk,\hh) = 0$}
        & 0 \\
        \lleqtag{$0$ is least element of domain $\E$}
        & \wp{\ASSIGNH{x}{\ee}}{\fg}(\sk,\hh).
\end{align}
\emph{The case} $\HASSIGN{\ee}{\ee'}$. Let $(\sk,\hh) \in \States$. We distinguish two cases: $\sk(\ee) \in \dom{\hh}$ and $\sk(\ee) \not\in \dom{\hh}$. \\ \\
First, assume $\sk(\ee) \in \dom{\hh}$. 
Suppose without loss of generality that $\hh(\sk(\ee)) = s(v')$ and let $h_{\ee,v'}$ be the heap with $\dom{h_{\ee,v'}} = \{ \sk(\ee) \}$ and $h_{\ee,v'}(\sk(\ee)) = v'$.
Furthermore, let $\hh'$ be the heap with $\hh' \sepcon h_{\ee,v'} = \hh$. Finally, let $h_{\ee,\ee'}$ denote the heap with $\dom{h_{\ee,\ee'}} = \{\sk(\ee) \}$ and $h_{\ee,\ee'}(\sk(\ee)) = \sk(\ee')$. Then
\begin{align}
        & \wp{\HASSIGN{\ee}{\ee'}}{\ff}(\sk,\hh) \\
        \eeqtag{by assumption, see above}
        & \wp{\HASSIGN{\ee}{\ee'}}{\ff}(\sk,\hh' \sepcon h_{\ee,v'}) \\
        \eeqtag{Table \ref{table:wp}} 
        & \Big( \validpointer{\ee} \sepcon \bigl(\singleton{\ee}{\ee'} \sepimp \ff \bigr) \Big) (\sk,\hh' \sepcon h_{\ee,v'}) \\
        \eeqtag{by assumption, $\validpointer{\ee}(\sk,\hh) = 1$}
        & \bigl(\singleton{\ee}{\ee'} \sepimp \ff \bigr)(\sk,\hh') \\
        \eeqtag{$\sk(\ee) \not\in \dom{\hh'}$} 
        & \ff(\sk,\hh' \sepcon h_{\ee,\ee'}) \\
        \lleqtag{$\ff \ppreceq \fg$ by assumption}
        & \fg(\sk,\hh' \sepcon h_{\ee,\ee'}) \\
        \eeqtag{$\sk(\ee) \not\in \dom{\hh'}$} 
        & \bigl(\singleton{\ee}{\ee'} \sepimp \fg \bigr)(\sk,\hh') \\
        \eeqtag{by assumption, $\validpointer{\ee}(\sk,\hh) = 1$}
        & \Big( \validpointer{\ee} \sepcon \bigl(\singleton{\ee}{\ee'} \sepimp \fg \bigr) \Big) (\sk,\hh' \sepcon h_{\ee,v'}) \\
        \eeqtag{by assumption, $\validpointer{\ee}(\sk,\hh) = 1$}
        & \Big( \validpointer{\ee} \sepcon \bigl(\singleton{\ee}{\ee'} \sepimp \fg \bigr) \Big) (\sk,\hh) \\
        \eeqtag{Table \ref{table:wp}} 
        & \wp{\HASSIGN{\ee}{\ee'}}{\fg}(\sk,\hh).
\end{align}
Second, assume $\sk(\ee) \not\in \dom{\hh}$. Then
\begin{align}
        & \wp{\HASSIGN{\ee}{\ee'}}{\ff}(\sk,\hh) \\
        \eeqtag{Table \ref{table:wp}} 
        & \Big( \validpointer{\ee} \sepcon \bigl(\singleton{\ee}{\ee'} \sepimp \ff \bigr) \Big) (\sk,\hh) \\
        \eeqtag{by assumption}
        & 0 \\
        \lleqtag{$0$ is least element of domain $\E$}
        & \wp{\HASSIGN{\ee}{\ee'}}{\fg}(\sk,\hh).
\end{align}
\emph{The case} $\FREE{x}$. 
Let $(\sk,\hh) \in \States$. 
We distinguish two cases: $\sk(\ee) \in \dom{\hh}$ and $\sk(\ee) \not\in \dom{\hh}$. \\ \\
First, $\sk(\ee) \in \dom{\hh}$. 
Suppose without loss of generality that $\hh(\sk(\ee)) = v'$ and let $h_{\ee,v'}$ be the heap with $\dom{h_{\ee,v'}} = \{ \sk(\ee) \}$ and $h_{\ee,v'}(\sk(\ee)) = v'$.
Furthermore, let $\hh'$ be the heap with $\hh' \sepcon h_{\ee,v'} = \hh$. Then
\begin{align}
        & \wp{\FREE{x}}{\ff}(\sk,\hh) \\
        \eeqtag{by assumption, see above}
        & \wp{\FREE{x}}{\ff}(\sk, \hh' \sepcon h_{\ee,v'}) \\
        \eeqtag{Table \ref{table:wp}}
        & \big( \validpointer{\ee} \sepcon \ff \big) (\sk,\hh' \sepcon h_{\ee,v'}) \\
        \eeqtag{$\validpointer{\ee}(\sk,h_{\ee,v'}) = 1$}
        & \ff(\sk,\hh') \\
        \lleqtag{$\ff \preceq \fg$ by assumption}
        & \fg(\sk,\hh') \\
        \eeqtag{$\validpointer{\ee}(\sk,h_{\ee,v'}) = 1$}
        & \big( \validpointer{\ee} \sepcon \fg \big) (\sk,\hh' \sepcon h_{\ee,v'}) \\
        \eeqtag{by assumption, see above}
        & \big( \validpointer{\ee} \sepcon \fg \big) (\sk,\hh) \\
        \eeqtag{Table \ref{table:wp}}
        & \wp{\FREE{x}}{\fg}(\sk,\hh).
\end{align}
%
%
%
%
\emph{The case} $\COMPOSE{\cc_1}{\cc_2}$.
\begin{align}
        & \wp{\COMPOSE{\cc_1}{\cc_2}}{\ff} \\
        \eeqtag{Table~\ref{table:wp}} 
        & \wp{\cc_1}{\wp{\cc_2}{\ff}} \\
        \ppreceqtag{By I.H. on $\cc_2$ it holds that $\wp{\cc_2}{\ff} \preceq \wp{\cc_2}{\fg}$. Hence, I.H. on $\cc_1$ yields}
        & \wp{\cc_1}{\wp{\cc_2}{\fg}} \\
        \eeqtag{Table~\ref{table:wp}} 
        & \wp{\COMPOSE{\cc_1}{\cc_2}}{\fg}. 
\end{align}
\emph{The case} $\ITE{\guard}{\cc_1}{\cc_2}$. 
Let $(\sk,\hh) \in \States$. 
We distinguish two cases: $\iverson{\guard}(\sk,\hh) = 1$ and $\iverson{\neg \guard}(\sk,\hh) = 1$.

For $\iverson{\guard}(\sk,\hh) = 1$, consider the following:
\begin{align}
        & \wp{\ITE{\guard}{\cc_1}{\cc_2}}{\ff}(\sk,\hh) \\
        \eeqtag{Table~\ref{table:wp}} 
        & \big( \iverson{\guard} \cdot \wp{\cc_1}{\ff} + \iverson{\neg \guard} \cdot \wp{\cc_2}{\ff} \big)(\sk,\hh) \\
        \eeqtag{$\iverson{\guard}(\sk,\hh) = 1$ by assumption}
        & \wp{\cc_1}{\ff}(\sk,\hh) \\
        \lleqtag{I.H.\ on $\cc_1$} 
        & \wp{\cc_1}{\fg}(\sk,\hh) \\
        \eeqtag{$\iverson{\guard}(\sk,\hh) = 1$ by assumption}
        & \big( \iverson{\guard} \cdot \wp{\cc_1}{\fg} + \iverson{\neg \guard} \cdot \wp{\cc_2}{\fg} \big)(\sk,\hh) \\
        \eeqtag{Table~\ref{table:wp}} 
        & \wp{\ITE{\guard}{\cc_1}{\cc_2}}{\fg}(\sk,\hh).
\end{align}
For $\iverson{\neg \guard}(\sk,\hh) = 1$, the proof is analogous.
\emph{The case} $\PCHOICE{\cc_1}{p}{\cc_2}$. 
\begin{align}
        & \wp{\PCHOICE{\cc_1}{p}{\cc_2}}{\ff} \\
        \eeqtag{Table \ref{table:wp}}
        &\eeq p \cdot \underbrace{\wp{\cc_1}{\ff}}_{\preceq \wp{\cc_1}{\fg}} + (1- p) \cdot \underbrace{\wp{\cc_2}{\ff}}_{\preceq \wp{\cc_2}{\fg}} \\
        \ppreceqtag{I.H.\ on $\cc_1$ and I.H.\ on $\cc_2$} 
        & p \cdot \wp{\cc_1}{\fg} + (1- p) \cdot \wp{\cc_2}{\fg} \\
        \eeqtag{Table \ref{table:wp}}
        & \wp{\PCHOICE{\cc_1}{p}{\cc_2}}{\fg}.
\end{align}
\emph{The case} $\cc = \WHILEDO{\guard}{\cc_1}$. 
First notice that, since $\wp{\cc_1}{\cdot}$ is a monotone function by I.H., the functional
\begin{align}
   \charwp{\guard}{\cc_1}{\fh}(\fk) \eeq \iverson{\neg \guard} \cdot \fh + \iverson{\guard} \cdot \wp{\cc_1}{\fk}
\end{align}
is also monotonic for all $\fh \in \E$. Hence, by the constructive version of the fixed point theorem by Tarski and Knaster (cf.~\cite{cousot1979constructive}), 
for every expectation $\fh\in \E$, there is an ordinal $\oa$ such that
\begin{align}
   \wp{\WHILEDO{\guard}{\cc_1}}{\fh} \eeq \lfp \fk . \charwp{\guard}{\cc_1}{\fh}(\fk) \eeq \charwpn{\guard}{\cc_1}{\fh}{\oa}(0).
\end{align}
Thus, in order to prove that $\wp{\WHILEDO{\guard}{\cc_1}}{\cdot}$ is monotonic, it suffices to show that for $\ff,\fg \in \E$ with $\ff \preceq \fg$ and all ordinals $\oa$
\begin{align}
   \charwpn{\guard}{\cc_1}{\ff}{\oa}(0) \ppreceq \charwpn{\guard}{\cc_1}{\fg}{\oa}(0).
\end{align}
We proceed by transfinite induction on $\oa$.\\ \\
\emph{Induction Base $\oa=0$}. This case is trivial since 
\begin{align}
   \charwpn{\guard}{\cc_1}{\ff}{0}(0) \eeq 0 \preceq 0 \eeq \charwpn{\guard}{\cc_1}{\fg}{0}(0).
\end{align}
\emph{Successor Ordinals.} For successor ordinals, assume that $\charwpn{\guard}{\cc_1}{\ff}{\oa}(0) \preceq \charwpn{\guard}{\cc_1}{\fg}{\oa}(0)$. We derive
\begin{align}
        & \charwpn{\guard}{\cc_1}{\ff}{\oa+1}(0) \\
        \eeqtag{Definition of $\charwp{\guard}{\cc_1}{\ff}$}
        & \iverson{\neg \guard} \cdot \ff + \iverson{\guard} \cdot \wp{\cc_1}{\charwpn{\guard}{\cc_1}{\ff}{\oa}(0)} \\
        \ppreceqtag{$\ff \preceq \fg$, I.H.}
        & \iverson{\neg \guard} \cdot \fg + \iverson{\guard} \cdot \wp{\cc_1}{\charwpn{\guard}{\cc_1}{\fg}{\oa}(0)} \\
        \eeqtag{Definition of $\charwp{\guard}{\cc_1}{\ff}$}
        & \charwpn{\guard}{\cc_1}{\fg}{\oa+1}(0).
\end{align}
\emph{Limit Ordinals.}  
Let $\oa$ be a limit ordinal and for all $\ob < \oa$, $\charwp{\guard}{\cc_1}{\ff}^{\ob}(0) \preceq \charwp{\guard}{\cc_1}{\fg}^{\ob}(0)$. We have
\begin{align}
        & \charwpn{\guard}{\cc_1}{\ff}{\oa}(0) \\
        \eeqtag{Def.\ $\charwpn{\guard}{\cc_1}{\ff}{\oa}(0)$ for $\oa$ limit ordinal} 
        & \sup_{\ob < \oa} \charwpn{\guard}{\cc_1}{\ff}{\ob}(0) \\
        \ppreceqtag{I.H. on $\charwpn{\guard}{\cc_1}{\ff}{\ob}(0)$} 
        & \sup_{\ob < \oa} \charwpn{\guard}{\cc_1}{\fg}{\ob}(0) \\
        \eeqtag{Def.\ $\charwpn{\guard}{\cc_1}{\ff}{\oa}(0)$ for $\oa$ limit ordinal} 
        & \charwpn{\guard}{\cc_1}{\fg}{\oa}(0).
\end{align}
\end{proof}
\subsection{Proof of Theorems~\ref{thm:wp:basic}.\ref{thm:wp:basic:super-linearity} and~\ref{thm:wp:basic}.\ref{thm:wp:basic:linearity} (Linearity)}
\label{app:wp:basic:linearity}

\begin{proof}
By induction on the structure of a $\hpgcl$ program $\cc$. First, we consider the base cases.
Notice that in case we show linearity, super-linearity follows immediately.\\ \\
\emph{The case} $\cc = \SKIP$. We show linearity as follows:
\begin{align}
   &\wp{\SKIP}{a \cdot \ff + \fg } \\
   \eeqtag{Table \ref{table:wp}}
   & a \cdot \ff + \fg \\
   \eeqtag{Table \ref{table:wp}} 
   & a \cdot \wp{\SKIP}{\ff} + \wp{\SKIP}{\fg}~.
\end{align}
\emph{The case} $\cc = \ASSIGN{x}{\ee}$. We show linearity as follows:
\begin{align}
   &\wp{\ASSIGN{x}{\ee}}{a \cdot \ff + \fg } \\
   \eeqtag{Table \ref{table:wp}}
   & (a \cdot \ff + \fg)\subst{x}{\ee} \\
   \eeqtag{Substitution is distributive}
   &a \cdot \ff\subst{x}{\ee} + \fg\subst{x}{\ee} \\
   \eeqtag{Table \ref{table:wp}} 
   & a \cdot \wp{\ASSIGN{x}{\ee}}{\ff} + \wp{\ASSIGN{x}{\ee}}{\fg}~.
\end{align}
\emph{The case} $\cc = \ALLOC{x}{\vec{\ee}}$. We show super-linearity point-wise as follows:
We make use of two facts. First, for two subsets of the non-negative real numbers $A, B \subseteq \PosReals$ and all $C \subseteq A$, $D \subseteq B$ it holds that
\begin{align}\label{eqn:linearity-inf}
     \inf \left\{ a+b ~\mid~ a \in C, b \in D \right\} ~{} \geq {} ~
     \inf \left\{ a ~\mid~ a \in A\} \right\}
     + \inf \left\{ b ~\mid~  b \in B \right\}~.
\end{align}
Second, for every $A \subseteq \PosReals$ and every $\cc \in \PosReals$, we have
\begin{align}\label{eqn:linearity-inf-mult}
     \inf \left\{ \cc \cdot a ~\mid~ a \in A \right\} ~{} \geq {} ~
     \cc \cdot \inf \left\{ a ~\mid~ a \in A \right\}~.
\end{align}
Now let $(\sk,\hh) \in \States$. We have
\begin{align}
   &\wp{\ALLOC{x}{\vec{\ee}}}{a \cdot \ff + \fg}(\sk,\hh) \\
   \eeqtag{Table \ref{table:wp}}
   &\big( \displaystyle\inf_{v \in \AVAILLOC{\vec{\ee}}} \singleton{v}{\vec{\ee}} \sepimp (a \cdot \ff + \fg)\subst{x}{v} \big) (\sk,\hh) \\
   \eeqtag{Definition of $\sepimp$}
   & \displaystyle\inf_{v \in \AVAILLOC{\vec{\ee}}} 
        \inf_{\hh'} \big\{ (a \cdot \ff + \fg)\subst{x}{v} \big) (\sk,\hh \sepcon \hh') ~\mid~ \hh \disjoint \hh', (\sk,\hh') \models \singleton{v}{\vec{\ee}}  \big\} \\
   \eeqtag{Substitution distributes}
   & \displaystyle\inf_{v \in \AVAILLOC{\vec{\ee}}} 
        \inf_{\hh'} \big\{ (a \cdot \ff)\subst{x}{v} + \fg\subst{x}{v} \big) (\sk,\hh \sepcon \hh') ~\mid~ \hh \disjoint \hh', (\sk,\hh') \models \singleton{v}{\vec{\ee}}  \big\} \\
   \ggeqtag{Equation \ref{eqn:linearity-inf}}
   & \displaystyle\inf_{v \in \AVAILLOC{\vec{\ee}}} \big(
        \inf_{\hh'} \big\{ (a \cdot \ff)\subst{x}{v} \big) (\sk,\hh \sepcon \hh') ~\mid~ \hh \disjoint \hh', (\sk,\hh') \models \singleton{v}{\vec{\ee}}  \big\} \\
   &\quad + \inf_{\hh'} \big\{ \fg\subst{x}{v} \big) (\sk,\hh \sepcon \hh') ~\mid~ \hh \disjoint \hh', (\sk,\hh') \models \singleton{v}{\vec{\ee}}  \big\}  \big)
    \notag \\
   \ggeqtag{Equation \ref{eqn:linearity-inf}}
   & \displaystyle\inf_{v \in \AVAILLOC{\vec{\ee}}} 
        \inf_{\hh'} \big\{ (a \cdot \ff)\subst{x}{v} \big) (\sk,\hh \sepcon \hh') ~\mid~ \hh \disjoint \hh', (\sk,\hh') \models \singleton{v}{\vec{\ee}}  \big\} \\
   &\quad +  \displaystyle\inf_{v \in \AVAILLOC{\vec{\ee}}} \inf_{\hh'} \big\{ \fg\subst{x}{v} \big) (\sk,\hh \sepcon \hh') ~\mid~ \hh \disjoint \hh', (\sk,\hh') \models \singleton{v}{\vec{\ee}}  \big\}  
    \notag \\
   \ggeqtag{Equation \ref{eqn:linearity-inf-mult}}
   &\displaystyle\inf_{v \in \AVAILLOC{\vec{\ee}}} 
        a \cdot \inf_{\hh'} \big\{  \ff\subst{x}{v} \big) (\sk,\hh \sepcon \hh') ~\mid~ \hh \disjoint \hh', (\sk,\hh') \models \singleton{v}{\vec{\ee}}  \big\} \\
   &\quad +  \displaystyle\inf_{v \in \AVAILLOC{\vec{\ee}}} \inf_{\hh'} \big\{ \fg\subst{x}{v} \big) (\sk,\hh \sepcon \hh') ~\mid~ \hh \disjoint \hh', (\sk,\hh') \models \singleton{v}{\vec{\ee}}  \big\}  
    \notag \\
     \ggeqtag{Equation \ref{eqn:linearity-inf-mult}}
   &a \cdot \displaystyle\inf_{v \in \AVAILLOC{\vec{\ee}}} 
         \inf_{\hh'} \big\{  \ff\subst{x}{v} \big) (\sk,\hh \sepcon \hh') ~\mid~ \hh \disjoint \hh', (\sk,\hh') \models \singleton{v}{\vec{\ee}}  \big\} \\
   &\quad +  \displaystyle\inf_{v \in \AVAILLOC{\vec{\ee}}} \inf_{\hh'} \big\{ \fg\subst{x}{v} \big) (\sk,\hh \sepcon \hh') ~\mid~ \hh \disjoint \hh', (\sk,\hh') \models \singleton{v}{\vec{\ee}}  \big\}  
    \notag \\
    \ggeqtag{Definition of $\sepimp$}
    &a \cdot \big( \displaystyle\inf_{v \in \AVAILLOC{\vec{\ee}}} 
         \singleton{v}{\vec{\ee}} \sepimp \ff\subst{x}{v} \big) (\sk,\hh) \\
   &\quad +  \big( \displaystyle\inf_{v \in \AVAILLOC{\vec{\ee}}} 
         \singleton{v}{\vec{\ee}} \sepimp \fg\subst{x}{v} \big) (\sk,\hh)
    \notag \\
    \eeqtag{Table \ref{table:wp}}
    &a \cdot \wp{\ALLOC{x}{\vec{\ee}}}{\ff}(\sk,\hh)  + \wp{\ALLOC{x}{\vec{\ee}}}{\fg}(\sk,\hh)~.
\end{align}
\emph{The case $\cc = \ASSIGNH{x}{\ee}$}. We show linearity as follows:
\begin{align}
   &\wp{\ASSIGNH{x}{\ee}}{a \cdot \ff +\fg} \\
   \eeqtag{Table \ref{table:wp}}
   &\displaystyle\sup_{v \in \Ints} \singleton{\ee}{v} \sepcon \bigl( \singleton{\ee}{v} \sepimp (a \cdot \ff +\fg)\subst{x}{v} \bigr) \\
   \eeqtag{Alternative version of the rule for heap lookup} 
   &\displaystyle\sup_{v \in \Ints} \containsPointer{\ee}{v} \cdot (a \cdot \ff +\fg)\subst{x}{v} \bigr) \\
   \eeqtag{Substitution distributes}
   &\displaystyle\sup_{v \in \Ints} \containsPointer{\ee}{v} \cdot (a \cdot \ff\subst{x}{v} +\fg \subst{x}{v}) \bigr) \\
   \eeqtag{Distributivity of $\cdot$}
   &\displaystyle\sup_{v \in \Ints} \left( \containsPointer{\ee}{v} \cdot (a \cdot \ff\subst{x}{v})+   \containsPointer{\ee}{v} \cdot \fg\subst{x}{v} \right)  \\
   \eeqtag{$\forall (\sk,\hh) \,\exists$ at most one $v \in \Ints$ such that $\containsPointer{\ee}{v}(\sk,\hh) = 1$}
   &\displaystyle\sup_{v \in \Ints} \containsPointer{\ee}{v} \cdot (a \cdot \ff\subst{x}{v})+  \displaystyle\sup_{v \in \Ints}  \containsPointer{\ee}{v} \cdot \fg\subst{x}{v}   \\
   \eeqtag{$a$ does not depend on $v$}
   &a \cdot \displaystyle\sup_{v \in \Ints} \containsPointer{\ee}{v} \cdot \ff\subst{x}{v}+  \displaystyle\sup_{v \in \Ints}  \containsPointer{\ee}{v} \cdot \fg\subst{x}{v}   \\
   \eeqtag{Alternative version of the rule for heap lookup}
   &a \cdot \big( \displaystyle\sup_{v \in \Ints} \singleton{\ee}{v} \sepcon \bigl( \singleton{\ee}{v} \sepimp \ff\subst{x}{v} \bigr) \big)
     +  \displaystyle\sup_{v \in \Ints} \singleton{\ee}{v} \sepcon \bigl( \singleton{\ee}{v} \sepimp \fg\subst{x}{v} \bigr)  \\
   \eeqtag{Table \ref{table:wp}}
   &a \cdot \wp{\ASSIGNH{x}{\ee}}{\ff}
     +  \wp{\ASSIGNH{x}{\ee}}{\fg}~.  
\end{align}
\emph{The case $\cc = \HASSIGN{\ee}{\ee'}$}. We show linearity point-wise as follows:
Let $(\sk,\hh) \in \States$. We distinguish the cases $\sk(\ee) \in \dom{\hh}$ and $\sk(\ee) \not\in \dom{\hh}$. If $\sk(\ee) \not\in \dom{\hh}$, then
\begin{align}
     &\wp{\HASSIGN{\ee}{\ee'}}{a \cdot \ff +\fg }(\sk,\hh) \\
     \eeqtag{Table \ref{table:wp}}
     &\big( \validpointer{\ee} \sepcon \bigl(\singleton{\ee}{\ee'} \sepimp (a \cdot \ff +\fg) \bigr) \big) (\sk,\hh)  \\
     \eeqtag{$\sk(\ee) \not\in \dom{\hh}$}
     & 0 \\
     \eeqtag{$\sk(\ee) \not\in \dom{\hh}$}
     &a \cdot \big( \validpointer{\ee} \sepcon \bigl(\singleton{\ee}{\ee'} \sepimp \ff \bigr) \big) (\sk,\hh)
     +\big( \validpointer{\ee} \sepcon \bigl(\singleton{\ee}{\ee'} \sepimp \fg \bigr) \big) (\sk,\hh) \\
     \eeqtag{Table \ref{table:wp}}
     &a \cdot \wp{\HASSIGN{\ee}{\ee'}}{\ff}(\sk,\hh) + \wp{\HASSIGN{\ee}{\ee'}}{\fg}(\sk,\hh)~.
\end{align}
Now let $\sk(\ee) \in \dom{\hh}$. For two arithmetic expressions $\ee_1, \ee_2$, we denote by $\hh_{\ee_1,\ee_2}$ the heap 
with $\{ \sk(\ee_1) \} = \dom{\hh_{\ee_1,\ee_2}}$ and $\hh_{\ee_1,\ee_2}(\sk(\ee_1)) = \sk(\ee_2)$. The heap $\hh$ is thus of the form
$\hh = \hh' \sepcon \hh_{\ee,v}$ for some heap $\hh'$ and some $v \in \Ints$. We have
\begin{align}
     &\wp{\HASSIGN{\ee}{\ee'}}{a \cdot \ff +\fg }(\sk,\hh) \\
     \eeqtag{Table \ref{table:wp}}
     &\big( \validpointer{\ee} \sepcon \bigl(\singleton{\ee}{\ee'} \sepimp (a \cdot \ff +\fg) \bigr) \big) (\sk,\hh) \\
     \eeqtag{Assumption}
     &\big( \validpointer{\ee} \sepcon \bigl(\singleton{\ee}{\ee'} \sepimp (a \cdot \ff +\fg) \bigr) \big) (\sk,\hh' \sepcon \hh_{\ee,v}) \\
     \eeqtag{$(\validpointer{\ee} \sepcon u)(\sk,\hh' \sepcon \hh_{\ee,v}) = u (\sk,\hh')$ for all $u \in \E$}
     &\big( \bigl(\singleton{\ee}{\ee'} \sepimp (a \cdot \ff +\fg) \bigr) \big) (\sk,\hh') \\
     \eeqtag{$\sk(\ee) \not\in \dom{\hh'}$}
     &(a \cdot \ff +\fg) (\sk,\hh' \sepcon \hh_{\ee,\ee'}) \\
     \eeqtag{Definition of $\cdot$ and $+$ w.r.t.\ $\E$}
     &a \cdot \ff(\sk,\hh' \sepcon \hh_{\ee,\ee'}) + \fg(\sk,\hh' \sepcon \hh_{\ee,\ee'})  \\
     \eeqtag{$\sk(\ee) \not\in \dom{\hh'}$}
     &a \cdot \bigl(\singleton{\ee}{\ee'} \sepimp \ff \bigr) (\sk,\hh') + \bigl(\singleton{\ee}{\ee'} \sepimp  \fg \bigr) (\sk,\hh')  \\
     \eeqtag{$u (\sk,\hh') = (\validpointer{\ee} \sepcon u)(\sk,\hh' \sepcon \hh_{\ee,v})$ for all $u \in \E$}
     &a \cdot \big( \validpointer{\ee} \sepcon \bigl(\singleton{\ee}{\ee'} \sepimp \ff \bigr) \big) (\sk,\hh' \sepcon \hh_{\ee,\ee'}) \\
     &\quad + \big(  \validpointer{\ee} \sepcon \bigl(\singleton{\ee}{\ee'} \sepimp  \fg \bigr)\big) (\sk,\hh' \sepcon \hh_{\ee,\ee'}) 
     \notag \\
     \eeqtag{Table \ref{table:wp}}
     &a \cdot \wp{\HASSIGN{\ee}{\ee'}}{\ff}(\sk,\hh) + \wp{\HASSIGN{\ee}{\ee'}}{\fg}(\sk,\hh)~.
\end{align}
\emph{The case $\cc = \FREE{\ee}$}. We show linearity point-wise as follows:
We distinguish the cases $\sk(\ee) \in \dom{\hh}$ and $\sk(\ee) \not\in \dom{\hh}$. If $\sk(\ee) \not\in \dom{\hh}$, then
\begin{align}
    &\wp{\FREE{\ee}}{a \cdot \ff +\fg }(\sk,\hh) \\
    \eeqtag{Table \ref{table:wp}}
    &\big( \validpointer{\ee} \sepcon ( a \cdot \ff +\fg) \big) (\sk,\hh) \\
    \eeqtag{$\sk(\ee) \not\in \dom{\hh}$}
    &0 \\
    \eeqtag{$\sk(\ee) \not\in \dom{\hh}$} 
    & a \cdot \big( \validpointer{\ee} \sepcon \ff \big) (\sk,\hh) + \big( \validpointer{\ee} \sepcon \fg \big) (\sk,\hh) \\
    \eeqtag{Table \ref{table:wp}}
    & a \cdot \wp{\FREE{\ee}}{\ff} (\sk,\hh) + \wp{\FREE{\ee}}{\fg}(\sk,\hh)~.
\end{align}
If $\sk(\ee) \in \dom{\hh}$, then the heap $\hh$ is of the form $\hh = \hh' \sepcon \hh_{\ee, v}$ for some heap $\hh'$ and some $v \in \Ints$.
We have
\begin{align}
     &\wp{\FREE{\ee}}{a \cdot \ff +\fg }(\sk,\hh) \\
    \eeqtag{Table \ref{table:wp}}
    &\big( \validpointer{\ee} \sepcon ( a \cdot \ff +\fg) \big) (\sk,\hh) \\
    \eeqtag{Assumption}
    &\big( \validpointer{\ee} \sepcon ( a \cdot \ff +\fg) \big) (\sk,\hh' \sepcon \hh_{\ee, v}) \\
    \eeqtag{$(\validpointer{\ee} \sepcon u)(\sk,\hh' \sepcon \hh_{\ee, v}) = u (\sk,\hh')$ for all $u \in \E$}
    &( a \cdot \ff +\fg) (\sk,\hh') \\
    \eeqtag{Definition of $\cdot$ and $+$ w.r.t.\ $\E$}
    &a \cdot \ff(\sk,\hh') + \fg(\sk,\hh') \\
    \eeqtag{$u (\sk,\hh')= (\validpointer{\ee} \sepcon u)(\sk,\hh' \sepcon \hh_{\ee, v})$ for all $u \in \E$}
    &a \cdot (\validpointer{\ee} \sepcon \ff) (\sk,\hh' \sepcon \hh_{\ee, v}) + (\validpointer{\ee} \sepcon \fg) (\sk,\hh' \sepcon \hh_{\ee, v}) \\
    \eeqtag{Table \ref{table:wp}}
    &a \cdot \wp{\FREE{\ee}}{\ff}(\sk,\hh' \sepcon \hh_{\ee, v}) + \wp{\FREE{\ee}}{\fg}(\sk,\hh' \sepcon \hh_{\ee, v}) \\
    \eeqtag{Assumption}
    &a \cdot \wp{\FREE{\ee}}{\ff}(\sk,\hh) + \wp{\FREE{\ee}}{\fg}(\sk,\hh)~.
\end{align}
As the induction hypothesis now assume that for some arbitrary, but fixed, $\cc_1, \cc_2 \in \hpgcl$, all $\ff,\fg \in \E$ and
all $a \in \PosReals$ it holds that
\begin{align}
  &\wp{\cc_1}{a \cdot \ff +\fg} \eeq a \cdot \wp{\cc_1}{\ff} + \wp{\cc_1}{\fg}\\
  \text{and} \quad & \wp{\cc_2}{a \cdot \ff +\fg} \eeq a \cdot \wp{\cc_2}{\ff} + \wp{\cc_1}{\fg}~.
\end{align}
Moreover, assume that for some arbitrary, but fixed, $\cc_1', \cc_2' \in \hpgcl$ containing no instances of $\ALLOC{x}{\ee}$, all $\ff,\fg \in \E$, and
all $a \in \PosReals$ it holds that
\begin{align}
  &\wp{\cc_1'}{a \cdot \ff +\fg} \ssucceq a \cdot \wp{\cc_1'}{\ff} + \wp{\cc_1'}{\fg}\\
  \text{and}\quad & \wp{\cc_2'}{a \cdot \ff +\fg} \ssucceq a \cdot \wp{\cc_2'}{\ff} + \wp{\cc_2'}{\fg}~.
\end{align}
\emph{The case $\cc = \COMPOSE{\cc_1}{\cc_2}$}. We have
\begin{align}
    &\wp{\COMPOSE{\cc_1}{\cc_2}}{a\cdot \ff + \fg} \\
    \eeqtag{Table \ref{table:wp}}
    &\wp{\cc_1}{\wp{\cc_2}{a\cdot \ff + \fg}} \\
    \eeqtag{I.H.\ on $\cc_2$}
    &\wp{\cc_1}{a \cdot \wp{\cc_2}{\ff} + \wp{\cc_2}{\fg}} \\
    \eeqtag{I.H.\ on $\cc_1$}
    &a\cdot \wp{\cc_1}{\wp{\cc_2}{\ff}} + \wp{\cc_1}{\wp{\cc_2}{\fg}} \\
    \eeqtag{Table \ref{table:wp}}
    &a \cdot \wp{\COMPOSE{\cc_1}{\cc_2}}{\ff} + \wp{\COMPOSE{\cc_1}{\cc_2}}{\fg}~.
\end{align}
The proof for super-linearity is completely analogous. \\ \\
\noindent
\emph{The case $\cc = \ITE{\guard}{\cc_1}{\cc_2}$}. We have
\begin{align}
   &\wp{\ITE{\guard}{\cc_1}{\cc_2}}{a \cdot \ff +\fg} \\
   \eeqtag{Table \ref{table:wp}}
   &\iverson{\guard} \cdot \wp{\cc_1}{a \cdot \ff +\fg} + \iverson{\neg \guard} \cdot \wp{\cc_2}{a \cdot \ff +\fg} \\
   \eeqtag{I.H.\ on $\cc_1$}
   &\iverson{\guard} \cdot (a \cdot \wp{\cc_1}{\ff} +\wp{\cc_1}{\fg}) +\iverson{\neg \guard} \cdot \wp{\cc_2}{a \cdot \ff +\fg} \\
   \eeqtag{I.H.\ on $\cc_2$}
   &\iverson{\guard} \cdot (a \cdot \wp{\cc_1}{\ff} +\wp{\cc_1}{\fg}) + \iverson{\neg \guard} \cdot (a \cdot \wp{\cc_2}{\ff} +\wp{\cc_2}{\fg}) \\
   \eeqtag{Algebra}
   &a \cdot (\iverson{\guard} \cdot \wp{\cc_1}{\ff} + \iverson{\neg \guard} \cdot \wp{\cc_2}{\ff})
    + (\iverson{\guard} \cdot \wp{\cc_1}{\fg} + \iverson{\neg \guard} \cdot \wp{\cc_2}{\fg}) \\
   \eeqtag{Table \ref{table:wp}}
   &a \cdot \wp{\ITE{\guard}{\cc_1}{\cc_2}}{\ff} + \wp{\ITE{\guard}{\cc_1}{\cc_2}}{\fg}~.
\end{align}
The proof for super-linearity is completely analogous. \\ \\
\emph{The case $\cc = \PCHOICE{\cc_1}{p}{\cc_2}$}. We have
\begin{align}
   &\wp{\PCHOICE{\cc_1}{p}{\cc_2}}{a \cdot \ff + \fg} \\
   \eeqtag{Table \ref{table:wp}}
   &p \cdot \wp{\cc_1}{a \cdot \ff + \fg} + (1- p) \cdot \wp{\cc_2}{a \cdot \ff + \fg} \\
   \eeqtag{I.H.\ on $\cc_1$}
   &p \cdot (a \cdot \wp{\cc_1}{\ff} +\wp{\cc_1}{\fg}) + (1- p) \cdot \wp{\cc_2}{a \cdot \ff + \fg} \\
   \eeqtag{I.H.\ on $\cc_2$}
   &p \cdot (a \cdot \wp{\cc_1}{\ff} +\wp{\cc_1}{\fg}) + (1- p) \cdot (a \cdot \wp{\cc_2}{\ff} +\wp{\cc_2}{\fg}) \\
   \eeqtag{Algebra}
   & p \cdot a \cdot \wp{\cc_1}{\ff} + p \cdot \wp{\cc_1}{\fg} \\
   &\quad  + (1-p) \cdot a \cdot \wp{\cc_2}{\ff} + (1-p) \cdot \wp{\cc_2}{\fg} 
   \notag \\
   \eeqtag{Algebra}
   & a \cdot(p \cdot \wp{\cc_1}{\ff} + (1-p) \cdot \wp{\cc_2}{\ff}) \\
   & \quad + (p \cdot \wp{\cc_1}{\fg} + (1-p) \cdot \wp{\cc_2}{\fg}) 
   \notag \\
   \eeqtag{Table \ref{table:wp}}
   & a \cdot \wp{\PCHOICE{\cc_1}{p}{\cc_2}}{\ff} + \wp{\PCHOICE{\cc_1}{p}{\cc_2}}{\fg}~.
\end{align}
The proof for super-linearity is completely analogous. \\ \\
\emph{The case} $\cc = \WHILEDO{\guard}{\cc_1}$. We make use of the fact that there is an ordinal $\oa$ such that
\begin{align}
     \wp{\WHILEDO{\guard}{\cc_1}}{a \cdot \ff + \fg} \eeq \charwpn{\guard}{\cc_1}{a \cdot \ff + \fg}{\oa}(0)~.
\end{align}
Suppose for the moment that we already established the following:
\begin{align}\label{eqn:linearty-loop-assumption}
   \charwpn{\guard}{\cc_1}{a \cdot \ff + \fg}{\oc}(0) = a \cdot \charwpn{\guard}{\cc_1}{\ff}{\oc}(0) + \charwpn{\guard}{\cc_1}{\fg}{\oc}(0) 
   \quad \forall \, \text{ordinals} ~ \oc~.
\end{align}
Now let $\oa$, $\ob$, and $\oc$ be ordinals such that
%
\begin{align}
    & \charwpn{\guard}{\cc_1}{a \cdot \ff + \fg}{\oa}(0) \eeq \lfp X\mydot \charwp{\guard}{\cc_1}{a \cdot \ff + \fg}(X) \\
    &\charwpn{\guard}{\cc_1}{\ff}{\ob}(0) \eeq \lfp X\mydot \charwp{\guard}{\cc_1}{\ff}(X) \\
    &\charwpn{\guard}{\cc_1}{\fg}{\oc}(0) \eeq \lfp X\mydot \charwp{\guard}{\cc_1}{\fg}(X)~.
\end{align}
By choosing $\vartheta = \max \{\oa, \ob, \oc \}$, we obtain
\begin{align}
    &\wp{\WHILEDO{\guard}{\cc_1}}{a \cdot \ff + \fg}  \\
    \eeqtag{By assumption}
    &\charwpn{\guard}{\cc_1}{a \cdot \ff + \fg}{\oa}(0) \\
    \eeqtag{$\charwpn{\guard}{\cc_1}{a \cdot \ff + \fg}{\oa}(0)$ is a fixed point and $\vartheta \geq \oa$}
    &\charwpn{\guard}{\cc_1}{a \cdot \ff + \fg}{\vartheta}(0) \\
    \eeqtag{By Equation \ref{eqn:linearty-loop-assumption}}
    & a \cdot \charwpn{\guard}{\cc_1}{\ff}{\vartheta}(0) + \charwpn{\guard}{\cc_1}{\fg}{\vartheta}(0) \\
    \eeqtag{$\charwpn{\guard}{\cc_1}{\ff}{\ob}(0)$ is a fixed point and $\vartheta \geq \ob$}
    & a \cdot \charwpn{\guard}{\cc_1}{\ff}{\ob}(0) + \charwpn{\guard}{\cc_1}{\fg}{\vartheta}(0) \\
    \eeqtag{$\charwpn{\guard}{\cc_1}{\fg}{\oc}(0)$ is a fixed point and $\vartheta \geq \oc$}
    &a \cdot \charwpn{\guard}{\cc_1}{\ff}{\ob}(0) + \charwpn{\guard}{\cc_1}{\fg}{\oc}(0) \\
    \eeqtag{By assumption} 
    &a \cdot (\lfp X\mydot \charwp{\guard}{\cc_1}{\ff}(X)) + (\lfp X\mydot \charwp{\guard}{\cc_1}{\fg}(X)) \\
    \eeqtag{Table \ref{table:wp}}
    &a \cdot \wp{\WHILEDO{\guard}{\cc_1}}{\ff} + \wp{\WHILEDO{\guard}{\cc_1}}{\fg} ~.
\end{align}
Hence, it suffices to prove Equation \ref{eqn:linearty-loop-assumption}.
We proceed by transfinite induction on $\oc$. \\ \\
\noindent 
\emph{The case $\oc = 0$.} We have
\begin{align}
   &\charwpn{\guard}{\cc_1}{a \cdot \ff + \fg}{0}(0) \\
   \eeqtag{By definition}
   &0 \\
   \eeqtag{By definition}
   &a \cdot \charwpn{\guard}{\cc_1}{\ff}{0} + \charwpn{\guard}{\cc_1}{\fg}{0}~.
\end{align}
\emph{The case $\oc$ successor ordinal.} We have
\begin{align}
   &\charwpn{\guard}{\cc_1}{a \cdot \ff + \fg}{\oc +1 }(0) \\
   \eeqtag{By definition}
   &\charwp{\guard}{\cc_1}{a \cdot \ff + \fg}(\charwpn{\guard}{\cc_1}{a \cdot \ff + \fg}{\oc}(0)) \\
   \eeqtag{I.H. on $\oc$}
   &\charwp{\guard}{\cc_1}{a \cdot \ff + \fg}(a \cdot \charwpn{\guard}{\cc_1}{\ff}{\oc}(0) + \charwpn{\guard}{\cc_1}{\fg}{\oc}(0)) \\
   \eeqtag{Definition of $\charwp{\guard}{\cc_1}{a \cdot \ff + \fg}(\cdot)$}
   &\iverson{\neg \guard} \cdot (a \cdot \ff + \fg) 
   + \iverson{\guard} \cdot \wp{\cc_1}{a \cdot \charwpn{\guard}{\cc_1}{\ff}{\oc}(0) + \charwpn{\guard}{\cc_1}{\fg}{\oc}(0)} \\
   \eeqtag{I.H.\ on $\cc_1$}
   &\iverson{\neg \guard} \cdot (a \cdot \ff + \fg) 
   + \iverson{\guard} \cdot  (a\cdot \wp{\cc_1}{\charwpn{\guard}{\cc_1}{\ff}{\oc}(0)}  + \wp{\cc_1}{\charwpn{\guard}{\cc_1}{\fg}{\oc}(0)}) \\
   \eeqtag{Algebra}
   & a \cdot \big( \iverson{\guard} \cdot \wp{\cc_1}{\charwpn{\guard}{\cc_1}{\ff}{\oc}(0)} + \iverson{\neg \guard} \cdot \ff  \big) \\
   &\quad + \big( \iverson{\guard} \cdot \wp{\cc_1}{\charwpn{\guard}{\cc_1}{\fg}{\oc}(0)} + \iverson{\neg \guard} \cdot \fg \big) 
   \notag \\
   \eeqtag{By definition}
   &a \cdot \charwpn{\guard}{\cc_1}{\ff}{\oc +1}(0) + \charwpn{\guard}{\cc_1}{\fg}{\oc +1}(0)~.
\end{align}
\emph{The case $\oc$ limit ordinal.} Suppose Equation~\ref{eqn:linearty-loop-assumption} holds for all $\ob < \oc$. We have
\begin{align}
   &\charwpn{\guard}{\cc_1}{a \cdot \ff + \fg}{\oc}(0) \\
   \eeqtag{By definition}
   &\displaystyle\sup_{\ob < \oc} \charwpn{\guard}{\cc_1}{a \cdot \ff + \fg}{\ob}(0) \\
   \eeqtag{I.H.\ on $\delta$}
   &\displaystyle\sup_{\ob< \oc} \big( a \cdot \charwpn{\guard}{\cc_1}{\ff}{\ob}(0) + \charwpn{\guard}{\cc_1}{\fg}{\ob}(0) \big) \\
   \eeqtag{$\charwpn{\guard}{\cc_1}{\ff}{\delta}(0)$ and $\charwpn{\guard}{\cc_1}{\fg}{\delta}(0)$ monotonic in $\delta$}
   &a \cdot \big( \displaystyle\sup_{\ob < \oc}  \charwpn{\guard}{\cc_1}{\ff}{\ob}(0) \big) +
    \big( \displaystyle\sup_{\ob < \oc}  \charwpn{\guard}{\cc_1}{\fg}{\ob}(0) \big) \\
   \eeqtag{By definition}
   &a \cdot  \charwpn{\guard}{\cc_1}{\ff}{\oc}(0)  +
     \charwpn{\guard}{\cc_1}{\fg}{\oc}(0) ~.
\end{align}
The proof for super-linearity is completely analogous.
\end{proof}
\subsection{Proof of Theorem~\ref{thm:wp:basic}.\ref{thm:wp:basic:preservation-of-0} (Strictness)}
\label{app:wp:basic:preservation-of-0}
\begin{proof}
In order to show $\wp{\cc}{0} = 0$ consider the following:
\begin{align}
        & \wp{\cc}{0} \eeq \wp{\cc}{0 \cdot 0} \\
        \ppreceqtag{Theorem~\ref{thm:wp:basic}.\ref{thm:wp:basic:super-linearity}}
        & 0 \cdot \wp{\cc}{0} \eeq 0.
\end{align}
\end{proof}
\subsection{Proof of Theorem~\ref{thm:wp:basic}.\ref{thm:wp:basic:1-boundedness} (One-Boundedness)}
\label{app:wp:basic:1-boundedness}
\begin{proof}
In order to show $\wp{\cc}{\iverson{\preda}} \ppreceq 1$, first notice that a straightforward induction on the structure of \hpgcl programs 
yields that for each $\cc \in \hpgcl$, we have $\wp{\cc}{1} \preceq 1$.
Since $\iverson{\preda} \preceq 1$ and $\wpsymbol$ is monotone (Theorem~\ref{thm:wp:basic}.\ref{thm:wp:basic:monotonicity}), we then conclude $\wp{\cc}{\iverson{\preda}} \preceq 1$.
\end{proof}
\subsection{Proof of Theorem~\ref{thm:wp:basic}.\ref{thm:wp:basic:continuity} (Continuity)}
\label{app:wp:basic:continuity}
\begin{proof}
Assume $\cc$ does not contain an allocation statement, i.e. a statement of the form $\ALLOC{x}{\vec{\ee}}$.
We then have to show that for every increasing $\omega$-chain $\ff_1 \preceq \ff_2 \ppreceq \ldots$ in $\E$, we have
\[ \sup_{n} \wp{\cc}{\ff_n} \eeq \wp{\cc}{\sup_{n} \ff_n}.\]
We proceed by induction on the structure of $c$. \\ \\
\noindent
\emph{The case} $\cc = \SKIP$. We have
\begin{align}
   &\wp{\SKIP}{\sup_{n} \ff_n} \\
   \eeqtag{Table \ref{table:wp}}
   & \sup_{n} \ff_n \\
   \eeqtag{Table \ref{table:wp}}
   &\sup_{n} \wp{\SKIP}{\ff_n}.
\end{align}
\emph{The case} $\cc = \ASSIGN{x}{\ee}$. We have
\begin{align}
   &\wp{\ASSIGN{x}{\ee}}{\sup_{n} \ff_n} \\
   \eeqtag{Table \ref{table:wp}}
   & \big( \sup_{n} \ff_n \big) \subst{x}{\ee} \\
   \eeqtag{Substitution is distributive}
   &\sup_n \ff_n \subst{x}{\ee} \\
   \eeqtag{Table \ref{table:wp}}
   &\sup_n \wp{\ASSIGN{x}{\ee}}{\ff_n}.
\end{align}
\emph{The case} $\ASSIGNH{x}{\ee}$. We show continuity point-wise by distinguishing the cases
$s(\ee) \in \dom{h}$ and $s(\ee) \not\in \dom{h}$ for a given state $(\sk,\hh) \in \States$.

First, assume $\sk (\ee) \not\in \dom{h}$. We have
\begin{align}
   &\wp{\ASSIGNH{x}{\ee}}{\sup_{n} \ff_n}(\sk,\hh) \\
   \eeqtag{Table \ref{table:wp}}
   &\Big( \displaystyle\sup_{v \in \Ints} \singleton{\ee}{v} \sepcon \bigl( \singleton{\ee}{v} \sepimp \big( \sup_{n} \ff_n \big) \subst{x}{v} \bigr)
   \Big) (\sk,\hh) \\
   \eeqtag{Alternative version of the rule for heap lookup}
   &\Big( \displaystyle\sup_{v \in \Ints} \containsPointer{\ee}{v} \cdot \big( \sup_{n} \ff_n \big) \subst{x}{v} \Big) (\sk,\hh) \\
   \eeqtag{$\sk (\ee) \not\in \dom{\hh}$}
   & 0 \\
   \eeqtag{$\sk (\ee) \not\in \dom{\hh}$}
   & \sup_{n} \Big( \displaystyle\sup_{v \in \Ints} \containsPointer{\ee}{v} \cdot  \ff_n \subst{x}{v} \Big) (\sk , \hh ) \\
   \eeqtag{Alternative version of the rule for heap lookup}
   &\sup_{n} \wp{\ASSIGNH{x}{\ee}}{\ff_n}(\sk,\hh ).
\end{align}
Now assume $\sk (\ee) \in \dom{\hh}$ and $h(\sk (\ee)) = v'$. Then
\begin{align}
   &\wp{\ASSIGNH{x}{\ee}}{\sup_{n} \ff_n}(\sk ,\hh) \\
   \eeqtag{Table \ref{table:wp}}
   &\Big( \displaystyle\sup_{v \in \Ints} \singleton{\ee}{v} \sepcon \bigl( \singleton{\ee}{v} \sepimp \big( \sup_{n} \ff_n \big) \subst{x}{v} \bigr)
   \Big) (\sk ,\hh ) \\
   \eeqtag{Alternative version of the rule for heap lookup}
   &\Big( \displaystyle\sup_{v \in \Ints} \containsPointer{\ee}{v} \cdot \big( \sup_{n} \ff_n \big) \subst{x}{v} \Big) (\sk ,\hh ) \\
   \eeqtag{By assumption: $(\displaystyle\sup_{v \in \Ints} \containsPointer{\ee}{v})(s,h) = 1 =  \containsPointer{\ee}{v'}(s,h) $}
   &\big( \big( \sup_{n} \ff_n \big) \subst{x}{v'} \big) (\sk, \hh ) \\
   \eeqtag{Substitution is distributive}
   & \sup_{n} \big( \ff_n\subst{x}{v'} (s,h) \big) \\
   \eeqtag{By assumption: $(\displaystyle\sup_{v \in \Ints} \containsPointer{\ee}{v})(\sk ,\hh ) = 1 =  \containsPointer{\ee}{v'}(\sk , \hh ) $}
   &\sup_{n} \big( \displaystyle\sup_{v \in \Ints} \containsPointer{\ee}{v} \cdot \ff_n\subst{x}{v} \big)(\sk ,\hh ) \\
   \eeqtag{Alternative version of the rule for heap lookup}
   &\sup_{n} \wp{\ASSIGNH{x}{\ee}}{\ff_n}(\sk , \hh ).
\end{align}
\emph{The case} $\cc = \HASSIGN{\ee}{\ee'}$. We show continuity point-wise by distinguishing the cases
$\sk (\ee) \in \dom{\hh }$ and $\sk (\ee) \not\in \dom{\hh }$ for a given state $(\sk , \hh) \in \States$.

First, assume $\sk (\ee) \not\in \dom{\hh }$. We have
\begin{align}
   &\wp{\HASSIGN{\ee}{\ee'}}{\sup_{n} \ff_n}(\sk , \hh ) \\
   \eeqtag{Table \ref{table:wp}}
   &\Big( \validpointer{\ee} \sepcon \bigl(\singleton{\ee}{\ee'} \sepimp \big( \sup_{n} \ff_n \big)  \bigr) \Big) (\sk ,\hh ) \\
   \eeqtag{$\sk (\ee) \not\in \dom{\hh }$}
   & 0 \\
   \eeqtag{$\sk(\ee) \not\in \dom{\hh}$}
   &\sup_{n} \Big( \validpointer{\ee} \sepcon \bigl(\singleton{\ee}{\ee'} \sepimp \big( \ff_n \big)  \bigr) \Big) (\sk, \hh ) \\
   \eeqtag{Table \ref{table:wp}}
   &\sup_{n} \wp{\HASSIGN{\ee}{\ee'}}{\ff_n}(\sk , \hh ).
\end{align}
Now assume $\sk (\ee) \in \dom{\hh}$.
Moreover, given two arithmetic expressions $\ee$ and $\ee'$, let $h_{\ee,\ee'}$ denote the heap with $\dom{h_{\ee,\ee'}} = \{\sk(\ee) \}$ and $h_{\ee,v'}(\sk(\ee)) = s(\ee)$. 
The heap $\hh$ is thus of the form
$\hh = \hh' \sepcon h_{\ee,v'}$ for some value $v'$. This gives us
\begin{align}
   &\wp{\HASSIGN{\ee}{\ee'}}{\sup_{n} \ff_n}(\sk , \hh ) \\
   \eeqtag{Table \ref{table:wp}}
   &\Big( \validpointer{\ee} \sepcon \bigl(\singleton{\ee}{\ee'} \sepimp \big( \sup_{n} \ff_n \big)  \bigr) \Big) (\sk ,\hh ) \\
   \eeqtag{By assumption: $\hh = \hh' \sepcon h_{\ee,v'}$}
   &\Big( \validpointer{\ee} \sepcon \bigl(\singleton{\ee}{\ee'} \sepimp \big( \sup_{n} \ff_n \big)  \bigr) \Big) (\sk , \hh' \sepcon h_{\ee,v'} ) \\
   \eeqtag{$\validpointer{\ee}(\sk,h_{\ee,v'}) = 1$}
   &\bigl( \singleton{\ee}{\ee'} \sepimp \big( \sup_{n} \ff_n \big)  \bigr) (\sk , \hh' ) \\
   \eeqtag{$s(e) \not\in \dom{h'}$}
   & \sup_{n} \ff_n   (\sk , \hh' \sepcon h_{\ee,\ee'} ) \\
   \eeqtag{$s(e) \not\in \dom{h'}$}
   & \sup_{n} \bigl( (\singleton{\ee}{\ee'} \sepimp \ff_n)(\sk , \hh' )    \bigr)  \\
   \eeqtag{$\validpointer{\ee}(\sk,h_{\ee,v'}) = 1$ and $\sk{\ee} \not \in \hh'$}
   &\sup_{n} \bigl( \bigl( \validpointer{\ee} \sepcon (\singleton{\ee}{\ee'} \sepimp \ff_n) \bigr)  (\sk , \hh' \sepcon h_{\ee,v'} )  \bigr) \\
   \eeqtag{By assumption: $\hh = \hh' \sepcon h_{\ee,v'}$ }
   &\sup_{n} \bigl( \bigl( \validpointer{\ee} \sepcon (\singleton{\ee}{\ee'} \sepimp \ff_n) \bigr) (\sk , \hh ) \bigr)  \\
   \eeqtag{Table \ref{table:wp}}
   &\sup_{n} \wp{\HASSIGN{\ee}{\ee'}}{\ff_n}(\sk , \hh ).
\end{align}
\emph{The case} $\cc= \FREE{\ee}$.  We show continuity point-wise by distinguishing the cases
$\sk (\ee) \in \dom{\hh }$ and $\sk (\ee) \not\in \dom{\hh }$ for a given state $(\sk , \hh) \in \States$.

If $\sk(\ee) \not\in \dom{\hh}$, then
\begin{align}
   &\wp{\FREE{\ee}}{\sup_{n} \ff_n}( \sk, \hh) \\
   \eeqtag{Table \ref{table:wp}}
   &\big( \validpointer{\ee} \sepcon (\sup_{n} \ff_n ) \big) (\sk, \hh) \\
   \eeqtag{$\sk(\ee) \not\in \dom{\hh}$}
   &0 \\
   \eeqtag{$\sk(\ee) \not\in \dom{\hh}$}
   &\sup_{n} \big( \big( \validpointer{\ee} \sepcon  \ff_n  \big) \big) (\sk, \hh) \\
   \eeqtag{Table \ref{table:wp}}
   &\sup_{n} \wp{\FREE{\ee}}{\ff_n} (\sk, \hh).
\end{align}
Now suppose $\sk(\ee)\in \dom{\hh}$, i.e.\ the heap $\hh$ is of the form $\hh = \hh' \sepcon \hh_{\ee,v'}$ for some value $v'$. We have
\begin{align}
   &\wp{\FREE{\ee}}{\sup_{n} \ff_n}( \sk, \hh) \\
   \eeqtag{Table \ref{table:wp}}
   &\big( \validpointer{\ee} \sepcon (\sup_{n} \ff_n ) \big) (\sk, \hh) \\
   \eeqtag{By assumption: $\hh = \hh' \sepcon \hh_{\ee,v'}$}
   &\big( \validpointer{\ee} \sepcon (\sup_{n} \ff_n ) \big) (\sk, \hh' \sepcon \hh_{\ee,v'}) \\
   \eeqtag{$ \validpointer{\ee}(\sk, \hh_{\ee,v'}) = 1$}
   &\sup_{n} \ff_n (\sk, \hh' ) \\
   \eeqtag{$ \validpointer{\ee}(\sk, \hh_{\ee,v'}) = 1$ and $\sk{\ee} \not \in \hh'$}
   &\sup_{n} \big( \validpointer{\ee} \sepcon \ff_n  \big)(\sk, \hh' \sepcon \hh_{\ee,v'}) \\
   \eeqtag{Table \ref{table:wp}}
   &\sup_{n} \wp{\FREE{\ee}}{\ff_n}(\sk, \hh' \sepcon \hh_{\ee,v'}) \\
   \eeqtag{By assumption: $\hh = \hh' \sepcon \hh_{\ee,v'}$}
   &\sup_{n} \wp{\FREE{\ee}}{\ff_n}(\sk, \hh). 
\end{align}
As the induction hypothesis now assume that for some arbitrary, but fixed, $\cc_1, \cc_2 \in \hpgcl$ and all increasing $\omega$-chains 
$\ff_1 \preceq \ff_2 \ppreceq \ldots$ and $\fg_1 \preceq \fg_2 \ppreceq \ldots$ in $\E$ it holds that both
\begin{align}
   &\wp{\cc_1}{\sup_{n} \ff_n} \eeq \sup_{n} \wp{\cc_1}{ \ff_n}\\
   \text{and} \quad  &\wp{\cc_2}{\sup_{n} \fg_n} \eeq  \sup_{n} \wp{\cc_2}{ \fg_n}.
\end{align}
Furthermore, we make use of Lebesgue's Monotone Convergence Theorem (LMCT); see e.g.\ \cite[p.~567]{schechter:1996}.

\emph{The case} $\cc = \COMPOSE{\cc_1}{\cc_2}$. We have
\begin{align}
   &\wp{\COMPOSE{\cc_1}{\cc_2}}{\sup_{n} \ff_n} \\
   \eeqtag{Table \ref{table:wp}}
   &\wp{\cc_1}{\wp{\cc_2}{\sup_{n} \ff_n}} \\
   \eeqtag{I.H.\ on $\cc_2$}
   &\wp{\cc_1}{\sup_{n} \wp{\cc_2}{ \ff_n}} \\
   \eeqtag{By mon.\ of $\wpsymbol$, $(\wp{\cc_2}{ \ff_n})_{n \geq 1}$ is an increasing chain, then apply I.H. on $\cc_1$}
   &\sup_{n} \wp{\cc_1}{ \wp{\cc_2}{\ff_n}} \\
   \eeqtag{Table \ref{table:wp}}
   &\sup_{n} \wp{\COMPOSE{\cc_1}{\cc_2}}{\ff_n}.
\end{align}
\emph{The case} $\cc = \ITE{\guard}{\cc_1}{\cc_2}$. We have
\begin{align}
   &\wp{\ITE{\guard}{\cc_1}{\cc_2}}{\sup_n \ff_n} \\
   \eeqtag{Table \ref{table:wp}}
   &\iverson{\guard} \cdot \wp{\cc_1}{\sup_n \ff_n} + \iverson{\neg \guard} \cdot \wp{\cc_2}{\sup_n \ff_n} \\
   \eeqtag{I.H.\ on $\cc_1$ and $\cc_1$}
   &\iverson{\guard} \cdot \sup_n \wp{\cc_1}{\ff_n} + \iverson{\neg \guard} \cdot \sup_n \wp{\cc_2}{ \ff_n} \\
   \eeqtag{Both $(\wp{\cc_1}{\ff_n})_{n \geq 1}$ and  $(\wp{\cc_2}{ \ff_n})_{n \geq 1}$ are increasing chains, then apply LMCT}
   &\sup_n \big( \iverson{\guard} \cdot \wp{\cc_1}{\ff_n} + \iverson{\neg \guard} \cdot \wp{\cc_2}{ \ff_n} \big) \\
   \eeqtag{Table \ref{table:wp}}
   &\sup_n \wp{\ITE{\guard}{\cc_1}{\cc_2}}{ \ff_n}.
\end{align}
\emph{The case} $\cc = \PCHOICE{\cc_1}{\pp}{\cc_2}$. We have
\begin{align}
   &\wp{\PCHOICE{\cc_1}{\pp}{\cc_2}}{\sup_n \ff_n} \\
   \eeqtag{Table \ref{table:wp}}
   &\pp \cdot \wp{\cc_1}{\sup_n \ff_n} + (1- \pp) \cdot \wp{\cc_2}{\sup_n \ff_n} \\
   \eeqtag{I.H.\ on $\cc_1$ and $\cc_2$}
   &\pp \cdot \sup_n \wp{\cc_1}{\ff_n} + (1- \pp) \cdot \sup_n \wp{\cc_2}{\ff_n} \\
   \eeqtag{Both $(\wp{\cc_1}{\ff_n})_{n \geq 1}$ and  $(\wp{\cc_2}{ \ff_n})_{n \geq 1}$ are increasing chains, then apply LMCT}
   &\sup_n \big( \pp \cdot \wp{\cc_1}{\ff_n} + (1- \pp) \cdot  \wp{\cc_2}{ \ff_n} \big) \\
   \eeqtag{Table \ref{table:wp}}
   &\sup_n \wp{\PCHOICE{\cc_1}{\pp}{\cc_2}}{\ff_n}.
\end{align}
\emph{The case} $\cc = \WHILEDO{\guard}{\cc_1}$. Since for every $\ff \in \E$ there is an ordinal $\oa$ such that
\begin{align}
   \wp{\WHILEDO{\guard}{\cc_1}}{\ff} \eeq \lfp \fk . \charwp{\guard}{\cc_1}{\ff}(\fk) \eeq \charwpn{\guard}{\cc_1}{\ff}{\oa}(0),
\end{align}
it suffices to show that
\begin{align}
   \charwpn{\guard}{\cc_1}{\sup_n \ff_n}{\ob}(0) \eeq \sup_n \charwpn{\guard}{\cc_1}{\ff_n}{\ob}(0)
\end{align}
for all ordinals $\ob$. We proceed by transfinite induction on $\ob$.  \\ \\
\noindent
\emph{The case} $\ob =0$. This case is trivial since
\begin{align}
   &\charwpn{\guard}{\cc_1}{\sup_n \ff_n}{0}(0) \\
   \eeqtag{By definition}
   &0 \\
   \eeqtag{By definition}
   &\sup_n \charwpn{\guard}{\cc_1}{\ff_n}{0}(0).
\end{align}
\emph{The case $\ob$ successor ordinal.} We have
\begin{align}
   &\charwpn{\guard}{\cc_1}{\sup_n \ff_n}{\ob +1}(0) \\
   \eeqtag{By definition}
   &\charwp{\guard}{\cc_1}{\sup_n \ff_n}(\charwpn{\guard}{\cc_1}{\sup_n \ff_n}{\ob}(0)) \\
   \eeqtag{I.H.\ on $\ob$}
   &\charwp{\guard}{\cc_1}{\sup_n \ff_n}( \sup_n \charwpn{\guard}{\cc_1}{ \ff_n}{\ob}(0)) \\
   \eeqtag{By definition}
   &\iverson{\guard} \cdot \wp{\cc_1}{\sup_n \charwpn{\guard}{\cc_1}{ \ff_n}{\ob}(0)} + \iverson{\neg \guard} \cdot \sup_n \ff_n \\
   \eeqtag{By monotonicity $(\charwpn{\guard}{\cc_1}{ \ff_n}{\ob}(0))_{n \geq 1}$ is an increasing chain, then apply I.H.\ on $\cc_1$}
   &\iverson{\guard} \cdot \sup_n \wp{\cc_1}{\charwpn{\guard}{\cc_1}{ \ff_n}{\ob}(0)} + \iverson{\neg \guard} \cdot \sup_n \ff_n \\
   \eeqtag{By mon.\ $(\wp{\cc_1}{\charwpn{\guard}{\cc_1}{ \ff_n}{\ob}(0)})_{n \geq 1}$ is an increasing chain, then apply LMCT }
   &\sup_n \big( \iverson{\guard} \cdot \wp{\cc_1}{\charwpn{\guard}{\cc_1}{ \ff_n}{\ob}(0)} + \iverson{\neg \guard} \cdot \ff_n \big) \\
   \eeqtag{By definition}
   &\sup_n \charwp{\guard}{\cc_1}{\ff_n}(\charwpn{\guard}{\cc_1}{\ff_n}{\ob}(0)) \\
   \eeqtag{By definition}
   &\sup_n \charwpn{\guard}{\cc_1}{\ff_n}{\ob +1}(0).
\end{align}
\emph{The case $\ob$ limit ordinal.} We have
\begin{align}
    &\charwpn{\guard}{\cc_1}{\sup_n \ff_n}{\ob}(0) \\
    \eeqtag{By definition for $\ob$ limit ordinal}
    &\sup_{\oc < \ob} \charwpn{\guard}{\cc_1}{\sup_n \ff_n}{\oc}(0) \\
    \eeqtag{I.H.\ on $\oc$}
    &\sup_{\oc < \ob} \sup_n  \charwpn{\guard}{\cc_1}{\ff_n}{\oc}(0) \\
    \eeqtag{Commutativity of $\sup$}
    &\sup_n \sup_{\oc \leq \ob}  \charwpn{\guard}{\cc_1}{\ff_n}{\oc}(0) \\
    \eeqtag{By definition for $\ob$ limit ordinal}
    &\sup_n \charwpn{\guard}{\cc_1}{\ff_n}{\ob}(0).
\end{align} 

\end{proof}

\subsection{Counterexample for continuity of weakest preexpectations}\label{app:wp:continuity-counterexample}

Consider an $\omega$-chain of expectations $\ff_n = \iverson{1 \leq x \leq n}$. 
Moreover, let $\emptyheap$ be the empty heap. 
Then, for an arbitrary stack $\sk$, 
\begin{align*}
        & \wp{\ALLOC{x}{0}}{\textstyle \sup_{n} \ff_n}(\sk,\emptyheap) \\
        \eeqtag{\autoref{table:wp}}
        & \inf_{v \in \AVAILLOC{0}} \left(\singleton{v}{0} \sepimp (sup_{n}\ff_n)\subst{x}{v}\right)(\sk,\emptyheap) \\
        \eeqtag{$\dom{\emptyheap} = \emptyset$} 
        & \inf_{v \in \Nats} \left(\singleton{v}{0} \sepimp (sup_{n}\ff_n)\subst{x}{v}\right)(\sk,\emptyheap) \\
        \eeqtag{$\sup_{n} \iverson{1 \leq v \leq n} = \iverson{1 \leq v \leq \infty}$} 
        & \inf_{v \in \Nats} \iverson{0 \leq v \leq \infty}(\sk,\emptyheap) \\
        \eeqtag{algebra}
        & \inf_{v \in \Nats} 1 \eeq 1.
\end{align*}
However, if we swap application of the weakest preexpectation and the supremum, we obtain
\begin{align*}
        & \sup_{n \in \Nats} \wp{\ALLOC{x}{0}}{\ff_n}(\sk,\emptyheap) \\
        \eeqtag{\autoref{table:wp}} 
        & \sup_{n \in \Nats} \inf_{v \in \AVAILLOC{0}} \left(\singleton{v}{0} \sepimp \ff_n\subst{x}{v}\right)(\sk,\emptyheap) \\
        \eeqtag{$\dom{\emptyheap} = \emptyset$} 
        & \sup_{n \in \Nats} \inf_{v \in \Nats} \left(\singleton{v}{0} \sepimp \ff_n\subst{x}{v}\right)(\sk,\emptyheap) \\
        \eeqtag{Definition of $\ff_n$} 
        & \sup_{n \in \Nats} \inf_{v \in \Nats} \iverson{0 \leq v \leq n}(\sk,\emptyheap)  \\
        \eeqtag{we can always choose $v > n$}
        & 0.
\end{align*}
Hence, continuity breaks for the $\ALLOC{x}{\ee}$ statement.

\subsection{Modus Ponens for Single Points-to Predicates}\label{app:wand-reynolds}
\begin{lemma}\label{lem:wand-reynolds}
  Let $\ff \in \E$. Then 
  \[ \singleton{x}{\ee} \sepcon \left(\singleton{x}{\ee} \sepimp \ff\right) \eeq \containsPointer{x}{\ee} \cdot \ff~. \]
\end{lemma}
\begin{proof}
Let $(\sk,\hh)$ be a stack-heap pair.
We distinguish two cases.

First, assume $\containsPointer{x}{\ee}(\sk,\hh) = 0$. Then
\begin{align}
& \left(\singleton{x}{\ee} \sepcon \left(\singleton{x}{\ee} \sepimp \ff\right)\right)(\sk,\hh) \\
\eeqtag{Definition of $\sepcon$}
& \max_{\hh_1,\hh_2} \setcomp{ \singleton{x}{\ee}(\sk,\hh_1) \cdot \left(\singleton{x}{\ee} \sepimp \ff\right)(\sk,\hh_2)}{ \hh = \hh_1 \sepcon \hh_2 } \\
\eeqtag{by assumption: $\containsPointer{x}{\ee}(\sk,\hh) = 0$} 
& \max_{\hh_1,\hh_2} \setcomp{ 0 \cdot \left(\singleton{x}{\ee} \sepimp \ff\right)(\sk,\hh_2) }{ \hh = \hh_1 \sepcon \hh_2 } \\ 
\eeqtag{algebra}
& 0 \\
\eeqtag{by assumption: $\containsPointer{x}{\ee}(\sk,\hh) = 0$} 
& \left(\containsPointer{x}{\ee} \cdot \ff\right)(\sk,\hh).
\end{align}

For the second case, assume $\containsPointer{x}{\ee}(\sk,\hh) = 1$. 
Then there exist unique heaps $\hh_1', \hh_2'$ such that $\singleton{x}{\ee}(\sk,\hh_1') = 1$ and $\hh = \hh_1' \sepcon \hh_2'$.
Consequently, we have
\begin{align}
& \left(\singleton{x}{\ee} \sepcon \left(\singleton{x}{\ee} \sepimp \ff\right)\right)(\sk,\hh) \\
\eeqtag{Definition of $\sepcon$}
& \max_{\hh_1,\hh_2} \setcomp{ \singleton{x}{\ee}(\sk,\hh_1) \cdot \left(\singleton{x}{\ee} \sepimp \ff\right)(\sk,\hh_2) }{ \hh = \hh_1 \sepcon \hh_2 } \\
\eeqtag{$\singleton{x}{\ee}(\sk,\hh_1) = 0$ for $\hh_1 \neq \hh_1'$} 
& \singleton{x}{\ee}(\sk,\hh_1') \cdot \left(\singleton{x}{\ee} \sepimp \ff\right)(\sk,\hh_2') \\
\eeqtag{$\singleton{x}{\ee}(\sk,\hh_1') = 1$} 
& \left(\singleton{x}{\ee} \sepimp \ff\right)(\sk,\hh_2') \\ 
\eeqtag{Definition of $\sepimp$} 
& \inf_{\hh'} \left\{ \ff(\sk,\hh_2'\sepcon \hh') ~|~ \hh' \disjoint \hh_2' \textnormal{ and } \sk,\hh' \models \singleton{x}{\ee} \right\} \\
\eeqtag{$\sk,\hh' \models \singleton{x}{\ee}$ iff $\dom{\hh'} = \sk(x)$ and $\hh'(\sk(x))=\sk(\ee)$. Hence, $\hh' = \hh_1'$ } 
&  \ff(\sk,\hh_2'\sepcon \hh_1') \\
\eeqtag{$\hh = \hh_1' \sepcon \hh_2'$}
&  \ff(\sk,\hh) \\
\eeqtag{Assumption: $\containsPointer{x}{\ee}(\sk,\hh) = 1$} 
& \containsPointer{x}{\ee}(\sk,\hh) \cdot \ff(\sk,\hh) \\
\eeqtag{Definition of $\cdot$~} 
&  \left(\containsPointer{x}{\ee} \cdot \ff\right)(\sk,\hh). 
\end{align}
In both cases, we obtain the claim, i.e. $\singleton{x}{\ee} \sepcon \left(\singleton{x}{\ee} \sepimp \ff\right) \eeq \containsPointer{x}{\ee} \cdot \ff$.
\end{proof}
%


\subsection{Proof of Theorem~\ref{thm:wp:soundness} (Soundness of Weakest Preexpectations)}\label{app:wp:soundness}

\paragraph{Preliminaries} Let us first collect a few important facts about our operational semantics:
\begin{enumerate}
        \item The execution relation $\ExecSymbol$ determining our operational semantics together with reward function $\OpRew$ specifies a Markov decision process with rewards~\cite{DBLP:books/daglib/0020348}.
              To be precise, the set of states is given by program configurations $\OpStates$, the set of actions is $\Nats$, the probability transition function is $\ProbSymbol$, and the reward function is $\OpRew$. Each of these items has been introduced in Section~\ref{sec:wp:soundness}.
              A reader familiar with MDPs might want to add a sink state with zero reward and a self-loop with probability one.
              Then all goal configurations, which have no outgoing transitions so far, additionally get a single transition with action $0$ and probability $1$ to the sink state.
              We chose to omit a sink state to improve readability.
  \item The set of program configurations $\OpStates$ and the set of actions $\Nats$ are countable.
  \item The reflexive, transitive closure of execution relation $\ExecSymbol$---denoted by $\ExecSymbol^{*}$---is well-founded if restricted to configurations that occur in $\PathsFromTo{\cc,\sk,\hh}{\Scheduler}$ for any scheduler $\Scheduler$.
  \item Only goal configurations, i.e.. configurations in $\Target = \{ (\Term, \sigma) ~|~ \sigma \in \Sigma \}$ are assigned positive reward. Hence, all paths that do not reach a goal configuration contribute zero reward.
\end{enumerate}

Furthermore, let us denote the set of actions available at configuration $t \in \OpStates$ by 
\begin{align*}
        \OpAct{t} \eeq \left\{ a \in \OpActions ~|~ \exists t' \in \OpStates \, \exists \pp > 0 ~:~ \ExecSimple{t}{a,\pp}{t'} \right\}.
\end{align*}

We use the following characterization for expected rewards of Markov decision processes (cf.~\cite[Theorem 7.1.3]{puterman2005markov}), which has been adapted to our notation and the fact that only goal configurations have positive rewards: 

\begin{theorem}[Characterization of Expected Rewards]\label{thm:exprew}
    Let $\ff \in \E$ and $t \in \OpStates$.
    Then the least expected reward $\ExpRewC{\ff}{t}$ satisfies the following equation system: 
    \begin{itemize}
            \item If $t = (\cc, \sk, \hh) \in \Target$ then $\ExpRewC{\ff}{t} = \OpRew(t) = \ff(\sk,\hh)$.
            \item If $t = (\Fault, \sigma)$, $\sigma \in \States$, then $\ExpRewC{\ff}{t} = 0$.
            \item Otherwise, we have
            \begin{align*}
                \ExpRewC{\ff}{t} \eeq \inf_{a \in \OpAct{t}} \sum_{\ExecSimple{t}{a,\pp}{t'}} \pp \cdot \ExpRewC{\ff}{t'}.
            \end{align*}
    \end{itemize}
\end{theorem}

Moreover, we need a few technical definitions.

\begin{definition}
  A function of type $\Fsymbol : \hpgcl \to (\E \to \E)$ is called an \emph{expectation transformer}. 
  We compare expectation transformers by pointwise application of $\leq$, i.e. $\Fsymbol \leq \Fsymbol'$ iff for all $\cc \in \hpgcl$, $\ff \in \E$, and $\sigma \in \States$, we have $\F{\cc}{\ff}(\sigma) \leq \Fsymbol'\llbracket\cc\rrbracket(\ff)(\sigma)$.
  \hfill$\triangle$
\end{definition}

Clearly, $\wpsymbol$ is an expectation transformer.
We next define an expectation transformer mapping each program $\cc$ and each expectation $\ff$ to the corresponding expected reward of our execution relation with respect to $\ff$ when running $\cc$ on a given initial state.
Consequently, we refer to this transformer as the operational semantics of $\hpgcl$-programs.

\begin{definition}[Operational Semantics of $\hpgcl$-Programs]\label{def:op}
   The \emph{operational semantics of $\hpgcl$-programs} is given by the expectation transformer
    \begin{align*}
            \Opsymbol ~:~ \hpgcl \to \E \to \E,
            \quad
            \Opf{\cc}{\ff}(\sigma) \eeq \ExpRewC{\ff}{\cc, \sigma}.
            \tag*{$\triangle$}
    \end{align*}
\end{definition}

The remaining two technical definitions are used to improve the proof structure.

\begin{definition}\label{def:ext}
        The \emph{extended expectation transformer} $\Ext{\Fsymbol}$ of expectation transformer $\Fsymbol$ is given by
        \begin{align*}
                & \Ext{\Fsymbol} ~:~ (\hpgcl \cup \{ \Term,\Fault \}) \to (\E \to \E) \\
                & \EF{\cc}{\ff} \eeq
                \begin{cases}
                    \ff & ~\text{if}~ \cc \eeq \Term \\
                    0 & ~\text{if}~ \cc \eeq \Fault \\
                    \F{\cc}{\ff} & ~\text{otherwise.}
                \end{cases}
        \end{align*}
  \hfill$\triangle$
\end{definition}

\begin{definition}\label{def:functional}
  $\Fsymbol$ is called an \hpgcl-\emph{functional} if and only if
  \begin{enumerate}
          \item $\Fsymbol$ is of type $\Fsymbol : \hpgcl \to (\E \to \E)$,
          \item for all $\cc \in \hpgcl$, $\ff \in \E$ and $\sigma \in \States$, we have
                \begin{align*}
                        \EF{\cc}{\ff}(\sigma) \eeq \inf_{n \in \OpAct{\cc, \sigma}} \sum_{\ExecSimple{\cc, \sigma}{n,\pp}{\cc', \sigma'}} \pp \cdot \EF{\cc'}{\ff}(\sigma')~.
                \end{align*}
  \end{enumerate}
  \hfill$\triangle$
\end{definition}

\paragraph{Soundness proof}
We are now in a position to show that $\wpsymbol$ is sound with respect to our operational semantics. 
The auxiliary results used within the proof below are found in Appendix~\ref{app:wp:soundness:auxiliaries}, p.~\pageref{app:wp:soundness:auxiliaries}.
Due to Definition~\ref{def:op}, our proof obligation can be conveniently restated as
\begin{align}
        \wpsymbol \eeq \Opsymbol. 
\end{align}

\begin{proof}[Proof of Theorem~\ref{thm:wp:soundness}]
        First, we show that our operational semantics $\Opsymbol$ is the \emph{least} $\hpgcl$-\emph{functional} with respect to pointwise application of the ordering $\leq$.
        That $\Opsymbol$ is an $\hpgcl$-functional follows from Theorem~\ref{thm:exprew} and a straightforward induction on the program structure.
        That $\Opsymbol$ is also the least $\hpgcl$ functional is proven by well-founded induction on the structure of $\hpgcl$-functionals (see Definition~\ref{def:functional}).
        Please confer Lemma~\ref{thm:op-least} for a detailed proof.

        Next, we show that our weakest preexpectation semantics $\wpsymbol$ is an $\hpgcl$-functional. This is shown by induction on the program structure. 
        Please confer Lemma~\ref{thm:wp-functional} for a detailed proof.
        Putting both results together, we immediately obtain $\Opsymbol \leq \wpsymbol$.

        To complete the soundness proof, we show the converse direction, i.e.\ $\wpsymbol \leq \Opsymbol$, by induction on the program structure.
        Please confer Lemma~\ref{thm:wp-leq-op} for a detailed proof.
\end{proof}

\subsection{Auxiliary Lemmas in the Proof of Theorem~\ref{thm:wp:soundness}}
\label{app:wp:soundness:auxiliaries}

\begin{lemma}\label{thm:op-least}
  $\Opsymbol$ is the least $\hpgcl$-functional with respect to $\leq$.
\end{lemma}

\begin{proof}
        Clearly, $\Opsymbol$ is an $\hpgcl$-functional due to Theorem~\ref{thm:exprew} and a straightforward induction on the structure of $\hpgcl$ programs.
        
        Next, consider the paths determined by $\Opf{\cc}{\ff}(\sigma)$.
        Every path $\pi$ starting in $(\cc, \sigma)$ that never reaches a goal configuration, i.e.. a configuration of the form $(\Term, \sigma')$, contributes zero reward.
        This is a direct consequence of the fact that only goal configurations may have a non-zero reward. Thus, every configuration belonging to path $\pi$ has a zero reward.
        We may thus restrict ourselves to paths reaching a goal configuration without changing the value of $\Opf{\cc}{\ff}(\sigma)$.

        Now, let $\Fsymbol$ be any $\hpgcl$-functional as of Definition~\ref{def:functional}.
        The probabilistic transition relation $\ExecSymbol$ is well-founded if we restrict it to configurations $(\cc', \sigma')$ reachable from the initial configuration, i.e..
        $(\cc, \sigma) ~\rightarrow^{*}~ (\cc', \sigma')$.
        We prove by induction with respect to this well-founded ordering that 
        \begin{align}
                \EOpf{\cc}{\ff}(\sigma) = \EF{\cc}{\ff}(\sigma).
        \end{align}
        For the two base cases, we have by Definition~\ref{def:ext}
        \begin{align}
            \EOpf{\Term}{\ff}(\sigma) \eeq \ff(\sigma) \eeq \EF{\Term}{\ff}(\sigma)
            \qand
            \EOpf{\Fault}{\ff}(\sigma) \eeq 0 \eeq \EF{\Fault}{\ff}(\sigma).
        \end{align}
        Otherwise, we have
        \begin{align}
                & \EF{\cc}{\ff}(\sigma) \\
                \eeqtag{Definition~\ref{def:functional}} 
                & \inf_{n \in \OpAct{\cc, \sigma}} \sum_{\ExecSimple{\cc, \sigma}{n,\pp}{\cc', \sigma'}} \pp \cdot \EF{\cc'}{\ff}(\sigma') \\
                \eeqtag{I.H.}
                & \inf_{n \in \OpAct{\cc, \sigma}} \sum_{\ExecSimple{\cc, \sigma}{n,\pp}{\cc', \sigma'}} \pp \cdot \EOpf{\cc'}{\ff}(\sigma') \\
                \eeqtag{Definition~\ref{def:op}, Theorem~\ref{thm:exprew}}
                & \EOpf{\cc}{\ff}(\sigma). 
        \end{align}
        Hence, $\EOpf{\cc}{\ff}(\sigma) = \EF{\cc}{\ff}(\sigma)$ if we consider only executions that successfully terminate, i.e.. paths that reach a goal configuration.
        Since all other paths given by $\Opf{\cc}{\ff}(\sigma)$ contribute zero reward, we conclude that $\EOpf{\cc}{\ff}(\sigma) \leq \EF{\cc}{\ff}(\sigma)$.
\end{proof}

\begin{lemma}\label{thm:wp-functional}
    $\wpsymbol$ is an $\hpgcl$-functional.
\end{lemma}

\begin{proof}
  Clearly $\wpsymbol$ is of type $\hpgcl \to (\E \to \E)$.
  It thus remains to show for all $\cc \in \hpgcl$, $\ff \in \E$ and $\sigma \in \States$ that
  \begin{align}
          \extwp{\cc}{\ff}(\sigma) 
          \eeq 
          \inf_{n \in \OpAct{\cc, \sigma}} \sum_{\ExecSimple{\cc, \sigma}{n,\pp}{\cc', \sigma'}} \pp \cdot \extwp{\cc'}{\ff}(\sigma').
  \end{align}
  We proceed by induction on the structure of inference rules of our operational semantics (cf.\ Figure~\ref{table:op}).
  We group the cases by statement.

  \emph{The case} $\SKIP$. 
  \begin{align}
          & \inf_{n \in \OpAct{\SKIP, \sigma}} \sum_{\ExecSimple{\SKIP, \sigma}{n,\pp}{\cc', \sigma'}} \pp \cdot \extwp{\cc'}{\ff}(\sigma') \\
          \eeqtag{Definition of op. semantics (Figure~\ref{table:op})}
          & \inf_{n \in \OpAct{\SKIP, \sigma}} \sum_{\ExecSimple{\SKIP, \sigma}{0,1}{\Term, \sigma}} 1 \cdot \extwp{\Term}{\ff}(\sigma) \\
          \eeqtag{algebra}
          & 1 \cdot \extwp{\Term}{\ff}(\sigma)  \\
          \eeqtag{Definition~\ref{def:ext}, Definition of $\wpsymbol$} 
          & \extwp{\SKIP}{\ff}(\sigma). 
  \end{align}
  
  \emph{The case} $\ASSIGN{x}{\ee}$.
  \begin{align}
          & \inf_{n \in \OpAct{\ASSIGN{x}{\ee}, \sk,\hh}} \sum_{\ExecSimple{\ASSIGN{x}{\ee}, \sk,\hh}{n,\pp}{\cc', \sigma'}} \pp \cdot \extwp{\cc'}{\ff}(\sigma') \\
          \eeqtag{Definition of op. semantics (Figure~\ref{table:op})}
          & \inf_{n \in \OpAct{\ASSIGN{x}{\ee}, \sk,\hh}} \sum_{\Exec{\ASSIGN{x}{\ee}}{\sk}{\hh}{0}{1}{\Term}{\sk\subst{x}{\sk(\ee)}}{\hh}} 1 \cdot \extwp{\Term}{\ff}(\sk\subst{x}{\sk(\ee)},\hh) \\
          \eeqtag{algebra}
          & 1 \cdot \extwp{\Term}{\ff}(\sk\subst{x}{\sk(\ee)},\hh) \\
          \eeqtag{Definition~\ref{def:ext}} 
          & \ff(\sk\subst{x}{\sk(\ee)},\hh) \\
          \eeqtag{Definition of $\ff\subst{x}{\ee}$} 
          & \ff\subst{x}{\ee}(\sk,\hh) \\
          \eeqtag{Definition of $\wpsymbol$} 
          & \extwp{\ASSIGN{x}{\ee}}{\ff}(\sk,\hh). 
  \end{align}
  
  \emph{The case} $\ALLOC{x}{\ee_1,\ldots,\ee_n}$.
  \begin{align}
          & \inf_{n \in \OpAct{\ALLOC{x}{\ee_1,\ldots,\ee_n},\sk,\hh}} \sum_{\Exec{\ALLOC{x}{\ee_1,\ldots,\ee_n}}{\sk}{\hh}{n}{\pp}{\cc'}{\sk'}{\hh'}} \pp \cdot \extwp{\cc'}{\ff}(\sk',\hh') \\
          \eeqtag{Definition of op. semantics (Figure~\ref{table:op}), algebra}
          & \inf_{u \in \PosNats : u,u+1,\ldots,u+n-1 \notin \dom{\hh}} \extwp{\Term}{\ff}(\sk\subst{x}{u},\hh \sepcon \HeapSet{u \mapsto v_1 \ldots v_n}) \\
          %
          %
          %
          \eeqtag{Definition~\ref{def:ext}} 
          & \inf_{u \in \PosNats : u,u+1,\ldots,u+n-1 \notin \dom{\hh}} \ff(\sk\subst{x}{u},\hh \sepcon \singleton{u}{v_1 \ldots v_n}) \\
          %
          \eeqtag{Definition of $\ff\subst{x}{u}$} 
          & \inf_{u \in \PosNats : u,u+1,\ldots,u+n-1 \notin \dom{\hh}} \ff\subst{x}{u}(\sk,\hh \sepcon \singleton{u}{v_1 \ldots v_n}) \label{proof:soundness:alloc:1}
          %
\end{align}
Let
\begin{align}
        M \eeq \{ u \in \PosNats ~|~ u,u+1,\ldots,u+n-1 \notin \dom{\hh} \}
\end{align}
be the set of all possible choices for address $u$ such that the block $u,u+1,\ldots,u+n-1$ can be allocated without overlapping with heap $\hh$.
Moreover, consider the expectation
\begin{align}
        f(u) \eeq \left(\singleton{u}{v_1 \ldots v_n} \sepimp \ff\subst{x}{u}\right)(\sk,\hh)~.
\end{align}
By definition of $\sepimp$, we obtain that for each $u \in M$,
\begin{align}
        f(u) \eeq \ff\subst{x}{u}(\sk,\hh \sepcon \singleton{u}{v_1 \ldots v_n})~. \label{proof:soundness:alloc:2}
\end{align}
Moreover, we have
\begin{align}
        & \inf_{u \in \PosNats} f(u) \\
        \eeqtag{algebra}
        & \min \{ \inf_{u \in M} f(u),~ \inf_{u \in \PosNats \setminus M} f(u) \} \\
        \eeqtag{for each $u \in \PosNats \setminus M$, $f(u) = \infty$ by definition of $\sepimp$}
        & \min \{ \inf_{u \in M} f(u),~ \infty \} \\
        \eeqtag{$\infty$ is the largest element of the lattice $(\E,{\preceq})$} 
        & \inf_{u \in \PosNats} f(u) \eeq \inf_{u \in M} f(u)~. \label{proof:soundness:alloc:3}
\end{align}
Hence, we can continue at equation~(\ref{proof:soundness:alloc:1}) as follows:
\begin{align}
        & \text{continuing from equation~(\ref{proof:soundness:alloc:1})} \\
        \eeqtag{by equation~(\ref{proof:soundness:alloc:2})}
        & \inf_{u \in \PosNats : u,u+1,\ldots,u+n-1 \notin \dom{\hh}} \left(\singleton{u}{v_1 \ldots v_n} \sepimp \ff\subst{x}{u}\right)(\sk,\hh) \\
        \eeqtag{by equation~(\ref{proof:soundness:alloc:3})}
        & \inf_{u \in \PosNats} \left(\singleton{u}{v_1 \ldots v_n} \sepimp \ff\subst{x}{u}\right)(\sk,\hh) \\
        \eeqtag{Definition of $\wpsymbol$}
        & \extwp{\ALLOC{x}{\ee_1, \ldots, \ee_n}}{\ff}(\sk,\hh)~. 
\end{align}
  
  \emph{The case} $\HASSIGN{\ee}{\ee'}$.
  Let $\sk(\ee') = v$. 
  We have to distinguish two cases: 
  
  First, assume $\sk(\ee)=u \in \dom{\hh}$. Then
  \begin{align}
          & \inf_{n \in \OpAct{\HASSIGN{\ee}{\ee'},\sk,\hh}} \sum_{\Exec{\HASSIGN{\ee}{\ee'}}{\sk}{\hh}{n}{\pp}{\cc'}{\sk'}{\hh'}} \pp \cdot \extwp{\cc'}{\ff}(\sk',\hh') \\
          \eeqtag{Definition of op. semantics (Figure~\ref{table:op})}
          & \inf_{n \in \OpAct{\HASSIGN{\ee}{\ee'},\sk,\hh}} \sum_{\Exec{\HASSIGN{\ee}{\ee'}}{\sk}{\hh}{0}{1}{\Term}{\sk}{\hh\subst{u}{v}}} \extwp{\Term}{\ff}(\sk,\hh\subst{u}{v}) \\
          \eeqtag{algebra} 
          & \extwp{\Term}{\ff}(\sk,\hh\subst{u}{v}) \\
          \eeqtag{Definition~\ref{def:ext}} 
          & \ff(\sk,\hh\subst{u}{v}) \\
          \eeqtag{$u \in \dom{\hh}$ by assumption} 
          & \left(\validpointer{u} \sepcon (\singleton{u}{v} \sepimp \ff)\right)(\sk,\hh) \\
          \eeqtag{$\sk(\ee)=u,\sk(\ee')=v$ by assumption} 
          & \left(\validpointer{\ee} \sepcon (\singleton{\ee}{\ee'} \sepimp \ff)\right)(\sk,\hh) \\
          \eeqtag{Definition of $\wpsymbol$}
          & \extwp{\HASSIGN{\ee}{\ee'}}{\ff}(\sk,\hh). 
  \end{align}
  Second, assume $\sk(\ee)=u \notin \dom{\hh}$. Then
  \begin{align}
          & \inf_{n \in \OpAct{\HASSIGN{\ee}{\ee'},\sk,\hh}} \sum_{\Exec{\HASSIGN{\ee}{\ee'}}{\sk}{\hh}{n}{\pp}{\cc'}{\sk'}{\hh'}} \pp \cdot \extwp{\cc'}{\ff}(\sk',\hh') \\
          \eeqtag{Definition of op. semantics (Figure~\ref{table:op})}
          & \inf_{n \in \OpAct{\HASSIGN{\ee}{\ee'}, \sk,\hh}} \sum_{\Exec{\HASSIGN{\ee}{\ee'}}{\sk}{\hh}{0}{1}{\Fault}{\sk}{\hh}} \extwp{\Fault}{\ff}(\sk,\hh) \\
          \eeqtag{algebra} 
          & \extwp{\Fault}{\ff}(\sk,\hh) \\
          \eeqtag{Definition~\ref{def:ext}} 
          & 0 \\
          \eeqtag{$u \notin \dom{\hh}$ by assumption} 
          & \left(\validpointer{u} \sepcon (\singleton{u}{v} \sepimp \ff)\right)(\sk,\hh) \\
          \eeqtag{$\sk(\ee)=u,\sk(\ee')=v$ by assumption} 
          & \left(\validpointer{\ee} \sepcon (\singleton{\ee}{\ee'} \sepimp \ff)\right)(\sk,\hh) \\
          \eeqtag{Definition of $\wpsymbol$}
          & \extwp{\HASSIGN{\ee}{\ee'}}{\ff}(\sk,\hh). 
  \end{align}

  \emph{The case} $\ASSIGNH{x}{\ee}$.
  We have to distinguish two cases: First, assume $\sk(\ee)=u \in \dom{\hh}$. Then, for $\hh(u) = v$, we have
  \begin{align}
          & \inf_{n \in \OpAct{\ASSIGNH{x}{\ee},\sk,\hh}} \sum_{\Exec{\ASSIGNH{x}{\ee}}{\sk}{\hh}{n}{\pp}{\cc'}{\sk'}{\hh'}} \pp \cdot \extwp{\cc'}{\ff}(\sk',\hh') \\
          \eeqtag{Definition of op. semantics (Figure~\ref{table:op})}
          & \inf_{n \in \OpAct{\ASSIGNH{x}{\ee}, \sk,\hh}} \sum_{\Exec{\ASSIGNH{x}{\ee}}{\sk}{\hh}{0}{1}{\Term}{\sk\subst{x}{v}}{\hh}} \extwp{\Term}{\ff}(\sk\subst{x}{v},\hh) \\
          \eeqtag{algebra} 
          & \extwp{\Term}{\ff}(\sk\subst{x}{v},\hh) \\
          \eeqtag{Definition~\ref{def:ext}} 
          & \ff(\sk\subst{x}{v},\hh) \\
          \eeqtag{algebra} 
          & \ff\subst{x}{v}(\sk,\hh) \\
          \eeqtag{$\sk(\ee) = u \in \dom{\hh}, \hh(u) = v$} 
          & (\sup_{v \in \Ints} \singleton{\ee}{v} \sepcon (\singleton{\ee}{v} \sepimp \ff\subst{x}{v}))(\sk,\hh) \\
          \eeqtag{Definition of $\wpsymbol$}
          & \extwp{\ASSIGNH{x}{\ee}}{\ff}(\sk,\hh). 
  \end{align}
  Second, assume $\sk(\ee)=u \notin \dom{\hh}$. Then
  \begin{align}
          & \inf_{n \in \OpAct{\ASSIGNH{x}{\ee},\sk,\hh}} \sum_{\Exec{\ASSIGNH{x}{\ee}}{\sk}{\hh}{n}{\pp}{\cc'}{\sk'}{\hh'}} \pp \cdot \extwp{\cc'}{\ff}(\sk',\hh') \\
          \eeqtag{Definition of op. semantics (Figure~\ref{table:op})}
          & \inf_{n \in \OpAct{\ASSIGNH{x}{\ee}, \sk,\hh}} \sum_{\Exec{\ASSIGNH{x}{\ee}}{\sk}{\hh}{0}{1}{\Fault}{\sk}{\hh}} \extwp{\Fault}{\ff}(\sk,\hh) \\
          \eeqtag{Algebra} 
          & \extwp{\Fault}{\ff}(\sk,\hh) \\
          \eeqtag{Definition~\ref{def:ext}} 
          & 0 \\
          \eeqtag{$\sk(\ee) \notin \dom{\hh}$} 
          & (\sup_{v \in \Ints}  \singleton{\ee}{v} \sepcon (\singleton{\ee}{v} \sepimp \ff\subst{x}{v}))(\sk,\hh) \\
          \eeqtag{Definition of $\wpsymbol$}
          & \extwp{\ASSIGNH{x}{\ee}}{\ff}(\sk,\hh). 
  \end{align}
  
  \emph{The case} $\FREE{x}$.
  We have to distinguish two cases: The heap is of the form $\hh \sepcon \HeapSet{\sk(x) \mapsto v}$ or the heap is not of this form.
  In the first case, we have
  \begin{align}
          & \inf_{n \in \OpAct{\FREE{x}, \sk,\hh \sepcon \HeapSet{\sk(x) \mapsto v}}} \sum_{\Exec{\FREE{x}}{\sk}{\hh \sepcon \HeapSet{\sk(x) \mapsto v}}{n}{\pp}{\cc'}{\sk'}{\hh'}} \pp \cdot \extwp{\cc'}{\ff}(\sk',\hh') \\
          \eeqtag{Definition of op. semantics (Figure~\ref{table:op})}
          & \inf_{n \in \OpAct{\FREE{x}, \sk,\hh \sepcon \HeapSet{\sk(x) \mapsto v}}} \sum_{\Exec{\FREE{x}}{\sk}{\hh \sepcon \HeapSet{\sk(x) \mapsto v}}{0}{1}{\Term}{\sk}{\hh}} 1 \cdot \extwp{\Term}{\ff}(\sk,\hh) \\
          \eeqtag{algebra} 
          & \extwp{\Term}{\ff}(\sk,\hh) \\
          \eeqtag{Definition~\ref{def:ext}} 
          & \ff(\sk,\hh) \\
          \eeqtag{$\ff(\sk,\hh) = (\singleton{u}{v} \sepcon \ff)(\sk,\hh\sepcon \singleton{u}{v})$ if $u \notin \dom{\hh}$} 
          & (\validpointer{\sk(x)} \sepcon \ff)(\sk,\hh \sepcon \singleton{\sk(x)}{v}) \\
          \eeqtag{Definition of $\wpsymbol$}
          & \extwp{\FREE{x}}{\ff}. 
  \end{align}
  Otherwise, we have
  \begin{align}
          & \inf_{n \in \OpAct{\FREE{x}, \sk,\hh}} \sum_{\Exec{\FREE{x}}{\sk}{\hh}{n}{\pp}{\cc'}{\sk'}{\hh'}} \pp \cdot \extwp{\cc'}{\ff}(\sk',\hh') \\
          \eeqtag{Definition of op. semantics (Figure~\ref{table:op})}
          & \inf_{n \in \OpAct{\FREE{x}, \sk,\hh}} \sum_{\Exec{\FREE{x}}{\sk}{\hh}{0}{1}{\Fault}{\sk}{\hh}} 1 \cdot \extwp{\Fault}{\ff}(\sk,\hh) \\
          \eeqtag{algebra}
          & \extwp{\Fault}{\ff}(\sk,\hh) \\
          \eeqtag{Definition~\ref{def:ext}} 
          & 0 \\
          \eeqtag{$\sk(x) \notin \dom{\hh}$ by assumption}
          & (\validpointer{\sk(x)} \sepcon \ff)(\sk,\hh) \\
          \eeqtag{Definition of $\wpsymbol$}
          & \extwp{\FREE{x}}{\ff}. 
  \end{align}

  \emph{The case} $\PCHOICE{\cc_1}{\pp}{\cc_2}$.
  \begin{align}
          & \inf_{n \in \OpAct{\PCHOICE{\cc_1}{\pp}{\cc_2}, \sigma}} \sum_{\ExecSimple{\PCHOICE{\cc_1}{\pp}{\cc_2}, \sigma}{n,q}{\cc', \sigma'}} q \cdot \extwp{\cc'}{\ff}(\sigma') \\
          \eeqtag{Definition of op. semantics (Figure~\ref{table:op})}
          & \inf_{n \in \OpAct{\PCHOICE{\cc_1}{\pp}{\cc_2}, \sigma}} \sum_{\ExecSimple{\PCHOICE{\cc_1}{\pp}{\cc_2}, \sigma}{0,q}{\cc', \sigma}} q \cdot \extwp{\cc'}{\ff}(\sigma) \\
          \eeqtag{Definition of op. semantics (Figure~\ref{table:op}), algebra}
          & \pp \cdot \extwp{\cc_1}{\ff}(\sigma) + (1-\pp) \cdot \extwp{\cc_2}{\ff}(\sigma) \\
          \eeqtag{Definition of $\wpsymbol$}
          & \extwp{\PCHOICE{\cc_1}{\pp}{\cc_2}}{\ff}(\sigma). 
  \end{align}
  
  \emph{The case} $\COMPOSE{\cc_1}{\cc_2}$.
  First, note that for every $\hpgcl$-program $\cc_1$, we have either 
  \begin{enumerate}
    \item $\ExecSimple{\cc_1, \sigma}{n,\pp}{\cc_1', \sigma'}$, where $\cc_1' \in \hpgcl$, or 
    \item $\ExecSimple{\cc_1, \sigma}{n,\pp}{\Term, \sigma'}$, or
    \item $\ExecSimple{\cc_1, \sigma}{n,\pp}{\Fault, \sigma'}$.
  \end{enumerate}
  In other words, within a single step, a $\hpgcl$-program either proceeds execution, terminates or fails due to a memory error, but it never goes into multiple of these successor configurations.
  We thus have to distinguish three mutually exclusive cases.
  In the first case, we have
  \begin{align}
          & \inf_{n \in \OpAct{\COMPOSE{\cc_1}{\cc_2}, \sigma}} \sum_{\ExecSimple{\COMPOSE{\cc_1}{\cc_2}, \sigma}{n,\pp}{\cc', \sigma'}} \pp \cdot \extwp{\cc'}{\ff}(\sigma') \\
          \eeqtag{Definition of op. semantics (Figure~\ref{table:op}), case assumption}
          & \inf_{n \in \OpAct{\COMPOSE{\cc_1}{\cc_2}, \sigma}} 
            \sum_{\ExecSimple{\COMPOSE{\cc_1}{\cc_2}, \sigma}{n,\pp}{\COMPOSE{\cc_1'}{\cc_2}, \sigma'}} \pp \cdot \extwp{\COMPOSE{\cc_1'}{\cc_2}}{\ff}(\sigma') \\
          \eeqtag{Definition of $\Ext{\wpsymbol}$} 
          & \inf_{n \in \OpAct{\COMPOSE{\cc_1}{\cc_2}, \sigma}} 
            \sum_{\ExecSimple{\COMPOSE{\cc_1}{\cc_2}, \sigma}{n,\pp}{\COMPOSE{\cc_1'}{\cc_2}, \sigma'}} \pp \cdot \extwp{\cc_1'}{\extwp{\cc_2}{\ff}}(\sigma') \\
          \eeqtag{Definition of op. semantics (Figure~\ref{table:op})}
          & \inf_{n \in \OpAct{\cc_1, \sigma}} \sum_{\ExecSimple{\cc_1, \sigma}{n,\pp}{\cc_1', \sigma'}} \pp \cdot \extwp{\cc_1'}{\extwp{\cc_2}{\ff}}(\sigma')  \\
          \eeqtag{I.H.} 
          & \extwp{\cc_1}{\extwp{\cc_2}{\ff}}(\sigma) \\
          \eeqtag{Definition of $\Ext{\wpsymbol}$}
          & \extwp{\COMPOSE{\cc_1}{\cc_2}}{\ff}(\sigma).
  \end{align}
  In the second case, we have
  \begin{align}
          & \inf_{n \in \OpAct{\COMPOSE{\cc_1}{\cc_2}, \sigma}} \sum_{\ExecSimple{\COMPOSE{\cc_1}{\cc_2}, \sigma}{n,\pp}{\cc', \sigma'}} \pp \cdot \extwp{\cc'}{\ff}(\sigma') \\
          \eeqtag{Definition of op. semantics (Figure~\ref{table:op}), case assumption}
          & \inf_{n \in \OpAct{\COMPOSE{\cc_1}{\cc_2}, \sigma}} \sum_{\ExecSimple{\COMPOSE{\cc_1}{\cc_2}, \sigma}{n,\pp}{\cc_2, \sigma'}} \pp \cdot \extwp{\cc_2}{\ff}(\sigma') \\
          \eeqtag{Definition of op. semantics (Figure~\ref{table:op})}
          & \inf_{n \in \OpAct{\COMPOSE{\cc_1}{\cc_2}, \sigma}} \sum_{\ExecSimple{\cc_1, \sigma}{n,\pp}{\Term, \sigma'}} \pp \cdot \extwp{\cc_2}{\ff}(\sigma') \\
          \eeqtag{Definition~\ref{def:ext}} 
          & \inf_{n \in \OpAct{\cc_1, \sigma}} \sum_{\ExecSimple{\cc_1, \sigma}{n,\pp}{\Term, \sigma'}} \pp \cdot \extwp{\Term}{\extwp{\cc_2}{\ff}}(\sigma') \\
          \eeqtag{I.H.} 
          & \extwp{\cc_1}{\extwp{\cc_2}{\ff}}(\sigma) \\
          \eeqtag{Definition of $\Ext{\wpsymbol}$}
          & \extwp{\COMPOSE{\cc_1}{\cc_2}}{\ff}(\sigma).
  \end{align}
  In the third case, we have
  \begin{align}
          & \inf_{n \in \OpAct{\COMPOSE{\cc_1}{\cc_2}, \sigma}} \sum_{\ExecSimple{\COMPOSE{\cc_1}{\cc_2}, \sigma}{n,\pp}{\cc', \sigma'}} \pp \cdot \extwp{\cc'}{\ff}(\sigma') \\
          \eeqtag{Definition of op. semantics (Figure~\ref{table:op}), case assumption}
          & \inf_{n \in \OpAct{\COMPOSE{\cc_1}{\cc_2}, \sigma}} \sum_{\ExecSimple{\COMPOSE{\cc_1}{\cc_2}, \sigma}{n,\pp}{\Fault, \sigma'}} \pp \cdot \extwp{\Fault}{\ff}(\sigma') \\
          \eeqtag{Definition of op. semantics (Figure~\ref{table:op})}
          & \inf_{n \in \OpAct{\COMPOSE{\cc_1}{\cc_2}, \sigma}} \sum_{\ExecSimple{\cc_1, \sigma}{n,\pp}{\Fault, \sigma'}} \pp \cdot \extwp{\Fault}{\ff}(\sigma') \\
          \eeqtag{$\extwp{\Fault}{\fg} = 0$ for all $\fg$}
          & \inf_{n \in \OpAct{\cc_1, \sigma}} \sum_{\ExecSimple{\cc_1, \sigma}{n,\pp}{\Fault, \sigma'}} \pp \cdot \extwp{\Fault}{\extwp{\cc_2}{\ff}}(\sigma') \\
          \eeqtag{I.H.} 
          & \extwp{\cc_1}{\extwp{\cc_2}{\ff}}(\sigma) \\
          \eeqtag{Definition of $\Ext{\wpsymbol}$}
          & \extwp{\COMPOSE{\cc_1}{\cc_2}}{\ff}(\sigma). 
  \end{align}
  
  \emph{The case} $\ITE{\guard}{\cc_1}{\cc_2}$.
  We have to distinguish two cases: $\sk(\guard) = \false$ and $\sk(\guard) = \true$.

  If $\sk(\guard) = \false$ then
  \begin{align}
          & \inf_{n \in \OpAct{\ITE{\guard}{\cc_1}{\cc_1}, \sigma}} 
            \sum_{\ExecSimple{\ITE{\guard}{\cc_1}{\cc_2}, \sigma}{0,1}{\cc', \sigma'}} \extwp{\cc'}{\ff}(\sigma') \\
          \eeqtag{Definition of op. semantics (Figure~\ref{table:op}), case assumption}
          & \inf_{n \in \OpAct{\ITE{\guard}{\cc_1}{\cc_1}, \sigma}} 
            \sum_{\ExecSimple{\ITE{\guard}{\cc_1}{\cc_2}, \sigma}{0,1}{\cc_2, \sigma}} \extwp{\cc_2}{\ff}(\sigma) \\
          \eeqtag{algebra} 
          & \extwp{\cc_2}{\ff}(\sigma) \\
          \eeqtag{$\iverson{\guard}(\sigma) = 0$ by assumption} 
          & \left(\iverson{\guard} \cdot \extwp{\cc_1}{\ff} + \iverson{\neg \guard} \cdot \extwp{\cc_2}{\ff}\right)(\sigma) \\
          \eeqtag{Definition of $\wpsymbol$}
          & \extwp{\ITE{\guard}{\cc_1}{\cc_2}}{\ff}(\sigma). 
  \end{align}
  If $\sk(\guard) = \true$ then
  \begin{align}
          & \inf_{n \in \OpAct{\ITE{\guard}{\cc_1}{\cc_1}, \sigma}} 
            \sum_{\ExecSimple{\ITE{\guard}{\cc_1}{\cc_2}, \sigma}{0,1}{\cc', \sigma'}} \extwp{\cc'}{\ff}(\sigma') \\
          \eeqtag{Definition of op. semantics (Figure~\ref{table:op}), case assumption}
          & \inf_{n \in \OpAct{\ITE{\guard}{\cc_1}{\cc_2}, \sigma}} \sum_{\ExecSimple{\ITE{\guard}{\cc_1}{\cc_2}, \sigma}{0,1}{\cc_1, \sigma}} \extwp{\cc_1}{\ff}(\sigma) \\
          \eeqtag{algebra} 
          & \extwp{\cc_1}{\ff}(\sigma) \\
          \eeqtag{$\iverson{\guard}(\sigma) = 1$ by assumption} 
          & \left(\iverson{\guard} \cdot \extwp{\cc_1}{\ff} + \iverson{\neg \guard} \cdot \extwp{\cc_2}{\ff}\right)(\sigma) \\
          \eeqtag{Definition of $\wpsymbol$}
          & \extwp{\ITE{\guard}{\cc_1}{\cc_2}}{\ff}(\sigma). 
  \end{align}

  \emph{The case} $\WHILEDO{\guard}{\cc'}$.
  We have to distinguish two cases: $\sk(\guard) = \false$ and $\sk(\guard) = \true$.

  If $\sk(\guard) = \false$ then
  \begin{align}
          & \inf_{n \in \OpAct{\WHILEDO{\guard}{\cc'}, \sigma}} \sum_{\ExecSimple{\WHILEDO{\guard}{\cc'}, \sigma}{n,\pp}{\cc'', \sigma'}} \extwp{\cc''}{\ff}(\sigma') \\
          \eeqtag{Definition of op. semantics (Figure~\ref{table:op})}
          & \inf_{n \in \OpAct{\WHILEDO{\guard}{\cc'}, \sigma}} \sum_{\ExecSimple{\WHILEDO{\guard}{\cc'}, \sigma}{0,1}{\Term, \sigma}} \extwp{\Term}{\ff}(\sigma) \\
          \eeqtag{algebra}
          & \extwp{\Term}{\ff}(\sigma) \\
          \eeqtag{Definition~\ref{def:ext}} 
          & \ff(\sigma) \\
          \eeqtag{$\iverson{\guard}(\sigma) = 0$ by assumption} 
          & \left(\iverson{\neg \guard} \cdot \ff + \iverson{\guard} \cdot \extwp{\COMPOSE{\cc'}{\WHILEDO{\guard}{\cc'}}}{\ff}\right)(\sigma) \\
          \eeqtag{Definition of $\wpsymbol$}
          & \extwp{\WHILEDO{\guard}{\cc'}}{\ff}(\sigma). 
  \end{align}
  
  Conversely, if $\sk(\guard) = \true$ then
  \begin{align}
          & \inf_{n \in \OpAct{\WHILEDO{\guard}{\cc'}, \sigma}} \sum_{\ExecSimple{\WHILEDO{\guard}{\cc'}, \sigma}{n,\pp}{\cc'', \sigma'}} \extwp{\cc''}{\ff}(\sigma') \\
          \eeqtag{Definition of op. semantics (Figure~\ref{table:op})}
          & \inf_{n \in \OpAct{\WHILEDO{\guard}{\cc'}, \sigma}} \sum_{\ExecSimple{\WHILEDO{\guard}{\cc'}, \sigma}{0,1}{\COMPOSE{\cc'}{\WHILEDO{\guard}{\cc'}}, \sigma}} \\
          \eeqtag{algebra} 
          & \extwp{\COMPOSE{\cc'}{\WHILEDO{\guard}{\cc'}}}{\ff}(\sigma) \\
          \eeqtag{$\iverson{\guard}(\sigma) = 1$ by assumption} 
          & \left(\iverson{\neg \guard} \cdot \ff + \iverson{\guard} \cdot \extwp{\COMPOSE{\cc'}{\WHILEDO{\guard}{\cc'}}}{\ff}\right)(\sigma) \\
          \eeqtag{Definition of $\wpsymbol$}
          & \extwp{\WHILEDO{\guard}{\cc'}}{\ff}(\sigma). 
  \end{align}
\end{proof}

\begin{lemma}\label{thm:op:while}
    $\EOpf{\WHILEDO{\guard}{\cc}}{\ff} = \iverson{\neg \guard} \cdot \ff + \iverson{\guard} \cdot \EOpf{\COMPOSE{\cc}{\WHILEDO{\guard}{\cc}}}{\ff}$.
\end{lemma}

\begin{proof}
Let $(\sk,\hh) \in \States$. We distinguish two cases: $\sk(\guard) = \false$ and $\sk(\guard) = \true$.

If $\sk(\guard) = \false$, we have
\begin{align}
        & \EOpf{\WHILEDO{\guard}{\cc}}{\ff}(\sk,\hh) \\
        \eeqtag{Theorem~\ref{thm:exprew}, Definition of op. semantics (Figure~\ref{table:op})} 
        & \sum_{\Exec{\WHILEDO{\guard}{\cc}}{\sk}{\hh}{0}{1}{\Term}{\sk}{\hh}} \EOpf{\Term}{\ff}(\sk,\hh) \\
        \eeqtag{algebra, Definition~\ref{def:ext}} 
        & \ff(\sk,\hh) \\
        \eeqtag{$\sk(\guard) = 1 - \sk(\neg \guard) = 0$ by assumption}
        & \left(\iverson{\neg \guard} \cdot \ff + \iverson{\guard} \cdot \EOpf{\COMPOSE{\cc}{\WHILEDO{\guard}{\cc}}}{\ff}\right)(\sk,\hh). 
\end{align}

If $\sk(\guard) = \true$, we have
\begin{align}
        & \EOpf{\WHILEDO{\guard}{\cc}}{\ff}(\sk,\hh) \\
        \eeqtag{Theorem~\ref{thm:exprew}, Definition of op. semantics (Figure~\ref{table:op})} 
        & \sum_{\Exec{\WHILEDO{\guard}{\cc}}{\sk}{\hh}{0}{1}{\COMPOSE{\cc}{\WHILEDO{\guard}{\cc}}}{\sk}{\hh}} \EOpf{\COMPOSE{\cc}{\WHILEDO{\guard}{\cc}}}{\ff}(\sk,\hh) \\
        \eeqtag{algebra, Definition~\ref{def:ext}} 
        & \EOpf{\COMPOSE{\cc}{\WHILEDO{\guard}{\cc}}}{\ff}(\sk,\hh) \\
        \eeqtag{$\sk(\guard) = 1 - \sk(\neg \guard) = 1$ by assumption}
        & \left(\iverson{\neg \guard} \cdot \ff + \iverson{\guard} \cdot \EOpf{\COMPOSE{\cc}{\WHILEDO{\guard}{\cc}}}{\ff}\right)(\sk,\hh). 
\end{align}
\end{proof}

\begin{lemma}\label{thm:op:compose}
    $\EOpf{\COMPOSE{\cc_1}{\cc_2}}{\ff} = \EOpf{\cc_1}{\EOpf{\cc_2}{\ff}}$.
\end{lemma}

\begin{proof}
  By induction on the structure of inference rules (cf.\ Figure~\ref{table:op}) for sequential composition.

  There are two base cases:
  
  First, consider $\ExecSimple{\cc_1, \sigma}{a,p}{\Term, \sigma'}$. Then
  \begin{align}
          & \EOpf{\COMPOSE{\cc_1}{\cc_2}}{\ff}(\sigma) \\
          \eeqtag{Definition of op. semantics (Figure~\ref{table:op})}
          & \inf_{a \in \OpAct{\COMPOSE{\cc_1}{\cc_2}, \sigma}} \sum_{\ExecSimple{\COMPOSE{\cc_1}{\cc_2}, \sigma}{a,\pp}{\cc_2, \sigma'}} \pp \cdot \EOpf{\cc_2}{\ff}(\sigma') \\
          \eeqtag{Definition of op. semantics (Figure~\ref{table:op}), Definition~\ref{def:ext}}
          & \inf_{a \in \OpAct{\cc_1, \sigma}} \sum_{\ExecSimple{\cc_1, \sigma}{a,\pp}{\Term, \sigma'}} \pp \cdot \EOpf{\Term}{\EOpf{\cc_2}{\ff}}(\sigma') \\
          \eeqtag{Theorem~\ref{thm:exprew}}
          & \EOpf{\cc_1}{\EOpf{\cc_2}{\ff}}(\sigma).
  \end{align}
  
  Second, consider $\ExecSimple{\cc_1, \sigma}{a,\pp}{\Fault, \sigma}$.
  \begin{align}
          & \EOpf{\COMPOSE{\cc_1}{\cc_2}}{\ff}(\sigma) \\
          \eeqtag{Definition of op. semantics (Figure~\ref{table:op})}
          & \inf_{a \in \OpAct{\COMPOSE{\cc_1}{\cc_2}, \sigma}} \sum_{\ExecSimple{\COMPOSE{\cc_1}{\cc_2}, \sigma}{a,\pp}{\Fault, \sigma}} \pp \cdot \EOpf{\Fault}{\ff}(\sigma) \\
          \eeqtag{Definition~\ref{def:ext}} 
          & \inf_{a \in \OpAct{\COMPOSE{\cc_1}{\cc_2}, \sigma}} \sum_{\ExecSimple{\COMPOSE{\cc_1}{\cc_2}, \sigma}{a,\pp}{\Fault, \sigma}} \pp \cdot \EOpf{\Fault}{\EOpf{\cc_2}{\ff}}(\sigma) \\
          \eeqtag{Definition of op. semantics (Figure~\ref{table:op})}
          & \inf_{a \in \OpAct{\cc_1, \sigma}} \sum_{\ExecSimple{\cc_1, \sigma}{a,\pp}{\Fault, \sigma}} \pp \cdot \EOpf{\Fault}{\EOpf{\cc_2}{\ff}}(\sigma) \\
          \eeqtag{Theorem~\ref{thm:exprew}}
          & \EOpf{\cc_1}{\EOpf{\cc_2}{\ff}}(\sigma).
  \end{align}

  For the composite case, assume $\ExecSimple{\cc_1, \sigma}{a,\pp}{\cc_1', \sigma'}$.
  \begin{align}
          & \EOpf{\COMPOSE{\cc_1}{\cc_2}}{\ff}(\sigma) \\
          \eeqtag{Definition of op. semantics (Figure~\ref{table:op})}
          & \inf_{a \in \OpAct{\COMPOSE{\cc_1}{\cc_2}, \sigma}} 
            \sum_{\ExecSimple{\COMPOSE{\cc_1}{\cc_2}, \sigma}{a,\pp}{\COMPOSE{\cc_1'}{\cc_2}, \sigma'}} \pp \cdot \EOpf{\COMPOSE{\cc_1'}{\cc_2}}{\ff}(\sigma') \\
          \eeqtag{I.H.}
          & \inf_{a \in \OpAct{\COMPOSE{\cc_1}{\cc_2}, \sigma}} 
            \sum_{\ExecSimple{\COMPOSE{\cc_1}{\cc_2}, \sigma}{a,\pp}{\COMPOSE{\cc_1'}{\cc_2}, \sigma'}} \pp \cdot \EOpf{\cc_1'}{\EOpf{\cc_2}{\ff}}(\sigma') \\
          \eeqtag{Definition of op. semantics (Figure~\ref{table:op}), using premise}
          & \inf_{a \in \OpAct{\cc_1, \sigma}} \sum_{\ExecSimple{\cc_1, \sigma}{a,\pp}{\cc_1', \sigma'}} \pp \cdot \EOpf{\cc_1'}{\EOpf{\cc_2}{\ff}}(\sigma') \\
          \eeqtag{Theorem~\ref{thm:exprew}}
          & \EOpf{\cc_1}{\EOpf{\cc_2}{\ff}}(\sigma).
  \end{align}
  
\end{proof}

\begin{lemma}\label{thm:wp-leq-op}
        $\Ext{\wpsymbol} \leq \Ext{\Opsymbol}$.
\end{lemma}

\begin{proof}
  By induction on the structure of $\hpgcl$ programs.

  \emph{The base cases} $\SKIP$, $\ASSIGN{x}{\ee}$, $\ALLOC{x}{\ee_1,\ldots,\ee_n}$, $\HASSIGN{\ee}{\ee'}$, $\ASSIGNH{x}{\ee}$, $\FREE{x}$.
  Let $\cc$ be one of the above base cases. We distinguish two disjoint cases (cf. Figure~\ref{table:op}): $\cc$ successfully terminates in one step or $\cc$ leads to a memory fault in one step (if possible).

  First, assume $\cc$ successfully terminates. Then
    \begin{align}
            & \extwp{\cc}{\ff}(\sigma) \\
            \eeqtag{Lemma~\ref{thm:wp-functional}, assumption} 
            & \inf_{n \in \OpAct{\cc, \sigma}} \sum_{\ExecSimple{\cc,\sigma}{n,\pp}{\Term, \sigma'}} \pp \cdot \extwp{\Term}{\ff}(\sigma') \\
            \eeqtag{Definition~\ref{def:ext}} 
            & \inf_{n \in \OpAct{\cc, \sigma}} \sum_{\ExecSimple{\cc,\sigma}{n,\pp}{\Term, \sigma'}} \pp \cdot \ff(\sigma') \\
            \eeqtag{Definition~\ref{def:ext}} 
            & \inf_{n \in \OpAct{\cc, \sigma}} \sum_{\ExecSimple{\cc,\sigma}{n,\pp}{\Term, \sigma'}} \pp \cdot \EOpf{\Term}{\ff}(\sigma') \\
            \eeqtag{Lemma~\ref{thm:op-least}}
            & \EOpf{\cc}{\ff}(\sigma). 
    \end{align}
  Now, assume $\cc$ leads to a memory fault. Then
    \begin{align}
            & \extwp{\cc}{\ff}(\sigma) \\
            \eeqtag{Lemma~\ref{thm:wp-functional}, assumption} 
            & \inf_{n \in \OpAct{\cc, \sigma}} \sum_{\ExecSimple{\cc,\sigma}{n,\pp}{\Fault, \sigma}} \pp \cdot \extwp{\Fault}{\ff}(\sigma') \\
            \eeqtag{Definition~\ref{def:ext}} 
            & \inf_{n \in \OpAct{\cc, \sigma}} \sum_{\ExecSimple{\cc,\sigma}{n,\pp}{\Fault, \sigma'}} \pp \cdot 0 \\
            \eeqtag{Definition~\ref{def:ext}} 
            & \inf_{n \in \OpAct{\cc, \sigma}} \sum_{\ExecSimple{\cc,\sigma}{n,\pp}{\Fault, \sigma'}} \pp \cdot \EOpf{\Fault}{\ff}(\sigma') \\
            \eeqtag{Lemma~\ref{thm:op-least}}
            & \EOpf{\cc}{\ff}(\sigma). 
    \end{align}

  \emph{The case $\COMPOSE{\cc_1}{\cc_2}$}
  \begin{align}
          & \extwp{\COMPOSE{\cc_1}{\cc_2}}{\ff} \\
          \eeqtag{Definition of $\wpsymbol$} 
          & \extwp{\cc_1}{\extwp{\cc_2}{\ff}} \\
          \ppreceqtag{I.H. on $\cc_2$} 
          & \extwp{\cc_1}{\EOpf{\cc_2}{\ff}} \\
          \ppreceqtag{I.H. on $\cc_1$} 
          & \EOpf{\cc_1}{\EOpf{\cc_2}{\ff}} \\
          \eeqtag{Lemma~\ref{thm:op:compose}}
          & \EOpf{\COMPOSE{\cc_1}{\cc_2}}{\ff}. 
  \end{align}
  
  \emph{The case $\ITE{\guard}{\cc_1}{\cc_2}$}.
  \begin{align}
          & \extwp{\ITE{\guard}{\cc_1}{\cc_2}}{\ff}(\sigma) \\
          \eeqtag{Lemma~\ref{thm:wp-functional}} 
          & \inf_{n \in \OpAct{\ITE{\guard}{\cc_1}{\cc_2}, \sigma}} 
            \sum_{\ExecSimple{\ITE{\guard}{\cc_1}{\cc_2},\sigma}{n,\pp}{\cc', \sigma}} \pp \cdot \extwp{\cc'}{\ff}(\sigma) \\
          \ppreceqtag{I.H.} 
          & \inf_{n \in \OpAct{\ITE{\guard}{\cc_1}{\cc_2}, \sigma}} 
            \sum_{\ExecSimple{\ITE{\guard}{\cc_1}{\cc_2},\sigma}{n,\pp}{\cc', \sigma}} \pp \cdot \EOpf{\cc'}{\ff}(\sigma) \\
          \eeqtag{Lemma~\ref{thm:op-least}}
          & \EOpf{\ITE{\guard}{\cc_1}{\cc_2}}{\ff}(\sigma). 
  \end{align}
  
  \emph{The case $\PCHOICE{\cc_1}{\pp}{\cc_2}$}.
  \begin{align}
          & \wp{\PCHOICE{\cc_1}{\pp}{\cc_2}}{\ff}(\sigma) \\
          \eeqtag{Definition $\wpsymbol$} \\
          & \pp \cdot \wp{\cc_1}{\ff}(\sigma) + (1-\pp) \cdot \wp{\cc_2}{\ff}(\sigma) \\
          \ppreceqtag{I.H.} 
          & \pp \cdot \Opf{\cc_1}{\ff}(\sigma) + (1-\pp) \cdot \Opf{\cc_2}{\ff}(\sigma) \\
          \eeqtag{Definition of op. semantics (Figure~\ref{table:op})}
          & \sum_{\ExecSimple{\PCHOICE{\cc_1}{\pp}{\cc_2}, \sigma}{0, q}{\cc', \sigma}} q \cdot \Opf{\cc'}{\ff}(\sigma) \\
          \eeqtag{Theorem~\ref{thm:exprew}}
          & \Opf{\PCHOICE{\cc_1}{\pp}{\cc_2}}{\ff}(\sigma). 
  \end{align}

  \emph{The case $\WHILEDO{\guard}{\cc}$}.
  Recall that $\wp{\WHILEDO{\guard}{\cc}}{\ff} = \lfp \fh. \charwp{\guard}{\cc}{\ff}(\fh)$, where the function $\charwp{\guard}{\cc}{\ff}(\fh)$ is given by 
  \begin{align}
          \charwp{\guard}{\cc}{\ff}(\fh) \eeq \iverson{\neg \guard} \cdot \ff + \iverson{\guard} \cdot \wp{\cc}{\fh}. \label{eq:op:charwp}
  \end{align}
  Now, let $\fk = \EOpf{\WHILEDO{\guard}{\cc}}{\ff}$. Then
  \begin{align}
          & \charwp{\guard}{\cc}{\ff}(\fk) \\
          \eeqtag{by (\ref{eq:op:charwp})}
          & \iverson{\neg \guard} \cdot \ff + \iverson{\guard} \cdot \wp{\cc}{\fk} \\
          \ppreceqtag{I.H.}
          & \iverson{\neg \guard} \cdot \ff + \iverson{\guard} \cdot \Opf{\cc}{\fk} \\
          \eeqtag{Lemma~\ref{thm:op:while}}
          & \fk. 
  \end{align}
  Hence, $\fk$ is a prefixed point of $F_{\ff}(\fh)$. Consequently,
  \begin{align}
          \wp{\WHILEDO{\guard}{\cc}}{\ff} \eeq \lfp \fh. \charwp{\guard}{\cc}{\ff}(\fh) \ppreceq \fk \eeq \Opf{\WHILEDO{\guard}{\cc}}{\ff}.
  \end{align}
\end{proof}

\subsection{Conservativity of \QSL as a verification system}\label{app:qsl:conservativity:wp}

For a non-probabilistic $\hpgcl$ program $\cc$ and a postcondition $\preda \in \SL$ in classical separation logic, we denote the classical \emph{weakest precondition} (cf.~\cite{DBLP:books/ph/Dijkstra76,DBLP:conf/popl/KrebbersTB17,DBLP:conf/lics/Reynolds02})  of program $\cc$ with respect to postcondition $\preda$ by $\SLwp{\cc}{\preda}$. 

The proof of Theorem~\ref{thm:qsl:conservativity:wp} relies on the following auxiliary result.

\begin{lemma}\label{thm:sl:conservativity}
    Let $\cc \in \hpgcl$ be a non-probabilistic program.
    Then, for all classical separation logic formulas $\preda \in \SL$, we have
     $\qslemb{\SLwp{\cc}{\preda}} \eeq \wp{\cc}{\qslemb{\preda}}$.
\end{lemma}

\begin{proof}
    By induction on the structure of $\hpgcl$ programs (excluding probabilistic choice).

    \emph{The case $\SKIP$:}
    \begin{align}
            & \qslemb{\SLwp{\SKIP}{\preda}} \\
            \eeqtag{Definition of weakest preconditions}
            & \qslemb{\preda} \\
            \eeqtag{Definition of weakest preexpectations}
            & \wp{\SKIP}{\qslemb{\preda}}~.
    \end{align}

    \emph{The case $\ASSIGN{x}{\ee}$:}
    \begin{align}
            & \qslemb{\SLwp{\ASSIGN{x}{\ee}}{\preda}} \\
            \eeqtag{Definition of weakest preconditions}
            & \qslemb{\preda\subst{x}{\ee}} \\
            \eeqtag{Substitution distributes over embedding}
            & \qslemb{\preda}\subst{x}{\ee} \\
            \eeqtag{Definition of weakest preexpectations}
            & \wp{\ASSIGN{x}{\ee}}{\qslemb{\preda}}~.
    \end{align}

    \emph{The case $\ASSIGNH{x}{\ee}$:}
    \begin{align}
            & \qslemb{\SLwp{\ASSIGNH{x}{\ee}}{\preda}} \\
            \eeqtag{Definition of weakest preconditions}
            & \qslemb{\exists z\colon \SLsingleton{\ee}{z} \sepcon \left(\SLsingleton{\ee}{z} \sepimp \preda\subst{x}{z}\right)} \\
            \eeqtag{applying embedding of \SL into \QSL}
            & \sup_{v \in \Ints} \singleton{\ee}{v} \sepcon \left(\singleton{\ee}{v} \sepimp \qslemb{\preda}\subst{x}{v}\right) \\
            \eeqtag{Definition of weakest preexpectations}
            & \wp{\ASSIGNH{x}{\ee}}{\qslemb{\preda}}~.
    \end{align}

    \emph{The case $\HASSIGN{\ee}{\ee'}$:}
    \begin{align}
            & \qslemb{\SLwp{\HASSIGN{\ee}{\ee'}}{\preda}} \\
            \eeqtag{Definition of weakest preconditions}
            & \qslemb{\SLvalidpointer{\ee} \sepcon \left(\SLsingleton{\ee}{\ee'} \sepimp \preda\right)} \\
            \eeqtag{applying embedding of \SL into \QSL}
            & \validpointer{\ee} \sepcon \left(\singleton{\ee}{\ee'} \sepimp \qslemb{\preda}\right) \\
            \eeqtag{Definition of weakest preexpectations}
            & \wp{\HASSIGN{\ee}{\ee'}}{\qslemb{\preda}}~.
    \end{align}

    \emph{The case $\FREE{x}$:}
    \begin{align}
            & \qslemb{\SLwp{\FREE{x}}{\preda}} \\
            \eeqtag{Definition of weakest preconditions}
            & \qslemb{\SLvalidpointer{x} \sepcon \preda} \\
            \eeqtag{applying embedding of \SL into \QSL}
            & \validpointer{x} \sepcon \qslemb{\preda} \\
            \eeqtag{Definition of weakest preexpecations}
            & \wp{\FREE{x}}{\qslemb{\preda}}~.
    \end{align}

    \emph{The case $\ALLOC{x}{\vec{e}}$:}
    \begin{align}
            & \qslemb{\SLwp{\ALLOC{x}{\vec{e}}}{\preda}} \\
            \eeqtag{Definition of weakest preconditions}
            & \qslemb{\forall z\colon \singleton{z}{\vec{e}} \sepimp \preda\subst{x}{z}} \\
            \eeqtag{applying embedding of \SL into \QSL}
            & \inf_{v \in \Ints} \singleton{v}{\vec{e}} \sepimp \qslemb{\preda}\subst{x}{v} \\
            \eeqtag{Expectation space is $\Eone$}
            & \inf_{v \in \AVAILLOC{\vec{e}}} \singleton{v}{\vec{e}} \sepimp \qslemb{\preda}\subst{x}{v} \\
            \eeqtag{Definition of weakest preexpectations}
            & \wp{\ALLOC{x}{\vec{e}}}{\qslemb{\preda}}~.
    \end{align}

    \emph{The case $\COMPOSE{\cc_1}{\cc_2}$:}
    \begin{align}
            & \qslemb{\SLwp{\COMPOSE{\cc_1}{\cc_2}}{\preda}} \\
            \eeqtag{Definition of weakest preconditions}
            & \qslemb{\SLwp{\cc_1}{\SLwp{\cc_2}{\preda}}} \\
            \eeqtag{I.H.}
            & \wp{\cc_1}{\qslemb{\SLwp{\cc_2}{\preda}}} \\
            \eeqtag{I.H.}
            & \wp{\cc_1}{\wp{\cc_2}{\qslemb{\preda}}} \\
            \eeqtag{Definition of weakest preexpectations}
            & \wp{\COMPOSE{\cc_1}{\cc_2}}{\qslemb{\preda}}~.
    \end{align}

    \emph{The case $\ITE{\guard}{\cc_1}{\cc_2}$:}
    \begin{align}
            & \qslemb{\SLwp{\ITE{\guard}{\cc_1}{\cc_2}}{\preda}} \\
            \eeqtag{Definition of weakest preconditions}
            & \qslemb{\left(\guard \wedge \SLwp{\cc_1}{\preda}\right) \vee \left(\neg \guard \wedge \SLwp{\cc_2}{\preda}\right)} \\
            \eeqtag{applying embedding of \SL into \QSL}
            & \max\{ \iverson{\guard} \cdot \qslemb{\SLwp{\cc_1}{\preda}}, \iverson{\neg \guard} \cdot \qslemb{\SLwp{\cc_2}{\preda}} \} \\
            \eeqtag{$\max$ amounts to $+$ for a set of mutually exclusive expectations}
            & \iverson{\guard} \cdot \qslemb{\SLwp{\cc_1}{\preda}} \,+\, \iverson{\neg \guard} \cdot \qslemb{\SLwp{\cc_2}{\preda}} \\
            \eeqtag{I.H. (twice)}
            & \iverson{\guard} \cdot \wp{\cc_1}{\qslemb{\preda}} \,+\, \iverson{\neg \guard} \cdot \wp{\cc_2}{\qslemb{\preda}} \\
            \eeqtag{Definition of weakest preexpectations}
            & \wp{\ITE{\guard}{\cc_1}{\cc_2}}{\qslemb{\preda}}~.
    \end{align}

    \emph{The case $\WHILEDO{\guard}{\cc}$:}
    By definition of weakest preconditions and weakest preexpectations, we have
    \begin{align}
            \SLwp{\WHILEDO{\guard}{\cc}}{\preda} \eeq & \lfp \predb\mydot \underbrace{\left(\neg \guard \wedge \preda\right) \vee \left(\guard \wedge \SLwp{\cc}{\predb}\right)}_{\eeq \Psi(\predb)} \\
            \wp{\WHILEDO{\guard}{\cc}}{\qslemb{\preda}} \eeq & \lfp \ff\mydot \underbrace{\iverson{\neg \guard} \cdot \qslemb{\preda} + \iverson{\guard} \wp{\cc}{\ff}}_{\eeq \Phi(\ff)}
    \end{align}
    Since both $\Psi$ and $\Phi$ are monotone, we may use the Tarski-Knaster fixed point theorem (cf.~\cite{cousot1979constructive}): There exists an ordinal $\oa$ such that 
    \begin{align}
            \SLwp{\WHILEDO{\guard}{\cc}}{\preda} \eeq \Psi^{\oa}(\false) \qquad\text{and}\qquad \wp{\WHILEDO{\guard}{\cc}}{\qslemb{\preda}} \eeq \Phi^{\oa}(0)~.
    \end{align}
    We proceed by showing by transfinite induction that for all ordinals $\ob$, we have
    \begin{align}
            \qslemb{\Psi^{\ob}(\false)} \eeq \Phi^{\ob}(0)~.
    \end{align}
    In particular, for $\ob = \oa$, this means that
    \begin{align}
            \qslemb{\SLwp{\WHILEDO{\guard}{\cc}}{\preda}} \eeq \qslemb{\Psi^{\oa}(\false)} \eeq \Phi^{\oa}(0) \eeq \wp{\WHILEDO{\guard}{\cc}}{\qslemb{\preda}}~.
    \end{align}

    For $\ob = 0$, we have
    \begin{align}
            & \qslemb{\Psi^{0}(\false)} \\
            \eeqtag{Definition of $\Psi^{0}$}
            & \qslemb{\false} \\
            \eeqtag{Definition~\ref{def:embedding-sl-qsl}}
            & 0 \\
            \eeqtag{Definition of $\Phi^{0}$}
            & \Phi^{0}(0)~.
    \end{align}
    For a successor ordinal $\ob+1$, we have
    \begin{align}
            & \qslemb{\Psi^{\ob+1}(\false)} \\
            \eeqtag{Definition of $\Psi^{\ob+1}$}
            & \qslemb{\Psi(\Psi^{\ob}(\false))} \\
            \eeqtag{Definition of $\Psi$}
            & \qslemb{\left(\neg \guard \wedge \preda\right) \vee \left(\guard \wedge \SLwp{\cc}{\Psi^{\ob}(\false)}\right)} \\
            \eeqtag{Definition~\ref{def:embedding-sl-qsl}}
            & \iverson{\neg \guard} \cdot \qslemb{\preda} + \iverson{\guard} \cdot \qslemb{\SLwp{\cc}{\Psi^{\ob}(\false)}} \\
            \eeqtag{outer I.H.}
            & \iverson{\neg \guard} \cdot \qslemb{\preda} + \iverson{\guard} \cdot \wp{\cc}{\qslemb{\Psi^{\ob}(\false)}} \\
            \eeqtag{inner I.H.}
            & \iverson{\neg \guard} \cdot \qslemb{\preda} + \iverson{\guard} \cdot \wp{\cc}{\Phi^{\ob}(0)} \\
            \eeqtag{Definition of $\Phi$}
            & \Phi(\Phi^{\ob}(0)) \\
            \eeqtag{Definition of $\Phi^{\ob+1}$}
            & \Phi^{\ob+1}(0)~.
    \end{align}
    For a limit ordinal $\ob$, we have
    \begin{align}
            & \qslemb{\Psi^{\ob}(\false)} \\
            \eeqtag{Definition of $\Psi^{\ob}$ for a limit ordinal $\ob$}
            & \qslemb{\sup_{\oc < \oa} \Psi^{\oc}(\false)} \\
            \eeqtag{Supremum distributes over embedding}
            & \sup_{\oc < \oa} \qslemb{\Psi^{\oc}(\false)} \\
            \eeqtag{inner I.H.}
            & \sup_{\oc < \oa} \Phi^{\oc}(0) \\
            \eeqtag{Definition of $\Phi^{\ob}$ for a limit ordinal $\ob$}
            & \Phi^{\ob}(0)~.
    \end{align}

\end{proof}

%
%
%
%
%
\begin{proof}[Proof of Theorem~\ref{thm:qsl:conservativity:wp}]
    We first notice a standard fact for Hoare triples in relation to weakest preconditions:
    \begin{align}
            \{ \,\preda\, \}\,\cc\,\{\,\predb\,\}~\text{is valid for total correctness} \quad\text{iff}\quad \preda \Longrightarrow \underline{\wp{\cc}{\predb}}~,
    \end{align}
    It thus suffices to prove that
    \begin{align}
            \preda \Longrightarrow \underline{\wp{\cc}{\predb}}
            \quad\text{iff}\quad \qslemb{\preda} \ppreceq \wp{\cc}{\qslemb{\predb}}~.
    \end{align}
    Since, by Theorem~\ref{thm:qsl:conservativity:language}.\ref{thm:qsl:conservativity:language:0-1},
    $\qslemb{\preda}(\sk,\hh) \in \{0,1\}$, it suffices to distinguish two cases.
    First, assume $\qslemb{\preda}(\sk,\hh) = 0$ and consequently $(\sk,\hh) \not\models \preda$ by Theorem~\ref{thm:qsl:conservativity:language}.\ref{thm:qsl:conservativity:language:equivalence}.
    Then we immediately obtain
    \begin{align}
            (\sk,\hh) \models \preda \Longrightarrow \SLwp{\cc}{\predb} \quad\text{and}\quad \qslemb{\preda}(\sk,\hh) \lleq \wp{\cc}{\qslemb{\predb}}~.
    \end{align}
    Second, assume $\qslemb{\preda}(\sk,\hh) = 1$ and consequently $(\sk,\hh) \models \preda$. 
    Then
    \begin{align}
            & (\sk,\hh) \models \preda \Longrightarrow \SLwp{\cc}{\predb} \\
            \leftrighttag{assumption}
            & (\sk,\hh) \models \SLwp{\cc}{\predb} \\
            \leftrighttag{Theorem~\ref{thm:qsl:conservativity:language}.\ref{thm:qsl:conservativity:language:equivalence}}
            & \qslemb{\SLwp{\cc}{\predb}}(\sk,\hh) \eeq 1 \\
            \leftrighttag{Lemma~\ref{thm:sl:conservativity}} 
            & \wp{\cc}{\qslemb{\predb}}(\sk,\hh) \eeq 1 \\
            \leftrighttag{assumption}
            & \qslemb{\preda}(\sk,\hh) \lleq \wp{\cc}{\qslemb{\predb}}(\sk,\hh)~.
    \end{align}
\end{proof}

%
\subsection{Proof of Theorem~\ref{thm:frame-rules} (Frame Rule)}
\label{app:frame-rule}
\begin{proof}
   We show Theorem~\ref{thm:frame-rules} by induction on the structure of $\hpgcl$ programs. 

   \emph{The case $\SKIP$}
   \begin{align}
           & \wp{\SKIP}{\ff} \sepcon \fg \\
           \eeqtag{Table~\ref{table:wp}}
           & \ff \sepcon \fg \\
           \eeqtag{Table~\ref{table:wp}}
           & \wp{\SKIP}{\ff \sepcon \fg}.
   \end{align}
%
%
   \emph{The case $\ASSIGN{x}{\ee}$}
   \begin{align}
           & \wp{\ASSIGN{x}{\ee}}{\ff} \sepcon \fg \\
           \eeqtag{Table~\ref{table:wp}}
           & \ff\subst{x}{\ee} \sepcon \fg \\
           \eeqtag{$x \in \Mod{\ASSIGN{x}{\ee}}$. Hence, $x \notin \Vars(\fg)$} 
           & \ff\subst{x}{\ee} \sepcon \fg\subst{x}{\ee} \\ 
           \eeqtag{algebra}
           & \left(\ff \sepcon \fg\right)\subst{x}{\ee} \\
           \eeqtag{Table~\ref{table:wp}}
           & \wp{\ASSIGN{x}{\ee}}{\ff \sepcon \fg}.
   \end{align}

   \emph{The case $\ALLOC{x}{\vec{\ee}}$}
   \begin{align}
           & \wp{\ALLOC{x}{\vec{\ee}}}{\ff \sepcon \fg}  \\
           \eeqtag{Table~\ref{table:wp}}
           & \inf_{v \in \AVAILLOC{\vec{\ee}}} \left\{ \singleton{v}{\vec{\ee}} \sepimp \left(\ff\sepcon \fg\right)\subst{x}{v} \right\} \\
           \eeqtag{$x \notin \Vars(\fg)$} 
           & \inf_{v \in \AVAILLOC{\vec{\ee}}} \left\{ \singleton{v}{\vec{\ee}} \sepimp \left(\ff\subst{x}{v} \sepcon \fg\right) \right\} \\ 
           \eeqtag{Definition of $\sepimp$}
           & \lambda (\sk,\hh) \mydot \inf_{v \in \AVAILLOC{\vec{\ee}}} \inf_{\hh'} 
                  \left\{ \left(\ff\subst{x}{v} \sepcon \fg \right)(\sk,\hh \sepcon \hh') ~|~ \hh' \disjoint \hh \textnormal{ and } (\sk,\hh') \models \singleton{v}{\vec{\ee}} \right\} \\
           \eeqtag{Definition of $\sepcon$}
           & \lambda (\sk,\hh) \mydot \inf_{v \in \AVAILLOC{\vec{\ee}}} \inf_{\hh'} \big\{ 
                \max_{\hh_1,\hh_2} \setcomp{ \ff\subst{x}{v}(\sk,\hh_1) \cdot \fg(\sk,\hh_2) }{ \hh \sepcon \hh' = \hh_1 \sepcon \hh_2 } \\
           & \qquad \qquad ~|~ \hh' \disjoint \hh \textnormal{ and } (\sk,\hh') \models \singleton{v}{\vec{\ee}} \big\} \notag \\
           \ssucceqtag{choose $\hh' \subseteq \hh_1$}
           & \lambda (\sk,\hh) \mydot \inf_{v \in \AVAILLOC{\vec{\ee}}} \inf_{\hh'} \big\{ 
                \max_{\hh_1,\hh_2} \setcomp{ \ff\subst{x}{v}(\sk,\hh_1 \sepcon \hh') \cdot \fg(\sk,\hh_2) }{ \hh = \hh_1 \sepcon \hh_2 } \\
           & \qquad \qquad ~|~ \hh' \disjoint \hh \textnormal{ and } (\sk,\hh') \models \singleton{v}{\vec{\ee}} \big\} \notag \\
           \eeqtag{replace $\max$ by $\sup$ for non-empty finite set}
           %
           %
           %
           & \lambda (\sk,\hh) \mydot \inf_{v \in \AVAILLOC{\vec{\ee}}} \inf_{\hh'} \big\{ 
                \sup_{\hh_1,\hh_2} \setcomp{ \ff\subst{x}{v}(\sk,\hh_1 \sepcon \hh') \cdot \fg(\sk,\hh_2) }{ \hh = \hh_1 \sepcon \hh_2 } \\
           & \qquad \qquad ~|~ \hh' \disjoint \hh \textnormal{ and } (\sk,\hh') \models \singleton{v}{\vec{\ee}} \big\} \notag \\
           \ssucceqtag{$\inf_{a \in A} \sup_{b \in B} f(a,b) \geq \sup_{b \in B} \inf_{a \in A} f(a,b)$ twice} 
           & \lambda (\sk,\hh) \mydot \sup_{\hh_1,\hh_2} \big\{ \\
           & \qquad \inf_{v \in \AVAILLOC{\vec{\ee}}} \inf_{\hh'} \setcomp{ \ff\subst{x}{v}(\sk,\hh_1 \sepcon \hh') \cdot \fg(\sk,\hh_2) }{ \hh' \disjoint \hh \textnormal{ and } (\sk,\hh') \models \singleton{v}{\vec{\ee}}} \notag \\
           & \big|~ \hh = \hh_1 \sepcon \hh_2 \big\} \notag \\
           %
           %
           \eeqtag{algebra ($\fg$ does not depend on $\hh'$)}
           & \lambda (\sk,\hh)\mydot \sup_{\hh_1,\hh_2} \big\{ \inf_{v \in \AVAILLOC{\vec{\ee}}} \inf_{\hh'} \\
           & \qquad \qquad \left\{ \ff\subst{x}{v}(\sk,\hh_1 \sepcon \hh') ~|~ \hh' \disjoint \hh_1 \textnormal{ and } (\sk,\hh') \models \singleton{v}{\vec{\ee}} \right\} \cdot \fg(\sk,\hh_2) \notag \\
           & \qquad ~|~ \hh = \hh_1 \sepcon \hh_2 \big\}
           \notag \\
           \eeqtag{Definition of $\sepimp$}
           & \lambda (\sk,\hh)\mydot \sup_{\hh_1,\hh_2} \left\{ \inf_{v \in \AVAILLOC{\vec{\ee}}}  
             \left(\singleton{v}{\vec{\ee}} \sepimp \ff\subst{x}{v}\right)(\sk,\hh_1) \cdot \fg(\sk,\hh_2) ~\middle|~ \hh = \hh_1 \sepcon \hh_2 \right\} \\
           \eeqtag{supremum is attained (the set of partitions $\hh = \hh_1 \sepcon \hh_1$ is non-empty and finite)}
           & \lambda (\sk,\hh)\mydot \max_{\hh_1,\hh_2} \left\{ \inf_{v \in \AVAILLOC{\vec{\ee}}}  
             \left(\singleton{v}{\vec{\ee}} \sepimp \ff\subst{x}{v}\right)(\sk,\hh_1) \cdot \fg(\sk,\hh_2) ~\middle|~ \hh = \hh_1 \sepcon \hh_2 \right\} \\
           \eeqtag{Definition of $\sepcon$}
           & \left(\inf_{v \in \AVAILLOC{\vec{\ee}}} \singleton{v}{\vec{\ee}} \sepimp \ff\subst{x}{v} \right) \sepcon \fg \\
           \eeqtag{Table~\ref{table:wp}}
           & \wp{\ALLOC{x}{\vec{\ee}}}{\ff} \sepcon \fg. 
   \end{align}
   \emph{The case $\HASSIGN{x}{\ee}$}
   \begin{align}
           & \wp{\HASSIGN{x}{\ee}}{\ff \sepcon \fg} \\
           \eeqtag{Table~\ref{table:wp}}
           & \validpointer{x} \sepcon \left(\singleton{x}{\ee} \sepimp (\ff \sepcon \fg) \right) \\
           \eeqtag{algebra}
           & \validpointer{x} \sepcon \lambda (\sk,\hh) \mydot \left(\singleton{x}{\ee} \sepimp (\ff \sepcon \fg) \right)(\sk,\hh) \\
           \eeqtag{Definition of $\sepimp$}
           & \validpointer{x} \sepcon \lambda (\sk,\hh) \mydot \big( \\
           & \qquad \inf_{\hh'} \setcomp{ (\ff \sepcon \fg)(\sk,\hh \sepcon \hh') }{ \hh \disjoint \hh', (\sk,\hh') \models \singleton{x}{\ee} } \notag \\
           & \big) \notag \\
           \eeqtag{Definition of $\sepcon$}
           & \validpointer{x} \sepcon \lambda (\sk,\hh) \mydot \big( \\
           & \qquad \inf_{\hh'} \setcomp{ \max_{\hh_1,\hh_2} \setcomp{ \ff(\sk,\hh_1) \cdot \fg(\sk,\hh_2) }{ \hh \sepcon \hh' = \hh_1 \sepcon \hh_2 } }{ \hh \disjoint \hh', (\sk,\hh') \models \singleton{x}{\ee} } \notag \\
           & \big) \notag \\
           \eeqtag{replace $\max$ by $\sup$ for non-empty finite set}
           %
           %
           %
           & \validpointer{x} \sepcon \lambda (\sk,\hh) \mydot \big( \\
           & \qquad \inf_{\hh'} \setcomp{ \sup_{\hh_1,\hh_2} \setcomp{ \ff(\sk,\hh_1) \cdot \fg(\sk,\hh_2) }{ \hh \sepcon \hh' = \hh_1 \sepcon \hh_2 } }{ \hh \disjoint \hh', (\sk,\hh') \models \singleton{x}{\ee} } \notag \\
           & \big) \notag \\
           \ssucceqtag{choose $\hh' \subseteq \hh_1$} 
           & \validpointer{x} \sepcon \lambda (\sk,\hh) \mydot \big( \\
           & \qquad \inf_{\hh'} \setcomp{ \sup_{\hh_1,\hh_2} \setcomp{ \ff(\sk,\hh_1 \sepcon \hh') \cdot \fg(\sk,\hh_2) }{ \hh = \hh_1 \sepcon \hh_2 } }{ \hh \disjoint \hh', (\sk,\hh') \models \singleton{x}{\ee} } \notag \\
           & \big) \notag \\
           %
           \ssucceqtag{$\inf_{a \in A} \sup_{b \in B} f(a,b) \geq \sup_{b \in B} \inf_{a \in A} f(a,b)$} 
           & \validpointer{x} \sepcon \lambda (\sk,\hh) \mydot \big( \\
           & \qquad 
           \sup_{\hh_1,\hh_2} \setcomp{ 
                \inf_{\hh'} \setcomp{ \ff(\sk,\hh_1 \sepcon \hh') \cdot \fg(\sk,\hh_2) }{ \hh \disjoint \hh', (\sk,\hh') \models \singleton{x}{\ee}}
            }{ \hh = \hh_1 \sepcon \hh_2 } \notag \\
           & \big) \notag \\
           %
           %
           \eeqtag{algebra ($\fg$ does not depend on $\hh'$)}
           & \validpointer{x} \sepcon \lambda (\sk,\hh) \mydot \big( \\
           & \qquad 
           \sup_{\hh_1,\hh_2} \setcomp{ 
                \inf_{\hh'} \setcomp{ \ff(\sk,\hh_1 \sepcon \hh') }{ \hh \disjoint \hh', (\sk,\hh') \models \singleton{x}{\ee}} \cdot \fg(\sk,\hh_2)
            }{ \hh = \hh_1 \sepcon \hh_2 } \notag \\
           & \big) \notag \\
           %
           \eeqtag{Definition of $\sepimp$}
           & \validpointer{x} \sepcon \lambda (\sk,\hh) \mydot \sup_{\hh_1,\hh_2} \setcomp{ (\singleton{x}{\ee} \sepimp \ff)(\sk, \hh_1) \cdot \fg(\sk,\hh_2) }{ \hh = \hh_1 \sepcon \hh_2 } \\
           %
           \eeqtag{supremum is attained (the set of partitions $\hh = \hh_1 \sepcon \hh_1$ is non-empty and finite)}
           & \validpointer{x} \sepcon \lambda (\sk,\hh) \mydot \max_{\hh_1,\hh_2} \setcomp{ (\singleton{x}{\ee} \sepimp \ff)(\sk, \hh_1) \cdot \fg(\sk,\hh_2) }{ \hh = \hh_1 \sepcon \hh_2 } \\
           \eeqtag{Definition of $\sepcon$}
           & \validpointer{x} \sepcon (\singleton{x}{\ee} \sepimp \ff) \sepcon \fg \\
           %
           %
           %
           %
           %
           \eeqtag{Table~\ref{table:wp}}
           & \wp{\HASSIGN{x}{\ee}}{\ff} \sepcon \fg.
   \end{align}
   \emph{The case $\ASSIGNH{x}{\ee}$}
   \begin{align}
           & \wp{\ASSIGNH{x}{\ee}}{\ff \sepcon \fg} \\
           \eeqtag{Table~\ref{table:wp}}
           & \sup_{v \in \Ints} \left\{ \singleton{\ee}{v} \sepcon \left( \singleton{\ee}{v} \sepimp (\ff\sepcon \fg)\subst{x}{v} \right) \right\} \\
           \eeqtag{$x \notin \Vars(\fg)$} 
           & \sup_{v \in \Ints} \left\{ \singleton{\ee}{v} \sepcon \left( \singleton{\ee}{v} \sepimp (\ff\subst{x}{v} \sepcon \fg) \right) \right\} \\
           \eeqtag{Lemma~\ref{lem:wand-reynolds}} 
           & \sup_{v \in \Ints} \left\{ \containsPointer{\ee}{v} \cdot (\ff\subst{x}{v} \sepcon \fg) \right\} \\
           \eeqtag{Definition of $\sepcon$}
           & \lambda (\sk,\hh) \mydot \sup_{v \in \Ints} \max_{\hh_1,\hh_2} \left\{ \containsPointer{\ee}{v}(\sk,\hh) \cdot (\ff\subst{x}{v}(\sk,\hh_1) \cdot \fg(\sk,\hh_2)) ~|~ \hh = \hh_1 \sepcon \hh_2 \right\} \\
           \eeqtag{algebra}
           & \lambda (\sk,\hh) \mydot \sup_{v \in \Ints} \max_{\hh_1,\hh_2} \left\{ (\containsPointer{\ee}{v}(\sk,\hh) \cdot \ff\subst{x}{v}(\sk,\hh_1)) \cdot \fg(\sk,\hh_2) ~|~ \hh = \hh_1 \sepcon \hh_2 \right\} \\
           \ssucceqtag{take subset in which $\containsPointer{\ee}{v}$ is evaluated in $\hh_1$ instead of $\hh$}
           & \lambda (\sk,\hh) \mydot \sup_{v \in \Ints} \max_{\hh_1,\hh_2} \left\{ (\containsPointer{\ee}{v}(\sk,\hh_1) \cdot \ff\subst{x}{v}(\sk,\hh_1)) \cdot \fg(\sk,\hh_2) ~|~ \hh = \hh_1 \sepcon \hh_2 \right\} \\
           \eeqtag{Definition of $\sepcon$}
           & \sup_{v \in \Ints} \left\{ (\containsPointer{\ee}{v} \cdot \ff\subst{x}{v}) \sepcon \fg \right\} \\
           \eeqtag{$v$ fresh, does not occur in $\fg$}
           & \sup_{v \in \Ints} \left\{ \containsPointer{\ee}{v} \cdot \ff\subst{x}{v} \right\} \sepcon \fg \\
           \eeqtag{Table~\ref{table:wp}}
           & \wp{\ASSIGNH{x}{\ee}}{\ff} \sepcon \fg.
   \end{align}
   \emph{The case $\FREE{x}$}
   \begin{align}
           & \wp{\FREE{x}}{\ff} \sepcon \fg \\
           \eeqtag{Table~\ref{table:wp}}
           & \left(\validpointer{x} \sepcon \ff\right) \sepcon \fg \\
           \eeqtag{Theorem~\ref{thm:sep-con-monoid}.\ref{thm:sep-con-monoid:ass}} 
           & \validpointer{x} \sepcon \left(\ff \sepcon \fg\right) \\
           \eeqtag{Table~\ref{table:wp}}
           & \wp{\FREE{x}}{\ff \sepcon \fg}.
   \end{align}
   \emph{The case $\COMPOSE{\cc_1}{\cc_2}$}
   \begin{align}
           & \wp{\COMPOSE{\cc_1}{\cc_2}}{\ff} \sepcon \fg \\
           \eeqtag{Table~\ref{table:wp}}
           & \wp{\cc_1}{\wp{\cc_2}{\ff}} \sepcon \fg \\
           \ppreceqtag{I.H. on $\cc_1$}
           & \wp{\cc_1}{\wp{\cc_2}{\ff} \sepcon \fg} \\
           \ppreceqtag{I.H. on $\cc_2$} 
           & \wp{\cc_1}{\wp{\cc_2}{\ff \sepcon \fg}} \\
           \eeqtag{Table~\ref{table:wp}}
           \eeq & \wp{\COMPOSE{\cc_1}{\cc_1}}{\ff \sepcon \fg}.
   \end{align}
   \emph{The case $\PCHOICE{\cc_1}{p}{\cc_2}$}
   \begin{align}
           & \wp{\PCHOICE{\cc_1}{p}{\cc_2}}{\ff} \sepcon \fg \\
           \eeqtag{Table~\ref{table:wp}}
           & \left(p \cdot \wp{\cc_1}{\ff} + (1-p) \cdot \wp{\cc_2}{\ff}\right) \sepcon \fg \\
           \ppreceqtag{Theorem~\ref{thm:sep-con-distrib}.\ref{thm:sep-con-distrib:sepcon-over-plus}}
           & \left(p \cdot \wp{\cc_1}{\ff}\right) \sepcon \fg + \left((1-p) \cdot \wp{\cc_2}{\ff}\right) \sepcon \fg \\
           \eeqtag{Theorem~\ref{thm:sep-con-algebra-pure}.3}
           & p \cdot \left(\wp{\cc_1}{\ff} \sepcon \fg\right) + (1-p) \cdot \left(\wp{\cc_2}{\ff} \sepcon \fg\right) \\
           \ppreceqtag{I.H. for $\cc_1$ and $\cc_2$} 
           & p \cdot \wp{\cc_1}{\ff \sepcon \fg} + (1-p) \cdot \wp{\cc_2}{\ff \sepcon \fg} \\ 
           \eeqtag{Table~\ref{table:wp}}
           & \wp{\PCHOICE{\cc_1}{p}{\cc_2}}{\ff \sepcon \fg}.
   \end{align}
   \emph{The case $\ITE{\guard}{\cc_1}{\cc_2}$}
   \begin{align}
           & \wp{\ITE{\guard}{\cc_1}{\cc_2}}{\ff} \sepcon \fg \\
           \eeqtag{Table~\ref{table:wp}}
           & \left(\iverson{\guard} \cdot \wp{\cc_1}{\ff} + \iverson{\neg \guard} \cdot \wp{\guard_2}{\ff}\right) \sepcon \fg \\
           \ppreceqtag{Theorem~\ref{thm:sep-con-distrib}.\ref{thm:sep-con-distrib:sepcon-over-plus}} 
           & \left(\iverson{\guard} \cdot \wp{\cc_1}{\ff}\right) \sepcon \fg + \left(\iverson{\neg \guard} \cdot \wp{\cc_2}{\ff}\right) \sepcon \fg \\ 
           \eeqtag{Theorem~\ref{thm:sep-con-algebra-pure}.3} 
           & \iverson{\guard} \cdot \left(\wp{\cc_1}{\ff} \sepcon \fg\right) + \iverson{\neg \guard} \cdot \wp{\cc_2}{\ff} \sepcon \fg \\
           \ppreceqtag{I.H. for $\cc_1$ and $\cc_2$} 
           & \iverson{\guard} \cdot \wp{\cc_1}{\ff \sepcon \fg} + \iverson{\neg \guard} \cdot \wp{\cc_2}{\ff \sepcon \fg} \\
           \eeqtag{Table~\ref{table:wp}}
           & \wp{\ITE{\guard}{\cc_1}{\cc_2}}{\ff \sepcon \fg}.
   \end{align}
   
   \emph{The case $\WHILEDO{\guard}{\cc}$}
   Recall the functional $\charwp{\guard}{\cc}{\fg}$ determining the unrollings of loop $\WHILEDO{\guard}{\cc}$ with respect to $\ff \in \E$ given by
   \begin{align}
           \charwp{\guard}{\cc}{\ff}(\fh) \eeq \iverson{\neg \guard} \cdot \ff + \iverson{\guard} \cdot \wp{\cc}{\fh}. \label{eq:proof:framerule:charwp}
   \end{align}
   Then, by Table~\ref{table:wp}, we have
   \begin{align}
           \wp{\WHILEDO{\guard}{\cc}}{\ff \sepcon \fg} \eeq \lfp \fh\mydot \charwp{\guard}{\cc}{\ff\sepcon \fg}(\fh).
   \end{align}
   Let $\Ord$ be the class of ordinals.
   By a constructive version of Tarski'\sk fixed point theorem (cf.~\cite{cousot1979constructive})
   we know that this fixed point exists and we have
   \begin{align}
           \lfp \fh\mydot \charwp{\guard}{\cc}{\ff\sepcon \fg}(\fh) \eeq \sup_{\oa \in \textit{Ord}} \charwpn{\guard}{\cc}{\ff \sepcon \fg}{\oa}(0).
   \end{align}
   In particular, there is some ordinal for which the least fixed point is reached.
   To complete the proof, we show that
   \begin{align}
           \forall \oa \in \Ord ~\colon~ \charwpn{\guard}{\cc}{\ff\sepcon \fg}{\oa}(0) \ssucceq \charwpn{\guard}{\cc}{\ff}{\oa}(0) \sepcon \fg
   \end{align}
   by transfinite induction on $\oa$.
  
    The case $\oa = 0$ is trivial.
    For $\oa = 1$, we have
   \begin{align}
           & \charwp{\guard}{\cc}{\ff\sepcon \fg}(0) \\
           \eeqtag{by equation~(\ref{eq:proof:framerule:charwp})}
           & \iverson{\neg \guard} \cdot (\ff\sepcon \fg) + \iverson{\guard} \cdot \wp{\cc}{0} \\
           \eeqtag{$\wp{\cc}{0} = 0$}
           & \iverson{\neg \guard} \cdot (\ff\sepcon \fg) \\
           \eeqtag{Theorem~\ref{thm:sep-con-algebra-pure}.3}
           & \left( \iverson{\neg \guard} \cdot \ff \right) \sepcon \fg \\
           \eeqtag{by equation~(\ref{eq:proof:framerule:charwp}), as above}
           & \charwp{\guard}{\cc}{\ff}(0) \sepcon \fg.
   \end{align}
   For successor ordinals, assume that $\charwpn{\guard}{\cc}{\ff\sepcon \fg}{\oa}(0) \ssucceq \charwpn{\guard}{\cc}{\ff}{\oa}(0) \sepcon \fg$.
   Then
   \begin{align}
           & \charwpn{\guard}{\cc}{\ff \sepcon \fg}{\oa+1}(0) \\
           \eeqtag{by definition: $\charwpn{\guard}{\cc}{\ff \sepcon \fg}{\oa+1}(0) \eeq \charwp{\guard}{\cc}{\ff \sepcon \fg}\left( \charwpn{\guard}{\cc}{\ff \sepcon \fg}{\oa}(0) \right)$}
           & \charwp{\guard}{\cc}{\ff \sepcon \fg}\left( \charwpn{\guard}{\cc}{\ff \sepcon \fg}{\oa}(0) \right) \\
           \eeqtag{by equation~(\ref{eq:proof:framerule:charwp})}
           & \iverson{\neg \guard} \cdot (\ff\sepcon \fg) + \iverson{\guard} \cdot \wp{\cc}{\charwpn{\guard}{\cc}{\ff \sepcon \fg}{\oa}(0)} \\
           \ssucceqtag{I.H.}
           & \iverson{\neg \guard} \cdot (\ff\sepcon \fg) + \iverson{\guard} \cdot \wp{\cc}{\charwpn{\guard}{\cc}{\ff}{\oa}(0) \sepcon \fg} \\
           \ssucceqtag{I.H. of outer induction}
           & \iverson{\neg \guard} \cdot (\ff\sepcon \fg) + \iverson{\guard} \cdot \left(\wp{\cc}{\charwpn{\guard}{\cc}{\ff}{\oa}(0)} \sepcon \fg\right) \\
           \eeqtag{Theorem~\ref{thm:sep-con-algebra-pure}.3} 
           & \left(\iverson{\neg \guard} \cdot \ff \right) \sepcon \fg + \left(\iverson{\guard} \cdot \wp{\cc}{\charwpn{\guard}{\cc}{\ff}{\oa}(0)}\right) \sepcon \fg \\
           \ssucceqtag{Theorem~\ref{thm:sep-con-distrib}.\ref{thm:sep-con-distrib:sepcon-over-plus}}
           & \left(\iverson{\neg \guard} \cdot \ff + \iverson{\guard} \cdot \left(\wp{\cc}{\charwpn{\guard}{\cc}{\ff}{\oa}(0)} \right)\right) \sepcon \fg \\
           \eeqtag{by equation~(\ref{eq:proof:framerule:charwp})}
           & \charwp{\guard}{\cc}{\ff \sepcon \fg}\left( \charwpn{\guard}{\cc}{\ff \sepcon \fg}{\oa}(0) \right) \\
           \eeqtag{by definition: $\charwpn{\guard}{\cc}{\ff \sepcon \fg}{\oa+1}(0) \eeq \charwp{\guard}{\cc}{\ff \sepcon \fg}\left( \charwpn{\guard}{\cc}{\ff \sepcon \fg}{\oa}(0) \right)$}
           & \charwpn{\guard}{\cc}{\ff}{\oa+1}(0) \sepcon \fg.
   \end{align}
   Finally, let $\oa$ be a limit ordinal and assume for all $\ob < \oa$ that
   $\charwpn{\guard}{\cc}{\ff \sepcon \fg}{\ob}(0) \ssucceq \charwpn{\guard}{\cc}{\ff}{\ob}(0) \sepcon \fg$.
   Then 
   \begin{align}
           & \charwpn{\guard}{\cc}{\ff \sepcon \fg}{\oa}(0) \\
           \eeqtag{Definition of $\charwpn{\guard}{\cc}{\ff \sepcon \fg}{\oa}(0)$ for $\oa$ limit ordinal}
           & \sup_{\ob < \oa} \charwpn{\guard}{\cc}{\ff \sepcon \fg}{\ob}(0) \\
           \ssucceqtag{I.H.}
           & \sup_{\ob < \oa} \left( \charwpn{\guard}{\cc}{\ff}{\ob}(0) \sepcon \fg \right) \\
           \eeqtag{algebra}
           & \sup_{\ob < \oa} \lambda(\sk,\hh)\mydot \left( \charwpn{\guard}{\cc}{\ff}{\ob}(0) \sepcon \fg \right)(\sk,\hh) \\
           \eeqtag{algebra}
           & \lambda(\sk,\hh) \mydot \sup_{\ob < \oa} \left( \charwpn{\guard}{\cc}{\ff}{\ob}(0) \sepcon \fg\right)(\sk,\hh) \\
           \eeqtag{Definition of $\sepcon$}
           & \lambda(\sk,\hh) \mydot \sup_{\ob < \oa} \max_{\hh_1, \hh_2} \setcomp{ \charwpn{\guard}{\cc}{\ff}{\ob}(0)(\sk,\hh_1) \cdot \fg(\sk,\hh_2) }{\hh = \hh_1 \sepcon \hh_2} \\
           \eeqtag{replace $\max$ by $\sup$ for non-empty finite set}
           & \lambda(\sk,\hh) \mydot \sup_{\ob < \oa} \sup_{\hh_1, \hh_2} \setcomp{ \charwpn{\guard}{\cc}{\ff}{\ob}(0)(\sk,\hh_1) \cdot \fg(\sk,\hh_2) }{\hh = \hh_1 \sepcon \hh_2} \\
           \eeqtag{commute suprema}
           & \lambda(\sk,\hh) \mydot \sup_{\hh_1, \hh_2} \setcomp{ \sup_{\ob < \oa} \charwpn{\guard}{\cc}{\ff}{\ob}(0)(\sk,\hh_1) \cdot \fg(\sk,\hh_2) }{\hh = \hh_1 \sepcon \hh_2} \\
           \eeqtag{algebra ($\fg$ does not depend on $\ob$)}
           & \lambda(\sk,\hh) \mydot \sup_{\hh_1,\hh_2} \setcomp{ \fg(\sk,\hh_2) \cdot \sup_{\ob < \oa} \charwpn{\guard}{\cc}{\ff}{\ob}(0)(\sk,\hh_1) }{ \hh = \hh_1 \sepcon \hh_2 } \\
           \eeqtag{Definition of $\charwpn{\guard}{\cc}{\ff}{\oa}(0)$ for $\oa$ limit ordinal}
           & \lambda(\sk,\hh) \mydot \sup_{\hh_1,\hh_2} \setcomp{ \fg(\sk,\hh_2) \cdot \charwpn{\guard}{\cc}{\ff}{\oa}(0)(\sk,\hh_1) }{ \hh = \hh_1 \sepcon \hh_2 } \\
           \eeqtag{supremum is attained (the set of partitions $\hh = \hh_1 \sepcon \hh_1$ is non-empty and finite)}
           & \lambda(\sk,\hh) \mydot \max_{\hh_1,\hh_2} \setcomp{ \fg(\sk,\hh_2) \cdot \charwpn{\guard}{\cc}{\ff}{\oa}(0)(\sk,\hh_1) }{ \hh = \hh_1 \sepcon \hh_2 } \\
           \eeqtag{Definition of $\sepcon$}
           & \fg \sepcon \charwpn{\guard}{\cc}{\ff}{\oa}(0) \\
           \eeqtag{commutativity of $\sepcon$}
           & \charwpn{\guard}{\cc}{\ff}{\oa}(0) \sepcon \fg.
   \end{align}
\end{proof}

\subsection{Proof of Theorem~\ref{thm:wlp-frame-rules} (Frame Rule for Weakest Liberal Preexpectations)}
\label{app:wlp-frame-rule}
\begin{proof}
    We show Theorem~\ref{thm:wlp-frame-rules} by induction on the structure of $\hpgcl$ programs $\cc$.
    As specified in Section~\ref{sec:wp:landscape}, the only difference between $\wpsymbol$ and $\wlpsymbol$ is---apart from renaming $\wpsymbol$ by $\wlpsymbol$---the definition for loops:
    \begin{align}
            \wlp{\WHILEDO{\guard}{\cc}}{\ff} \eeq \gfp \fh\mydot \charwp{\guard}{\cc}{\ff}(\fh), 
    \end{align}
    where 
    \begin{align}
            \charwp{\guard}{\cc}{\ff}(\fh) \eeq \iverson{\neg \guard} \cdot \ff + \iverson{\guard} \cdot \wlp{\cc}{\fh}. \label{eq:proof:framerule:charwlp}
    \end{align}
    We thus only consider loops. All other cases are analogous to the proof of Theorem~\ref{thm:frame-rules} (see Appendix~\ref{app:frame-rule}, p.~\pageref{app:frame-rule}).

   \emph{The case $\WHILEDO{\guard}{\cc}$.}
   Let $\Ord$ be the class of ordinals.
   By a constructive version of Tarski's fixed point theorem (cf.~\cite{cousot1979constructive})
   we know that this fixed point exists and we have
   \begin{align}
           \gfp \fh\mydot \charwp{\guard}{\cc}{\ff\sepcon \fg}(\fh) \eeq \inf_{\oa \in \textit{Ord}} \charwpn{\guard}{\cc}{\ff \sepcon \fg}{\oa}(1).
   \end{align}
   In particular, there is some ordinal for which the greatest fixed point  is reached.
   To complete the proof, we show that
   \begin{align}
           \forall \oa \in \Ord ~\colon~ \charwpn{\guard}{\cc}{\ff\sepcon \fg}{\oa}(1) \ssucceq \charwpn{\guard}{\cc}{\ff}{\oa}(1) \sepcon \fg
   \end{align}
   by transfinite induction on $\oa$.
  
    The case $\oa = 0$ is trivial.
    For $\oa = 1$, we have
   \begin{align}
           & \charwp{\guard}{\cc}{\ff\sepcon \fg}(1) \\
           \eeqtag{by equation~(\ref{eq:proof:framerule:charwlp})}
           & \iverson{\neg \guard} \cdot (\ff\sepcon \fg) + \iverson{\guard} \cdot \wlp{\cc}{1} \\
           \eeqtag{Theorem~\ref{thm:sep-con-algebra-pure}.3}
           & \left( \iverson{\neg \guard} \cdot \ff \right) \sepcon \fg + \iverson{\guard} \cdot \wlp{\cc}{1} \\
           \ssucceqtag{$\fg \in \Eone$} 
           & \left( \iverson{\neg \guard} \cdot \ff \right) \sepcon \fg + \left(\iverson{\guard} \cdot \wlp{\cc}{1}\right) \sepcon \fg \\
           \ssucceqtag{Theorem~\ref{thm:sep-con-distrib}.\ref{thm:sep-con-distrib:sepcon-over-plus}}
           & \left( \iverson{\neg \guard} \cdot \ff + \iverson{\guard} \cdot \wlp{\cc}{1}\right) \sepcon \fg \\
           \eeqtag{by equation~(\ref{eq:proof:framerule:charwlp})}
           & \charwp{\guard}{\cc}{\ff}(1) \sepcon \fg.
   \end{align}
   For successor ordinals, assume that $\charwpn{\guard}{\cc}{\ff\sepcon \fg}{\oa}(1) \ssucceq \charwpn{\guard}{\cc}{\ff}{\oa}(1) \sepcon \fg$.
   Then
   \begin{align}
           & \charwpn{\guard}{\cc}{\ff \sepcon \fg}{\oa+1}(1) \\
           \eeqtag{by definition: $\charwpn{\guard}{\cc}{\ff \sepcon \fg}{\oa+1}(1) \eeq \charwp{\guard}{\cc}{\ff \sepcon \fg}\left( \charwpn{\guard}{\cc}{\ff \sepcon \fg}{\oa}(1) \right)$}
           & \charwp{\guard}{\cc}{\ff \sepcon \fg}\left( \charwpn{\guard}{\cc}{\ff \sepcon \fg}{\oa}(1) \right) \\
           \eeqtag{by equation~(\ref{eq:proof:framerule:charwlp})}
           & \iverson{\neg \guard} \cdot (\ff\sepcon \fg) + \iverson{\guard} \cdot \wlp{\cc}{\charwpn{\guard}{\cc}{\ff \sepcon \fg}{\oa}(1)} \\
           \ssucceqtag{I.H.}
           & \iverson{\neg \guard} \cdot (\ff\sepcon \fg) + \iverson{\guard} \cdot \wlp{\cc}{\charwpn{\guard}{\cc}{\ff}{\oa}(1) \sepcon \fg} \\
           \ssucceqtag{I.H. of outer induction (on the program structure)}
           & \iverson{\neg \guard} \cdot (\ff\sepcon \fg) + \iverson{\guard} \cdot \left(\wlp{\cc}{\charwpn{\guard}{\cc}{\ff}{\oa}(1)} \sepcon \fg\right) \\
           \eeqtag{Theorem~\ref{thm:sep-con-algebra-pure}.3} 
           & \left(\iverson{\neg \guard} \cdot \ff \right) \sepcon \fg + \left(\iverson{\guard} \cdot \wlp{\cc}{\charwpn{\guard}{\cc}{\ff}{\oa}(1)}\right) \sepcon \fg \\
           \ssucceqtag{Theorem~\ref{thm:sep-con-distrib}.\ref{thm:sep-con-distrib:sepcon-over-plus}}
           & \left(\iverson{\neg \guard} \cdot \ff + \iverson{\guard} \cdot \left(\wp{\cc}{\charwpn{\guard}{\cc}{\ff}{\oa}(1)} \right)\right) \sepcon \fg \\
           \eeqtag{by equation~(\ref{eq:proof:framerule:charwlp})}
           & \charwp{\guard}{\cc}{\ff \sepcon \fg}\left( \charwpn{\guard}{\cc}{\ff \sepcon \fg}{\oa}(1) \right) \\
           \eeqtag{by definition: $\charwpn{\guard}{\cc}{\ff \sepcon \fg}{\oa+1}(1) \eeq \charwp{\guard}{\cc}{\ff \sepcon \fg}\left( \charwpn{\guard}{\cc}{\ff \sepcon \fg}{\oa}(1) \right)$}
           & \charwpn{\guard}{\cc}{\ff}{\oa+1}(1) \sepcon \fg.
   \end{align}
   Finally, let $\oa$ be a limit ordinal and assume for all $\ob < \oa$ that
   \begin{align} 
     \charwpn{\guard}{\cc}{\ff \sepcon \fg}{\ob}(1) \ssucceq \charwpn{\guard}{\cc}{\ff}{\ob}(1) \sepcon \fg.
   \end{align} 
   Then 
   \begin{align}
           & \charwpn{\guard}{\cc}{\ff \sepcon \fg}{\oa}(1) \\
           \eeqtag{Definition of $\charwpn{\guard}{\cc}{\ff \sepcon \fg}{\oa}(1)$ for $\oa$ limit ordinal}
           & \sup_{\ob < \oa} \charwpn{\guard}{\cc}{\ff \sepcon \fg}{\ob}(1) \\
           \ssucceqtag{I.H.}
           & \sup_{\ob < \oa} \left( \charwpn{\guard}{\cc}{\ff}{\ob}(1) \sepcon \fg \right) \\
           \eeqtag{algebra}
           & \sup_{\ob < \oa} \lambda(\sk,\hh)\mydot \left( \charwpn{\guard}{\cc}{\ff}{\ob}(1) \sepcon \fg \right)(\sk,\hh) \\
           \eeqtag{algebra}
           & \lambda(\sk,\hh) \mydot \sup_{\ob < \oa} \left( \charwpn{\guard}{\cc}{\ff}{\ob}(1) \sepcon \fg\right)(\sk,\hh) \\
           \eeqtag{Definition of $\sepcon$}
           & \lambda(\sk,\hh) \mydot \sup_{\ob < \oa} \max_{\hh_1,\hh_2} \setcomp{ \charwpn{\guard}{\cc}{\ff}{\ob}(1)(\sk,\hh_1) \cdot \fg(\sk,\hh_2) }{ \hh = \hh_1 \sepcon \hh_2 } \\
           \eeqtag{replace $\max$ by $\sup$ (over a finite non-empty set)}
           & \lambda(\sk,\hh) \mydot \sup_{\ob < \oa} \sup_{\hh_1,\hh_2} \setcomp{ \charwpn{\guard}{\cc}{\ff}{\ob}(1)(\sk,\hh_1) \cdot \fg(\sk,\hh_2) }{ \hh = \hh_1 \sepcon \hh_2 } \\
           \eeqtag{commute suprema}
           & \lambda(\sk,\hh) \mydot \sup_{\hh_1,\hh_2} \setcomp{ \sup_{\ob < \oa} \charwpn{\guard}{\cc}{\ff}{\ob}(1)(\sk,\hh_1) \cdot \fg(\sk,\hh_2) }{ \hh = \hh_1 \sepcon \hh_2 } \\
           %
           %
           \eeqtag{algebra ($\fg$ does not depend on $\ob$) }
           & \lambda(\sk,\hh) \mydot \sup_{\hh_1,\hh_2} \setcomp{ \fg(\sk,\hh_2) \cdot \sup_{\ob < \oa} \charwpn{\guard}{\cc}{\ff}{\ob}(1)(\sk,\hh_1) }{ \hh = \hh_1 \sepcon \hh_2 } \\
           \eeqtag{Definition of $\charwpn{\guard}{\cc}{\ff}{\oa}(1)$ for $\oa$ limit ordinal}
           & \lambda(\sk,\hh) \mydot \sup_{\hh_1,\hh_2} \setcomp{ \fg(\sk,\hh_2) \cdot \charwpn{\guard}{\cc}{\ff}{\oa}(1)(\sk,\hh_1) }{ \hh = \hh_1 \sepcon \hh_2 } \\
           \eeqtag{supremum is attained (the set of partitions $\hh = \hh_1 \sepcon \hh_1$ is non-empty and finite)}
           & \lambda(\sk,\hh) \mydot \max_{\hh_1,\hh_2} \setcomp{ \fg(\sk,\hh_2) \cdot \charwpn{\guard}{\cc}{\ff}{\oa}(1)(\sk,\hh_1) }{ \hh = \hh_1 \sepcon \hh_2 } \\
           \eeqtag{Definition of $\sepcon$}
           & \fg \sepcon \charwpn{\guard}{\cc}{\ff}{\oa}(1) \\
           \eeqtag{commutativity of $\sepcon$}
           & \charwpn{\guard}{\cc}{\ff}{\oa}(1) \sepcon \fg.
   \end{align}
\end{proof}

\subsection{Proof of Theorem~\ref{thm:wp-duality} (Duality of Weakest Preexpectations)}
\label{app:wp-duality}

Each of the statements in Theorem~\ref{thm:wp-duality} is proven by induction on the structure of $\hpgcl$ programs. 
We consider the relationship between $\wpsymbol$ and $\awlepsymbol$ in detail. The other relationships are shown analogously.
According to Section~\ref{sec:wp:landscape}, $\awlepsymbol$ is given by the rules in Table~\ref{table:awlep}.
In particular, we have 
\begin{align}
        \ff \esepcon \fg \eeq & \lambda (\sk,\hh) \mydot \min_{\hh_1,\hh_2} \setcomp{ 1-\ff(\sk,\hh_1) + \ff(\sk,\hh_1) \cdot \fg(\sk,\hh_2) }{ \hh = \hh_1 \sepcon \hh_2 } \\
        \iverson{\varphi} \esepimp \fg \eeq & \lambda (\sk,\hh) \mydot \sup_{\hh'} \left\{ \fg(\sk,\hh \sepcon \hh') ~|~ \hh \disjoint \hh', (\sk,\hh') \models \varphi \right\}
\end{align}

\begin{table*}[t]
\renewcommand{\arraystretch}{1.5}
\begin{tabular}{@{\hspace{1em}}l@{\hspace{2em}}l}
	\hline\hline
	$\boldsymbol{\cc}$			& $\boldsymbol{\textbf{\textsf{awlep}}\,\left \llbracket \cc\right\rrbracket  \left(\ff \right)}$ \\
	\hline\hline
	$\SKIP$					& $\ff$ 																					\\
	$\ASSIGN{x}{\ee}$			& $\ff\subst{x}{\ee}$ \\
    $\ALLOC{x}{\vec{\ee}}$	& $\displaystyle\lambda(\sk,\hh)\mydot \sup_{v \in \PosNats:v,v+1,\ldots,v+|\vec{\ee}|-1 \notin \dom{\hh}} \left(\singleton{v}{\vec{\ee}} \esepimp \ff\subst{x}{v}\right)(\sk,\hh)$ \\
	$\ASSIGNH{x}{\ee}$			& $\displaystyle\inf_{v \in \Ints} \singleton{\ee}{v} \esepcon \bigl( \singleton{\ee}{v} \esepimp \ff\subst{x}{v} \bigr)$ \\
	$\HASSIGN{\ee}{\ee'}$			& $\validpointer{\ee} \esepcon \bigl(\singleton{\ee}{\ee'} \esepimp \ff \bigr)$ \\
	$\FREE{\ee}$				& $\validpointer{\ee} \esepcon \ff$ \\
	$\COMPOSE{\cc_1}{\cc_2}$		& $\awlep{\cc_1}{\vphantom{\big(}\awlep{\cc_2}{\ff}}$ \\
	$\ITE{\guard}{\cc_1}{\cc_2}$		& $\iverson{\guard} \cdot \awlep{\cc_1}{\ff} + \iverson{\neg \guard} \cdot \awlep{\cc_2}{\ff}$ \\
	$\PCHOICE{\cc_1}{\pp}{\cc_2}$		& $\pp \cdot \awlep{\cc_1}{\ff} + (1- \pp) \cdot \awlep{\cc_2}{\ff}$ \\
	$\WHILEDO{\guard}{\cc'}$		& $\gfp \fh\mydot \iverson{\neg \guard} \cdot \ff + \iverson{\guard} \cdot \awlep{\cc'}{\fh}$ \\
	\hline\hline
\end{tabular}
\caption{Rules for the angelic weakest liberal preexpectation transformer with extrinsic memory safety. Here $\ff \in \Eone$ is a (post)expectation, 
$\ff\subst{x}{v} =  \lambda (\sk,\hh)\mydot \ff(\sk\subst{x}{\sk(v)}, \hh)$ is the ``syntactic replacement'' of $x$ by $v$ in $\ff$, and
$\vec{\ee} = (\ee_1,\ldots,\ee_n)$ is a tuple of expressions. 
Moreover, $\AVAILLOC{\vec{\ee}} = \lambda(\sk,\hh)\mydot \{ v \in \Nats ~|~ v,v+1,\ldots,v+|\vec{\ee}|-1 \notin \dom{\hh} \}$ collects all memory locations for allocation of $\vec{\ee}$ in heap $\hh$ and $\gfp \fh\mydot \Phi(\fh)$ is the greatest fixed point of $\Phi$.
}
\label{table:awlep}
\end{table*}

Now, our goal is to show that
\begin{align}
  \wp{\cc}{\ff} \eeq 1 - \awlep{\cc}{1-\ff}.
\end{align}

\begin{proof}
We proceed by induction on the structure of $\hpgcl$-programs.

\emph{The case $\SKIP$}
\begin{align}
        & \wp{\SKIP}{\ff} \\
        \eeqtag{Table~\ref{table:wp}}
        & \ff \\
        \eeqtag{algebra}
        & 1 - (1-\ff) \\
        \eeqtag{Table~\ref{table:awlep}}
        & 1 - \awlep{\SKIP}{1-\ff}.
\end{align}

\emph{The case $\ASSIGN{x}{\ee}$}
\begin{align}
        & \wp{\ASSIGN{x}{\ee}}{\ff} \\
        \eeqtag{Table~\ref{table:wp}} 
        & \ff\subst{x}{\ee} \\
        \eeqtag{algebra}
        & 1 - (1-\ff\subst{x}{\ee}) \\
        \eeqtag{algebra}
        & 1 - (1-\ff)\subst{x}{\ee} \\
        \eeqtag{Table~\ref{table:awlep}} 
        & 1 - \awlep{\ASSIGN{x}{\ee}}{1-\ff}.
\end{align}

\emph{The case $\ALLOC{x}{\vec{\ee}}$}
\begin{align}
        & 1 - \awlep{\ALLOC{x}{\vec{\ee}}}{1-\ff} \\
        \eeqtag{Table~\ref{table:awlep}} 
        & \lambda(\sk,\hh)\mydot 1 - 
        \sup_{v \in \PosNats : v,v+1,\ldots,v+|\vec{\ee}|-1 \notin \dom{\hh}} \left(\singleton{v}{\vec{\ee}} \esepimp (1-\ff)\subst{x}{v}\right)(\sk,\hh) \\
        \eeqtag{algebra}
        & \lambda(\sk,\hh)\mydot 1 - 
        \sup_{v \in \PosNats : v,v+1,\ldots,v+|\vec{\ee}|-1 \notin \dom{\hh}} \left(\singleton{v}{\vec{\ee}} \esepimp (1-\ff\subst{x}{v})\right)(\sk,\hh) \\
        %
        \eeqtag{algebra}
        & \lambda(\sk,\hh)\mydot \inf_{v \in \PosNats : v,v+1,\ldots,v+|\vec{\ee}|-1 \notin \dom{\hh}} 
        1 - \left(\singleton{v}{\vec{\ee}} \esepimp (1-\ff\subst{x}{v})\right)(\sk,\hh) \\
        %
        \eeqtag{Definition of $\esepimp$}
        & \lambda(\sk,\hh)\mydot \inf_{v \in \PosNats : v,v+1,\ldots,v+|\vec{\ee}|-1 \notin \dom{\hh}}  
          1 - \sup_{\hh'} \left\{ 1 - \ff\subst{x}{v}(\sk,\hh \sepcon \hh') ~|~ \hh \disjoint \hh', (\sk,\hh') \models \singleton{v}{\vec{\ee}} \right\} \\
        %
        \eeqtag{algebra}
        & \lambda(\sk,\hh)\mydot \inf_{v \in \PosNats : v,v+1,\ldots,v+|\vec{\ee}|-1 \notin \dom{\hh}} 
        1 - \left(1 - \inf_{\hh'} \left\{ \ff\subst{x}{v}(\sk,\hh \sepcon \hh') ~|~ \hh \disjoint \hh', (\sk,\hh') \models \singleton{v}{\vec{\ee}} \right\}\right) \\
        %
        \eeqtag{algebra}
        & \lambda(\sk,\hh)\mydot \inf_{v \in \PosNats : v,v+1,\ldots,v+|\vec{\ee}|-1 \notin \dom{\hh}} 
          \inf_{\hh'} \left\{ \ff\subst{x}{v}(\sk,\hh \sepcon \hh') ~|~ \hh \disjoint \hh', (\sk,\hh') \models \singleton{v}{\vec{\ee}} \right\} \\
        %
        %
        \eeqtag{using equations (\ref{proof:soundness:alloc:2}) and (\ref{proof:soundness:alloc:3})}
        & \lambda(\sk,\hh)\mydot \inf_{v \in \PosNats}
          \inf_{\hh'} \left\{ \ff\subst{x}{v}(\sk,\hh \sepcon \hh') ~|~ \hh \disjoint \hh', (\sk,\hh') \models \singleton{v}{\vec{\ee}} \right\} \\
        \eeqtag{Definition $\sepimp$}
        & \inf_{v \in \AVAILLOC{\vec{\ee}}} \singleton{v}{\vec{\ee}} \sepimp \ff\subst{x}{v} \\
        \eeqtag{Table~\ref{table:wp}} 
        & \wp{\ALLOC{x}{\vec{\ee}}}{\ff}. 
\end{align}

\emph{The case $\ASSIGNH{x}{\ee}$}
\begin{align}
        & 1 - \awlep{\ASSIGNH{x}{\ee}}{1-\ff} \\
        \eeqtag{Table~\ref{table:awlep}} 
        & 1 - \sup_{v \in \Ints} \singleton{\ee}{v} \esepcon \bigl( \singleton{\ee}{v} \esepimp (1-\ff)\subst{x}{v} \bigr) \\
        \eeqtag{algebra}
        & \inf_{v \in \Ints} 1 - \singleton{\ee}{v} \esepcon \bigl( \singleton{\ee}{v} \esepimp 1-\ff\subst{x}{v} \bigr) \\
        \eeqtag{Definition~$\esepcon$} 
        & \inf_{v \in \Ints} 1 - \lambda (\sk,\hh) \mydot \min_{\hh_1,\hh_2} \big\{ \\
        & \qquad 1 - \singleton{\ee}{v}(\sk,\hh_1) + \singleton{\ee}{v}(\sk,\hh_1) \cdot \bigl( \singleton{\ee}{v} \esepimp 1-\ff\subst{x}{v} \bigr)(\sk,\hh_2) \notag \\ 
        & ~\big|~ \hh = \hh_1 \sepcon \hh_2 \big\} \notag \\
        \eeqtag{Definition~$\esepimp$} 
        & \inf_{v \in \Ints} 1 - \lambda (\sk,\hh) \mydot \min_{\hh_1,\hh_2} \big\{ \\
        & \qquad 1 - \singleton{\ee}{v}(\sk,\hh_1) + \singleton{\ee}{v}(\sk,\hh_1) ~\cdot \notag\\
        & \qquad \qquad \sup_{\hh'} \left\{ 1-\ff\subst{x}{v}(\sk,\hh_2\sepcon \hh') ~|~ \hh \disjoint \hh', (\sk,\hh') \models \singleton{\ee}{v} \right\} \notag\\
        & ~\big|~ \hh = \hh_1 \sepcon \hh_2 \big\} \notag \\
        \eeqtag{algebra}
        & \inf_{v \in \Ints} 1 - \lambda (\sk,\hh) \mydot \min_{\hh_1,\hh_2} \big\{ \\
        & \qquad 1 - \singleton{\ee}{v}(\sk,\hh_1) + \singleton{\ee}{v}(\sk,\hh_1) ~\cdot 1 \notag\\
        & \qquad \qquad - \singleton{\ee}{v}(\sk,\hh_1) \cdot \inf_{\hh'} \left\{ \ff\subst{x}{v}(\sk,\hh_2\sepcon \hh') ~|~ \hh \disjoint \hh', (\sk,\hh') \models \singleton{\ee}{v} \right\} \notag\\
        & ~\big|~ \hh = \hh_1 \sepcon \hh_2 \big\} \notag \\
        \eeqtag{algebra}
        & \inf_{v \in \Ints} 1 - \lambda (\sk,\hh) \mydot \min_{\hh_1,\hh_2} \big\{ \\
        & \qquad 1 - \singleton{\ee}{v}(\sk,\hh_1) \cdot \inf_{\hh'} \left\{ \ff\subst{x}{v}(\sk,\hh_2\sepcon \hh') ~|~ \hh \disjoint \hh', (\sk,\hh') \models \singleton{\ee}{v} \right\} \notag\\
        & ~\big|~ \hh = \hh_1 \sepcon \hh_2 \big\} \notag \\
        \eeqtag{algebra}
        & \inf_{v \in \Ints} 1 - 1 + \lambda (\sk,\hh) \mydot \max_{\hh_1,\hh_2} \big\{ \\
        & \qquad \singleton{\ee}{v}(\sk,\hh_1) \cdot \inf_{\hh'} \left\{ \ff\subst{x}{v}(\sk,\hh_2\sepcon \hh') ~|~ \hh \disjoint \hh', (\sk,\hh') \models \singleton{\ee}{v} \right\} \notag\\
        & ~\big|~ \hh = \hh_1 \sepcon \hh_2 \big\} \notag \\
        \eeqtag{algebra}
        & \inf_{v \in \Ints} \lambda (\sk,\hh) \mydot \max_{\hh_1,\hh_2} \big\{ \\
        & \qquad \singleton{\ee}{v}(\sk,\hh_1) \cdot \inf_{\hh'} \left\{ \ff\subst{x}{v}(\sk,\hh_2\sepcon \hh') ~|~ \hh \disjoint \hh', (\sk,\hh') \models \singleton{\ee}{v} \right\} \notag\\
        & ~\big|~ \hh = \hh_1 \sepcon \hh_2 \big\} \notag \\
        \eeqtag{Definition $\sepimp$} 
        & \inf_{v \in \Ints} \lambda (\sk,\hh) \mydot \max_{\hh_1,\hh_2} \setcomp{ \singleton{\ee}{v}(\sk,\hh_1) \cdot (\singleton{\ee}{v} \sepimp \ff\subst{x}{v})(\sk,\hh_2) }{\hh = \hh_1 \sepcon \hh_2} \\
        \eeqtag{Definition $\sepcon$} 
        & \inf_{v \in \Ints} \singleton{\ee}{v} \sepcon (\singleton{\ee}{v} \sepimp \ff\subst{x}{v})(\sk,\hh_2) \\
        \eeqtag{Table~\ref{table:wp}}
        & \wp{\ASSIGNH{x}{\ee}}{\ff}. 
\end{align}

\emph{The case $\HASSIGN{\ee}{\ee'}$}
\begin{align}
        & 1 - \awlep{\HASSIGN{\ee}{\ee'}}{1-\ff} \\
        \eeqtag{Table~\ref{table:awlep}} 
        & 1 - \validpointer{\ee} \esepcon (\singleton{\ee}{\ee'} \esepimp 1-\ff) \\
        \eeqtag{Definition of $\esepcon$} 
        & 1 - \lambda(\sk,\hh)\mydot \min_{\hh_1,\hh_2} \big\{ 1 - \validpointer{\ee}(\sk,\hh_1) \\
        & \qquad + \validpointer{\ee}(\sk,\hh_1) \cdot (\singleton{\ee}{\ee'} \esepimp 1-\ff)(\sk,\hh_2) \notag\\
        & ~\big|~ \hh = \hh_1 \sepcon \hh_2 \big\} \notag \\
        \eeqtag{Definition of $\esepimp$} 
        & 1 - \lambda(\sk,\hh)\mydot \min_{\hh_1,\hh_2} \big\{ 1 - \validpointer{\ee}(\sk,\hh_1) \\
        & \qquad + \validpointer{\ee}(\sk,\hh_1) \cdot \sup_{\hh'} \left\{ 1 - \ff(\sk,\hh_2 \sepcon \hh') ~|~ \hh \disjoint \hh', (\sk,\hh') \models \singleton{\ee}{\ee'} \right\} \notag\\
        & ~\big|~ \hh = \hh_1 \sepcon \hh_2 \big\} \notag \\
        \eeqtag{algebra} 
        & 1 - \lambda(\sk,\hh)\mydot \min_{\hh_1,\hh_2} \big\{ 1 \\
        & \qquad - \validpointer{\ee}(\sk,\hh_1) \cdot \inf_{\hh'} \left\{ \ff(\sk,\hh_2 \sepcon \hh') ~|~ \hh \disjoint \hh', (\sk,\hh') \models \singleton{\ee}{\ee'} \right\} \notag\\
        & ~\big|~ \hh = \hh_1 \sepcon \hh_2 \big\} \notag \\
        \eeqtag{algebra}
        & \lambda(\sk,\hh)\mydot \max_{\hh_1,\hh_2} \big\{ \\
        & \qquad \validpointer{\ee}(\sk,\hh_1) \cdot \inf_{\hh'} \left\{ \ff(\sk,\hh_2 \sepcon \hh') ~|~ \hh \disjoint \hh', (\sk,\hh') \models \singleton{\ee}{\ee'} \right\} \notag \\
        & ~\big|~ \hh = \hh_1 \sepcon \hh_2 \big\} \notag \\
        \eeqtag{Definition of $\sepimp$} 
        & \lambda(\sk,\hh)\mydot \max_{\hh_1,\hh_2} \setcomp{ \validpointer{\ee}(\sk,\hh_1) \cdot (\singleton{\ee}{\ee'} \sepimp \ff)(\sk,\hh_2) }{ \hh = \hh_1 \sepcon \hh_2 } \\
        \eeqtag{Definition of $\sepcon$} 
        & \validpointer{\ee} \sepcon (\singleton{\ee}{\ee'} \sepimp \ff) \\
        \eeqtag{Table~\ref{table:wp}}
        & \wp{\HASSIGN{\ee}{\ee'}}{\ff}.
\end{align}

\emph{The case $\FREE{\ee}$}
\begin{align}
        & 1 - \awlep{\FREE{x}}{1-\ff} \\
        \eeqtag{Table~\ref{table:awlep}}
        & 1 - \validpointer{\ee} \esepcon (1-\ff) \\
        \eeqtag{Definition of $\esepcon$}
        & 1 - \lambda (\sk,\hh) \mydot \min_{\hh_1,\hh_2} \setcomp{ 1 - \validpointer{\ee}(\sk,\hh_1) + \singleton{\ee}{v}(\sk,\hh_1) \cdot (1-\ff)(\sk,\hh_2) }{ \hh = \hh_1 \sepcon \hh_2 } \\
        \eeqtag{algebra}
        & 1 - \lambda (\sk,\hh) \mydot \min_{\hh_1,\hh_2} \big\{ \\
        & \qquad 1 - \validpointer{\ee}(\sk,\hh_1) + \validpointer{\ee}(\sk,\hh_1) - \validpointer{\ee}(\sk,\hh_1) \cdot \ff(\sk,\hh_2) \notag\\
        & \big| \hh = \hh_1 \sepcon \hh_2 \big\}  \notag \\
        \eeqtag{algebra}
        & 1 - \lambda (\sk,\hh) \mydot \min_{\hh_1,\hh_2} \setcomp{ 1 - \validpointer{\ee}(\sk,\hh_1) \cdot \ff(\sk,\hh_2) }{ \hh = \hh_1 \sepcon \hh_2 } \\
        \eeqtag{algebra}
        & 1 - 1 + \lambda (\sk,\hh) \mydot \max_{\hh_1,\hh_2} \setcomp{ \validpointer{\ee}(\sk,\hh_1) \cdot \ff(\sk,\hh_2) }{ \hh = \hh_1 \sepcon \hh_2 } \\
        \eeqtag{algebra}
        & \lambda (\sk,\hh) \mydot \max_{\hh_1,\hh_2} \setcomp{ \validpointer{\ee}(\sk,\hh_1) \cdot \ff(\sk,\hh_2) }{ \hh = \hh_1 \sepcon \hh_2 } \\
        \eeqtag{Definition of $\sepcon$}
        & \validpointer{\ee} \sepcon \ff \\
        \eeqtag{Table~\ref{table:wp}}
        & \wp{\FREE{x}}{\ff}.
\end{align}

\emph{The case $\COMPOSE{\cc_1}{\cc_2}$}
\begin{align}
        & \wp{\COMPOSE{\cc_1}{\cc_2}}{\ff} \\
        \eeqtag{Table~\ref{table:wp}} 
        & \wp{\cc_1}{\wp{\cc_2}{\ff}} \\
        \eeqtag{I.H.}
        & 1 - \awlep{\cc_1}{1 - \wp{\cc_2}{\ff}} \\
        \eeqtag{I.H.}
        & 1 - \awlep{\cc_1}{1 - (1 - \awlep{\cc_2}{1-\ff})} \\
        \eeqtag{algebra}
        & 1 - \awlep{\cc_1}{\awlep{\cc_2}{1-\ff}} \\
        \eeqtag{Table~\ref{table:awlep}} 
        & 1 - \awlep{\COMPOSE{\cc_1}{\cc_2}}{1-\ff}.
\end{align}

\emph{The case $\ITE{\guard}{\cc_1}{\cc_2}$}
\begin{align}
        & \wp{\ITE{\guard}{\cc_1}{\cc_2}}{\ff} \\
        \eeqtag{Table~\ref{table:wp}} 
        & \iverson{\guard} \cdot \wp{\cc_1}{\ff} + \iverson{\neg \guard} \cdot \wp{\cc_2}{\ff} \\
        \eeqtag{I.H.}
        & \iverson{\guard} \cdot (1 - \awlep{\cc_1}{1-\ff}) + \iverson{\neg \guard} \cdot \wp{\cc_2}{\ff} \\
        \eeqtag{I.H.}
        & \iverson{\guard} \cdot (1 - \awlep{\cc_1}{1-\ff}) + \iverson{\neg \guard} \cdot (1 - \awlep{\cc_2}{1-\ff}) \\
        \eeqtag{algebra}
        & \underbrace{(\iverson{\guard} + \iverson{\neg \guard})}_{\eeq 1} - \iverson{\guard} \cdot \awlep{\cc_1}{1-\ff} - \iverson{\neg \guard} \cdot \awlep{\cc_2}{1-\ff} \\
        \eeqtag{algebra}
        & 1 - (\iverson{\guard} \cdot \awlep{\cc_1}{1-\ff} + \iverson{\neg \guard} \cdot \awlep{\cc_2}{1-\ff}) \\
        \eeqtag{Table~\ref{table:awlep}}
        & 1 - \awlep{\ITE{\guard}{\cc_1}{\cc_2}}{1-\ff}.
\end{align}

\emph{The case $\PCHOICE{\cc_1}{\pp}{\cc_2}$}
\begin{align}
        & \wp{\PCHOICE{\cc_1}{\pp}{\cc_2}}{\ff} \\
        \eeqtag{Table~\ref{table:wp}} 
        & \pp \cdot \wp{\cc_1}{\ff} + (1-\pp) \cdot \wp{\cc_2}{\ff} \\
        \eeqtag{I.H.}
        & \pp \cdot (1-\awlep{\cc_1}{1-\ff}) + (1-\pp) \cdot \wp{\cc_2}{\ff} \\
        \eeqtag{I.H.}
        & \pp \cdot (1-\awlep{\cc_1}{1-\ff}) + (1-\pp) \cdot (1-\awlep{\cc_2}{1-\ff}) \\
        \eeqtag{algebra}
        & \underbrace{(\pp + (1-\pp))}_{\eeq 1} - \pp \cdot \awlep{\cc_1}{1-\ff} - (1-\pp) \cdot \awlep{\cc_2}{1-\ff} \\
        \eeqtag{algebra}
        & 1 - (\pp \cdot \awlep{\cc_1}{1-\ff} + (1-\pp) \cdot \awlep{\cc_2}{1-\ff}) \\
        \eeqtag{Table~\ref{table:awlep}}
        & 1 - \awlep{\PCHOICE{\cc_1}{\pp}{\cc_2}}{1-\ff}.
\end{align}

\emph{The case $\WHILEDO{\guard}{\cc}$}
Recall that $\wp{\WHILEDO{\guard}{\cc}}{\ff} = \lfp \fh . F_{\ff}{\fh}$, where
\begin{align}
        F_{\ff}(\fh) \eeq \iverson{\neg \guard} \cdot \ff + \iverson{\guard} \cdot \wp{\cc}{\fh}.
\end{align}
Moreover, we have $\awlep{\WHILEDO{\guard}{\cc}}{1-\ff} = \gfp \fh. G_{1-\ff}{\fh}$, where
\begin{align}
        G_{1-\ff}(\fh) \eeq \iverson{\neg \guard} \cdot (1-\ff) + \iverson{\guard} \cdot \awlep{\cc}{\fh}.
\end{align}

Then, using a constructive version of the Tarski and Knaster fixed point theorem (cf.~\cite{cousot1979constructive}), it suffices to show that
\begin{align}
        \sup_{\oa \in \Ord} F_{\ff}^{\oa}(0) \eeq 1 - \inf_{\oa \in \Ord} G_{1-\ff}^{\oa}(1),
\end{align}
where $\Ord$ denotes the class of all ordinals.
We proceed by transfinite induction on $\oa \in \Ord$ to show 
\begin{align}
        F_{\ff}^{\oa}(0) \eeq 1 - G_{1-\ff}^{\oa}(1).
\end{align}

For $\oa = 0$, we have
\begin{align}
        F_{\ff}^{0}(0) \eeq 0 \eeq 1 - 1 \eeq 1 - G_{1-\ff}^{0}(1).
\end{align}

If $\oa$ is a successor ordinal, we have
\begin{align}
        & F_{\ff}^{\oa+1}(0) \\
        \eeqtag{Definition of $F_{\ff}^{\oa+1}$} 
        & F_{\ff}(F_{\ff}^{\oa}(0)) \\
        \eeqtag{I.H.}
        & F_{\ff}(1 - G_{1-\ff}^{\oa}(1)) \\
        \eeqtag{Definition of $F_{\ff}$} 
        & \iverson{\neg\guard} \cdot \ff + \iverson{\guard} \cdot \wp{\cc}{1-G_{1-\ff}^{\oa}(1)} \\
        \eeqtag{outer I.H.}
        & \iverson{\neg\guard} \cdot \ff + \iverson{\guard} \cdot \left(1 - \awlep{\cc}{G_{1-\ff}^{\oa}(1)}\right) \\
        \eeqtag{algebra}
        & \iverson{\neg \guard} - (\iverson{\neg\guard} \cdot (1-\ff)) + \iverson{\guard} \cdot \left(1 - \awlep{\cc}{G_{1-\ff}^{\oa}(1)}\right) \\
        \eeqtag{algebra}
        & 1 - \left( \iverson{\neg\guard} \cdot (1-\ff) + \iverson{\guard} \cdot \awlep{\cc}{G_{1-\ff}^{\oa}(1)}\right) \\
        \eeqtag{Definition of $G_{1-\ff}$} 
        & 1 - G_{1-\ff}(G_{1-\ff}^{\oa}(1)) \\
        \eeqtag{Definition of $G_{1-\ff}^{\oa+1}$}
        & 1 - G_{1-\ff}^{\oa+1}(1). 
\end{align}

If $\oa$ is a limit ordinal, we have
\begin{align}
        & F_{\ff}^{\oa}(0) \\
        \eeqtag{Definition of $F_{\ff}^{\oa}$} 
        & \sup_{\ob < \oa} F_{\ff}^{\ob}(0) \\
        \eeqtag{I.H.}
        & \sup_{\ob < \oa} 1 - G_{1-\ff}^{\ob}(1) \\
        \eeqtag{algebra}
        & 1 - \inf_{\ob < \oa} G_{1-\ff}^{\ob}(1)  \\
        \eeqtag{Definition of $G_{1-\ff}^{\oa}$}
        & 1 - G_{1-\ff}^{\oa}(1). 
\end{align}

Hence, for all ordinals $\oa \in \Ord$, we have
\begin{align}
        F_{\ff}^{\oa}(0) \eeq 1 - G_{1-\ff}^{\oa}(1)
\end{align}
and thus also $\wp{\WHILEDO{\guard}{\cc}}{\ff} = 1 - \awlep{\WHILEDO{\guard}{\cc}}{1-\ff}$.
\end{proof}

\newpage
\section{Appendix to Section~\ref{sec:extensions} (Beyond \hpgcl Programs)}
\label{app:sec:extensions}

\subsection{Incorporating Recursive Procedure Calls}\label{app.extensions:recursion}

\paragraph{Syntax}
To incorporate procedure calls with parameters and local variables, the syntax of \hpgcl programs has to be adapted.
To this end, let $\ProcNames$ be a set of \emph{procedure names}. 
Then the set $\rhpgcl$ of \emph{recursive} \hpgcl programs $d$ is given by the following context-free grammar:
\begin{align*}
    d ~\longrightarrow~ & \cc \\
    ~|~ & \ProcDeclIs{P}{\vec{x}}{\vec{y}}{\cc} \SEMI d \tag{\textbf{procedure declarations}} \\
    \cc ~\longrightarrow~ & \SKIP \tag{effectless program} \\
	~|~ & \ASSIGN{x}{\ee} \tag{assignment} \\
	~|~ & \COMPOSE{\cc}{\cc} \tag{sequential composition} \\
	~|~ & \ITE{\guard}{\cc}{\cc}\quad{} & \tag{conditional choice} \\
	~|~ & \WHILEDO{\guard}{\cc} & \tag{loop} \\
	~|~ & \PCHOICE{\cc}{\pp}{\cc} \tag{probabilistic choice}\\
	~|~ & \ALLOC{x}{\ee_1,\, \ldots,\, \ee_n} \tag{allocation} \\
	~|~ & \HASSIGN{\ee}{\ee'} \tag{mutation} \\
	~|~ & \ASSIGNH{x}{\ee} \tag{lookup} \\
	~|~ & \FREE{\ee} \tag{deallocation} \\
    ~|~ & \ProcCall{P}{\vec{\ee}}{\vec{y}}, \tag{\textbf{procedure call}}  \\
\end{align*}
%
%
where $\ProcName{P} \in \ProcNames$, $\vec{x}$ is a tuple of variables $\ee, \ee', \ee_1,\ldots,\ee_n$ are arithmetic expressions, $n \in \Nats$, $\vec{\ee}$ is a tuple of arithmetic expressions with $|\vec{\ee}| = |\vec{x}|$,
$\guard$ is a predicate, i.\ee. an expression over variables evaluating to either true or false, and
$\pp \in [0,\, 1] \cap \Rats$ is a probability. 

Let us briefly consider the added program statements.
The statement
\begin{align}
    \ProcCall{P}{\vec{\ee}}{\vec{y}}
\end{align}
calls a procedure $\ProcName{P}$ with call-by-value parameters $\vec{\ee}$.

%

The meaning of such a procedure call is specified by a preceding procedure declaration 
\begin{align}
    \ProcDeclIs{P}{\vec{x}}{\vec{y}}{\cc} 
\end{align}
to declare a procedure $\ProcName{P}$ with parameters $\vec{x}$ and procedure body $\cc$.
All variables in $\cc$ that do not occur in $\vec{x}$ are considered local variables initialized with $0$.
Since a procedure declaration by itself does not modify any variables, we set
\begin{align}
        \Mod{\ProcDeclIs{P}{\vec{x}}{\vec{y}}{\cc}} \eeq \emptyset.
\end{align}

\paragraph{Static semantics}
For simplicity, we require that each procedure name is declared at most once and that every procedure only calls previously declared procedures or itself. Hence, we do not consider mutual recursion.
Moreover, we assume that the number of parameters passed to a procedure matches with the number of declared parameters.
Furthermore, as stated before, we require for a procedure $\ProcDeclIs{P}{\vec{x}}{\vec{y}}{\cc}$ that variables $\vec{x}$ are not modified by $\cc$.

\paragraph{Local variables}
Towards a formal semantics of \rhpgcl programs, we have to define how local variable of procedures are incorporated in our previous notion of program states.
As is standard in denotational semantics, we extend the type of our wp-calculus by a call-stack representation for local variables,
Thus, we split our previous notion of stacks, i.e. evaluations of variables of the form $\sk : \Vars \to \Ints$, 
into \emph{stores} mapping locations on the call stack---represented by natural numbers---to values
and \emph{variable environments} mapping variables to locations on the call stack.
Moreover, stores are assumed to contain a special symbol $\Next$  holding the next free location.
Formally, we introduce the sets:
\begin{align}
        \Stores \eeq & \left\{ \store ~\middle|~ \store : N \cup \{ \Next \} \to \Ints, N \subseteq \Nats, |N| < \infty \right\} \tag{stores} \\
        \VarEnv \eeq & \left\{ \varenv ~\middle|~ \varenv : \Vars \to \Nats \right\} \tag{variable environments}
\end{align}
Since guards, arithmetic expressions, etc. are evaluated by stacks in $\Stacks$, we introduce a function that recovers our original notion of stacks $\sk : \Vars \to \Ints$ used throughout the paper:
\begin{align}
        \stack : \VarEnv \to \Stores \to \Stacks, \quad \stack(\varenv)(\store) \eeq \lookup.
\end{align}
Since expressions are not allowed to depend on the heap, we consider arithmetic expressions $\ee$ and Iverson brackets for predicates $\iverson{\preda}$ as functions
\begin{align}
   \ee : \Stacks\to \Ints \qand \iverson{\preda} : \Stacks \to \{0,1\}.
\end{align}
Then the corresponding evaluation functions for a given variable environment (but an arbitrary store) are given by
\begin{align}
        \ee \circ \stack(\varenv) : \Stores \to \Ints \qand \iverson{\preda} \circ \stack(\varenv) : \Stores \to \{0,1\}.
\end{align}

\paragraph{Expectations}
We also have to adapt the continuations used within our $\wpsymbol$-calculus, i.e. the notion of expectations.
Originally, an expectations maps stack-heap pairs, to positive real numbers or infinity.
In our new setting, in which stacks are split into variable environments and stores, the domain of our continuations consists of \emph{store-heap} pairs instead.
Hence, we consider the set of expectations
\begin{align}
        \Estore \eeq \{ \ff ~|~ \ff : \Stores \times \Heaps \to \PosRealsInf \}.
\end{align}

To enable local reasoning, we restrict ourselves to expectations that cannot measure quantities across variable environments, e.g. measuring the size of a store. Formally:
\begin{definition}\label{def:admissible-exp}
        Let $\varenv \in \VarEnv$ be a variable environment with $V = \{\ell \in \Nats ~|~ \exists x . \varenv(x) = \ell \}$ finite.
        Then an expectation $\ff \in \Estore$ is \emph{admissible for $\varenv$} if and only if
        \begin{align*}
                \forall \hh \, \forall \store_1 \forall \store_2 \colon \dom{\store_1} \cap \dom{\store_2} \subseteq V 
                ~\text{implies}~ \ff(\store_1,\hh) \eeq \ff(\store_2,\hh).
        \end{align*}
        The set of \emph{admissible expectations} for $\varenv$ is denoted by $\Eadm$.
\end{definition}
In particular, given a variable environment $\varenv$ and a ``classical'' expectation $\ff \in \E$ as used throughout the paper, we obtain an admissible expectation in $\fg \in \Eadm$ as follows:
\begin{align}
        \fg \eeq \lambda(\store,\hh)\mydot \ff(\stack(\varenv)(\store),\hh).
\end{align}
Admissible expectations thus suffice to express expectations considered in our original setting.

\paragraph{Semantics of procedure declarations}
A \emph{procedure environment} is a mapping from procedure names in $\ProcNames$ together with values for its parameters to an expectation transformer.
Consequently, the \emph{set of procedure environments} is given by
\begin{align}
        \ProcEnv \eeq \{ \procenv ~|~ \procenv : \ProcNames \to \Ints^{*} \to \left(\Estore \to \Estore\right) \},
\end{align}
where $\Ints^{*}$ denotes the set of all sequences over integers.
If the number of parameters does not match the parameter list in the procedure's declaration, we require that 
$\procenvP{P}{\vec{z}}{\vec{z_2}}$ is undefined.
The semantics of procedure declarations is then defined in terms of a transformer
\begin{align}
        \DeclPSym \colon d \to \VarEnv \to \left(\ProcEnv \to \ProcEnv\right)
\end{align}
given by
\begin{align} 
        \label{eq:recursion:declp}
        \DeclP{\varepsilon}{\varenv}{\procenv} \eeq & \procenv \\
        \DeclP{\ProcDeclIs{P}{\vec{x}}{\vec{y}}{\cc}\SEMI d}{\varenv}{\procenv} \eeq & \DeclP{d}{\varenv}{\procenv\subst{\ProcName{P}}{\lfp f. \ProcPsi{\varenv}{\procenv}(f)}},
\end{align}
where $\lfp f. \ProcPsi{\varenv}{\procenv}(f)$ denotes the least fixed point of $\ProcPsi{\varenv}{\procenv}$.
To describe this transformer formally, let 
\begin{itemize}
  \item $\vec{a} = a_1 \ldots a_i \in \Ints^{i}$ be the values supplied to call-by-value parameters in $\vec{x}$, and
  \item $\{z_1,\ldots,z_k\}$ be the set of all procedure-local variables that are not parameters, i.e. variables that occur in $\Vars(\cc)$, but neither in $\vec{x}$. 
\end{itemize}
Then the transformer 
\begin{align}
    \ProcPsi{\varenv}{\procenv}\colon \left(\Ints^{*} \to (\Estore \to \Estore)\right) \to \left(\Ints^{*} \to (\Estore \to \Estore)\right)
\end{align}
is given by
\begin{align}
        \label{eq:recursion:psi}
        \ProcPsi{\varenv}{\procenv}(f) \eeq & \lambda\vec{a}\mydot \lambda\ff\mydot \lambda(\store,\hh)\mydot \wpsymbol\llbracket \cc,\varenv',\procenv\subst{\ProcName{P}}{f}\rrbracket(\ff)(\store',\hh) \\
        \varenv' \eeq & \varenv''\subst{x_1}{\ell_1}\ldots\subst{x_i}{\ell_k} \tag{add call-by-value parameters} \\
        \varenv'' \eeq & \varenv\subst{z_1}{r_1}\ldots\subst{z_k}{r_k} \tag{add local variables} \\
        \store' \eeq & \store''\subst{\ell_1}{a_1}\ldots\subst{\ell_i}{a_i} \tag{initialize call-by-value parameters} \\
        \store'' \eeq & \store\subst{r_1}{0}\ldots\subst{r_k}{0}\subst{\Next}{\store(\Next)+i+k}, \tag{initialize local variables with $0$}
\end{align}
where $\ell_n = \store(\Next)+(n-1)$, $1 \leq n \leq i$, and $r_n = \store(\Next) + i + n-1$, $1 \leq n \leq k$,
are new locations for local variables and call-by-value parameters.
Intuitively, $\varenv''$ and $\store''$ account for the local variables that are initialized with $0$.
In $\varenv'$ and $\store'$ , we additionally update the values assigned to procedure parameters $\vec{x}$.
Furthermore,
\begin{align}
        \wpsymbol : \hpgcl \to \VarEnv \to \ProcEnv \to (\Estore \to \Estore)
\end{align}
is the $\wpsymbol$-semantics for \rhpgcl programs, which we introduce next.

\paragraph{Weakest preexpectation semantics of $\rhpgcl$}
Our updated semantics for $\rhpgcl$ programs is defined inductively as shown in Table~\ref{table:rec-wp}.
Notice that the semantics of all original \hpgcl statements remains unchanged (except for the use of variable environments and stores which requires an additional indirection using the $\stack$ function).
The semantics of procedure declarations first updates the procedure environment as introduced above (see equation~\ref{eq:recursion:psi}).
The semantics of procedure calls then boils down to applying the current procedure environment to an evaluation of the procedure's parameters.
\begin{table*}[t]
\renewcommand{\arraystretch}{1.5}
\begin{tabular}{@{\hspace{0.5em}}l@{\hspace{2em}}l}
	\hline\hline
    $\boldsymbol{d}$			& $\boldsymbol{\textbf{\textsf{wp}}\,\left \llbracket d, \varenv, \procenv\right\rrbracket\left(\ff\right)}$ \\
	\hline\hline
    $\cc$                                           & $\wp{\cc,\varenv,\procenv}{\ff}$ (see below) \\
    $\ProcDeclIs{P}{\vec{x}}{\vec{y}}{\cc} \SEMI d$ & $\wp{d,\varenv,\procenv'}{\ff}$ \\
                                                  & \qquad where $\DeclP{d}{\varenv}{\procenv} \eeq \procenv'$ \\
	\hline\hline
    $\boldsymbol{\cc}$			& $\boldsymbol{\textbf{\textsf{wp}}\,\left \llbracket \cc, \varenv, \procenv\right\rrbracket\left(\ff\right)}$ \\
	\hline\hline
    $\SKIP$					& $\ff$ 																					\\
    $\ASSIGN{x}{\ee}$			& $\ff\subst{\varenv(x)}{\evv{\ee}}$ \\
    $\ALLOC{x}{\vec{\ee}}$	& $\displaystyle\inf_{v \in \AVAILLOC{\vec{\ee}}} \singleton{v}{\evv{\vec{\ee}}} \sepimp \ff\subst{\varenv(x)}{v}$ \\
    $\ASSIGNH{x}{\ee}$			& $\displaystyle\sup_{v \in \Ints} \singleton{\evv{\ee}}{v} \sepcon \bigl( \singleton{\evv{\ee}}{v} \sepimp \ff\subst{\varenv(x)}{v} \bigr)$ \\
    $\HASSIGN{\ee}{\ee'}$			& $\validpointer{\evv{\ee}} \sepcon \bigl(\singleton{\evv{\ee}}{\evv{\ee'}} \sepimp \ff \bigr)$ \\
    $\FREE{\ee}$				& $\validpointer{\evv{\ee}} \sepcon \ff$ \\
    $\COMPOSE{\cc_1}{\cc_2}$		& $\wp{\cc_1,\varenv,\procenv}{\wp{\cc_2,\varenv,\procenv}{\ff}}$ \\
    $\ITE{\guard}{\cc_1}{\cc_2}$		& $(\evv{\iverson{\guard}}) \cdot \wp{\cc_1,\varenv,\procenv}{\ff}$ \\ 
                                  & \qquad $+ (\evv{\iverson{\neg \guard}}) \cdot \wp{\cc_2,\varenv,\procenv}{\ff}$ \\
	$\PCHOICE{\cc_1}{\pp}{\cc_2}$		& $\pp \cdot \wp{\cc_1,\varenv,\procenv}{\ff} + (1- \pp) \cdot \wp{\cc_2,\varenv,\procenv}{\ff}$ \\
    $\WHILEDO{\guard}{\cc'}$		& $\lfp \fg\mydot (\evv{\iverson{\neg \guard}}) \cdot \ff + (\evv{\iverson{\guard}}) \cdot \wp{\cc',\varenv,\procenv}{\fg}$ \\
    $\ProcCall{P}{\vec{\ee}}{\vec{y}}$ & $\lambda(\store,\hh) \mydot \procenvP{P}{\evv{\vec{\ee}}(\store)}{\varenv(\vec{y})}(\ff)(\store,\hh)$ \\
	\hline\hline
\end{tabular}
\caption{%
Rules of the weakest preexpectation transformer for \rhpgcl programs. Here $\varenv \in \VarEnv$ is a variable environment, $\procenv \in \ProcEnv$ is a procedure environment, and $\ff \in \Estore$ is a (post)expectation.
$\ff\subst{x}{\ee} =  \lambda (\store,\hh)\mydot \ff(\store\subst{\ell}{\store(\ee)}, \hh)$ is the ``syntactic replacement'' of $x$ by $\ee$ in $\ff$, where $\store(\ee)$ is the evaluation of expression $\ee : \Stores \to \Ints$. In particular, if $\ee = v \in \Ints$, we assume the constant function given by $v(\store) = v$.
$\vec{\ee} = (\ee_1,\ldots,\ee_n)$ is a tuple of expressions, and $\vec{x}$, $\vec{y}$ are tuples of variables. 
We write $\evv{\vec{\ee}}$ as a shortcut for $(\evv{\ee_1},\ldots,\evv{\ee_n})$.
Moreover, $\AVAILLOC{\vec{\ee}} = \lambda(\store,\hh)\mydot \{ v \in \Nats ~|~ v,v+1,\ldots,v+|\vec{\ee}|-1 \notin \dom{\hh} \}$ collects all suitable memory locations for allocation of $\vec{\ee}$ in heap $\hh$.
}
\label{table:rec-wp}
\end{table*}
We conclude our introduction of procedure calls and local variables with a few facts that each can be shown by induction on the program structure.
\begin{proposition}
The following facts hold for \rhpgcl programs:
\begin{itemize}
    \item $(\Estore,\preceq)$ is a complete lattice for 
            \[ \ff \preceq \fg \qiff \forall (\store,\hh) \colon \ff(\store,\hh) \leq \fg(\store,\hh). \]
    \item $\wp{\cc,\varenv,\procenv} \colon \Estore \to \Estore$ is monotone with respect to $\preceq$.
    \item $(\ProcEnv,\preceq)$ is a complete lattice for 
            \[ \procenv \preceq \procenv' \qiff \forall \ProcName{P} \in \ProcNames \forall \vec{z_1} \forall \vec{z_2} \colon 
                    \procenvP{P}{\vec{z_1}}{\vec{z_2}} 
                    \ppreceq \procenvPrime{P}{\vec{z_1}}{\vec{z_2}} 
            \]
    \item $\ProcPsi{\varenv}{\procenv} \colon \Ints^{*} \to \Estore \to \Estore$ is a monotone function with respect to $\preceq$. 
\end{itemize}
\end{proposition}

In particular, by the Tarski-Knaster fixed point theorem, the least fixed point $\lfp f. \ProcPsi{\varenv}{\procenv}(f)$ used in the semantics of procedure declarations exists and is given by
\begin{align}
        \lfp f. \ProcPsi{\varenv}{\procenv} \eeq \sup_{\oa \in \Ord} \ProcPsi{\varenv}{\procenv}^{\oa}(0),
\end{align}
where $\Ord$ is the set of ordinals and $0$ is the smallest element of the lattice.

\subsection{Proof Rule for Recursion}

The proof rules to deal with recursion presented in the paper are standard (cf.~\cite{DBLP:journals/fac/Hesselink94}).
Let us briefly discuss how these proof rules are connected to our formalization of procedures.
A formal proof is outside the scope of this paper.
We refer the interested reader to~\cite{DBLP:conf/lics/OlmedoKKM16} for a formal correctness proof of this rule (for a simpler probabilistic programming language).

For simplicity, we consider only a single procedure, say $\ProcName{P}$, which is declared by the statement $\ProcDeclIs{P}{\vec{x}}{\vec{y}}{\cc}$.
Moreover, let us fix a variable environment $\varenv \in \VarEnv$ and a procedure environment $\procenv \in \ProcEnv$. 
By definition of the $\wpsymbol$ semantics of procedure declaration, we obtain a procedure environment $\procenv'$ given by
\begin{align}
    \procenv' \eeq \procenv\subst{\ProcName{P}}{
        \lfp f \mydot \ProcPsi{\varenv}{\procenv}{f}
    }.
\end{align}
The semantics of calls  of procedure $\ProcName{P}$ with respect to $\varenv$ and $\procenv$ is then given by 
\begin{align}
        & \wp{\ProcCall{P}{\vec{\ee}}{\vec{y}}, \varenv, \procenv'}{\ff} \\
        \eeqtag{Table~\ref{table:rec-wp}}
        & \lambda(\store,\hh) \mydot \procenvPrime{P}{\evv{\vec{\ee}}(\store)}{\varenv(\vec{y})}(\ff)(\store,\hh) \\
        \eeqtag{Definition of $\procenv'$}
        &\lambda(\store,\hh) \mydot \left(\lfp f\mydot\ProcPsi{\varenv}{\procenv}{f}\right)\left(\evv{\vec{\ee}}(\store)\right)(\ff)(\store,\hh).
\end{align}
Now, if $g$ is a pre-fixed point of $\ProcPsi{\varenv}{\procenv}$, i.e. $\ProcPsi{\varenv}{\procenv}{g} \preceq g$, then $\lfp f \mydot \ProcPsi{\varenv}{\procenv} \preceq g$.
By definition of $\ProcPsi{\varenv}{\procenv}$ (see equation~\ref{eq:recursion:psi}) the fact that $g$ is a pre-fixed point means that
\begin{align}
    \lambda\vec{a}\mydot \lambda\ff\mydot \lambda(\store,\hh)\mydot \wpsymbol\llbracket \cc,\varenv',\procenv\subst{\ProcName{P}}{g}\rrbracket(\ff)(\store',\hh) \ppreceq g
\end{align}
Consequently, we obtain the rule
\begin{align}
\infer{
        \forall \vec{\ee} \colon
        \wpC{\ProcCall{P}{\vec{\ee}}{\vec{y}}, \varenv, \procenv'} \ppreceq 
        \lambda(\store,\hh) \mydot g\left(\evv{\vec{\ee}}(\store)\right)(\ff)(\store,\hh)
  }{
        \forall \vec{a} \colon \lambda\ff\mydot \lambda(\store,\hh)\mydot \wpsymbol\llbracket \cc,\varenv',\procenv\subst{\ProcName{P}}{g}\rrbracket(\ff)(\store',\hh) \ppreceq g(\vec{a}),
  }
\end{align}
where we replaced the lambdas for parameters $\vec{a}$ by universal quantifiers.
Now, let $I(\vec{x}) \in \E$ be a classical expectation that depends on $\vec{x}$. 
We can then define a corresponding function $g$ as follows:
\begin{align}
        g \eeq \lambda \vec{x} \mydot \lambda(\store,\hh) \mydot \, I(\stack(\store)(\vec{x}))(\stack(\store), \hh).
\end{align}
Inserting this definition in the above proof rule for a fixed $X \in \Estore$ and evaluations of parameters $\vec{a} = \stack{\varenv}(\store)(\vec{\ee})$ in the definition of $\varenv'$, we obtain
\begin{align}
\infer{
        \forall \vec{\ee} \colon 
        \wp{\ProcCall{P}{\vec{\ee}}{\vec{y}}, \varenv, \procenv'}{\ff} \ppreceq 
        I(\vec{\ee})
  }{
        \forall \vec{\ee} \colon \lambda\ff\mydot \lambda(\store,\hh)\mydot \wpsymbol\llbracket \cc,\varenv',\procenv\subst{\ProcName{P}}{g}\rrbracket(\ff)(\store',\hh) \ppreceq I(\vec{\ee}),
  }
\end{align}
In the next step, let us remove the fixed procedure environment $\procenv$. To highlight that $\procenv$ is updated by $g$ for procedure $\ProcName{P}$, write the premise of the above rule as
\begin{align}
  \forall \vec{e} \colon \wp{\ProcCall{P}{\vec{\ee}}{\vec{y}}, \varenv'}{\ff} \ppreceq \inv(\vec{\ee}) ~\Vdash~ \wp{\cc, \varenv'}{\ff} \ppreceq \inv(\vec{x})
\end{align}
Finally, since variable environment $\varenv$ is fixed, let us remove it from our notation as well. We then obtain the proof rule presented in the paper, i.e.
\begin{align*}
\infer[\text{[rec]}]{
     \forall \vec{\ee} \colon \wp{\ProcCall{P}{\vec{\ee}}{\vec{y}}}{\ff} \ppreceq \inv(\vec{\ee}),
  }{
     \forall \vec{\ee} \colon \wp{\ProcCall{P}{\vec{\ee}}{\vec{y}}}{\ff} \ppreceq \inv(\vec{\ee}) ~\Vdash~ \wp{\cc}{\ff} \ppreceq \inv(\vec{x})
  }
\end{align*}

Analogously, if we consider weakest \emph{liberal} preexpectations, i.e. take greatest instead of least fixed points, we obtain a proof rule for lower bounds:
\begin{align*}
\infer[\text{[rec]}]{
        \forall \vec{\ee} \colon \wlp{\ProcCall{P}{\vec{\ee}}{\vec{y}}}{\ff} \ssucceq \inv(\vec{\ee})
  }{
        \forall \vec{\ee} \colon \wlp{\ProcCall{P}{\vec{\ee}}{\vec{y}}}{\ff} \ssucceq \inv(\vec{\ee}) ~\Vdash~ \wlp{\cc}{\ff} \ssucceq \inv(\vec{x}).
  }
\end{align*}

\subsection{Lifting Properties of $\wpsymbol$ to \rhpgcl}
All previously introduced properties of $\wpsymbol$ for \hpgcl programs that have been shown by structural induction on the program structure can be lifted to $\wpsymbol$ for \rhpgcl programs, i.e. programs with recursive procedures. 
In this section, we briefly explain the main steps to adapt our proofs to account for recursion.

Suppose we want to show that for all $\rhpgcl$ programs $d$ (including procedure declarations) and (admissible) expectations $\ff \in \Eadm$, $\wp{d,\varenv,\procenv}{\ff}$ has a property of interest, say $\texttt{Prop}$. 
By assumption, we have already shown $\texttt{Prop}$ for all \hpgcl programs $\cc$ by induction on the program structure.
We then proceed in three steps:
\begin{enumerate}
        \item\label{app:recursion:lifting:1}
              First, we show for all \hpgcl programs $\cc$ with procedure calls (but without declarations) that 
                $\wp{\cc,\varenv,\procenv}{\ff}$ has property \texttt{Prop} \emph{if procedure environment $\procenv$ satisfies property \texttt{Prop} for every procedure used in $\cc$} by induction on the program structure.
              For all cases except procedure calls, this is analogous to our proof for \hpgcl programs.
              For procedure calls, we apply the procedure environment, which, by assumption, satisfies \texttt{Prop}.
        \item\label{app:recursion:lifting:2}
              Next, \emph{assuming a given procedure environment already satisfies \texttt{Prop} for every procedure occurring in an additional procedure declaration}, we show that the transformer for this procedure declaration satisfies property \texttt{Prop}.
              Formally, if $M$ is the set of procedure calls used in the body $\cc$ of $\ProcDeclIs{P}{\vec{x}}{\vec{y}}{\cc}$, we show by transfinite induction that
              \begin{align}
                      \forall \ProcName{Q} \in M \setminus \{\ProcName{P}\} \colon \procenv(\ProcName{Q}) ~\text{satisfies}~\texttt{Prop} 
                      \qimplies 
                      \forall \oa \in \Ord \colon \ProcPsi{\varenv}{\procenv}^{\oa}(0)~\text{satisfies}~\texttt{Prop},
              \end{align}
              where $\ProcPsi{\varenv}{\procenv}$ is the functional used to determine the semantics of procedure \ProcName{P} in the definition of procedure declarations (see equations~\ref{eq:recursion:declp} and~\ref{eq:recursion:psi}).
        \item\label{app:recursion:lifting:3}
              Finally, since every procedure may only call already declared procedures or itself in its procedure body, the premise of the above property is initially satisfied for every procedure environment.
              We then show by a (rather straightforward) induction on the structure of $\rhpgcl$ programs $d$ that
              \begin{align}
                  \wp{d,\varenv,\procenv}{\ff}~\text{satisfies property}~\texttt{Prop}.
              \end{align}
\end{enumerate}
We do not explicitly perform the above steps for every statement that has been proven by induction on the structure of \hpgcl programs before.
Let us, however, consider the frame rule in detail as an example of the above scheme.
The proofs for linearity of $\wpsymbol$, etc. are very similar.
\subsection{Lifting the Frame Rule to \rhpgcl}
Since expectations in $\Estore$ are functions of the form $\ff : \Stores \times \Heaps \to \PosRealsInf$, let us first update our notion of variables that ``occur'' in $\ff$
with respect to a given variable environment $\varenv \in \VarEnv$:
\begin{align*}
        \Vars_{\varenv}(\ff) \eeq \left\{ x \in \Vars ~|~ \exists(\store,\hh) \, \exists v,v' \in \Ints \colon 
        \ff(\store\subst{\varenv(x)}{v},\hh) \neq \ff(\store\subst{\varenv(x)}{v'},\hh) \right\}
\end{align*}
%
%
\begin{theorem}[Quantitative Frame Rule for \rhpgcl]\label{thm:rhpgcl-frame-rules}
  For every $\rhpgcl$-program $d$, variable environment $\varenv \in \VarEnv$, procedure environment $\procenv \in \ProcEnv$,
  and admissible expectations $\ff,\fg \in \Eadm$ with $\Mod{\cc} \cap \Vars_{\varenv}(\fg) = \emptyset$, we have
  \begin{align*}
          \wp{d,\varenv,\procenv}{\ff} \sepcon \fg \ppreceq \wp{d,\varenv,\procenv}{\ff \sepcon \fg}.
  \end{align*}
\end{theorem}
\begin{proof}
We proceed according to the scheme to lift our results from \hpgcl to \rhpgcl . \\
\noindent \textbf{Step~\ref{app:recursion:lifting:1}.} 
Let $\procenv$ be a procedure environment such that for every procedure call $\ProcName{P}$ and parameters $\vec{z}$ occurring in a program $\cc$ and all expectations $\ff,\fg \in \Estore$  with 
$\Mod{\cc} \cap \Vars_{\varenv}(\fg) = \emptyset$, we have 
\begin{align}
    \procenvP{P}{\vec{z}}{\vec{z_2}}(\ff) \sepcon \fg \preceq \procenvP{P}{\vec{z}}{\vec{z_2}}(\ff \sepcon \fg).
\end{align}
We then show by induction on $\cc$, i.e. \hpgcl programs with procedure calls, that 
\begin{align}
        \wp{\cc,\varenv,\procenv}{\ff} \sepcon \fg \ppreceq \wp{\cc,\varenv,\procenv}{\ff \sepcon \fg}. \label{eq:recursive:frame:1}
\end{align}
For all cases except procedure calls the proof is analogous to the proof of the quantitative frame rule for \hpgcl programs, see Theorem~\ref{thm:frame-rules}.
For procedure calls, i.e. $\cc = \ProcCall{P}{\vec{\ee}}{\vec{y}}$, we have
\begin{align}
        & \wp{\ProcCall{P}{\vec{\ee}}{\vec{y}},\varenv,\procenv}{\ff \sepcon \fg} \\
        \eeqtag{Table~\ref{table:rec-wp}}
        & \lambda(\store,\hh) \mydot \procenvP{P}{\evv{\vec{e}}}{\varenv(\vec{y})}(\ff \sepcon \fg)(\store,\hh) \\
        \ssucceqtag{Assumption from above}
        & \lambda(\store,\hh) \mydot \left(\procenvP{P}{\evv{\vec{\ee}}}{\varenv(\vec{y})}(\ff) \sepcon \fg\right)(\store,\hh) \\
        \eeqtag{algebra}
        & \left(\lambda(\store,\hh) \mydot \procenvP{P}{\evv{\vec{\ee}}}{\store(\vec{y})}(\ff)(\store,\hh)\right) \sepcon \fg \\
        \eeqtag{Table~\ref{table:rec-wp}}
        & \wp{\ProcCall{P}{\vec{\ee}}{\vec{y}},\varenv,\procenv}{\ff} \sepcon \fg.
\end{align}
\noindent \textbf{Step~\ref{app:recursion:lifting:2}.} 
Now, let $\ProcDeclIs{P}{\vec{\ee}}{\vec{y}}{\cc}$ be a procedure declaration and $\ProcPsi{\varenv}{\procenv}$ be the corresponding functional such that
\begin{align}
        \DeclP{\ProcDeclIs{P}{\vec{\ee}}{\vec{y}}{\cc}}{\varenv}{\procenv} 
        \eeq \underbrace{\procenv\subst{P}{\lfp f\mydot \ProcPsi{\varenv}{\procenv}(f)}}_{\eeq \procenv'}.
\end{align}
Moreover, assume that for all procedure names $\ProcName{Q}$ except $\ProcName{P}$ that occur in $\cc$ and parameters $\vec{z}$ that 
\begin{align}
        \procenvP{Q}{\vec{z}}{\vec{z_2}}(\ff) \sepcon \fg \ppreceq \procenvP{Q}{\vec{z}}{\vec{z_2}}(\ff \sepcon \fg).
\end{align}
Our goal is to show that for every procedure $\ProcName{Q}$ that occurs in $\cc$ (including $\ProcName{P}$) that (for all suitable parameters $\vec{z}$)
\begin{align}
        \procenvP{Q}{\vec{z}}{\vec{z_2}}(\ff) \sepcon \fg \ppreceq \procenvP{Q}{\vec{z}}{\vec{z_2}}(\ff \sepcon \fg). \label{eq:recursive:frame:2}
\end{align}
Since $\procenv' = \procenv\subst{P}{\lfp f \mydot \ProcPsi{\varenv}{\procenv}(f)}$ we only have to show that
\begin{align}
    \procenvP{P}{\vec{z}}{\vec{z_2}}(\ff \sepcon \fg) 
    \eeq (\lfp f\mydot \ProcPsi{\varenv}{\procenv}(f))(\vec{z})(\ff \sepcon \fg) 
    \ssucceq (\lfp f\mydot \ProcPsi{\varenv}{\procenv}(f))(\vec{z})(\ff) \sepcon \fg.
\end{align}
Furthermore, by the Tarski-Knaster fixed point theorem, we have 
\begin{align}
        \lfp f\mydot \ProcPsi{\varenv}{\procenv}(f) = \sup_{\oa \in \Ord} \ProcPsi{\varenv}{\procenv}^{\oa}(0).
\end{align}
It thus suffices to show that for all ordinals $\oa$ and all suitable parameters $\vec{z}$ that
\begin{align}
        \ProcPsi{\varenv}{\procenv}^{\alpha}(0)(\vec{z})(\ff \sepcon \fg)  \ssucceq
        \ProcPsi{\varenv}{\procenv}^{\alpha}(0)(\vec{z})(\ff) \sepcon \fg.
\end{align}
Since parameters $\vec{z}$ never change in the proof below, let us write 
$\ProcPsi{\varenv}{\procenv}(\ff)$ instead of the more convoluted $\ProcPsi{\varenv}{\procenv}(\vec{z})(\ff)$.
We then proceed by transfinite induction on $\oa$.

\emph{The case $\oa = 0$.}
\begin{align}
        & \ProcPsi{\varenv}{\procenv}^{0}(0)(\ff \sepcon \fg) \\
        \eeqtag{$\ProcPsi{\varenv}{\procenv}^{0}(0) = 0$}
        & 0 \\
        \eeqtag{algebra}
        & 0 \sepcon \fg \\
        \eeqtag{$\ProcPsi{\varenv}{\procenv}^{0}(0) = 0$}
        & \ProcPsi{\varenv}{\procenv}^{0}(0)(\ff) \sepcon \fg.
\end{align}
\emph{The case $\oa$ successor ordinal.}
\begin{align}
        & \ProcPsi{\varenv}{\procenv}^{\oa+1}(0)(\ff \sepcon \fg) \\
        \eeqtag{$\ProcPsi{\varenv}{\procenv}^{\oa+1}(0) = \ProcPsi{\varenv}{\procenv}(\ProcPsi{\varenv}{\procenv}^{\oa}(0))$}
        & \ProcPsi{\varenv}{\procenv}\left(\ProcPsi{\varenv}{\procenv}^{\oa}(0)\right)(\ff \sepcon \fg) \\
        \eeqtag{Definition of $\ProcPsi{\varenv}{\procenv}$}
        & \lambda(\store,\hh)\mydot \wpC{\cc,\varenv',\procenv\subst{\ProcName{P}}{\ProcPsi{\varenv}{\procenv}^{\oa}(0)}}(\ff \sepcon \fg)(\store',\hh) \\
        \ssucceqtag{equation~\ref{eq:recursive:frame:1} (premise is satisfied by I.H.)}
        & \lambda(\store,\hh)\mydot \left(\wpC{\cc,\varenv',\procenv\subst{\ProcName{P}}{\ProcPsi{\varenv}{\procenv}^{\oa}(0)}}(\ff) \sepcon \fg \right)(\store',\hh) \\
        \eeqtag{algebra}
        & \lambda(\store,\hh)\mydot \left(
             \left(\lambda(\store_1,\hh_1)\mydot \wpC{\cc,\varenv',\procenv\subst{\ProcName{P}}{\ProcPsi{\varenv}{\procenv}^{\oa}(0)}}(\ff)(\store_1,\hh_1) \right)
             \sepcon \left(\lambda(\store_2,\hh_2)\mydot \fg(\store_2,\hh_2) \right)
        \right)(\store',\hh) \\
        \eeqtag{algebra}
        &\left(\lambda(\store,\hh_1)\mydot \wpC{\cc,\varenv',\procenv\subst{\ProcName{P}}{\ProcPsi{\varenv}{\procenv}^{\oa}(0)}}(\ff)(\store',\hh_1) \right)
        \sepcon \left(\lambda(\store,\hh_2)\mydot \fg(\store',\hh_2) \right) \\
        \eeqtag{$\fg \in \Eadm$. By Definition~\ref{def:admissible-exp} and equation~\ref{eq:recursion:psi} this means $\fg(\store',\hh_2) = \fg(\store,\hh_2)$}
        &\left(\lambda(\store,\hh_1)\mydot \wpC{\cc,\varenv',\procenv\subst{\ProcName{P}}{\ProcPsi{\varenv}{\procenv}^{\oa}(0)}}(\ff)(\store',\hh_1) \right)
        \sepcon \left(\lambda(\store,\hh_2)\mydot \fg(\store,\hh_2) \right) \\
        \eeqtag{algebra}
        &\left(\lambda(\store,\hh_1)\mydot \wpC{\cc,\varenv',\procenv\subst{\ProcName{P}}{\ProcPsi{\varenv}{\procenv}^{\oa}(0)}}(\ff)(\store',\hh_1) \right)
        \sepcon \fg \\
        \eeqtag{Definition of $\ProcPsi{\varenv}{\procenv}$}
        & \ProcPsi{\varenv}{\procenv}\left(\ProcPsi{\varenv}{\procenv}^{\oa}(0)\right)(\ff) \sepcon \fg \\
        \eeqtag{$\ProcPsi{\varenv}{\procenv}^{\oa+1}(0) = \ProcPsi{\varenv}{\procenv}(\ProcPsi{\varenv}{\procenv}^{\oa}(0))$}
        & \ProcPsi{\varenv}{\procenv}^{\oa+1}(0)(\ff) \sepcon \fg.
\end{align}

\emph{The case $\oa$ limit ordinal.}
\begin{align}
        & \ProcPsi{\varenv}{\procenv}^{\oa}(0)(\ff \sepcon \fg) \\
        \eeqtag{Definition of $\ProcPsi{\varenv}{\procenv}^{\oa}(0)$ for $\oa$ limit ordinal}
        & \sup_{\ob < \oa} \ProcPsi{\varenv}{\procenv}^{\ob}(0)(\ff \sepcon \fg) \\
        \ssucceqtag{I.H.}
        & \sup_{\ob < \oa} \left(\ProcPsi{\varenv}{\procenv}^{\ob}(0)(\ff) \sepcon \fg\right) \\
        \eeqtag{by definition of $\sepcon$ and algebra (analogously to proof of Theorem~\ref{thm:frame-rules})}
        & \left( \sup_{\ob < \oa} \ProcPsi{\varenv}{\procenv}^{\ob}(0)(\ff) \right) \sepcon \fg \\
        \eeqtag{Definition of $\ProcPsi{\varenv}{\procenv}^{\oa}(0)$ for $\oa$ limit ordinal}
        & \ProcPsi{\varenv}{\procenv}^{\oa}(0)(\ff) \sepcon \fg.
\end{align}
\noindent \textbf{Step~\ref{app:recursion:lifting:3}.}
We are now in a position to prove Theorem~\ref{thm:rhpgcl-frame-rules} , i.e. for all $d \in \rhpgcl$, $\varenv \in \VarEnv$, $\procenv \in \ProcEnv$, and $\ff,\fg \in \Eadm$ with
$\Mod{d} \cap \Vars_{\varenv}(\fg) = \emptyset$, we have
\begin{align}
        \wp{d,\varenv,\procenv}{\ff \sepcon \fg} \ssucceq
        \wp{d,\varenv,\procenv}{\ff} \sepcon \fg.
\end{align}

We proceed by induction on the structure of \rhpgcl programs.
More precisely, we show that the claim holds if the initial procedure environment $\procenv$ satisfies equation~\ref{eq:recursive:frame:2} for all procedure calls in program $d \in \rhpgcl$ that have not been declared in $d$.
Since there are no procedure calls without a preceding declaration in a complete \rhpgcl program, this implies the claim.

The base case, i.e. $d = \cc$, is covered by our first step, see equation~\ref{eq:recursive:frame:1}.

Otherwise, if $d = \COMPOSE{\ProcDeclIs{P}{\vec{x}}{\vec{y}}{\cc}}{d'}$, we have
\begin{align}
        & \wp{\COMPOSE{\ProcDeclIs{P}{\vec{x}}{\vec{y}}{\cc}}{d'}, \varenv, \procenv}{\ff \sepcon \fg} \\
        \eeqtag{Table~\ref{table:rec-wp}}
        & \wp{d', \varenv, \procenv'}{\ff \sepcon \fg} \\
        \ssucceqtag{by step 2, see equation~\ref{eq:recursive:frame:2}, we may apply the I.H.}
        & \wp{d', \varenv, \procenv'}{\ff} \sepcon \fg \\
        \eeqtag{by definition of $\procenv'$, see Table~\ref{table:rec-wp}}
        & \wp{\COMPOSE{\ProcDeclIs{P}{\vec{x}}{\vec{y}}{\cc}}{d'}, \varenv, \procenv}{\ff} \sepcon \fg.
\end{align}
Hence, the quantitative frame rule also holds in the presence of recursion.
\end{proof}

\subsection{Incorporating Random Number Generators}\label{app:extensions:random-numbers}

Technically, the statement $\ASSIGNUNIFORM{x}{\ee}{\ee'}$ is syntactic sugar for \hpgcl, because we can write a program with the same behavior.
Intuitively, such a program first generates $\ee' - \ee$ many random bits by flipping coins in a loop. 
The program then checks whether exactly one bit is set to one.
If yes, then the result is the number between $\ee$ and $\ee'$ at that position.
Otherwise, we perform rejection sampling and start all over again.
A corresponding \hpgcl program is found below.

\begin{lstlisting}[escapeinside={(*}{*)}, showstringspaces=false]
r := -1; // stores the final result
l := e' - e; // length
while(r == -1) { // rejection sampling
    x := 0; // stores randomly generated bits
    i := 0;
    while(i < l) { // generate l random bits in x  
        i := i+1;
        x := 2*x;
        { skip } [0.5] { x:= x+1 }
    }
    y := 1;
    j := 0;
    while(y < x && j < l) { // check whether x is a power of two
        j := j +1;
        y := 2 * y;
    }
    if(y == x) {
         r := j; // position found, terminate
    }
}
r := e + j; // fetch actual value
\end{lstlisting}

In particular, notice that $\sk(e') < \sk(e)$ implies that the above program does not terminate, i.e. the weakest preexpectation will be $0$.
Analogously, our $\wp{\ASSIGNUNIFORM{x}{\ee}{\ee'}}{\iverson{\ee' < \ee} \cdot \ff} = 0$ according to our direct definition presented in the paper.

Since $\ASSIGNUNIFORM{x}{\ee}{\ee'}$ is syntactic sugar, all results shown for \hpgcl transfer automatically. 
However, since we did not explicitly show correctness of the above program, let us briefly check that Theorems~\ref{thm:wp:basic},~\ref{thm:batz}, and~\ref{thm:frame-rules} also hold for $\ASSIGNUNIFORM{x}{\ee}{\ee'}$.

\paragraph{Correctness of Theorem~\ref{thm:wp:basic}}
For \emph{linearity}, i.e. Theorem~\ref{thm:wp:basic}.\ref{thm:wp:basic:super-linearity},\ref{thm:wp:basic:linearity}, consider the following:
\begin{align}
        & \wp{\ASSIGNUNIFORM{x}{\ee}{\ee'}}{k \cdot \ff + \fg} \\
        \eeqtag{Definition of $\wpsymbol$}
        & \lambda(\sk,\hh)\mydot \frac{1}{\sk(\ee') - \sk(\ee)} \cdot \sum_{\ell = \sk(\ee)}^{\sk(\ee')} \left(k \cdot \ff + \fg\right)\subst{x}{\ell}(\sk,\hh) \\
        \eeqtag{algebra}
        & k \cdot \lambda(\sk,\hh)\mydot \frac{1}{\sk(\ee') - \sk(\ee)} \cdot \left( 
        \sum_{\ell = \sk(\ee)}^{\sk(\ee')} \ff\subst{x}{\ell}(\sk,\hh)
        + \sum_{\ell = \sk(\ee)}^{\sk(\ee')} \fg\subst{x}{\ell}(\sk,\hh) \right) \\
        \eeqtag{algebra}
        & k \cdot \lambda(\sk,\hh)\mydot \frac{1}{\sk(\ee') - \sk(\ee)} \cdot \sum_{\ell = \sk(\ee)}^{\sk(\ee')} \ff\subst{x}{\ell}(\sk,\hh) \\
        & \qquad + \lambda(\sk,\hh)\mydot \frac{1}{\sk(\ee') - \sk(\ee)} \cdot \sum_{\ell = \sk(\ee)}^{\sk(\ee')} \fg\subst{x}{\ell}(\sk,\hh) \notag \\
        \eeqtag{Definition of $\wpsymbol$}
        & k \cdot \wp{\ASSIGNUNIFORM{x}{\ee}{\ee'}}{\ff} + \wp{\ASSIGNUNIFORM{x}{\ee}{\ee'}}{\fg}.
\end{align}
For \emph{monotonicity}, i.e.\ Theorem~\ref{thm:wp:basic}.\ref{thm:wp:basic:monotonicity}, let $\ff, \fg \in \E$ with $\ff \preceq \fg$. We have
\begin{align}
        & \wp{\ASSIGNUNIFORM{x}{\ee}{\ee'}}{\ff} \\
        \eeqtag{Definition of $\wpsymbol$}
        & \lambda(\sk,\hh)\mydot \frac{1}{\sk(\ee') - \sk(\ee)} \cdot \sum_{\ell = \sk(\ee)}^{\sk(\ee')} \ff \subst{x}{\ell}(\sk,\hh) \\
        \lleqtag{By assumption: $\ff(s,h) < \fg (s,h)$ for all $(s,h) \in \States$}
        & \lambda(\sk,\hh)\mydot \frac{1}{\sk(\ee') - \sk(\ee)} \cdot \sum_{\ell = \sk(\ee)}^{\sk(\ee')} \fg \subst{x}{\ell}(\sk,\hh) \\
        \eeqtag{Definition of $\wpsymbol$}
        & \wp{\ASSIGNUNIFORM{x}{\ee}{\ee'}}{\fg}.
\end{align}
For \emph{continuity}, i.e.\ Theorem~\ref{thm:wp:basic}.\ref{thm:wp:basic:continuity}, let $\ff_1 \preceq \ff_2 \ppreceq \ldots$ be an increasing $\omega$-chain in $\E$. 
The proof relies on Lebesgue's Monotone Convergence Theorem (LMCT); see e.g.\ \cite[p.~567]{schechter:1996}.
\begin{align}   
   & \wp{\ASSIGNUNIFORM{x}{\ee}{\ee'}}{\sup_n \ff_n} \\
   \eeqtag{Definition of $\wpsymbol$}
   & \lambda(\sk,\hh)\mydot \frac{1}{\sk(\ee') - \sk(\ee)} \cdot \sum_{\ell = \sk(\ee)}^{\sk(\ee')} (\sup_n \ff_n) \subst{x}{\ell}(\sk,\hh) \\
   \eeqtag{Definition of substitution and supremum over $\E$}
   & \lambda(\sk,\hh)\mydot \frac{1}{\sk(\ee') - \sk(\ee)} \cdot \sum_{\ell = \sk(\ee)}^{\sk(\ee')} \sup_n \ff_n(\sk\subst{x}{\ell},\hh) \\
   \eeqtag{LMCT}
   & \lambda(\sk,\hh)\mydot \frac{1}{\sk(\ee') - \sk(\ee)} \cdot \sup_n  \sum_{\ell = \sk(\ee)}^{\sk(\ee')} \ff_n(\sk\subst{x}{\ell},\hh) \\
   \eeqtag{Algebra}
   & \lambda(\sk,\hh)\mydot \sup_n \frac{1}{\sk(\ee') - \sk(\ee)} \cdot   \sum_{\ell = \sk(\ee)}^{\sk(\ee')} \ff_n(\sk\subst{x}{\ell},\hh) \\
   \eeqtag{Definition of supremum over $\E$}
   &\sup_n \lambda(\sk,\hh)\mydot  \frac{1}{\sk(\ee') - \sk(\ee)} \cdot   \sum_{\ell = \sk(\ee)}^{\sk(\ee')} \ff_n(\sk\subst{x}{\ell},\hh) \\
   \eeqtag{Definition of substitution}
   &\sup_n \lambda(\sk,\hh)\mydot  \frac{1}{\sk(\ee') - \sk(\ee)} \cdot   \sum_{\ell = \sk(\ee)}^{\sk(\ee')} \ff_n\subst{x}{\ell} (\sk, \hh) \\
   \eeqtag{Definition of $\wpsymbol$}
   &\sup_n \wp{\ASSIGNUNIFORM{x}{\ee}{\ee'}}{ \ff_n}.
\end{align}
%
%
\paragraph{Correctness of Theorem~\ref{thm:batz}}
Let $\ff,\fg \in \E$ such that $\fg$ is pure and $\Vars (\fg) \cap \{x\} = \emptyset$. Then
\begin{align}
   & \wp{\ASSIGNUNIFORM{x}{\ee}{\ee'}}{\ff \cdot \fg} \\
   \eeqtag{Definition of $\wpsymbol$}
   & \lambda(\sk,\hh)\mydot \frac{1}{\sk(\ee') - \sk(\ee)} \cdot \sum_{\ell = \sk(\ee)}^{\sk(\ee')} (\ff \cdot \fg) \subst{x}{\ell}(\sk,\hh) \\
   \eeqtag{By assumption: $x$ does not occur in $\fg$}
   & \lambda(\sk,\hh)\mydot \frac{1}{\sk(\ee') - \sk(\ee)} \cdot \sum_{\ell = \sk(\ee)}^{\sk(\ee')} \ff \subst{x}{\ell}(\sk,\hh) \cdot \fg(\sk, \hh)  \\
   \eeqtag{$\fg(\sk,\hh)$ does not depend on $k$}
   & \lambda(\sk,\hh)\mydot \frac{1}{\sk(\ee') - \sk(\ee)} \cdot \fg(\sk, \hh) \cdot \sum_{\ell = \sk(\ee)}^{\sk(\ee')} \ff \subst{x}{\ell}(\sk,\hh)  \\
   \eeqtag{Algebra}
   &\fg \cdot  \lambda(\sk,\hh)\mydot \frac{1}{\sk(\ee') - \sk(\ee)} \cdot \sum_{\ell = \sk(\ee)}^{\sk(\ee')} \ff \subst{x}{\ell}(\sk,\hh)  \\
   \eeqtag{Definition of $\wpsymbol$}
   &\fg \cdot \wp{\ASSIGNUNIFORM{x}{\ee}{\ee'}}{\ff}.
\end{align}
%
%
\paragraph{Correctness of Theorem~\ref{thm:frame-rules}}
Let $\ff, \fg \in \E$ with $\{ x \} \cap \Vars(\fg) = \emptyset$. We have
\begin{align}
   & \wp{\ASSIGNUNIFORM{x}{\ee}{\ee'}}{\ff \sepcon\fg} \\
   \eeqtag{Definition of $\wpsymbol$}
   & \lambda(\sk,\hh)\mydot \frac{1}{\sk(\ee') - \sk(\ee)} \cdot \sum_{\ell = \sk(\ee)}^{\sk(\ee')} (\ff \sepcon \fg) \subst{x}{\ell}(\sk,\hh) \\
   \eeqtag{$x$ does not occur in $\fg$}
   & \lambda(\sk,\hh)\mydot \frac{1}{\sk(\ee') - \sk(\ee)} \cdot \sum_{\ell = \sk(\ee)}^{\sk(\ee')} (\ff\subst{x}{\ell} \sepcon \fg) (\sk,\hh) \\
   \eeqtag{Algebra}
   & \lambda(\sk,\hh)\mydot \frac{1}{\sk(\ee') - \sk(\ee)} \cdot \big( \sum_{\ell = \sk(\ee)}^{\sk(\ee')} (\ff\subst{x}{\ell} \sepcon \fg) \big) (\sk,\hh) \\
   \ssucceqtag{Subdistributivity of $\sepcon$ over $+$ (Theorem \ref{thm:sep-con-distrib}.\ref{thm:sep-con-distrib:sepcon-over-plus})}
   & \lambda(\sk,\hh)\mydot \frac{1}{\sk(\ee') - \sk(\ee)} \cdot \big( \fg \sepcon \sum_{\ell = \sk(\ee)}^{\sk(\ee')} \ff\subst{x}{\ell} \big) (\sk,\hh) \\
   \eeqtag{Algebra}
   &\fg \sepcon \lambda(\sk,\hh)\mydot \frac{1}{\sk(\ee') - \sk(\ee)} \cdot \sum_{\ell = \sk(\ee)}^{\sk(\ee')} \ff\subst{x}{\ell}  (\sk,\hh) \\
   \eeqtag{Definition of $\wpsymbol$}
   &\fg \sepcon \wp{\ASSIGNUNIFORM{x}{\ee}{\ee'}}{\ff}.
\end{align}

\newpage
\section{Appendix to Section~\ref{sec:case-studies} (Case Studies)}
\label{app:sec:case-studies}
\subsection{Verification of Invariant for Lossy List Reversal}\label{app:case-studies:lossy-reversal}

Recall the invariant proposed in the paper:
\begin{align*}
        \inv \eeq \Len{\lsRev}{0} \sepcon \Ls{\lsHead}{0} + \sfrac{1}{2} \cdot \iverson{\lsHead \neq 0} \cdot \left( \Len{\lsHead}{0} \sepcon \Ls{\lsRev}{0} \right).
\end{align*}
Moreover, let $\cc$ be the loop body in procedure $\csPLossyReversal$. We then have to show that
\begin{align}
        \charwp{\lsHead \neq 0}{\cc}{\Len{\lsRev}{0}}\left(\inv\right) \eeq \iverson{\lsHead \neq 0} \cdot \wp{\cc}{\inv} + \iverson{\lsHead = 0} \cdot \Len{\lsRev}{0} \ppreceq \inv
\end{align}
in order to prove that $\inv$ is indeed an upper invariant, i.e. $\wp{\WHILEDO{\lsHead \neq 0}{\cc}}{\Len{\lsRev}{0}} \preceq \inv$.

\paragraph{Weakest preexpectation of loop body}

We first consider $\wp{\cc}{\ff}$ for an arbitrary expectation $\ff \in \E$:
\begin{align}
        & \wp{\cc}{\ff} \\
        \eeqtag{Let $\cc = \COMPOSE{\cc_1}{\ASSIGN{\lsHead}{\lsTmp}}$, apply Table~\ref{table:wp}}
        &\wp{\cc_1}{\ff\subst{\lsHead}{\lsTmp}} \\
        \eeqtag{Let $\cc_1 = \COMPOSE{\cc_2}{\PCHOICE{\cc_3}{\sfrac{1}{2}}{\FREE{\lsHead}}}$, apply Table~\ref{table:wp}}
        & \wp{\cc_2}{\sfrac{1}{2} \cdot \wp{\cc_3}{\ff\subst{\lsHead}{\lsTmp}} + \sfrac{1}{2} \cdot \wp{\FREE{\lsHead}}{\ff\subst{\lsHead}{\lsTmp}}} \\
        \eeqtag{Table~\ref{table:wp}}
        & \wp{\cc_2}{\sfrac{1}{2} \cdot \wp{\cc_3}{\ff\subst{\lsHead}{\lsTmp}} + \sfrac{1}{2} \cdot \validpointer{\lsHead} \sepcon \ff\subst{\lsHead}{\lsTmp}} \\
        \eeqtag{Let $\cc_3 = \COMPOSE{\HASSIGN{\lsHead}{\lsRev}}{\ASSIGN{\lsRev}{\lsHead}}$, apply Table~\ref{table:wp}}
        & \wp{\cc_2}{\sfrac{1}{2} \cdot \wp{\HASSIGN{\lsHead}{\lsRev}}{\ff\subst{\lsHead}{\lsTmp}\subst{\lsRev}{\lsHead}} + \sfrac{1}{2} \cdot \validpointer{\lsHead} \sepcon \ff\subst{\lsHead}{\lsTmp}} \\
        \eeqtag{Table~\ref{table:wp}}
        & \wp{\cc_2}{\sfrac{1}{2} \cdot \validpointer{\lsHead} \sepcon \left( \singleton{\lsHead}{\lsRev} \sepimp \ff\subst{\lsHead}{\lsTmp}\subst{\lsRev}{\lsHead} \right) + \sfrac{1}{2} \cdot \validpointer{\lsHead} \sepcon \ff\subst{\lsHead}{\lsTmp}} \\
        \eeqtag{$\cc_2 = \ASSIGNH{\lsTmp}{\lsHead}$, apply Table~\ref{table:wp}}
        & \sup_{v \in \Ints} \singleton{\lsHead}{v} \sepcon \big(\singleton{\lsHead}{v} \sepimp \\
        & \qquad \sfrac{1}{2} \cdot \validpointer{\lsHead} \sepcon \left( \singleton{\lsHead}{\lsRev} \sepimp \ff\subst{\lsHead}{\lsTmp}\subst{\lsRev}{\lsHead}\subst{\lsTmp}{v} \right) + \sfrac{1}{2} \cdot \validpointer{\lsHead} \sepcon \ff\subst{\lsHead}{\lsTmp}\subst{\lsTmp}{v} \big) \notag \\
        \eeqtag{Lemma~\ref{lem:wand-reynolds}}
        & \sup_{v \in \Ints} \containsPointer{\lsHead}{v} \cdot \big( \sfrac{1}{2} \cdot \validpointer{\lsHead} \sepcon \left( \singleton{\lsHead}{\lsRev} \sepimp \ff\subst{\lsHead}{\lsTmp}\subst{\lsRev}{\lsHead}\subst{\lsTmp}{v} \right) \\
        & \qquad + \sfrac{1}{2} \cdot \validpointer{\lsHead} \sepcon \ff\subst{\lsHead}{\lsTmp}\subst{\lsTmp}{v} \big) \notag \\
        \eeqtag{$\containsPointer{\lsHead}{v} \cdot \validpointer{\lsHead} = \singleton{\lsHead}{v}$}
        & \sup_{v \in \Ints} \sfrac{1}{2} \cdot \singleton{\lsHead}{v} \sepcon \left( \singleton{\lsHead}{\lsRev} \sepimp \ff\subst{\lsHead}{\lsTmp}\subst{\lsRev}{\lsHead}\subst{\lsTmp}{v} \right) + \sfrac{1}{2} \cdot \singleton{\lsHead}{v} \sepcon \ff\subst{\lsHead}{\lsTmp}\subst{\lsTmp}{v} \\
\end{align}

\paragraph{Invariant verification}
\begin{align}
        & \charwp{\lsHead \neq 0}{\cc}{\Len{\lsRev}{0}}\left(\inv\right) \\
        \eeqtag{Definition of $\charwp{\lsHead\neq0}{\cc}{\Len{\lsRev}{0}}$}
        & \iverson{\lsHead \neq 0} \cdot \wp{\cc}{\inv} + \iverson{\lsHead = 0} \cdot \Len{\lsRev}{0} \\
        \eeqtag{by above computation}
        & \iverson{\lsHead \neq 0} \cdot \sup_{v \in \Ints} \big( \sfrac{1}{2} \cdot \singleton{\lsHead}{v} \sepcon \left( \singleton{\lsHead}{\lsRev} \sepimp \inv\subst{\lsHead}{\lsTmp}\subst{\lsRev}{\lsHead}\subst{\lsTmp}{v} \right) \\
        & \qquad + \sfrac{1}{2} \cdot \singleton{\lsHead}{v} \sepcon \inv\subst{\lsHead}{\lsTmp}\subst{\lsTmp}{v} \notag \\
        & \big) + \iverson{\lsHead = 0} \cdot \Len{\lsRev}{0} \notag \\
        \eeqtag{Definition of $\inv$}
        & \iverson{\lsHead \neq 0} \cdot \sup_{v \in \Ints} \big( \sfrac{1}{2} \cdot \singleton{\lsHead}{v} \sepcon \big( \singleton{\lsHead}{\lsRev} \label{eq:lossy-rev:before-entailment} \\
        & \qquad \quad \sepimp \left( \Len{\lsHead}{0} \sepcon \Ls{v}{0} + \iverson{v \neq 0} \cdot \sfrac{1}{2} \cdot (\Len{v}{0} \sepcon \Ls{\lsHead}{0}) \right) \notag \\
        & \qquad + \sfrac{1}{2} \cdot \singleton{\lsHead}{v} \sepcon \left( \Len{\lsRev}{0} \sepcon \Ls{v}{0} + \iverson{v \neq 0} \cdot \sfrac{1}{2} \cdot (\Len{v}{0} \sepcon \Ls{\lsRev}{0}) \right) \notag \\ 
        & \big) + \iverson{\lsHead = 0} \cdot \Len{\lsRev}{0} \notag \\
\end{align}
It then remains to prove that expectation $\ff$ in equation~\ref{eq:lossy-rev:before-entailment} entails our invariant $\inv$, i.e. $\ff \preceq \inv$.
To this end, we proceed as follows:
\begin{align}
        \eeqtag{continuing from equation~\ref{eq:lossy-rev:before-entailment}, Theorem~\ref{thm:sep-con-distrib}.\ref{thm:sep-con-distrib:sepcon-over-plus-full}}
        & \iverson{\lsHead \neq 0} \cdot \sup_{v \in \Ints} \big( \sfrac{1}{2} \cdot \singleton{\lsHead}{v} \sepcon \big( \singleton{\lsHead}{\lsRev} \\
        & \qquad \quad \sepimp \left( \Len{\lsHead}{0} \sepcon \Ls{v}{0} + \iverson{v \neq 0} \cdot \sfrac{1}{2} \cdot (\Len{v}{0} \sepcon \Ls{\lsHead}{0}) \right) \notag \\
        & \qquad + \sfrac{1}{2} \cdot \singleton{\lsHead}{v} \sepcon \Len{\lsRev}{0} \sepcon \Ls{v}{0} \notag \\
        & \qquad + \sfrac{1}{2} \cdot \singleton{\lsHead}{v} \sepcon (\iverson{v \neq 0} \cdot \sfrac{1}{2} \cdot (\Len{v}{0} \sepcon \Ls{\lsRev}{0})) \notag \\ 
        & \big) + \iverson{\lsHead = 0} \cdot \Len{\lsRev}{0} \notag \\
        \eeqtag{Lemma~\ref{thm:single-pointer-wand-plus}}
        & \iverson{\lsHead \neq 0} \cdot \sup_{v \in \Ints} \big( \\
        & \qquad \sfrac{1}{2} \cdot \singleton{\lsHead}{v} \sepcon \left( \singleton{\lsHead}{\lsRev} \sepimp \left( \Len{\lsHead}{0} \sepcon \Ls{v}{0} \right) \right)  \notag \\
        & \qquad + \sfrac{1}{2} \cdot \singleton{\lsHead}{v} \sepcon \big( \singleton{\lsHead}{\lsRev} \sepimp \left(\iverson{v \neq 0} \cdot \sfrac{1}{2} \cdot (\Len{v}{0} \sepcon \Ls{\lsHead}{0}) \right) \notag \\
        & \qquad + \sfrac{1}{2} \cdot \singleton{\lsHead}{v} \sepcon \Len{\lsRev}{0} \sepcon \Ls{v}{0} \notag \\
        & \qquad + \sfrac{1}{2} \cdot \singleton{\lsHead}{v} \sepcon ( \iverson{v \neq 0} \cdot \sfrac{1}{2} \cdot (\Len{v}{0} \sepcon \Ls{\lsRev}{0})) \notag \\ 
        & \big) + \iverson{\lsHead = 0} \cdot \Len{\lsRev}{0} \notag \\
        \eeqtag{algebra}
        & \sup_{v \in \Ints} \big( \\
        & \qquad 
          \underbrace{
                  \sfrac{1}{2} \cdot \iverson{\lsHead \neq 0} \cdot \singleton{\lsHead}{v} \sepcon \left( \singleton{\lsHead}{\lsRev} \sepimp \left( \Len{\lsHead}{0} \sepcon \Ls{v}{0} \right) \right)
          }_{
                  \eeq \sfrac{1}{2} \cdot \iverson{\lsHead \neq 0} \cdot (\Ls{\lsHead}{0} \sepcon (\Len{\lsRev}{0} + \Ls{\lsRev}{0}))
          }   
          \label{eq:lossy-rev:case-1} \\
        & \qquad +
          \underbrace{
                  \sfrac{1}{2} \cdot \iverson{\lsHead \neq 0} \cdot \singleton{\lsHead}{v} \sepcon \left( \singleton{\lsHead}{\lsRev} \sepimp \iverson{v \neq 0} \cdot \sfrac{1}{2} \cdot (\Len{v}{0} \sepcon \Ls{\lsHead}{0}) \right)
          }_{
                  \ppreceq \sfrac{1}{4} \cdot \iverson{\lsHead \neq 0} \cdot (\Ls{\lsRev}{0} \sepcon (\Len{\lsHead}{0} - \Ls{\lsHead}{0}))
          }
          \label{eq:lossy-rev:case-2} \\
        & \qquad + 
          \underbrace{
                  \sfrac{1}{2} \cdot \iverson{\lsHead \neq 0} \cdot \singleton{\lsHead}{v} \sepcon \Len{\lsRev}{0} \sepcon \Ls{v}{0} 
          }_{
                  \ppreceq \sfrac{1}{2} \cdot \iverson{\lsHead \neq 0} \cdot (\Len{\lsRev}{0} \sepcon \Ls{\lsHead}{0}) 
          }
          \label{eq:lossy-rev:case-3} \\
        & \qquad + 
          \underbrace{
                  \sfrac{1}{2} \cdot \iverson{\lsHead \neq 0} \cdot \singleton{\lsHead}{v} \sepcon ( \iverson{v \neq 0} \cdot \sfrac{1}{2} \cdot (\Len{v}{0} \sepcon \Ls{\lsRev}{0})) 
          }_{
                  \ppreceq \sfrac{1}{4} \cdot \iverson{\lsHead \neq 0} \cdot (\Ls{\lsRev}{0} \sepcon (\Len{\lsHead}{0} - \Ls{\lsHead}{0}))
          }
          \label{eq:lossy-rev:case-4} \\
        & \big) + \iverson{\lsHead = 0} \cdot \Len{\lsRev}{0} \notag \\
        \ppreceqtag{Each of the above properties is considered separately below.}
        & \sup_{v \in \Ints} \big( \\
        & \qquad \sfrac{1}{2} \cdot \iverson{\lsHead \neq 0} \cdot (\Ls{\lsHead}{0} \sepcon (\Len{\lsRev}{0} + \Ls{\lsRev}{0})) \notag \\
        & \qquad + \sfrac{1}{4} \cdot \iverson{\lsHead \neq 0} \cdot (\Ls{\lsRev}{0} \sepcon (\Len{\lsHead}{0} - \Ls{\lsHead}{0})) \notag \\
        & \qquad + \sfrac{1}{2} \cdot \iverson{\lsHead \neq 0} \cdot (\Len{\lsRev}{0} \sepcon \Ls{\lsHead}{0})  \notag \\
        & \qquad + \sfrac{1}{4} \cdot \iverson{\lsHead \neq 0} \cdot (\Ls{\lsRev}{0} \sepcon (\Len{\lsHead}{0} - \Ls{\lsHead}{0})) \notag \\
        & \big) + \iverson{\lsHead = 0} \cdot \Len{\lsRev}{0} \notag \\
        \eeqtag{algebra}
        & \sfrac{1}{2} \cdot \iverson{\lsHead \neq 0} \cdot (\Len{\lsRev}{0} \sepcon \Ls{\lsHead}{0}) \\
        & + \sfrac{1}{2} \cdot \iverson{\lsHead \neq 0} \cdot (\Ls{\lsHead}{0} \sepcon (\Len{\lsRev}{0} + \Ls{\lsRev}{0})) \notag \\
        & + \sfrac{1}{2} \cdot \iverson{\lsHead \neq 0} \cdot (\Ls{\lsRev}{0} \sepcon (\Len{\lsHead}{0} - \Ls{\lsHead}{0})) \notag \\
        & + \iverson{\lsHead = 0} \cdot \Len{\lsRev}{0} \notag \\
        \eeqtag{Lemma~\ref{thm:sepcon-distrib-domain-disjoint},~\ref{thm:ls:domain-disjoint}}
        & \sfrac{1}{2} \cdot \iverson{\lsHead \neq 0} \cdot (\Len{\lsRev}{0} \sepcon \Ls{\lsHead}{0}) \\
        & + \sfrac{1}{2} \cdot \iverson{\lsHead \neq 0} \cdot \left( \Ls{\lsHead}{0} \sepcon \Len{\lsRev}{0} + \Ls{\lsHead}{0} \sepcon \Ls{\lsRev}{0} \right) \notag \\
        & + \sfrac{1}{2} \cdot \iverson{\lsHead \neq 0} \cdot \left( \Ls{\lsRev}{0} \sepcon \Len{\lsHead}{0} - \Ls{\lsRev}{0} \sepcon \Ls{\lsHead}{0} \right) \notag \\
        & + \iverson{\lsHead = 0} \cdot \Len{\lsRev}{0} \notag \\
        \eeqtag{Theorem~\ref{thm:sep-con-monoid}, algebra}
        & \iverson{\lsHead \neq 0} \cdot (\Len{\lsRev}{0} \sepcon \Ls{\lsHead}{0}) \\
        & + \sfrac{1}{2} \cdot \iverson{\lsHead \neq 0} \cdot \left( \Ls{\lsRev}{0} \sepcon \Len{\lsHead}{0} \right) \notag \\
        & \underbrace{- \sfrac{1}{2} \cdot \iverson{\lsHead \neq 0} \cdot \left(  \Ls{\lsRev}{0} \sepcon \Ls{\lsHead}{0} \right) + \sfrac{1}{2} \cdot \iverson{\lsHead \neq 0} \cdot \left( \Ls{\lsHead}{0} \sepcon \Ls{\lsRev}{0} \right)}_{\eeq 0} \notag \\
        & + \iverson{\lsHead = 0} \cdot \Len{\lsRev}{0} \notag \\
        \eeqtag{algebra}
        & \iverson{\lsHead \neq 0} \cdot (\Len{\lsRev}{0} \sepcon \Ls{\lsHead}{0}) + \sfrac{1}{2} \cdot \iverson{\lsHead \neq 0} \cdot \left( \Ls{\lsRev}{0} \sepcon \Len{\lsHead}{0} \right) \\
        & + \iverson{\lsHead = 0} \cdot \Len{\lsRev}{0} \notag \\
        \eeqtag{Theorem~\ref{thm:sep-con-monoid}, algebra}
        & \iverson{\lsHead \neq 0} \cdot (\Len{\lsRev}{0} \sepcon \Ls{\lsHead}{0}) + \sfrac{1}{2} \cdot \iverson{\lsHead \neq 0} \cdot \left( \Ls{\lsRev}{0} \sepcon \Len{\lsHead}{0} \right) \\
        & + \iverson{\lsHead = 0} \cdot \left(\Len{\lsRev}{0} \sepcon \left(\iverson{\lsHead = 0} \cdot \emp\right)\right) \notag \\
        \eeqtag{by definition of $\Ls{\lsHead}{0}$, $\iverson{\lsHead = 0} \cdot \Ls{\lsHead}{0} = \iverson{\lsHead=0} \cdot \emp$}
        & \iverson{\lsHead \neq 0} \cdot (\Len{\lsRev}{0} \sepcon \Ls{\lsHead}{0}) + \sfrac{1}{2} \cdot \iverson{\lsHead \neq 0} \cdot \left( \Ls{\lsRev}{0} \sepcon \Len{\lsHead}{0} \right) \\
        & + \iverson{\lsHead = 0} \cdot (\Len{\lsRev}{0} \sepcon \Ls{\lsHead}{0}) \notag \\
        \eeqtag{algebra}
        & \Len{\lsRev}{0} \sepcon \Ls{\lsHead}{0} + \sfrac{1}{2} \cdot \iverson{\lsHead \neq 0} \cdot \left( \Ls{\lsRev}{0} \sepcon \Len{\lsHead}{0} \right) \\
        \eeqtag{Definition of $\inv$}
        & \inv.
\end{align}
To conclude the proof, we verify the relationships used in equations~\ref{eq:lossy-rev:case-1}--\ref{eq:lossy-rev:case-4}.

\paragraph{Verification of equation~\ref{eq:lossy-rev:case-1}}
\begin{align}
        & \sfrac{1}{2} \cdot \iverson{\lsHead \neq 0} \cdot \singleton{\lsHead}{v} \sepcon \left( \singleton{\lsHead}{\lsRev} \sepimp \Len{\lsHead}{0} \sepcon \Ls{v}{0} \right) \\
        \eeqtag{Lemma~\ref{thm:single-pointer-wand:sepcon}, Theorem~\ref{thm:sep-con-monoid}}
        & \sfrac{1}{2} \cdot \iverson{\lsHead \neq 0} \cdot \Ls{v}{0} \sepcon \underbrace{\singleton{\lsHead}{v} \sepcon \left( \singleton{\lsHead}{\lsRev} \sepimp \Len{\lsHead}{0} \right)}_{\eeq \singleton{\lsHead}{v} \sepcon (\iverson{\lsHead \neq 0} \cdot(\Ls{\lsRev}{0}+\Len{\lsRev}{0}))}  \\
        \eeqtag{Lemma~\ref{thm:list-length:sepimp-simple}}
        & \sfrac{1}{2} \cdot \underbrace{\iverson{\lsHead \neq 0} \cdot \Ls{v}{0} \sepcon \singleton{\lsHead}{v}}_{\preceq \Ls{\lsHead}{0}} \sepcon (\iverson{\lsHead \neq 0} \cdot(\Ls{\lsRev}{0}+\Len{\lsRev}{0}))  \\
        \ppreceqtag{Definition of $\Ls{\lsHead}{0}$}
        & \sfrac{1}{2} \cdot \Ls{\lsHead}{0} \sepcon (\iverson{\lsHead \neq 0} \cdot(\Ls{\lsRev}{0}+\Len{\lsRev}{0}))  \\
        \eeqtag{algebra, Theorem~\ref{thm:sep-con-algebra-pure}}
        & \sfrac{1}{2} \cdot \iverson{\lsHead \neq 0} \cdot (\Ls{\lsHead}{0} \sepcon (\Len{\lsRev}{0} + \Ls{\lsRev}{0})).
\end{align}

\paragraph{Verification of equation~\ref{eq:lossy-rev:case-2}}
\begin{align}
        & \sfrac{1}{2} \cdot \iverson{\lsHead \neq 0} \cdot \singleton{\lsHead}{v} \sepcon \left( \singleton{\lsHead}{\lsRev} \sepimp \iverson{v \neq 0} \cdot \sfrac{1}{2} \cdot (\Len{v}{0} \sepcon \Ls{\lsHead}{0}) \right) \\
        \eeqtag{Theorem~\ref{thm:sep-con-algebra-pure}, Lemma~\ref{thm:single-pointer-wand:pure}}
        & \sfrac{1}{2} \cdot \iverson{\lsHead \neq 0} \cdot \singleton{\lsHead}{v} \sepcon \left( \singleton{\lsHead}{\lsRev} \sepimp ((\iverson{v \neq 0} \cdot \sfrac{1}{2} \cdot \Len{v}{0}) \sepcon (\iverson{\lsHead \neq 0} \cdot \Ls{\lsHead}{0})) \right) \\
        \eeqtag{Lemma~\ref{thm:single-pointer-wand:sepcon}, Theorem~\ref{thm:sep-con-monoid}}
        & \sfrac{1}{2} \cdot \iverson{\lsHead \neq 0} \cdot \singleton{\lsHead}{v} \sepcon (\iverson{v \neq 0} \cdot \sfrac{1}{2} \cdot \Len{v}{0}) 
          \sepcon \left( \singleton{\lsHead}{\lsRev} \sepimp \iverson{\lsHead \neq 0} \cdot \Ls{\lsHead}{0} \right) \\
        \eeqtag{algebra}
        & \sfrac{1}{4} \cdot \iverson{\lsHead \neq 0} \cdot \iverson{v \neq 0} \cdot \singleton{\lsHead}{v} \sepcon \Len{v}{0} 
          \sepcon \left( \singleton{\lsHead}{\lsRev} \sepimp \iverson{\lsHead \neq 0} \cdot \Ls{\lsHead}{0} \right) \\
        \eeqtag{Definition of $\Ls{\lsHead}{0}$}
        & \sfrac{1}{4} \cdot \iverson{\lsHead \neq 0} \cdot \iverson{v \neq 0} \cdot \singleton{\lsHead}{v} \sepcon \Len{v}{0} \\
        & \qquad \sepcon \left( \singleton{\lsHead}{\lsRev} \sepimp \iverson{\lsHead \neq 0} \cdot \sup_{\za \in \Ints} \singleton{\lsHead}{\za} \sepcon \Ls{\za}{0} \right) \notag \\
        \eeqtag{Lemma~\ref{thm:misc:sepimp-contains}, algebra}
        & \sfrac{1}{4} \cdot \iverson{\lsHead \neq 0} \cdot \iverson{v \neq 0} \cdot \singleton{\lsHead}{v} \sepcon \Len{v}{0} \\
        & \qquad \sepcon \left( \singleton{\lsHead}{\lsRev} \sepimp \singleton{\lsHead}{\lsRev} \sepcon \left(\iverson{\lsHead \neq 0} \cdot \Ls{\lsRev}{0} \right)\right) \notag \\
        \eeqtag{Lemma~\ref{thm:list-length:sepimp-simple}, algebra using $\singleton{\lsHead}{v} \sepcon \ldots$}
        & \sfrac{1}{4} \cdot \iverson{\lsHead \neq 0} \cdot \iverson{v \neq 0} \cdot \underbrace{\singleton{\lsHead}{v} \sepcon \Len{v}{0}}_{\preceq \iverson{\lsHead \neq 0} \cdot (\Len{\lsHead}{0} - \Ls{\lsHead}{0})}
        \sepcon \left(\iverson{\lsHead \neq 0} \cdot \Ls{\lsRev}{0} \right) \\
        \ppreceqtag{Lemma~\ref{thm:list-length:sepcon}}
        & \sfrac{1}{4} \cdot \iverson{\lsHead \neq 0} \cdot \iverson{v \neq 0} \cdot (\iverson{\lsHead \neq 0} \cdot (\Len{\lsHead}{0} - \Ls{\lsHead}{0}))
        \sepcon \left(\iverson{\lsHead \neq 0} \cdot \Ls{\lsRev}{0} \right) \\
        \ppreceqtag{$\iverson{\lsHead \neq 0}, \iverson{v \neq 0} \preceq 1$}
        & \sfrac{1}{4} \cdot \iverson{\lsHead \neq 0} \cdot (\Len{\lsHead}{0} - \Ls{\lsHead}{0}) \sepcon \Ls{\lsRev}{0} \\
        \eeqtag{algebra}
        & \sfrac{1}{4} \cdot \iverson{\lsHead \neq 0} \cdot (\Ls{\lsRev}{0} \sepcon (\Len{\lsHead}{0} - \Ls{\lsHead}{0})).
\end{align}

\paragraph{Verification of equation~\ref{eq:lossy-rev:case-3}}
\begin{align}
        & \sfrac{1}{2} \cdot \iverson{\lsHead \neq 0} \cdot \singleton{\lsHead}{v} \sepcon \Len{\lsRev}{0} \sepcon \Ls{v}{0} \\
        \eeqtag{$\iverson{\lsHead \neq 0} = \iverson{\lsHead \neq 0} \cdot \iverson{\lsHead \neq 0}$, algebra (Theorem~\ref{thm:sep-con-monoid}, Theorem~\ref{thm:sep-con-algebra-pure})}
        & \sfrac{1}{2} \cdot \iverson{\lsHead \neq 0} \cdot \Len{\lsRev}{0} \sepcon (\iverson{\lsHead \neq 0} \cdot \singleton{\lsHead}{v} \sepcon \Ls{v}{0}) \\
        \eeqtag{$\iverson{\lsHead \neq 0} \cdot \iverson{\lsHead = 0} = 0$}
        & \sfrac{1}{2} \cdot \iverson{\lsHead \neq 0} \cdot \Len{\lsRev}{0} \sepcon (\iverson{\lsHead = 0} \cdot \emp + \iverson{\lsHead \neq 0} \cdot \singleton{\lsHead}{v} \sepcon \Ls{v}{0}) \\
        \ppreceqtag{Definition of $\Ls{\lsHead}{0}$}
        & \sfrac{1}{2} \cdot \iverson{\lsHead \neq 0} \cdot (\Len{\lsRev}{0} \sepcon \Ls{\lsHead}{0}).
\end{align}

\paragraph{Verification of equation~\ref{eq:lossy-rev:case-4}}
\begin{align}
        & \sfrac{1}{2} \cdot \iverson{\lsHead \neq 0} \cdot \singleton{\lsHead}{v} \sepcon ( \iverson{v \neq 0} \cdot \sfrac{1}{2} \cdot (\Len{v}{0} \sepcon \Ls{\lsRev}{0})) \\
        \eeqtag{Theorem~\ref{thm:sep-con-algebra-pure}, algebra}
        & \sfrac{1}{4} \cdot \underbrace{\iverson{v \neq 0}}_{\preceq 1} \cdot \iverson{\lsHead \neq 0} \cdot (\singleton{\lsHead}{v} \sepcon \Len{v}{0} \sepcon \Ls{\lsRev}{0}) \\
        \ppreceqtag{algebra}
        & \sfrac{1}{4} \cdot \iverson{\lsHead \neq 0} \cdot (\underbrace{\singleton{\lsHead}{v} \sepcon \Len{v}{0}}_{\preceq \iverson{\lsHead \neq 0} \cdot (\Len{\lsHead}{0}-\Ls{\lsHead}{0})} \sepcon \Ls{\lsRev}{0}) \\
        \ppreceqtag{Lemma~\ref{thm:list-length:sepcon}}
        & \sfrac{1}{4} \cdot \iverson{\lsHead \neq 0} \cdot ((\iverson{\lsHead \neq 0} \cdot (\Len{\lsHead}{0}-\Ls{\lsHead}{0})) \sepcon \Ls{\lsRev}{0}) \\
        \eeqtag{algebra}
        & \sfrac{1}{4} \cdot \iverson{\lsHead \neq 0} \cdot ((\Len{\lsHead}{0}-\Ls{\lsHead}{0}) \sepcon \Ls{\lsRev}{0}) \\
        \eeqtag{algebra}
        & \sfrac{1}{4} \cdot \iverson{\lsHead \neq 0} \cdot (\Ls{\lsRev}{0} \sepcon (\Len{\lsHead}{0} - \Ls{\lsHead}{0})).
\end{align}

\subsection{Probability of Successful Garbage Collection}\label{app:tree-delete}
We use of the following lemma:
\begin{lemma}\label{app:lem:tree-combine}
Let $\ff,\fg \in \E$ and $\pp \in \Rats$. Then
\begin{align*}
        \left(\ff \cdot \pp^{\heapSize}\right) \sepcon \left(\fg \cdot \pp^{\heapSize}\right) 
        \eeq 
        \pp^{\heapSize} \cdot \left(\ff \sepcon \fg\right).
\end{align*}
\end{lemma}

\begin{proof}
Let $(\sk,\hh)$ be a stack-heap pair. Then
\begin{align}
        & \left(\left(\ff \cdot \pp^{\heapSize}\right) \sepcon \left(\fg \cdot \pp^{\heapSize}\right)\right)(\sk,\hh) \\ 
        \eeqtag{Definition of $\sepcon$}
        & \max_{\hh_1,\hh_2} \left\{ \left(\ff \cdot \pp^{\heapSize}\right)(\sk,\hh_1) \cdot \left(\fg \cdot \pp^{\heapSize}\right)(\sk,\hh_2) ~|~ \hh = \hh_1 \sepcon \hh_2 \right\} \\ 
        \eeqtag{algebra}
        & \max_{\hh_1,\hh_2} \left\{ \ff(\sk,\hh_1) \cdot \pp^{\heapSize(\sk,\hh_1)} \cdot \fg(\sk,\hh_2) \cdot \pp^{\heapSize(\sk,\hh_2)} ~|~ \hh = \hh_1 \sepcon \hh_2 \right\} \\ 
        \eeqtag{algebra}
        & \max_{\hh_1,\hh_2} \left\{ \ff(\sk,\hh_1) \cdot \fg(\sk,\hh_2) \cdot \pp^{\heapSize(\sk,\hh_1)+\heapSize(\sk,\hh_2)} ~|~ \hh = \hh_1 \sepcon \hh_2 \right\} \\ 
        \eeqtag{algebra}
        & \max_{\hh_1,\hh_2} \left\{ \ff(\sk,\hh_1) \cdot \fg(\sk,\hh_2) \cdot \pp^{\heapSize(\sk,\hh)} ~|~ \hh = \hh_1 \sepcon \hh_2 \right\} \\ 
        \eeqtag{Definition of $\sepcon$}
        & \pp^{\heapSize}(\sk,\hh) \cdot \max_{\hh_1,\hh_2} \left\{ \ff(\sk,\hh_1) \cdot \fg(\sk,\hh_2) ~|~ \hh = \hh_1 \sepcon \hh_2 \right\} \\ 
        \eeqtag{algebra}
        & \pp^{\heapSize}(\sk,\hh) \cdot \left(\ff \sepcon \fg\right)(\sk,\hh) \\
        \eeqtag{algebra}
        & \left(\pp^{\heapSize} \cdot \left(\ff \sepcon \fg\right)\right)(\sk,\hh).
\end{align}

Now, recall from Figure~\ref{fig:faulty-gc}, p.~\pageref{fig:faulty-gc}, the procedure $\ProcName{delete}$. 
Additionally, we write $\textit{body}$ to refer to the procedure's body and $\textit{block}$ to refer to the program contained in the right branch of the probabilistic choice, respectively.

We are confronted with the following proof obligation:
Assuming 
\begin{align}
        \forall y \colon \wlp{\ProcCall{delete}{y}{}}{\emp} 
        \ssucceq
        \underbrace{\Tree{y} \cdot \left(1-\pp\right)^{\heapSize/2}}_{=: t(y)}.
        \label{eq:proof:delete1}
\end{align}
we have to show that 
\begin{align}
    \wlp{\textit{body}}{\emp} 
    \ssucceq 
    \underbrace{\Tree{x} \cdot \left(1-\pp\right)^{\heapSize/2}}_{= t(x)}.
\end{align}

We proceed as follows:

\begin{align}
        & \wlp{\textit{body}}{\emp} \\
        \eeqtag{Table~\ref{table:wp}} 
        & \wlp{\ITE{x \neq \nil}{\PCHOICE{\SKIP}{\pp}{\textit{block}}}{\SKIP} }{\emp} \\
        \eeqtag{Table~\ref{table:wp}} 
        & \iverson{x \neq \nil} \cdot \wlp{\PCHOICE{\SKIP}{\pp}{\textit{block}}}{\emp} + \iverson{x=0} \cdot \emp \\
        \eeqtag{Table~\ref{table:wp}} 
        & \iverson{x \neq \nil} \cdot \left( \pp \cdot \emp + (1-\pp) \cdot \wlp{\textit{block}}{\emp} \right) + \iverson{x=0} \cdot \emp \\
        \eeqtag{algebra}
        & \iverson{x \neq \nil} \cdot \left( \pp \cdot \emp \right) + \iverson{x \neq \nil} \cdot \left( (1-\pp) \cdot \wlp{\textit{block}}{\emp} \right) + \iverson{x=0} \cdot \emp \\
        \ssucceqtag{$\iverson{x \neq \nil} \geq 0$}
        & \iverson{x \neq \nil} \cdot \left( (1-\pp) \cdot \wlp{\textit{block}}{\emp} \right) \\
        \eeqtag{Let $\textit{block} = \COMPOSE{\cc_1}{\COMPOSE{\FREE{x}}{\FREE{x+1}}}$} 
        & \iverson{x \neq \nil} \cdot (1-\pp) \cdot \wlp{\cc_1\SEMI\FREE{x}\SEMI\FREE{x+1}}{\emp} + \iverson{x=0} \cdot \emp \\
        \eeqtag{Table~\ref{table:wp}} 
        & \iverson{x \neq \nil} \cdot (1-\pp) \cdot \wlp{\cc_1}{\singleton{x}{-,-} \sepcon \emp} + \iverson{x=0} \cdot \emp  \\
        \eeqtag{Let $\cc_1 = \COMPOSE{\cc_2}{\ProcCall{delete}{r}{}}$} 
        & \iverson{x \neq \nil} \cdot (1-\pp) \cdot \wlp{\cc_2\SEMI\ProcCall{delete}{r}{}}{\singleton{x}{-,-} \sepcon \emp} + \iverson{x=0} \cdot \emp  \\
        \ssucceqtag{Frame rule (Theorem~\ref{thm:frame-rules}) and (\ref{eq:proof:delete1})} 
        & \iverson{x \neq \nil} \cdot (1-\pp) \cdot \wlp{\cc_2}{\singleton{x}{-,-} \sepcon t(r)} + \iverson{x=0} \cdot \emp  \\
        \eeqtag{Theorem~\ref{thm:sep-con-monoid}.\ref{thm:sep-con-monoid:neut}, let $\cc_2 = \COMPOSE{\cc_3}{\ProcCall{delete}{l}{}}$} 
        & \iverson{x \neq \nil} \cdot (1-\pp) \cdot \wlp{\cc_3\SEMI\ProcCall{delete}{l}{}}{\singleton{x}{-,-} \sepcon t(r) \sepcon \emp} \\
        & \qquad + \iverson{x=0} \cdot \emp \notag \\
        \ssucceqtag{Frame rule (Theorem~\ref{thm:frame-rules}) and (\ref{eq:proof:delete1})} 
        & \iverson{x \neq \nil} \cdot (1-\pp) \cdot \wlp{\cc_3}{\singleton{x}{-,-} \sepcon t(r) \sepcon t(l)} + \iverson{x=0} \cdot \emp  \\
        \eeqtag{Theorem~\ref{thm:sep-con-monoid}.\ref{thm:sep-con-monoid:ass}} 
        & \iverson{x \neq \nil} \cdot (1-\pp) \cdot \wlp{\cc_3}{\singleton{x}{-,-} \sepcon (t(r) \sepcon t(l))} + \iverson{x=0} \cdot \emp  \\
        \eeqtag{Lemma~\ref{app:lem:tree-combine}} 
        & \iverson{x \neq \nil} \cdot (1-\pp) \\
        & \quad \cdot \wlp{\cc_3}{\singleton{x}{-,-} \sepcon \left((1-\pp)^{\heapSize/2} \cdot (\Tree{r} \sepcon \Tree{l})\right)} \notag \\
        & \qquad + \iverson{x=0} \cdot \emp \notag \\
        \eeqtag{$\cc_3 = \COMPOSE{\ASSIGNH{l}{x}}{\ASSIGNH{r}{x+1}}$}
        & \iverson{x \neq \nil} \cdot (1-\pp) \\
        & \quad \cdot \wlp{\ASSIGNH{l}{x}\SEMI\ASSIGNH{r}{x+1}}{\singleton{x}{-,-} \sepcon \left((1-\pp)^{\heapSize/2} \cdot (\Tree{r} \sepcon \Tree{l})\right)} \notag \\
        & \qquad + \iverson{x=0} \cdot \emp \notag \\
        \eeqtag{Table~\ref{table:wp}, Lemma~\ref{lem:wand-reynolds}} 
        & \iverson{x \neq \nil} \cdot (1-\pp) \\
        & \quad \cdot \wlp{\ASSIGNH{l}{x}}{\sup_{v\in\Ints} \containsPointer{x+1}{v} \cdot \left(\singleton{x}{-,-} \sepcon \left((1-\pp)^{\heapSize/2} \cdot (\Tree{v} \sepcon \Tree{l})\right)\right)} \notag \\
        & \qquad + \iverson{x=0} \cdot \emp \notag \\
        \eeqtag{Table~\ref{table:wp}, Lemma~\ref{lem:wand-reynolds}} 
        & \iverson{x \neq \nil} \cdot (1-\pp) \\
        & \quad \cdot \sup_{u,v\in\Ints} \containsPointer{x}{u,v} \cdot \left(\singleton{x}{-,-} \sepcon \left((1-\pp)^{\heapSize/2} \cdot (\Tree{v} \sepcon \Tree{u})\right)\right) + \iverson{x=0} \cdot \emp \notag \\
        \eeqtag{$\containsPointer{a}{b} \cdot (\validpointer{a} \sepcon \ff) = \singleton{a}{b} \sepcon \ff$} 
        & \iverson{x \neq \nil} \cdot (1-\pp) \\
        & \quad \cdot \sup_{u,v\in\Ints} \singleton{x}{u,v} \sepcon \left((1-\pp)^{\heapSize/2} \cdot (\Tree{v} \sepcon \Tree{u})\right) + \iverson{x=0} \cdot \emp \notag \\
        \eeqtag{Theorem~\ref{thm:sep-con-algebra-pure}.3}
        & \iverson{x \neq \nil} \\
        & \quad \cdot \sup_{u,v\in\Ints}\left((1-\pp) \cdot \singleton{x}{u,v}\right) \sepcon \left((1-\pp)^{\heapSize/2} \cdot (\Tree{v} \sepcon \Tree{u})\right) + \iverson{x=0} \cdot \emp \notag \\
        \eeqtag{$\singleton{x}{u,v} \cdot (1-\pp) = \singleton{x}{u,v} \cdot (1-\pp)^{\heapSize/2}$} 
        & \iverson{x \neq \nil} \\
        & \quad \cdot \sup_{u,v\in\Ints} \left((1-\pp)^{\heapSize/2} \cdot \singleton{x}{u,v}\right) \sepcon \left((1-\pp)^{\heapSize/2} \cdot (\Tree{v} \sepcon \Tree{u})\right) + \iverson{x=0} \cdot \emp \notag \\
        \eeqtag{Lemma~\ref{app:lem:tree-combine}} 
        & \iverson{x \neq \nil} \\
        & \quad \cdot \sup_{u,v\in\Ints} \left((1-\pp)^{\heapSize/2} \cdot (\singleton{x}{u,v} \sepcon \Tree{v} \sepcon \Tree{u})\right) + \iverson{x=0} \cdot \emp \notag \\
        \eeqtag{algebra}
        & \iverson{x \neq \nil} \cdot (1-\pp)^{\heapSize/2} \cdot \sup_{u,v\in\Ints} \singleton{x}{u,v} \sepcon \Tree{v} \sepcon \Tree{u}) + \iverson{x=0} \cdot \emp \\  
        \eeqtag{algebra}
        & (1-\pp)^{\heapSize/2} \cdot \left(\iverson{x \neq \nil} \cdot \sup_{u,v\in\Ints} \singleton{x}{u,v} \sepcon \Tree{v} \sepcon \Tree{u}\right) + \iverson{x=0} \cdot \emp \\  
        \eeqtag{$\singleton{x}{u,v}$ implies $\iverson{x \neq nil} = 1$}   
        & (1-\pp)^{\heapSize/2} \cdot \left(\sup_{u,v\in\Ints} \singleton{x}{u,v} \sepcon \Tree{v} \sepcon \Tree{u}\right) + \iverson{x=0} \cdot \emp \\
        \eeqtag{algebra}
        \eeq & (1-\pp)^{\heapSize/2} \cdot \left(\sup_{u,v\in\Ints} \singleton{x}{u,v} \sepcon \Tree{v} \sepcon \Tree{u}\right) + \iverson{x=0} \cdot \emp \cdot \underbrace{(1-\pp)^{\heapSize/2}}_{\text{=1 due to $\emp$}} \\
        \eeqtag{algebra}
        \eeq & (1-\pp)^{\heapSize/2} \cdot \left(\sup_{u,v\in\Ints} \singleton{x}{u,v} \sepcon \Tree{v} \sepcon \Tree{u} + \iverson{x=0} \cdot \emp \right) \\
        \eeqtag{Definition of $\Tree{x}$}
        \eeq & (1-\pp)^{\heapSize/2} \cdot \Tree{x}. 
\end{align}
\end{proof}

\subsection{Invariant Verification for Section~\ref{sec:list-length}}\label{app:list-length}
Recall the definition of our proposed invariant $\inv$:
\begin{align}
  \inv \eeq \Len{x}{\nil} + \iverson{c = 1}~.
\end{align}
To show that $\inv$ is an (upper) invariant of the loop in program $c_{\textrm{list}}$ with respect to postexpectation $\Sll{x}{\nil} \cdot \heapSize$, we have to prove that
\begin{align}
        &\charwp{\cc=1}{\textrm{loopBody}}{\inv} \\
        \eeq &\iverson{\cc \neq 1} \cdot \Len{x}{\nil} + \iverson{\cc = 1} \cdot \wp{\textrm{loopBody}}{\Len{x}{\nil} + \iverson{\cc = 1}} \\
        {}~\overset{!}{\preceq}~{}& \inv
\end{align}
where $\textrm{loopBody}$ denotes the loop body of $c_{\textrm{list}}$, i.e.
\begin{align}
    \textrm{loopBody} \eeq \PCHOICE{\ASSIGN{\cc}{0}}{0.5}{\ASSIGN{\cc}{1}\SEMI\ALLOC{x}{x}}~.
\end{align}
Since $\iverson{\cc \neq 1} \cdot \iverson{\cc = 1} = 0$, we subdivide our proof obligation into
\begin{enumerate}
   \item $\iverson{\cc \neq 1} \cdot \Len{x}{\nil} \preceq \inv$, and \label{eqn:list-length-to-show-1}
   \item $\iverson{\cc = 1} \cdot \wp{\textrm{loopBody}}{\Len{x}{\nil} + \iverson{\cc = 1}} \preceq \inv$. \label{eqn:list-length-to-show-2}
\end{enumerate}
The validity of (\ref{eqn:list-length-to-show-1}) is immediate since 
\begin{align}
   &\iverson{\cc \neq 1} \cdot \Len{x}{\nil} \\
   \ppreceqtag{$\iverson{\cc \neq 1} \preceq 1$}
   &\Len{x}{\nil} \\
   \ppreceqtag{$0 \preceq \iverson{c = 1}$}
   &\Len{x}{\nil} + \iverson{c = 1}.
\end{align}
For the validity of (\ref{eqn:list-length-to-show-2}), we first compute
\begin{align}
        & \wp{\ALLOC{x}{x}}{\Len{x}{0}} 
        \label{eqn:increase-list-length} \\
        \eeqtag{Definition of $\wpsymbol$}
        & \inf_{v \in \AVAILLOC{x}} \singleton{v}{x} \sepimp \Len{v}{0} \\
        \eeqtag{Definition of $\Lensymbol$}
        & \inf_{v \in \AVAILLOC{x}} \singleton{v}{x} \sepimp \iverson{v \neq 0} \cdot \sup_{\alpha} \singleton{v}{\alpha} \sepcon (\Ls{\alpha}{0} + \Len{\alpha}{0}) \\
        \eeqtag{Lemma~\ref{thm:misc:sepimp-contains}} 
        & \inf_{v \in \AVAILLOC{x}} \singleton{v}{x} \sepimp \containsPointer{v}{x} \cdot \iverson{v \neq 0} \cdot \sup_{\alpha} \singleton{v}{\alpha} \sepcon (\Ls{\alpha}{0} + \Len{\alpha}{0}) \\
        \eeqtag{algebra}
        & \inf_{v \in \AVAILLOC{x}} \singleton{v}{x} \sepimp \underbrace{\iverson{v \neq 0}}_{\eeq 1} \cdot \singleton{v}{x} \sepcon (\Ls{x}{0} + \Len{x}{0}) \\
        \eeqtag{algebra}
        & \inf_{v \in \AVAILLOC{x}} \singleton{v}{x} \sepimp \singleton{v}{x} \sepcon (\Ls{x}{0} + \Len{x}{0}) \\
        \eeqtag{Lemma~\ref{thm:misc:sepimp-sepcon}}
        & \inf_{v \in \AVAILLOC{x}} \containsPointer{v}{-} \cdot \infty + (1-\containsPointer{v}{-}) \cdot (\Ls{x}{0} + \Len{x}{0}) \\
        \eeqtag{Lemma \ref{lem:inf-over-addresses}}
        & \Ls{x}{0} + \Len{x}{0}. 
\end{align}
Using this result, we proceed as follows:
\begin{align}
   &\iverson{\cc = 1} \cdot \wp{\textrm{loopBody}}{\Len{x}{\nil} + \iverson{\cc = 1}} \\
   \eeqtag{Theorem \ref{thm:wp:basic} (\ref{thm:wp:basic:linearity})}
   &\iverson{\cc = 1} \cdot \big(  \wp{\textrm{loopBody}}{\Len{x}{\nil}} + \wp{\textrm{loopBody}}{\iverson{\cc = 1}} \big) \\
   \eeqtag{Definition of $\textrm{loopBody}$}
   &\iverson{\cc = 1} \cdot \big(  \wp{\PCHOICE{\ASSIGN{\cc}{0}}{0.5}{\ASSIGN{\cc}{1}\SEMI\ALLOC{x}{x}}}{\Len{x}{\nil}} \\
   &\qquad \qquad + \wp{\PCHOICE{\ASSIGN{\cc}{0}}{0.5}{\ASSIGN{\cc}{1}\SEMI\ALLOC{x}{x}}}{\iverson{\cc = 1}} \big) \notag \\
   \eeqtag{Table \ref{table:wp}}
   &\iverson{\cc = 1} \cdot \big(  \wp{\PCHOICE{\ASSIGN{\cc}{0}}{0.5}{\ASSIGN{\cc}{1}\SEMI\ALLOC{x}{x}}}{\Len{x}{\nil}} \\
   &\qquad  \qquad+ 0.5 \cdot \wp{\ASSIGN{\cc}{0}}{\iverson{c=1}} + 0.5 \cdot \wp{\ASSIGN{\cc}{1}\SEMI\ALLOC{x}{x}}{\iverson{\cc =1}} \big) \notag \\
   \eeqtag{By Table \ref{table:wp}: $\wp{\ASSIGN{\cc}{0}}{\iverson{c=1}} = \iverson{c=1}\subst{c}{0} =0$}
   &\iverson{\cc = 1} \cdot \big(  \wp{\PCHOICE{\ASSIGN{\cc}{0}}{0.5}{\ASSIGN{\cc}{1}\SEMI\ALLOC{x}{x}}}{\Len{x}{\nil}} \\
   &\qquad \qquad + 0.5 \cdot \wp{\ASSIGN{\cc}{1}\SEMI\ALLOC{x}{x}}{\iverson{\cc =1}} \big) \notag \\
   \eeqtag{Table \ref{table:wp}}
   &\iverson{\cc = 1} \cdot \big(  \wp{\PCHOICE{\ASSIGN{\cc}{0}}{0.5}{\ASSIGN{\cc}{1}\SEMI\ALLOC{x}{x}}}{\Len{x}{\nil}} \\
   &\qquad \qquad + 0.5 \cdot \wp{\ASSIGN{\cc}{1}\SEMI\ALLOC{x}{x}}{\iverson{\cc =1}} \big) \notag \\
   \eeqtag{Table \ref{table:wp}, $\iverson{\cc = 1}\subst{\cc}{1} =1$}
   &\iverson{\cc = 1} \cdot \big(  \wp{\PCHOICE{\ASSIGN{\cc}{0}}{0.5}{\ASSIGN{\cc}{1}\SEMI\ALLOC{x}{x}}}{\Len{x}{\nil}} \\
   &\qquad \qquad + 0.5 \cdot \inf_{v \in \AVAILLOC{x}} \singleton{v}{x} \sepimp 1 \big)  \notag \\
   \eeqtag{Lemma \ref{lem:wand-pure-expectation-on-rhs}}
   &\iverson{\cc = 1} \cdot \big(  \wp{\PCHOICE{\ASSIGN{\cc}{0}}{0.5}{\ASSIGN{\cc}{1}\SEMI\ALLOC{x}{x}}}{\Len{x}{\nil}} \\
   &\qquad \qquad +  0.5 \cdot \inf_{v \in \AVAILLOC{x}} (\containsPointer{v}{-} \cdot \infty + (1-\containsPointer{v}{-}) \cdot 1 \big)   \notag \\
   \eeqtag{Lemma \ref{lem:inf-over-addresses}}
   &\iverson{\cc = 1} \cdot \big(  \wp{\PCHOICE{\ASSIGN{\cc}{0}}{0.5}{\ASSIGN{\cc}{1}\SEMI\ALLOC{x}{x}}}{\Len{x}{\nil}} \\
   &\qquad \qquad +  0.5 \cdot 1   \notag \\
   \eeqtag{Algebra}
   &\iverson{\cc = 1} \cdot \big(  \wp{\PCHOICE{\ASSIGN{\cc}{0}}{0.5}{\ASSIGN{\cc}{1}\SEMI\ALLOC{x}{x}}}{\Len{x}{\nil}} \\
   &\qquad \qquad +  0.5    \notag \\
   \eeqtag{Table \ref{table:wp}}
   &\iverson{\cc = 1} \cdot \big(  0.5 \cdot \wp{\ASSIGN{\cc}{0}}{\Len{x}{\nil}} + 0.5 \cdot \wp{\ASSIGN{\cc}{1}\SEMI\ALLOC{x}{x}}{\Len{x}{\nil}} \\
   &\qquad \qquad + 0.5 \big)   \notag \\
   \eeqtag{Table \ref{table:wp}, $\cc$ does not occur in $\Len{x}{\nil}$}
   &\iverson{\cc = 1} \cdot \big(  0.5 \cdot \Len{x}{\nil} + 0.5 \cdot \wp{\ASSIGN{\cc}{1}}{\wp{\ALLOC{x}{x}}{\Len{x}{\nil}}} \\
   &\qquad \qquad + 0.5 \big)   \notag \\
   \eeqtag{Equation \ref{eqn:increase-list-length}}
   &\iverson{\cc = 1} \cdot \big(  0.5 \cdot \Len{x}{\nil} + 0.5 \cdot \wp{\ASSIGN{\cc}{1}}{\Ls{x}{0} + \Len{x}{0}} \\
   &\qquad \qquad + 0.5 \big)   \notag \\
   \eeqtag{$\cc$ does not occur in $\Ls{x}{0} + \Len{x}{0}$}
   &\iverson{\cc = 1} \cdot \big(  0.5 \cdot \Len{x}{\nil} + 0.5 \cdot (\Ls{x}{0} + \Len{x}{0}) 
   + 0.5 \big)    \\
   \ppreceqtag{$\Ls{x}{0} \preceq 1$}
   &\iverson{\cc = 1} \cdot \big(  0.5 \cdot \Len{x}{\nil} + 0.5 \cdot (1 + \Len{x}{0}) 
    + 0.5 \big)    \\
   \eeqtag{Algebra}
   &\iverson{\cc = 1} \cdot \big(   \Len{x}{\nil} + 1 \big)  \\
   \eeqtag{Algebra}
   &\iverson{\cc = 1} \cdot \Len{x}{\nil} + \iverson{\cc = 1}  \\
   \ppreceqtag{$\iverson{\cc = 1} \preceq 1$}
   &\Len{x}{\nil} + \iverson{\cc = 1}.
\end{align}
This completes the proof.

\subsection{Verification of Invariant for Randomize Array}\label{app:case-studies:randomize-array}
Recall the invariant $I$ proposed in the paper:

\begin{align}
   I \eeq & \iverson{0 \leq i < n} \cdot \frac{1}{(n-i)!} \cdot \bbigsepcon{k=0}{i-1} \singleton{\aarray+k}{\alpha_k}
             \sepcon \sum\limits_{\pi \in \perm{i}{n-1}} \bbigsepcon{k=i}{n-1} \singleton{\aarray+k}{\alpha_{\pi(k)}}  \\
          &\quad + \iverson{\neg (0 \leq i < n)} \cdot \singleton{\aarray}{\alpha_0,\ldots, \alpha_{n-1}} \notag
\end{align}
\noindent
In order to verify $I$ as an invariant of loop $\crand$ w.r.t.\ postexpectation $\singleton{\aarray}{\alpha_0,\ldots, \alpha_{n-1}}$,
we have to show that
\begin{align}
  &\charwp{0 \leq i < n}{\cbody}{\singleton{\aarray}{\alpha_0, \ldots, \alpha_{n-1}}}(I) \\
  \eeq&\iverson{0 \leq i < n} \cdot \wp{\cbody}{I} + \iverson{\neg (0 \leq i < n)} \cdot \singleton{\aarray}{\alpha_0,\ldots, \alpha_{n-1}} 
  \notag \\
  {}~\overset{!}{\preceq}~{} & I  .
  \notag
\end{align}
Since $\iverson{0 \leq i < n} \cdot \iverson{\neg (0 \leq i < n)} = 0$, we subdivide our proof obligation into
\begin{enumerate}
    \item $\iverson{\neg (0 \leq i < n)} \cdot \singleton{\aarray}{\alpha_0,\ldots, \alpha_{n-1}} \preceq I$,~\text{and} \label{eqn:rand_array_to_show_1}
    \item $\iverson{0 \leq i < n} \cdot \wp{\cbody}{I} \preceq I$. \label{eqn:rand_array_to_show_2}
\end{enumerate}
\emph{Proof of \ref{eqn:rand_array_to_show_1}}. We have 
\begin{align}
   &\iverson{\neg (0 \leq i < n)} \cdot \singleton{\aarray}{\alpha_0,\ldots, \alpha_{n-1}} \\
   \ppreceqtag{$0 \preceq X$ for all $X \in \E$}
   &\iverson{0 \leq i < n} \cdot \frac{1}{(n-i)!} \cdot \bbigsepcon{k=0}{i-1} \singleton{\aarray+k}{\alpha_k}
             \sepcon \sum\limits_{\pi \in \perm{i}{n-1}} \bbigsepcon{k=i}{n-1} \singleton{\aarray+k}{\alpha_{\pi(k)}}
             \notag \\
              &\quad + \iverson{\neg (0 \leq i < n)} \cdot \singleton{\aarray}{\alpha_0,\ldots, \alpha_{n-1}}
              \notag \\
   \eeqtag{Definition of $I$}
   & I.
   \notag
\end{align}
\emph{Proof of \ref{eqn:rand_array_to_show_2}}. 
%
%
%
Let 
$\cbody = \COMPOSE{\cc_1}{\COMPOSE{\cc_2}{\cc_3}}$. Moreover, let 
\begin{align}
   \FinPermutations{q} \eeq& \bigcup_{\substack{p \in \Nats \\ p \leq q}} \bigcup_{\substack{r \in \Nats \\ r \leq p}} \big\{ f ~\mid~ f: \{r,\ldots,p \} \mapsto \{r,\ldots,p \} \big\},~\text{and} \\
   \Permutations \eeq& \bigcup_{q \in \Nats} \FinPermutations{q}, \label{eqn:def-permutations}
\end{align}
We proceed as follows:
\begin{align}
   &\iverson{0 \leq i < n} \cdot \wp{\cbody}{I} 
   \label{eqn:rand_array_main_eqns} \\
   %
   %
   %
   \eeqtag{Definition of $I$}
   &\iverson{0\leq i <n} \cdot \wpsymbol\llbracket \cbody \rrbracket \big( 
      \iverson{0 \leq i < n} \cdot \frac{1}{(n-i)!}
      \\
          & \quad \cdot \bbigsepcon{k=0}{i-1} \singleton{\aarray+k}{\alpha_k}
             \sepcon \sum\limits_{\pi \in \perm{i}{n-1}} \bbigsepcon{k=i}{n-1} \singleton{\aarray+k}{\alpha_{\pi(k)}} \notag \\
          &\qquad + \iverson{\neg (0 \leq i < n)} \cdot \singleton{\aarray}{\alpha_0,\ldots, \alpha_{n-1}} \big) \notag \\
   \eeqtag{Let $\cbody = \cc_2 \SEMI \ASSIGN{i}{i+1}$}
   &\iverson{0\leq i <n} \cdot \wpsymbol\llbracket \cc_2 \SEMI \ASSIGN{i}{i+1}\rrbracket \big( 
      \iverson{0 \leq i < n} \cdot \frac{1}{(n-i)!}
      \\
          & \quad \cdot \bbigsepcon{k=0}{i-1} \singleton{\aarray+k}{\alpha_k}
             \sepcon \sum\limits_{\pi \in \perm{i}{n-1}} \bbigsepcon{k=i}{n-1} \singleton{\aarray+k}{\alpha_{\pi(k)}} \notag \\
          &\qquad + \iverson{\neg (0 \leq i < n)} \cdot \singleton{\aarray}{\alpha_0,\ldots, \alpha_{n-1}} \big) \notag \\
   \eeqtag{Table \ref{table:wp}: substituting $i$ by $i+1$}
   &\iverson{0\leq i <n} \cdot \wpsymbol\llbracket \cc_2 \rrbracket \big(  
      \iverson{0 \leq i+1 < n} \cdot \frac{1}{(n-i-1)!}
      \\
          & \quad \cdot \bbigsepcon{k=0}{i} \singleton{\aarray+k}{\alpha_k}
             \sepcon \sum\limits_{\pi \in \perm{i+1}{n-1}} \bbigsepcon{k=i+1}{n-1} \singleton{\aarray+k}{\alpha_{\pi(k)}} \notag \\
          &\qquad + \iverson{\neg (0 \leq i+1 < n)} \cdot \singleton{\aarray}{\alpha_0,\ldots, \alpha_{n-1}} \big) \notag \\
   \eeqtag{Pure Frame Rule (Theorem \ref{thm:batz}) on $\iverson{0\leq i <n})$}
   &\wpsymbol\llbracket \cc_2 \rrbracket \big( \iverson{0\leq i <n} \cdot \big(  
      \iverson{0 \leq i+1 < n} \cdot \frac{1}{(n-i-1)!} 
      \\
          & \quad \cdot \bbigsepcon{k=0}{i} \singleton{\aarray+k}{\alpha_k}
             \sepcon \sum\limits_{\pi \in \perm{i+1}{n-1}} \bbigsepcon{k=i+1}{n-1} \singleton{\aarray+k}{\alpha_{\pi(k)}} \notag \\
          &\qquad + \iverson{\neg (0 \leq i+1 < n)} \cdot \singleton{\aarray}{\alpha_0,\ldots, \alpha_{n-1}} \big) \big) \notag \\
   \eeqtag{$\iverson{0\leq i <n} \cdot \iverson{\neg (0 \leq i+1 < n)} = 0$}
   &\wpsymbol\llbracket \cc_2 \rrbracket \big( \iverson{0\leq i <n}  
       \iverson{0 \leq i+1 < n} \cdot \frac{1}{(n-i-1)!} 
       \\
          & \quad \cdot  \bbigsepcon{k=0}{i} \singleton{\aarray+k}{\alpha_k}
             \sepcon \sum\limits_{\pi \in \perm{i+1}{n-1}} \bbigsepcon{k=i+1}{n-1} \singleton{\aarray+k}{\alpha_{\pi(k)}} \big) \notag \\
   \ppreceqtag{$\iverson{0 \leq i+1 < n} \preceq 1$, then apply monotonicity (Theorem \ref{thm:wp:basic} (\ref{thm:wp:basic:monotonicity}))}
   &\wpsymbol\llbracket \cc_2 \rrbracket \big(  
      \iverson{0\leq i <n} \cdot \frac{1}{(n-i-1)!} 
      \\
          & \quad \cdot \bbigsepcon{k=0}{i} \singleton{\aarray+k}{\alpha_k}
             \sepcon \sum\limits_{\pi \in \perm{i+1}{n-1}} \bbigsepcon{k=i+1}{n-1} \singleton{\aarray+k}{\alpha_{\pi(k)}} \big) \notag \\
   \eeqtag{$\bbigsepcon{k=0}{i} \singleton{\aarray+k}{\alpha_k}$ is domain exact, then apply Theorem \ref{thm:sep-con-distrib} (\ref{thm:sep-con-distrib:sepcon-over-plus-full})}
   &\wpsymbol\llbracket \cc_2 \rrbracket \big( 
    \iverson{0\leq i <n} \cdot \frac{1}{(n-i-1)!} 
    \\
          & \quad \cdot \sum\limits_{\pi \in \perm{i+1}{n-1}}
           \bbigsepcon{k=0}{i} \singleton{\aarray+k}{\alpha_k}
             \sepcon \bbigsepcon{k=i+1}{n-1} \singleton{\aarray+k}{\alpha_{\pi(k)}} \big) \notag \\
   \eeqtag{Pure Frame Rule (Theorem \ref{thm:batz}) on $\frac{1}{(n-i-1)!}$}
   &\frac{1}{(n-i-1)!} \cdot \wpsymbol\llbracket \cc_2 \rrbracket \big(  
      \iverson{0\leq i <n}
      \\
          & \quad \cdot \sum\limits_{\pi \in \perm{i+1}{n-1}}
           \bbigsepcon{k=0}{i} \singleton{\aarray+k}{\alpha_k}
             \sepcon \bbigsepcon{k=i+1}{n-1} \singleton{\aarray+k}{\alpha_{\pi(k)}} \big) \notag \\
   \eeqtag{Rewrite sum using Equation \ref{eqn:def-permutations} and $0\leq i <n$}
   &\frac{1}{(n-i-1)!} \cdot \wpsymbol\llbracket \cc_2 \rrbracket \big( 
             \iverson{0\leq i <n} \cdot \sum\limits_{\pi \in \Permutations} \iverson{\pi \in \perm{i+1}{n-1}}  \\
             &\quad \cdot
           \bbigsepcon{k=0}{i} \singleton{\aarray+k}{\alpha_k}
             \sepcon \bbigsepcon{k=i+1}{n-1} \singleton{\aarray+k}{\alpha_{\pi(k)}} \big) \notag \\
   \eeqtag{Algebra}
   &\frac{1}{(n-i-1)!} \cdot \wpsymbol\llbracket \cc_2 \rrbracket \big( 
             \iverson{0\leq i <n} \cdot \sum\limits_{k=0}^{\infty} \sum\limits_{\pi \in \FinPermutations{k}} \iverson{\pi \in \perm{i+1}{n-1}}  \\
             &\quad \cdot
           \bbigsepcon{k=0}{i} \singleton{\aarray+k}{\alpha_k}
             \sepcon \bbigsepcon{k=i+1}{n-1} \singleton{\aarray+k}{\alpha_{\pi(k)}} \big) \notag \\
   \eeqtag{Algebra}
   &\frac{1}{(n-i-1)!} \cdot \wpsymbol\llbracket \cc_2 \rrbracket \big( 
             \iverson{0\leq i <n} \cdot \sup_{l \in \Nats} \sum\limits_{k=0}^{l} \sum\limits_{\pi \in \FinPermutations{k}} \iverson{\pi \in \perm{i+1}{n-1}}  \\
             &\quad \cdot
           \bbigsepcon{k=0}{i} \singleton{\aarray+k}{\alpha_k}
             \sepcon \bbigsepcon{k=i+1}{n-1} \singleton{\aarray+k}{\alpha_{\pi(k)}} \big) \notag \\
   \eeqtag{Continuity of $\wpsymbol$ (Theorem \ref{thm:wp:basic} (\ref{thm:wp:basic:continuity}))}
   &\frac{1}{(n-i-1)!} \cdot \sup_{l \in \Nats} \wpsymbol\llbracket \cc_2 \rrbracket \big(
             \iverson{0\leq i <n} \cdot \sum\limits_{k=0}^{l} \sum\limits_{\pi \in \FinPermutations{k}} \iverson{\pi \in \perm{i+1}{n-1}}  \\
             &\quad \cdot
           \bbigsepcon{k=0}{i} \singleton{\aarray+k}{\alpha_k}
             \sepcon \bbigsepcon{k=i+1}{n-1} \singleton{\aarray+k}{\alpha_{\pi(k)}} \big) \notag \\
   \eeqtag{Linearity of $\wpsymbol$ (Theorem \ref{thm:wp:basic} (\ref{thm:wp:basic:linearity}))}
   &\frac{1}{(n-i-1)!} \cdot \sup_{l \in \Nats} \sum\limits_{k=0}^{l} \sum\limits_{\pi \in \FinPermutations{k}} \wpsymbol\llbracket \cc_2 \rrbracket \big(
             \iverson{0\leq i <n} \cdot  \iverson{\pi \in \perm{i+1}{n-1}} \\
             &\quad \cdot
           \bbigsepcon{k=0}{i} \singleton{\aarray+k}{\alpha_k}
             \sepcon \bbigsepcon{k=i+1}{n-1} \singleton{\aarray+k}{\alpha_{\pi(k)}} \big) \notag \\
   \eeqtag{Pure Frame Rule (Theorem \ref{thm:batz}) on $\iverson{\pi \in \perm{i+1}{n-1}}$}
   &\frac{1}{(n-i-1)!} \cdot \sup_{n \in \Nats} \sum\limits_{k=0}^{l} \sum\limits_{\pi \in \FinPermutations{k}} \iverson{\pi \in \perm{i+1}{n-1}} \cdot \wpsymbol\llbracket \cc_2 \rrbracket \big(
             \iverson{0\leq i <n}  \\
             &\quad \cdot
           \bbigsepcon{k=0}{i} \singleton{\aarray+k}{\alpha_k}
             \sepcon \bbigsepcon{k=i+1}{n-1} \singleton{\aarray+k}{\alpha_{\pi(k)}} \big) \notag \\
   \eeqtag{Rewrite sum as above}
   &\frac{1}{(n-i-1)!} \cdot \sum\limits_{\pi \in \perm{i+1}{n-1}} \wpsymbol\llbracket \cc_2 \rrbracket \big(
             \iverson{0\leq i <n}    \\
             &\quad \cdot
           \bbigsepcon{k=0}{i} \singleton{\aarray+k}{\alpha_k}
             \sepcon \bbigsepcon{k=i+1}{n-1} \singleton{\aarray+k}{\alpha_{\pi(k)}} \big) \notag \\
   \eeqtag{$0 \leq i$}
   &\frac{1}{(n-i-1)!} \cdot \sum\limits_{\pi \in \perm{i+1}{n-1}} \wpsymbol\llbracket \cc_2 \rrbracket \big(  \\
          & \quad    \iverson{0\leq i <n} \cdot 
           \bbigsepcon{k=0}{i-1} \singleton{\aarray+k}{\alpha_k}
             \sepcon \singleton{\aarray+i}{\alpha_i}
             \sepcon \bbigsepcon{k=i+1}{n-1} \singleton{\aarray+k}{\alpha_{\pi(k)}} \big) \notag \\
   \eeqtag{Pure Frame Rule (Theorem \ref{thm:batz}) on $\iverson{0\leq i <n}$}    
   &\iverson{0\leq i <n} \cdot \frac{1}{(n-i-1)!} \cdot \sum\limits_{\pi \in \perm{i+1}{n-1}} \wpsymbol\llbracket \cc_2 \rrbracket \big(  \\
          & \quad    \underbrace{
           \bbigsepcon{k=0}{i-1} \singleton{\aarray+k}{\alpha_k}
             \sepcon \singleton{\aarray+i}{\alpha_i}
             \sepcon \bbigsepcon{k=i+1}{n-1} \singleton{\aarray+k}{\alpha_{\pi(k)}}}_{\eqqcolon \fg}\big) \notag \\
   \eeqtag{Case distinction: $\iverson{i=j} + \iverson{i < j} + \iverson{i>j} = 1$}    
   &\iverson{0\leq i <n} \cdot \frac{1}{(n-i-1)!} \cdot \sum\limits_{\pi \in \perm{i+1}{n-1}} \wpsymbol\llbracket \cc_2 \rrbracket \big(   \iverson{i=j} \cdot \fg  +  \iverson{i > j} \cdot \fg  + \iverson{i<j} \cdot \fg \big) \\
    \eeqtag{Linearity of $\wpsymbol$ (Theorem \ref{thm:wp:basic} (\ref{thm:wp:basic:linearity}))}    
   &\iverson{0\leq i <n} \cdot \frac{1}{(n-i-1)!}   \\
          & \quad  \cdot \sum\limits_{\pi \in \perm{i+1}{n-1}} \underbrace{\wpsymbol\llbracket \cc_2 \rrbracket \big( \iverson{i=j} \cdot \fg \big)}_{\eqqcolon \fh_1}
           + \underbrace{\wpsymbol\llbracket \COMPOSE{\cc_1}{\cc_2} \rrbracket \big( \iverson{i > j} \cdot \fg \big)}_{\eqqcolon \fh_2}
             + \underbrace{\wpsymbol\llbracket \COMPOSE{\cc_1}{\cc_2} \rrbracket \big( \iverson{i<j} \cdot \fg \big)}_{\eqqcolon \fh_3} 
             \notag
\end{align}
We continue by calculating $\fh_1$, $\fh_2$, and $\fh_3$ separately. For $\fh_1$, we have
\begin{align}
   &\fh_1 \\
   \eeqtag{Definition of $\fh_1$}
   &\wpsymbol\llbracket \cc_2 \rrbracket \big( \iverson{i=j} \cdot \fg \big) \\
   \eeqtag{Let $\cc_2 = \cc_1 \SEMI \ProcCall{swap}{\aarray, i , j}{}$}
   &\wpsymbol\llbracket \cc_1 \SEMI \ProcCall{swap}{\aarray, i , j}{} \rrbracket \big( \iverson{i=j} \cdot \fg \big) \\
   \eeqtag{Table \ref{table:wp}}
   &\wpsymbol\llbracket \cc_1 \rrbracket \big( \wp{\ProcCall{swap}{\aarray, i , j}{}}{ \iverson{i=j} \cdot \fg } \big) \\
   \eeqtag{Lemma \ref{lem:swap-i-equal-j}}
   &\wpsymbol\llbracket \cc_1 \rrbracket \big(  \iverson{i=j} \cdot \fg  \big)\\
   \eeqtag{$\cc_1 = \ASSIGNUNIFORM{j}{i}{n-1}$, definition of $\wpsymbol$ for random assignments}
   &\frac{1}{n-i} \cdot \sum\limits_{k=i}^{n-1} (\iverson{i=j} \cdot \fg) \subst{j}{k}  \\
   \eeqtag{$\iverson{i=j}\subst{j}{k} = 0$ for $k \neq i$}
   &\frac{1}{n-i} \cdot (\iverson{i=j} \cdot \fg) \subst{j}{i} \\
   \eeqtag{$\iverson{i=j} \subst{j}{i} = 1$ and $j$ does not occur in $\fg$}
   &\frac{1}{n-i} \cdot \fg. 
\end{align}
For $\fh_2$, we have
\begin{align}
   &\fh_2 \\
   \eeqtag{Definition of $\fh_2$}
   &\wpsymbol\llbracket \cc_2 \rrbracket \big( \iverson{i>j} \cdot \fg \big) \\
   \eeqtag{Let $\cc_2 = \cc_1 \SEMI \ProcCall{swap}{\aarray, i , j}{}$}
   &\wpsymbol\llbracket \cc_1 \SEMI \ProcCall{swap}{\aarray, i , j}{} \rrbracket \big( \iverson{i>j} \cdot \fg \big) \\
   \eeqtag{Table \ref{table:wp}}
   &\wpsymbol\llbracket \cc_1 \rrbracket \big( \wp{\ProcCall{swap}{\aarray, i , j}{}}{ \iverson{i>j} \cdot \fg } \big)\\
   \eeqtag{Pure Frame Rule (Theorem \ref{thm:batz}) on $\iverson{i>j}$}
   &\wpsymbol\llbracket \cc_1 \rrbracket \big( \iverson{i>j} \cdot \wp{\ProcCall{swap}{\aarray, i , j}{}}{ \fg } \big) \\
   \eeqtag{$\cc_1 = \ASSIGNUNIFORM{j}{i}{n-1}$, definition of $\wpsymbol$ for random assignments}
   &\frac{1}{n-i} \cdot \sum\limits_{k=i}^{n-1} (\iverson{i>j} \cdot \wp{\cc_2}{ \fg }) \subst{j}{k}  \\
   \eeqtag{$\iverson{i>j}\subst{j}{k} = 0$ for $k\geq i$}
   & 0.
\end{align}
For $\fh_3$, we have
\begin{align}
   &\fh_3 \\
   \eeqtag{Definition of $\fh_3$}
   &\wpsymbol\llbracket \cc_2 \rrbracket \big( \iverson{i<j} \cdot \fg \big)\\
   \eeqtag{Let $\cc_2 = \cc_1 \SEMI \ProcCall{swap}{\aarray, i , j}{}$}
   &\wpsymbol\llbracket \cc_1 \SEMI \ProcCall{swap}{\aarray, i , j}{} \rrbracket \big( \iverson{i<j} \cdot \fg \big)\\
   \eeqtag{Table \ref{table:wp}}
   &\wpsymbol\llbracket \cc_1 \rrbracket \big( \wpsymbol\llbracket \ProcCall{swap}{\aarray, i , j}{} \rrbracket \big( \iverson{i<j} \cdot \fg \big) \big)\\
   \eeqtag{Case distinction: $\iverson{j \leq n-1} + \iverson{j > n-1} = 1$}
   &\wpsymbol\llbracket \cc_1 \rrbracket \big( \wpsymbol\llbracket \ProcCall{swap}{\aarray, i , j}{} \rrbracket \big(\iverson{j \leq n-1} \cdot \iverson{i<j} \cdot \fg
        + \iverson{j > n-1} \cdot \iverson{i<j} \cdot \fg \big) \big)  \\
   \eeqtag{Linearity of $\wpsymbol$ (Theorem \ref{thm:wp:basic} (\ref{thm:wp:basic:linearity}))}
   &\wpsymbol\llbracket \cc_1 \rrbracket \big( \wpsymbol\llbracket \ProcCall{swap}{\aarray, i , j}{} \rrbracket \big(\iverson{j \leq n-1} \cdot \iverson{i<j} \cdot \fg \big) \\
   &\quad     + \wpsymbol\llbracket \ProcCall{swap}{\aarray, i , j}{} \rrbracket \big( \iverson{j > n-1} \cdot \iverson{i<j} \cdot \fg \big) \big) \\
   \eeqtag{Linearity of $\wpsymbol$ (Theorem \ref{thm:wp:basic} (\ref{thm:wp:basic:linearity}))}
   &\wpsymbol\llbracket \cc_1 \rrbracket \big( \wpsymbol\llbracket \ProcCall{swap}{\aarray, i , j}{} \rrbracket \big(\iverson{j \leq n-1} \cdot \iverson{i<j} \cdot \fg \big) \big) \\
   &\quad + \wpsymbol\llbracket \cc_1 \rrbracket \big(    \wpsymbol\llbracket \ProcCall{swap}{\aarray, i , j}{} \rrbracket \big( \iverson{j > n-1} \cdot \iverson{i<j} \cdot \fg \big) \big)
   \notag \\
   \eeqtag{Theorem \ref{thm:batz} on $\iverson{j > n-1}$}
   &\wpsymbol\llbracket \cc_1 \rrbracket \big( \wpsymbol\llbracket \ProcCall{swap}{\aarray, i , j}{} \rrbracket \big(\iverson{j \leq n-1} \cdot \iverson{i<j} \cdot \fg \big) \big) \\
   &\quad + \wpsymbol\llbracket \cc_1 \rrbracket \big(  \iverson{j > n-1} \cdot \wpsymbol\llbracket \ProcCall{swap}{\aarray, i , j}{} \rrbracket \big(  \iverson{i<j} \cdot \fg \big) \big)
   \notag \\
   \eeqtag{$\cc_1 = \ASSIGNUNIFORM{j}{i}{n-1}$, definition of $\wpsymbol$ for random assignments}
   &\wpsymbol\llbracket \cc_1 \rrbracket \big( \wpsymbol\llbracket \ProcCall{swap}{\aarray, i , j}{} \rrbracket \big(\iverson{j \leq n-1} \cdot \iverson{i<j} \cdot \fg \big) \big)  \\
   &\quad + \frac{1}{n-i} \cdot \sum\limits_{k=i}^{n-1}  \big(\iverson{j > n-1} \cdot \wpsymbol\llbracket \ProcCall{swap}{\aarray, i , j}{} \rrbracket \big( \iverson{j > n-1} \cdot \iverson{i<j} \cdot \fg \big) \big)\subst{j}{k}
   \notag \\
   \eeqtag{$\iverson{j > n-1}\subst{j}{k} = 0$ for $k \leq n-1$}
   &\wpsymbol\llbracket \cc_1 \rrbracket \big( \wpsymbol\llbracket \ProcCall{swap}{\aarray, i , j}{} \rrbracket \big(\iverson{j \leq n-1} \cdot \iverson{i<j} \cdot \fg \big) \big)  \\
   &\quad + 0
    \notag \\
   \eeqtag{Neutrality of $0$ w.r.t.\ $+$}
   &\wpsymbol\llbracket \cc_1 \rrbracket \big( \wpsymbol\llbracket \ProcCall{swap}{\aarray, i , j}{} \rrbracket \big(\iverson{j \leq n-1} \cdot \iverson{i<j} \cdot \fg \big) \big) \\
   \eeqtag{Definition of $\fg$}
   &\wpsymbol\llbracket \cc_1 \rrbracket \big( \wpsymbol\llbracket \ProcCall{swap}{\aarray, i , j}{} \rrbracket \big(\iverson{j \leq n-1} \cdot \iverson{i<j}  
     \cdot \bbigsepcon{k=0}{i-1} \singleton{\aarray+k}{\alpha_k}\\
     &\qquad \quad
             \sepcon \singleton{\aarray+i}{\alpha_i} 
    \sepcon \bbigsepcon{k=i+1}{n-1} \singleton{\aarray+k}{\alpha_{\pi(k)}}
     \big) \big)
     \notag \\
   \eeqtag{$i < j \leq n-1$}
   &\wpsymbol\llbracket \cc_1 \rrbracket \big( \wpsymbol\llbracket \ProcCall{swap}{\aarray, i , j}{} \rrbracket \big(\iverson{j \leq n-1} \cdot \iverson{i<j}
    \cdot 
    \bbigsepcon{k=0}{i-1} \singleton{\aarray+k}{\alpha_k}\\
  & \qquad \quad
             \sepcon \singleton{\aarray+i}{\alpha_i}  \sepcon \singleton{\aarray+j}{\alpha_{\pi(j)}}
             \sepcon \bbigsepcon{\substack{k=i+1 \\ k \neq j}}{n-1} \singleton{\aarray+k}{\alpha_{\pi(k)}}
     \big) \big)
     \notag \\
   \eeqtag{Pure Frame Rule (Theorem \ref{thm:batz}) on $\iverson{j \leq n-1} \cdot \iverson{i<j}$}
   &\wpsymbol\llbracket \cc_1 \rrbracket \big( \iverson{j \leq n-1} \cdot \iverson{i<j} \cdot \wpsymbol\llbracket \ProcCall{swap}{\aarray, i , j}{} \rrbracket \big( 
    \bbigsepcon{k=0}{i-1} \singleton{\aarray+k}{\alpha_k}\\
   &\qquad \quad
             \sepcon \singleton{\aarray+i}{\alpha_i}   \sepcon \singleton{\aarray+j}{\alpha_{\pi(j)}}
             \sepcon \bbigsepcon{\substack{k=i+1 \\ k \neq j}}{n-1} \singleton{\aarray+k}{\alpha_{\pi(k)}}
     \big) \big)
     \notag \\
   \eeqtag{Lemma \ref{lem:swap-i-unequal-j}}
   &\wpsymbol\llbracket \cc_1 \rrbracket \big( \iverson{j \leq n-1} \cdot \iverson{i<j}   
    \cdot \bbigsepcon{k=0}{i-1} \singleton{\aarray+k}{\alpha_k}
             \sepcon \singleton{\aarray+i}{\alpha_{\pi(j)}}  \\
   &\qquad \quad \sepcon \singleton{\aarray+j}{\alpha_i}
             \sepcon \bbigsepcon{\substack{k=i+1 \\ k \neq j}}{n-1} \singleton{\aarray+k}{\alpha_{\pi(k)}}
      \big) 
      \notag \\
   \eeqtag{$\cc_1 = \ASSIGNUNIFORM{j}{i}{n-1}$, definition of $\wpsymbol$ for random assignments}
   &\frac{1}{n-i} \cdot \sum\limits_{u=i}^{n-1} \big( \iverson{j \leq n-1} \cdot \iverson{i<j}  
    \cdot  \bbigsepcon{k=0}{i-1} \singleton{\aarray+k}{\alpha_k}
             \sepcon \singleton{\aarray+i}{\alpha_{\pi(j)}}  \\
   & \qquad \quad \sepcon \singleton{\aarray+j}{\alpha_i}
             \sepcon \bbigsepcon{\substack{k=i+1 \\ k \neq j}}{n-1} \singleton{\aarray+k}{\alpha_{\pi(k)}}
      \big)\subst{j}{u}
      \notag \\
   \eeqtag{$\iverson{j \leq n-1} \subst{j}{u} = 1$ for $u \leq n-1$ and $\iverson{i<j}\subst{j}{u} = 0$ for $u=i$}
   &\frac{1}{n-i} \cdot \sum\limits_{u=i+1}^{n-1}   
    \big( \bbigsepcon{k=0}{i-1} \singleton{\aarray+k}{\alpha_k}
             \sepcon \singleton{\aarray+i}{\alpha_{\pi(j)}} \\
   & \qquad \quad \sepcon \singleton{\aarray+j}{\alpha_i}
             \sepcon \bbigsepcon{\substack{k=i+1 \\ k \neq j}}{n-1} \singleton{\aarray+k}{\alpha_{\pi(k)}}
      \big)\subst{j}{u}
      \notag \\
   \eeqtag{Applying the substitution}
   &\frac{1}{n-i} \cdot \sum\limits_{u=i+1}^{n-1}   
     \bbigsepcon{k=0}{i-1} \singleton{\aarray+k}{\alpha_k}
             \sepcon \singleton{\aarray+i}{\alpha_{\pi(u)}} \\
   &\qquad \quad \sepcon \singleton{\aarray+u}{\alpha_i}
             \sepcon \bbigsepcon{\substack{k=i+1 \\ k \neq u}}{n-1} \singleton{\aarray+k}{\alpha_{\pi(k)}}.
             \notag
\end{align}
Using our calculations for $\fh_1, \fh_2$, and $\fh_3$, we continue with Equation~\ref{eqn:rand_array_main_eqns}:
\begin{align}
   &\iverson{0\leq i <n} \cdot \frac{1}{(n-i-1)!} \cdot \sum\limits_{\pi \in \perm{i+1}{n-1}}  
          \fh_1 + \fh_2 + \fh_3 \\
   \eeqtag{Plugging in the results for $\fh_1, \fh_2$, and $\fh_3$ (omitting $\fh_2 = 0$)}
   &\iverson{0\leq i <n} \cdot \frac{1}{(n-i-1)!}\cdot \sum\limits_{\pi \in \perm{i+1}{n-1}}   \\
   &\quad  \Big( \big( \frac{1}{n-i} \cdot \bbigsepcon{k=0}{i-1} \singleton{\aarray+k}{\alpha_k}
             \sepcon \singleton{\aarray+i}{\alpha_i}
             \sepcon \bbigsepcon{k=i+1}{n-1} \singleton{\aarray+k}{\alpha_{\pi(k)}} \big) 
             \notag \\
   &\quad +\big( \frac{1}{n-i} \cdot \sum\limits_{u=i+1}^{n-1}    \bbigsepcon{k=0}{i-1} \singleton{\aarray+k}{\alpha_k}
    \sepcon \singleton{\aarray+i}{\alpha_{\pi(u)}}
    \notag \\
   &\qquad  \sepcon \singleton{\aarray+u}{\alpha_i}
             \sepcon \bbigsepcon{\substack{k=i+1 \\ k \neq u}}{n-1} \singleton{\aarray+k}{\alpha_{\pi(k)}} \big) \Big) 
             \notag\\
   \eeqtag{Multiplying out $\frac{1}{n-i}$ and using $\frac{1}{(n-i-1)!} \cdot \frac{1}{n-i} = \frac{1}{(n-i)!}$}
   &\iverson{0\leq i <n} \cdot \frac{1}{(n-i)!}\cdot \sum\limits_{\pi \in \perm{i+1}{n-1}}  \\
   &\quad  \Big( \big(  \bbigsepcon{k=0}{i-1} \singleton{\aarray+k}{\alpha_k}
             \sepcon \singleton{\aarray+i}{\alpha_i}
             \sepcon \bbigsepcon{k=i+1}{n-1} \singleton{\aarray+k}{\alpha_{\pi(k)}} \big)
             \notag \\
   &\quad +\big(  \sum\limits_{u=i+1}^{n-1}    \bbigsepcon{k=0}{i-1} \singleton{\aarray+k}{\alpha_k} 
    \sepcon \singleton{\aarray+i}{\alpha_{\pi(u)}}
    \notag \\
   &\qquad  \sepcon \singleton{\aarray+u}{\alpha_i}
             \sepcon \bbigsepcon{\substack{k=i+1 \\ k \neq u}}{n-1} \singleton{\aarray+k}{\alpha_{\pi(k)}} \big) \Big)
             \notag \\
   \eeqtag{Pulling apart sum}
   &\iverson{0\leq i <n} \cdot \frac{1}{(n-i)!} \\
   &\quad \cdot \Big( \sum\limits_{\pi \in \perm{i+1}{n-1}}   \big(  \bbigsepcon{k=0}{i-1} \singleton{\aarray+k}{\alpha_k}
             \sepcon \singleton{\aarray+i}{\alpha_i}
             \sepcon \bbigsepcon{k=i+1}{n-1} \singleton{\aarray+k}{\alpha_{\pi(k)}} \big) 
             \notag \\
   &\quad +  \sum\limits_{\pi \in \perm{i+1}{n-1}}   \sum\limits_{u=i+1}^{n-1}    \bbigsepcon{k=0}{i-1} \singleton{\aarray+k}{\alpha_k} 
     \sepcon \singleton{\aarray+i}{\alpha_{\pi(u)}} 
     \notag \\
   &\qquad \quad\sepcon \singleton{\aarray+u}{\alpha_i}
             \sepcon \bbigsepcon{\substack{k=i+1 \\ k \neq u}}{n-1} \singleton{\aarray+k}{\alpha_{\pi(k)}}  \Big) 
             \notag \\
   \eeqtag{Swap sums, $\bbigsepcon{k=0}{i-1} \singleton{\aarray+k}{\alpha_k}$ is domain exact, then apply Theorem \ref{thm:sep-con-distrib} (\ref{thm:sep-con-distrib:sepcon-over-plus-full})}
   &\iverson{0\leq i <n} \cdot \frac{1}{(n-i)!}  \\
   &\quad \cdot \Big( \sum\limits_{\pi \in \perm{i+1}{n-1}}   \big(  \bbigsepcon{k=0}{i-1} \singleton{\aarray+k}{\alpha_k}
             \sepcon \singleton{\aarray+i}{\alpha_i}
             \sepcon \bbigsepcon{k=i+1}{n-1} \singleton{\aarray+k}{\alpha_{\pi(k)}} \big) 
             \notag \\
   &\quad +   \big(  \sum\limits_{u=i+1}^{n-1}    \bbigsepcon{k=0}{i-1} \singleton{\aarray+k}{\alpha_k}
     \sepcon \sum\limits_{\pi \in \perm{i+1}{n-1}} \singleton{\aarray+i}{\alpha_{\pi(u)}}
     \notag \\
   &\qquad \quad \sepcon \singleton{\aarray+u}{\alpha_i}
             \sepcon  \bbigsepcon{\substack{k=i+1 \\ k \neq u}}{n-1} \singleton{\aarray+k}{\alpha_{\pi(k)}} \big) \Big) 
             \notag \\
   \eeqtag{$i+1 \leq u \leq n-1$ and reorder sums by fixing $\pi(u) =i$ and $\pi(i)=u$}
   &\iverson{0\leq i <n} \cdot \frac{1}{(n-i)!}  \\
   &\quad \cdot \Big( \sum\limits_{\pi \in \perm{i+1}{n-1}}   \big(  \bbigsepcon{k=0}{i-1} \singleton{\aarray+k}{\alpha_k}
             \sepcon \singleton{\aarray+i}{\alpha_i}
             \sepcon \bbigsepcon{k=i+1}{n-1} \singleton{\aarray+k}{\alpha_{\pi(k)}} \big) 
             \notag \\
   &\quad +   \big(  \sum\limits_{u=i+1}^{n-1} \bbigsepcon{k=0}{i-1} \singleton{\aarray+k}{\alpha_k}
   \sepcon  \sum\limits_{\substack{\pi \in \perm{i}{n-1} \\ \pi (u) = i \\ \pi (i) = u}}
              \bbigsepcon{\substack{k=i}}{n-1} \singleton{\aarray+k}{\alpha_{\pi(k)}} \big) \Big) 
             \notag \\
   \eeqtag{$\bbigsepcon{k=0}{i-1} \singleton{\aarray+k}{\alpha_k}$ is domain exact, then apply Theorem \ref{thm:sep-con-distrib} (\ref{thm:sep-con-distrib:sepcon-over-plus-full})}
   &\iverson{0\leq i <n} \cdot \frac{1}{(n-i)!}  \\
   &\quad \cdot \Big( \sum\limits_{\pi \in \perm{i+1}{n-1}}   \big(  \bbigsepcon{k=0}{i-1} \singleton{\aarray+k}{\alpha_k}
             \sepcon \singleton{\aarray+i}{\alpha_i}
             \sepcon \bbigsepcon{k=i+1}{n-1} \singleton{\aarray+k}{\alpha_{\pi(k)}} \big) 
             \notag \\
   &\quad +   \big(  \sum\limits_{u=i+1}^{n-1} \sum\limits_{\substack{\pi \in \perm{i}{n-1} \\ \pi (u) = i \\ \pi (i) = u}} \bbigsepcon{k=0}{i-1} \singleton{\aarray+k}{\alpha_k}
   \sepcon  
              \bbigsepcon{\substack{k=i}}{n-1} \singleton{\aarray+k}{\alpha_{\pi(k)}} \big) \Big) 
             \notag \\
   \eeqtag{$u$ ranges from $i+1$ to $n-1$}
   &\iverson{0\leq i <n} \cdot \frac{1}{(n-i)!}  \\
   &\quad \cdot \Big( \sum\limits_{\pi \in \perm{i+1}{n-1}}   \big(  \bbigsepcon{k=0}{i-1} \singleton{\aarray+k}{\alpha_k}
             \sepcon \singleton{\aarray+i}{\alpha_i}
             \sepcon \bbigsepcon{k=i+1}{n-1} \singleton{\aarray+k}{\alpha_{\pi(k)}} \big) 
             \notag \\
   &\quad +   \big( \sum\limits_{\substack{\pi \in \perm{i}{n-1} \\ \pi (i) \neq i}}   \bbigsepcon{k=0}{i-1} \singleton{\aarray+k}{\alpha_k}
   \sepcon  
              \bbigsepcon{\substack{k=i}}{n-1} \singleton{\aarray+k}{\alpha_{\pi(k)}} \big) \Big) 
             \notag \\
   \eeqtag{fixing $\pi(i)=i$}
   &\iverson{0\leq i <n} \cdot \frac{1}{(n-i)!}  \\
   &\quad \cdot \Big( \sum\limits_{\substack{\pi \in \perm{i}{n-1}\\ \pi(i)=i}}   \big(  \bbigsepcon{k=0}{i-1} \singleton{\aarray+k}{\alpha_k}
             \sepcon \bbigsepcon{k=i}{n-1} \singleton{\aarray+k}{\alpha_{\pi(k)}} \big)
             \notag \\
   &\quad +   \big( \sum\limits_{\substack{\pi \in \perm{i}{n-1} \\ \pi (i) \neq i}}   \bbigsepcon{k=0}{i-1} \singleton{\aarray+k}{\alpha_k}
   \sepcon  
              \bbigsepcon{\substack{k=i}}{n-1} \singleton{\aarray+k}{\alpha_{\pi(k)}} \big) \Big) 
             \notag \\
   \eeqtag{$\sum\limits_{\substack{\pi \in \perm{i}{n-1} \\ \pi(i) =i}} X + \sum\limits_{\substack{\pi \in \perm{i}{n-1} \\ \pi(i) \neq i}} X
           = \sum\limits_{\substack{\pi \in \perm{i}{n-1} }} X$}
   &\iverson{0\leq i <n} \cdot \frac{1}{(n-i)!} \cdot \sum\limits_{\pi \in \perm{i}{n-1} }  
    \big( \bbigsepcon{k=0}{i-1} \singleton{\aarray+k}{\alpha_k}
             \sepcon \bbigsepcon{k=i}{n-1} \singleton{\aarray+k}{\alpha_{\pi(k)}} \big)  \\
   \ppreceqtag{$0 \preceq \iverson{\neg (0 \leq i < n)} \cdot \singleton{\aarray}{\alpha_0,\ldots, \alpha_{n-1}}$}
   &\iverson{0\leq i <n} \cdot \frac{1}{(n-i)!} \cdot \sum\limits_{\pi \in \perm{i}{n-1} }  
    \big( \bbigsepcon{k=0}{i-1} \singleton{\aarray+k}{\alpha_k}
             \sepcon \bbigsepcon{k=i}{n-1} \singleton{\aarray+k}{\alpha_{\pi(k)}} \big)  \\
   &\quad + \iverson{\neg (0 \leq i < n)} \cdot \singleton{\aarray}{\alpha_0,\ldots, \alpha_{n-1}} 
   \notag \\
   \eeqtag{Definition of $I$}
   & I.
\end{align}
This completes the proof. \hfill \qed
\paragraph{Auxiliary Results}
Let us first provide the exact implementation of procedure $\ProcName{swap}$:
\begin{align*}
        & \ProcDecl{swap}{\aarray, i, j} {} \\
        & \qquad \ASSIGNH{y}{\aarray + i}\SEMI \\
        & \qquad \ASSIGNH{z}{\aarray + j}\SEMI \\
        & \qquad \HASSIGN{\aarray + i}{z}\SEMI \\
        & \qquad \HASSIGN{\aarray + j}{y} \\
        & \}
\end{align*}
Notice that analyzing procedure $\ProcName{swap}$ amounts to analyzing its body, because the procedure is not recursive. 
\begin{lemma}\label{lem:swap-i-unequal-j}
Let $\ff \in \E$ \CHANGED{such that $y,z \notin \Vars(\ff)$}. We have
\begin{align*}
  &\wp{\ProcCall{swap}{\aarray, i , j}{}}{\ff \sepcon \singleton{\aarray+i}{\alpha} \sepcon \singleton{\aarray+j}{\beta}} \\
  {}\eeq{}~& \ff \sepcon \singleton{\aarray+j}{\alpha} \sepcon \singleton{\aarray+i}{\beta}.
\end{align*}
\end{lemma}
\begin{proof}
   Let $\cc$ denote the body of procedure $\ProcName{swap}$.
   Using the rules depicted in Table \ref{table:wp} and the alternative version of the rule for heap lookup, we compute
   \begin{align}
      &\wp{\ProcCall{swap}{\aarray, i , j}{}}{\ff \sepcon \singleton{\aarray+i}{\alpha} \sepcon \singleton{\aarray+j}{\beta}}
        \label{eqn:swap-wp-1} \\
      \eeqtag{$\ProcName{swap}$ is not recursive}
      &\wp{\cc}{\ff \sepcon \singleton{\aarray+i}{\alpha} \sepcon \singleton{\aarray+j}{\beta}} \\
      \eeqtag{Table \ref{table:wp}}
      &\displaystyle\sup_{v_1 \in \Ints} \containsPointer{\aarray +i}{v_1} 
        \cdot \displaystyle\sup_{v_2 \in \Ints} \containsPointer{\aarray +j}{v_2} \\
      &\quad \cdot \validpointer{\aarray +i} \sepcon \big( \singleton{\aarray +i}{v_2} \notag \\
      &\quad \sepimp
        \big( 
        \validpointer{\aarray +j} \sepcon \big( \singleton{\aarray +j}{v_1} \notag \\
      &\quad  \sepimp
        \big(
        \ff \sepcon \singleton{\aarray+i}{\alpha} \sepcon \singleton{\aarray+j}{\beta}
        \big)        
        \big)
        \big)
        \big) \notag
   \end{align}
Now let $(\sk,\hh) \in \States$. We distinguish the cases $s(\aarray +i) \not\in \dom{\hh} \vee \sk (\aarray +j) \not\in \dom{\hh}$ and $\sk (\aarray +i) \in \dom{\hh} \wedge \sk(\aarray +j) \in \dom{\hh}$. For the first case, we have
\begin{align}
   &\Big( \displaystyle\sup_{v_1 \in \Ints} \containsPointer{\aarray +i}{v_1} 
        \cdot \displaystyle\sup_{v_2 \in \Ints} \containsPointer{\aarray +j}{v_2} \\
      &\quad \cdot \validpointer{\aarray +i} \sepcon \big( \singleton{\aarray +i}{v_2} \notag \\
      &\quad \sepimp
        \big( 
        \validpointer{\aarray +j} \sepcon \big( \singleton{\aarray +j}{v_1} \notag \\
      &\quad  \sepimp
        \big(
        \ff \sepcon \singleton{\aarray+i}{\alpha} \sepcon \singleton{\aarray+j}{\beta}
        \big)        
        \big)
        \big)
        \big) \Big) (\sk, \hh ) \notag \\
    \eeqtag{By assumption: $\displaystyle\sup_{v_1 \in \Ints} \containsPointer{\aarray +i}{v_1}(\sk ,\hh ) = 0$ or $\displaystyle\sup_{v_2 \in \Ints} \containsPointer{\aarray +j}{v_2}(\sk ,\hh ) = 0$}
    & 0 \\
    \eeqtag{By assumption: $\singleton{\aarray+i}{\alpha}(\sk, \hh_1) = 0$ or $\singleton{\aarray+j}{\alpha}(s,h_1) = 0$ for $\hh_1 \sepcon \hh_2 = \hh$}
    &\big( \ff \sepcon \singleton{\aarray+i}{\alpha} \sepcon \singleton{\aarray+j}{\beta} \big) (\sk ,\hh ).
\end{align}
For the second case, suppose w.l.o.g.\ that $\hh(\sk(\aarray +i)) = v_i$ and $\hh(\sk(\aarray +j)) = v_j$.
The heap $h$ is thus of the form $\hh = \hh' \sepcon \HeapSet{\sk(\aarray +i) \mapsto v_i} \sepcon \HeapSet{\sk(\aarray +j) \mapsto v_j}$ 
for some heap $\hh'$. This yields
\begin{align}
   &\Big( \displaystyle\sup_{v_1 \in \Ints} \containsPointer{\aarray +i}{v_1} 
        \cdot \displaystyle\sup_{v_2 \in \Ints} \containsPointer{\aarray +j}{v_2} \\
      &\quad \cdot \validpointer{\aarray +i} \sepcon \big( \singleton{\aarray +i}{v_2} \notag \\
      &\quad \sepimp
        \big( 
        \validpointer{\aarray +j} \sepcon \big( \singleton{\aarray +j}{v_1} \notag \\
      &\quad  \sepimp
        \big(
        \ff \sepcon \singleton{\aarray+i}{\alpha} \sepcon \singleton{\aarray+j}{\beta}
        \big)        
        \big)
        \big)
        \big) \Big) (\sk, \hh ) \notag \\
   \eeqtag{By assumption: $\hh = \hh' \sepcon \HeapSet{\sk(\aarray +i) \mapsto v_i} \sepcon \HeapSet{\sk(\aarray +j) \mapsto v_j}$}
   &\Big( \displaystyle\sup_{v_1 \in \Ints} \containsPointer{\aarray +i}{v_1} 
        \cdot \displaystyle\sup_{v_2 \in \Ints} \containsPointer{\aarray +j}{v_2} \\
      &\quad \cdot \validpointer{\aarray +i} \sepcon \big( \singleton{\aarray +i}{v_2} \notag \\
      &\quad \sepimp
        \big( 
        \validpointer{\aarray +j} \sepcon \big( \singleton{\aarray +j}{v_1} \notag \\
      &\quad  \sepimp
        \big(
        \ff \sepcon \singleton{\aarray+i}{\alpha} \sepcon \singleton{\aarray+j}{\beta}
        \big)        
        \big)
        \big)
        \big) \Big) \notag \\
      & \quad (\sk, \hh' \sepcon \HeapSet{\sk(\aarray +i) \mapsto v_i} \sepcon \HeapSet{\sk(\aarray +j) \mapsto v_j} ) \notag \\
   \eeqtag{$\displaystyle\sup_{v_1 \in \Ints} \containsPointer{\aarray +i}{v_1}(\sk, \hh) = \containsPointer{\aarray +i}{v_i}(\sk, \hh) = 1$}
   &\Big(  \displaystyle\sup_{v_2 \in \Ints} \containsPointer{\aarray +j}{v_2} \\
      &\quad \cdot \validpointer{\aarray +i} \sepcon \big( \singleton{\aarray +i}{v_2} \notag \\
      &\quad \sepimp
        \big( 
        \validpointer{\aarray +j} \sepcon \big( \singleton{\aarray +j}{v_i} \notag \\
      &\quad  \sepimp
        \big(
        \ff \sepcon \singleton{\aarray+i}{\alpha} \sepcon \singleton{\aarray+j}{\beta}
        \big)        
        \big)
        \big)
        \big) \Big) \notag \\
      & \quad (\sk, \hh' \sepcon \HeapSet{\sk(\aarray +i) \mapsto v_i} \sepcon \HeapSet{\sk(\aarray +j) \mapsto v_j} ) \notag \\
   \eeqtag{$\displaystyle\sup_{v_2 \in \Ints} \containsPointer{\aarray +j}{v_2}(\sk, \hh) = \containsPointer{\aarray +j}{v_j}(\sk, \hh) = 1$}
   &\Big( \validpointer{\aarray +i} \sepcon \big( \singleton{\aarray +i}{v_j}  \\
      &\quad \sepimp
        \big( 
        \validpointer{\aarray +j} \sepcon \big( \singleton{\aarray +j}{v_i} \notag \\
      &\quad  \sepimp
        \big(
        \ff \sepcon \singleton{\aarray+i}{\alpha} \sepcon \singleton{\aarray+j}{\beta}
        \big)        
        \big)
        \big)
        \big) \Big) \notag \\
      & \quad (\sk, \hh' \sepcon \HeapSet{\sk(\aarray +i) \mapsto v_i} \sepcon \HeapSet{\sk(\aarray +j) \mapsto v_j} ) \notag \\
   \eeqtag{$\validpointer{\aarray +i}(\sk, \HeapSet{\sk(\aarray +i) \mapsto v_i}) = 1$}
   &\Big(  \singleton{\aarray +i}{v_j}  \\
      &\quad \sepimp
        \big( 
        \validpointer{\aarray +j} \sepcon \big( \singleton{\aarray +j}{v_i} \notag \\
      &\quad  \sepimp
        \big(
        \ff \sepcon \singleton{\aarray+i}{\alpha} \sepcon \singleton{\aarray+j}{\beta}
        \big)        
        \big)
        \big) \Big) \notag \\
      & \quad (\sk, \hh'   \sepcon \HeapSet{\sk(\aarray +j) \mapsto v_j} ) \notag \\
   \eeqtag{$s(\aarray +i ) \not\in \dom{\hh' \sepcon \HeapSet{\sk(\aarray +j) \mapsto v_j}}$}
   &\Big( 
        \validpointer{\aarray +j} \sepcon \big( \singleton{\aarray +j}{v_i}  \\
      &\quad  \sepimp
        \big(
        \ff \sepcon \singleton{\aarray+i}{\alpha} \sepcon \singleton{\aarray+j}{\beta}
        \big)        
        \big) \Big) \notag \\
      & \quad (\sk, \hh' \sepcon  \HeapSet{\sk(\aarray +i) \mapsto v_j} \sepcon \HeapSet{\sk(\aarray +j) \mapsto v_j} ) \notag \\
   \eeqtag{$\validpointer{\aarray +j}(\sk, \HeapSet{\sk(\aarray +j) \mapsto v_j}) = 1$}
   &\Big( 
       \singleton{\aarray +j}{v_i}  \\
      &\quad  \sepimp
        \big(
        \ff \sepcon \singleton{\aarray+i}{\alpha} \sepcon \singleton{\aarray+j}{\beta}   
        \big) \Big) \notag \\
      & \quad (\sk, \hh' \sepcon  \HeapSet{\sk(\aarray +i) \mapsto v_j} ) \notag \\
   \eeqtag{$s(\aarray +j ) \not\in \dom{\hh' \sepcon \HeapSet{\sk(\aarray +i) \mapsto v_j}}$}
   &\Big( 
        \ff \sepcon \singleton{\aarray+i}{\alpha} \sepcon \singleton{\aarray+j}{\beta}   
         \Big) \notag \\
      & \quad (\sk, \hh' \sepcon  \HeapSet{\sk(\aarray +i) \mapsto v_j} \sepcon \HeapSet{\sk(\aarray +j) \mapsto v_i} ). \notag 
\end{align}
Now, if $v_j \neq \alpha$ or $v_i \neq \beta$, then clearly
\begin{align}
   &\Big( 
        \ff \sepcon \singleton{\aarray+i}{\alpha} \sepcon \singleton{\aarray+j}{\beta}   
         \Big)  \\
      & \quad (\sk, \hh' \sepcon  \HeapSet{\sk(\aarray +i) \mapsto v_j} \sepcon \HeapSet{\sk(\aarray +j) \mapsto v_i} ) \notag \\
   \eeqtag{$v_j \neq \alpha$ or $v_i \neq \beta$}
   0 \\
   \eeqtag{$v_j \neq \alpha$ or $v_i \neq \beta$}
   &\Big( 
        \ff \sepcon \singleton{\aarray+j}{\alpha} \sepcon \singleton{\aarray+i}{\beta}   
         \Big)  \\
      & \quad (\sk, \hh' \sepcon  \HeapSet{\sk(\aarray +i) \mapsto v_i} \sepcon \HeapSet{\sk(\aarray +j) \mapsto v_j} ). \notag \\
   \eeqtag{By assumption: $\hh = \hh' \sepcon  \HeapSet{\sk(\aarray +i) \mapsto v_i} \sepcon \HeapSet{\sk(\aarray +j) \mapsto v_j}$}
   &\Big( 
        \ff \sepcon \singleton{\aarray+j}{\alpha} \sepcon \singleton{\aarray+i}{\beta}   
         \Big)   (\sk, \hh ). 
\end{align}
Otherwise, i.e.\ if $v_j = \alpha$ and $v_i = \beta$, then

\begin{align}
   &\Big( 
        \ff \sepcon \singleton{\aarray+i}{\alpha} \sepcon \singleton{\aarray+j}{\beta}   
         \Big)  \\
      & \quad(\sk, \hh' \sepcon  \HeapSet{\sk(\aarray +i) \mapsto v_j} \sepcon \HeapSet{\sk(\aarray +j) \mapsto v_i} ) \notag \\
   \eeqtag{$\singleton{\aarray+i}{\alpha}(\sk,\HeapSet{\sk(\aarray +i) \mapsto v_j}) = 1$, if $v_j = \alpha$}
   &\Big( 
        \ff \sepcon \singleton{\aarray+j}{\beta}   
         \Big)  \\
      & \quad (\sk, \hh' \sepcon \HeapSet{\sk(\aarray +j) \mapsto v_i} ) \notag \\
   \eeqtag{$\singleton{\aarray+j}{\beta}(\sk,\HeapSet{\sk(\aarray +j) \mapsto v_i}) = 1$, if $v_i = \beta$}
   &
        \ff    
      (\sk, \hh'  )  \\
   \eeqtag{$\singleton{\aarray+j}{\alpha}(\sk,\HeapSet{\sk(\aarray +j) \mapsto v_j}) = 1$, if $v_j = \alpha$}
   &\Big( 
        \ff  \sepcon \singleton{\aarray+j}{\alpha}
         \Big)  \\
      &  \quad (\sk, \hh' \sepcon \HeapSet{\sk(\aarray +j) \mapsto v_j}  ) \notag \\
   \eeqtag{$\singleton{\aarray+i}{\beta}(\sk,\HeapSet{\sk(\aarray +i) \mapsto v_i}) = 1$, if $v_i = \beta$}
   &\Big( 
        \ff  \sepcon \singleton{\aarray+j}{\alpha} \sepcon \singleton{\aarray+i}{\beta}
         \Big)  \\
      &\quad (\sk, \hh' \sepcon \HeapSet{\sk(\aarray +i) \mapsto v_i} \sepcon \HeapSet{\sk(\aarray +j) \mapsto v_j}  ). \notag \\
   \eeqtag{By assumption: $\hh = \hh' \sepcon  \HeapSet{\sk(\aarray +i) \mapsto v_i} \sepcon \HeapSet{\sk(\aarray +j) \mapsto v_j}$}
   &\Big( 
        \ff \sepcon \singleton{\aarray+j}{\alpha} \sepcon \singleton{\aarray+i}{\beta}   
         \Big)  \quad (\sk, \hh ). 
\end{align}
This completes the proof.
\end{proof}
\begin{lemma}
  \label{lem:swap-i-equal-j}
   Let $\ff \in \E$ \CHANGED{such that $y,z \notin \Vars(\ff)$}. We have
   \begin{align*}
      &\wp{\ProcCall{swap}{\aarray, i , j}{}}{\iverson{i=j} \cdot \big(\ff \sepcon \singleton{\aarray+i}{\alpha} \big)} \\
   {}\eeq{}~& \iverson{i=j} \cdot \big(\ff \sepcon \singleton{\aarray+i}{\alpha} \big)
   \end{align*}
\end{lemma}
\begin{proof}
      Let $\cc$ denote the body of procedure $\ProcName{swap}$.
   Using the rules depicted in Table \ref{table:wp} and the alternative version of the rule for heap lookup, we compute
   \begin{align}
      &\wp{\ProcCall{swap}{\aarray, i , j}{}}{\iverson{i=j} \cdot \big(\ff \sepcon \singleton{\aarray+i}{\alpha} \big)}
        \label{eqn:swap-wp-1} \\
      \eeqtag{$\ProcName{swap}$ is not recursive}
      &\wp{\cc}{\iverson{i=j} \cdot \big(\ff \sepcon \singleton{\aarray+i}{\alpha} \big)} \\
      \eeqtag{Table \ref{table:wp}}
      &\displaystyle\sup_{v_1 \in \Ints} \containsPointer{\aarray +i}{v_1} 
        \cdot \displaystyle\sup_{v_2 \in \Ints} \containsPointer{\aarray +j}{v_2} \\
      &\quad \cdot \validpointer{\aarray +i} \sepcon \big( \singleton{\aarray +i}{v_2} \notag \\
      &\quad \sepimp
        \big( 
        \validpointer{\aarray +j} \sepcon \big( \singleton{\aarray +j}{v_1} \notag \\
      &\quad  \sepimp
        \big(
        \iverson{i=j} \cdot \big(\ff \sepcon \singleton{\aarray+i}{\alpha} \big)
        \big)        
        \big)
        \big)
        \big) \notag
   \end{align}
Now let $(\sk,\hh) \in \States$. We distinguish the cases $s(\aarray +i) \not\in \dom{\hh}$ and $\sk (\aarray +i) \in \dom{\hh}$. For the first case, we have
\begin{align}
   &\Big( \displaystyle\sup_{v_1 \in \Ints} \containsPointer{\aarray +i}{v_1} 
        \cdot \displaystyle\sup_{v_2 \in \Ints} \containsPointer{\aarray +j}{v_2} \\
      &\quad \cdot \validpointer{\aarray +i} \sepcon \big( \singleton{\aarray +i}{v_2} \notag \\
      &\quad \sepimp
        \big( 
        \validpointer{\aarray +j} \sepcon \big( \singleton{\aarray +j}{v_1} \notag \\
      &\quad  \sepimp
        \big(
        \iverson{i=j} \cdot \big(\ff \sepcon \singleton{\aarray+i}{\alpha} \big)
        \big)        
        \big)
        \big)
        \big) \Big) (\sk, \hh ) \notag \\
    \eeqtag{By assumption: $\displaystyle\sup_{v_1 \in \Ints} \containsPointer{\aarray +i}{v_1}(\sk ,\hh ) = 0$}
    & 0 \\
    \eeqtag{By assumption: $\singleton{\aarray+i}{\alpha}(\sk, \hh_1) = 0$ for all $\hh_1 \sepcon \hh_2 = h$}
    &\big( \iverson{i=j} \cdot \big(\ff \sepcon \singleton{\aarray+i}{\alpha} \big) \big) (\sk ,\hh ).
\end{align}
For the second case, i.e.\ $\sk (\aarray +i) \in \dom{\hh}$, suppose w.l.o.g.\ that $\hh(\sk(\aarray +i)) = v_i$.
The heap $h$ is thus of the form $h = h' \sepcon \HeapSet{\sk(\aarray +i) \mapsto v_i}$ for some heap $h'$.
Again, we distinguish two cases: $\iverson{i = j}(s,h)=0$ and $\iverson{i = j}(s,h)=1$.
If $\iverson{i = j}(s,h)=0$, then either $\sk (\aarray +j) \not\in \dom{\hh}$ or $\sk (\aarray +j) \in \dom{\hh}$.
For $\sk (\aarray +j) \not\in \dom{\hh}$, we have
\begin{align}
   &\Big( \displaystyle\sup_{v_1 \in \Ints} \containsPointer{\aarray +i}{v_1} 
        \cdot \displaystyle\sup_{v_2 \in \Ints} \containsPointer{\aarray +j}{v_2} \\
      &\quad \cdot \validpointer{\aarray +i} \sepcon \big( \singleton{\aarray +i}{v_2} \notag \\
      &\quad \sepimp
        \big( 
        \validpointer{\aarray +j} \sepcon \big( \singleton{\aarray +j}{v_1} \notag \\
      &\quad  \sepimp
        \big(
        \iverson{i=j} \cdot \big(\ff \sepcon \singleton{\aarray+i}{\alpha} \big)
        \big)        
        \big)
        \big)
        \big) \Big) (\sk, \hh ) \notag \\
   \eeqtag{$\displaystyle\sup_{v_1 \in \Ints} \containsPointer{\aarray +i}{v_1}(\sk, \hh ) = \containsPointer{\aarray +i}{v_i} = 1$}
   &\Big(  \displaystyle\sup_{v_2 \in \Ints} \containsPointer{\aarray +j}{v_2} \\
      &\quad \cdot \validpointer{\aarray +i} \sepcon \big( \singleton{\aarray +i}{v_2} \notag \\
      &\quad \sepimp
        \big( 
        \validpointer{\aarray +j} \sepcon \big( \singleton{\aarray +j}{v_i} \notag \\
      &\quad  \sepimp
        \big(
        \iverson{i=j} \cdot \big(\ff \sepcon \singleton{\aarray+i}{\alpha} \big)
        \big)        
        \big)
        \big)
        \big) \Big) (\sk, \hh ) \notag \\
   \eeqtag{By assumption: $\sk (\aarray +j) \not\in \dom{\hh}$}
   &0 \\
   \eeqtag{By assumption: $\iverson{i = j}(s,h)=0$}
   &\big( \iverson{i=j} \cdot \big(\ff \sepcon \singleton{\aarray+i}{\alpha} \big) \big) (\sk ,\hh ).
\end{align}
For $\sk (\aarray +j) \in \dom{\hh}$, suppose w.l.o.g.\ that $\hh (\sk (\aarray +j) ) = v_j$, which implies that the heap 
$\hh$ is of the form $\hh = \hh'' \sepcon \HeapSet{\sk(\aarray +i) \mapsto v_i} \sepcon \HeapSet{\sk(\aarray +j) \mapsto v_j}$ for some heap $\hh''$.
This yields
\begin{align}
 &\Big( \displaystyle\sup_{v_1 \in \Ints} \containsPointer{\aarray +i}{v_1} 
        \cdot \displaystyle\sup_{v_2 \in \Ints} \containsPointer{\aarray +j}{v_2} \\
      &\quad \cdot \validpointer{\aarray +i} \sepcon \big( \singleton{\aarray +i}{v_2} \notag \\
      &\quad \sepimp
        \big( 
        \validpointer{\aarray +j} \sepcon \big( \singleton{\aarray +j}{v_1} \notag \\
      &\quad  \sepimp
        \big(
         \iverson{i=j} \cdot \big(\ff \sepcon \singleton{\aarray+i}{\alpha}  \big)
        \big)        
        \big)
        \big)
        \big) \Big) (\sk, \hh ) \notag \\
   \eeqtag{By assumption: $\hh = \hh'' \sepcon \HeapSet{\sk(\aarray +i) \mapsto v_i} \sepcon \HeapSet{\sk(\aarray +j) \mapsto v_j}$}
   &\Big( \displaystyle\sup_{v_1 \in \Ints} \containsPointer{\aarray +i}{v_1} 
        \cdot \displaystyle\sup_{v_2 \in \Ints} \containsPointer{\aarray +j}{v_2} \\
      &\quad \cdot \validpointer{\aarray +i} \sepcon \big( \singleton{\aarray +i}{v_2} \notag \\
      &\quad \sepimp
        \big( 
        \validpointer{\aarray +j} \sepcon \big( \singleton{\aarray +j}{v_1} \notag \\
      &\quad  \sepimp
        \big(
         \iverson{i=j} \cdot \big(\ff \sepcon \singleton{\aarray+i}{\alpha} \big)
        \big)        
        \big)
        \big)
        \big) \Big) \notag \\
      & \quad (\sk, \hh'' \sepcon \HeapSet{\sk(\aarray +i) \mapsto v_i} \sepcon \HeapSet{\sk(\aarray +j) \mapsto v_j} ) \notag \\
   \eeqtag{$\displaystyle\sup_{v_1 \in \Ints} \containsPointer{\aarray +i}{v_1}(\sk, \hh) = \containsPointer{\aarray +i}{v_i}(\sk, \hh) = 1$}
   &\Big(  \displaystyle\sup_{v_2 \in \Ints} \containsPointer{\aarray +j}{v_2} \\
      &\quad \cdot \validpointer{\aarray +i} \sepcon \big( \singleton{\aarray +i}{v_2} \notag \\
      &\quad \sepimp
        \big( 
        \validpointer{\aarray +j} \sepcon \big( \singleton{\aarray +j}{v_i} \notag \\
      &\quad  \sepimp
        \big(
         \iverson{i=j} \cdot \big(\ff \sepcon \singleton{\aarray+i}{\alpha}  \big)
        \big)        
        \big)
        \big)
        \big) \Big) \notag \\
      & \quad (\sk, \hh'' \sepcon \HeapSet{\sk(\aarray +i) \mapsto v_i} \sepcon \HeapSet{\sk(\aarray +j) \mapsto v_j} ) \notag \\
   \eeqtag{$\displaystyle\sup_{v_2 \in \Ints} \containsPointer{\aarray +j}{v_2}(\sk, \hh) = \containsPointer{\aarray +j}{v_j}(\sk, \hh) = 1$}
   &\Big( \validpointer{\aarray +i} \sepcon \big( \singleton{\aarray +i}{v_j}  \\
      &\quad \sepimp
        \big( 
        \validpointer{\aarray +j} \sepcon \big( \singleton{\aarray +j}{v_i} \notag \\
      &\quad  \sepimp
        \big(
         \iverson{i=j} \cdot \big(\ff \sepcon \singleton{\aarray+i}{\alpha}  \big)
        \big)        
        \big)
        \big)
        \big) \Big) \notag \\
      & \quad (\sk, \hh'' \sepcon \HeapSet{\sk(\aarray +i) \mapsto v_i} \sepcon \HeapSet{\sk(\aarray +j) \mapsto v_j} ) \notag \\
   \eeqtag{$\validpointer{\aarray +i}(\sk, \HeapSet{\sk(\aarray +i) \mapsto v_i}) = 1$}
   &\Big(  \singleton{\aarray +i}{v_j}  \\
      &\quad \sepimp
        \big( 
        \validpointer{\aarray +j} \sepcon \big( \singleton{\aarray +j}{v_i} \notag \\
      &\quad \sepimp
        \big(
         \iverson{i=j} \cdot \big(\ff \sepcon \singleton{\aarray+i}{\alpha}  \big)
        \big)        
        \big)
        \big) \Big) \notag \\
      & \quad (\sk, \hh''   \sepcon \HeapSet{\sk(\aarray +j) \mapsto v_j} ) \notag \\
   \eeqtag{$s(\aarray +i ) \not\in \dom{\hh' \sepcon \HeapSet{\sk(\aarray +j) \mapsto v_j}}$}
   &\Big( 
        \validpointer{\aarray +j} \sepcon \big( \singleton{\aarray +j}{v_i}  \\
      &\quad  \sepimp
        \big(
         \iverson{i=j} \cdot \big(\ff \sepcon \singleton{\aarray+i}{\alpha}  \big)
        \big)        
        \big) \Big) \notag \\
      & \quad (\sk, \hh'' \sepcon  \HeapSet{\sk(\aarray +i) \mapsto v_j} \sepcon \HeapSet{\sk(\aarray +j) \mapsto v_j} ) \notag \\
   \eeqtag{$\validpointer{\aarray +j}(\sk, \HeapSet{\sk(\aarray +j) \mapsto v_j}) = 1$}
   &\Big( 
       \singleton{\aarray +j}{v_i}  \\
      &\quad  \sepimp
        \big(
         \iverson{i=j} \cdot \big(\ff \sepcon \singleton{\aarray+i}{\alpha} \big) 
        \big) \Big) \notag \\
      & \quad (\sk, \hh'' \sepcon  \HeapSet{\sk(\aarray +i) \mapsto v_j} ) \notag \\
   \eeqtag{$s(\aarray +j ) \not\in \dom{\hh' \sepcon \HeapSet{\sk(\aarray +i) \mapsto v_j}}$}
   &\Big( 
        \iverson{i=j} \cdot \big(\ff \sepcon \singleton{\aarray+i}{\alpha}  \big)  
         \Big) \notag \\
      & \quad (\sk, \hh'' \sepcon  \HeapSet{\sk(\aarray +i) \mapsto v_j} \sepcon \HeapSet{\sk(\aarray +j) \mapsto v_i} ) \\
   \eeqtag{By assumption: $\iverson{i=j}(\sk, \hh'' \sepcon  \HeapSet{\sk(\aarray +i) \mapsto v_j} \sepcon \HeapSet{\sk(\aarray +j) \mapsto v_i} ) = 0$}
   & 0 \\
   \eeqtag{By assumption: $\iverson{i=j}(\sk, \hh) = 0$}
   &\Big( \iverson{i=j} \cdot \big(\ff \sepcon \singleton{\aarray+i}{\alpha} \big) \Big) (\sk, \hh).
\end{align}
Finally, if $\iverson{i =j}(\sk, \hh) =1$, we get
\begin{align}
 &\Big( \displaystyle\sup_{v_1 \in \Ints} \containsPointer{\aarray +i}{v_1} 
        \cdot \displaystyle\sup_{v_2 \in \Ints} \containsPointer{\aarray +j}{v_2} \\
      &\quad \cdot \validpointer{\aarray +i} \sepcon \big( \singleton{\aarray +i}{v_2} \notag \\
      &\quad \sepimp
        \big( 
        \validpointer{\aarray +j} \sepcon \big( \singleton{\aarray +j}{v_1} \notag \\
      &\quad  \sepimp
        \big(
         \iverson{i=j} \cdot \big(\ff \sepcon \singleton{\aarray+i}{\alpha}  \big)
        \big)        
        \big)
        \big)
        \big) \Big) (\sk, \hh ) \notag \\
   \eeqtag{By assumption: $\hh = \hh' \sepcon \HeapSet{\sk(\aarray +i) \mapsto v_i}$}
   &\Big( \displaystyle\sup_{v_1 \in \Ints} \containsPointer{\aarray +i}{v_1} 
        \cdot \displaystyle\sup_{v_2 \in \Ints} \containsPointer{\aarray +j}{v_2} \\
      &\quad \cdot \validpointer{\aarray +i} \sepcon \big( \singleton{\aarray +i}{v_2} \notag \\
      &\quad \sepimp
        \big( 
        \validpointer{\aarray +j} \sepcon \big( \singleton{\aarray +j}{v_1} \notag \\
      &\quad  \sepimp
        \big(
         \iverson{i=j} \cdot \big(\ff \sepcon \singleton{\aarray+i}{\alpha} \big)
        \big)        
        \big)
        \big)
        \big) \Big) \notag \\
      & \quad (\sk, \hh' \sepcon \HeapSet{\sk(\aarray +i) \mapsto v_i}  ) \notag \\
   \eeqtag{$\displaystyle\sup_{v_1 \in \Ints} \containsPointer{\aarray +i}{v_1}(\sk, \hh) = \containsPointer{\aarray +i}{v_i}(\sk, \hh) = 1$}
   &\Big(  \displaystyle\sup_{v_2 \in \Ints} \containsPointer{\aarray +j}{v_2} \\
      &\quad \cdot \validpointer{\aarray +i} \sepcon \big( \singleton{\aarray +i}{v_2} \notag \\
      &\quad \sepimp
        \big( 
        \validpointer{\aarray +j} \sepcon \big( \singleton{\aarray +j}{v_i} \notag \\
      &\quad  \sepimp
        \big(
         \iverson{i=j} \cdot \big(\ff \sepcon \singleton{\aarray+i}{\alpha}  \big)
        \big)        
        \big)
        \big)
        \big) \Big) \notag \\
      & \quad (\sk, \hh' \sepcon \HeapSet{\sk(\aarray +i) \mapsto v_i} ) \notag \\
   \eeqtag{Since $s(i) = s(j)$: $\displaystyle\sup_{v_2 \in \Ints} \containsPointer{\aarray +j}{v_2} (\sk, \hh) =  \displaystyle\sup_{v_1 \in \Ints} \containsPointer{\aarray +i}{v_1}(\sk, \hh)$}
   &\Big( \validpointer{\aarray +i} \sepcon \big( \singleton{\aarray +i}{v_i}  \\
      &\quad \sepimp
        \big( 
        \validpointer{\aarray +j} \sepcon \big( \singleton{\aarray +j}{v_i} \notag \\
      &\quad  \sepimp
        \big(
         \iverson{i=j} \cdot \big(\ff \sepcon \singleton{\aarray+i}{\alpha}  \big)
        \big)        
        \big)
        \big)
        \big) \Big) \notag \\
      & \quad (\sk, \hh' \sepcon \HeapSet{\sk(\aarray +i) \mapsto v_i}  ) \notag \\
   \eeqtag{$\validpointer{\aarray +i}(\sk, \HeapSet{\sk(\aarray +i) \mapsto v_i}) = 1$}
   &\Big(  \singleton{\aarray +i}{v_i}  \\
      &\quad \sepimp
        \big( 
        \validpointer{\aarray +j} \sepcon \big( \singleton{\aarray +j}{v_i} \notag \\
      &\quad \sepimp
        \big(
         \iverson{i=j} \cdot \big(\ff \sepcon \singleton{\aarray+i}{\alpha}  \big)
        \big)        
        \big)
        \big) \Big) \notag \\
      & \quad (\sk, \hh'  ) \notag \\
   \eeqtag{$s(\aarray +i ) \not\in \dom{\hh'}$}
   &\Big( 
        \validpointer{\aarray +j} \sepcon \big( \singleton{\aarray +j}{v_i}  \\
      &\quad  \sepimp
        \big(
         \iverson{i=j} \cdot \big(\ff \sepcon \singleton{\aarray+i}{\alpha}  \big)
        \big)        
        \big) \Big) \notag \\
      & \quad (\sk, \hh' \sepcon  \HeapSet{\sk(\aarray +i) \mapsto v_i}  ) \notag \\
   \eeqtag{Since $s(i) = s(j)$: $\validpointer{\aarray +j}(\sk, \HeapSet{\sk(\aarray +i) \mapsto v_i}) = 1$}
   &\Big( 
       \singleton{\aarray +j}{v_i}  \\
      &\quad  \sepimp
        \big(
         \iverson{i=j} \cdot \big(\ff \sepcon \singleton{\aarray+i}{\alpha} \big) 
        \big) \Big) \notag \\
      & \quad (\sk, \hh'  ) \notag \\
   \eeqtag{Since $s(i) = s(j)$: $s(\aarray +j ) \not\in \dom{\hh'}$}
   &\Big( 
        \iverson{i=j} \cdot \big(\ff \sepcon \singleton{\aarray+i}{\alpha}  \big)  
         \Big) \notag \\
      & \quad (\sk, \hh' \sepcon  \HeapSet{\sk(\aarray +i) \mapsto v_i}) \\
   \eeqtag{By assumption: $h = \hh' \sepcon  \HeapSet{\sk(\aarray +i) \mapsto v_i}$}
   &\Big( 
         \iverson{i=j} \cdot \big(\ff \sepcon \singleton{\aarray+i}{\alpha}  \big)  
         \Big) \notag  (\sk, \hh ) 
\end{align}
This completes the proof.
\end{proof}

\newpage
\section{Additional Simple Inference Rules}\label{app:sec:misc}

This section collects a few rather straightforward facts to compute with expectations in \QSL.

\begin{lemma}\label{thm:misc:single-sepimp:mult}
    $\singleton{\ee}{\ee'} \sepimp (\ff \cdot \fg) \eeq (\singleton{\ee}{\ee'} \sepimp \ff) \cdot (\singleton{\ee}{\ee'} \sepimp \fg)$.
\end{lemma}

\begin{proof}
\begin{align}
        & \singleton{\ee}{\ee'} \sepimp (\ff \cdot \fg) \\
        \eeqtag{Definition of $\sepimp$}
        & \lambda(\sk,\hh)\mydot \inf_{\hh'} \{ (\ff \cdot \fg)(\sk,\hh \sepcon \hh') ~|~ \hh \disjoint \hh', \sk,\hh' \models \singleton{\ee}{\ee'} \} \\
        \eeqtag{algebra} 
        & \lambda(\sk,\hh)\mydot \inf_{\hh'} \{ (\ff(\sk,\hh \sepcon \hh') \cdot \fg(\sk,\hh \sepcon \hh') ~|~ \hh \disjoint \hh', \sk,\hh' \models \singleton{\ee}{\ee'} \} \\
        \eeqtag{$|\{\hh' ~|~ \hh \disjoint \hh', \sk,\hh' \models \singleton{\ee}{\ee'}\}| \leq 1$}
        & \lambda(\sk,\hh)\mydot \inf_{\hh'} \{ (\ff(\sk,\hh \sepcon \hh') ~|~ \hh \disjoint \hh', \sk,\hh' \models \singleton{\ee}{\ee'} \} \\
        & \quad \cdot \inf_{\hh'} \{ (\ff(\sk,\hh \sepcon \hh') \cdot \fg(\sk,\hh \sepcon \hh') ~|~ \hh \disjoint \hh', \sk,\hh' \models \singleton{\ee}{\ee'} \} \notag \\
        \eeqtag{Definition of $\sepimp$}
        & (\singleton{\ee}{\ee'} \sepimp \ff) \cdot (\singleton{\ee}{\ee'} \sepimp \fg).
\end{align}
\end{proof}

\begin{lemma}\label{thm:misc:sepimp-contains}
        $\singleton{\ee}{\ee'} \sepimp \ff \eeq \singleton{\ee}{\ee'} \sepimp \containsPointer{\ee}{\ee'} \cdot \ff$.
\end{lemma}

\begin{proof}
\begin{align}
        & \singleton{\ee}{\ee'} \sepimp \ff \\
        \eeqtag{algebra}
        & \singleton{\ee}{\ee'} \sepimp (1 \cdot \ff) \\
        \eeqtag{Definition of $\sepimp$}
        & \lambda(\sk,\hh)\mydot \inf_{\hh'} \left\{ (1 \cdot \ff)(\sk,\hh \sepcon \hh') ~\middle|~ \hh \disjoint \hh', \sk,\hh' \models \singleton{\ee}{\ee'} \right\} \\
        \eeqtag{algebra}
        & \lambda(\sk,\hh)\mydot \inf_{\hh'} \left\{ (\underbrace{\containsPointer{\ee}{\ee'}}_{\eeq 1} \cdot \ff)(\sk,\hh \sepcon \hh') ~\middle|~ \hh \disjoint \hh', \sk,\hh' \models \singleton{\ee}{\ee'} \right\} \\
        \eeqtag{Definition of $\sepimp$}
        & \singleton{\ee}{\ee'} \sepimp \containsPointer{\ee}{\ee'} \cdot \ff.
\end{align}
\end{proof}

\begin{lemma}\label{thm:misc:sepimp-sepcon}
        $\singleton{\ee}{\ee'} \sepimp (\singleton{\ee}{\ee'} \sepcon \ff) \eeq \containsPointer{\ee}{-} \cdot \infty + (1-\containsPointer{\ee}{-}) \cdot \ff$.
\end{lemma}

\begin{proof}
\begin{align}
        & \singleton{\ee}{\ee'} \sepimp (\singleton{\ee}{\ee'} \sepcon \ff) \\
        \eeqtag{Definition of $\sepimp$}
        & \lambda(\sk,\hh)\mydot \inf_{\hh'} \left\{ (\singleton{\ee}{\ee'} \sepcon \ff)(\sk,\hh \sepcon \hh') ~\middle|~ \hh \disjoint \hh', \sk,\hh' \models \singleton{\ee}{\ee'} \right\} \\
        \eeqtag{Definition of $\sepcon$}
        & \lambda(\sk,\hh)\mydot \inf_{\hh'} \left\{ \max_{\hh_1, \hh_2} \setcomp{ \singleton{\ee}{\ee'}(\sk,\hh_1) \cdot \ff(\sk,\hh_2) }{ \hh \sepcon \hh' = \hh_1 \sepcon \hh_2 } ~\middle|~ \hh \disjoint \hh', \sk,\hh' \models \singleton{\ee}{\ee'} \right\} \\
        \eeqtag{By definition of $\singleton{\ee}{\ee'}$ the $\max$ is attained for $\hh_1 = \hh'$}
        & \lambda(\sk,\hh)\mydot \inf_{\hh'} \underbrace{\left\{ \ff(\sk,\hh) ~\middle|~ \hh \disjoint \hh', \hh' = \singleton{\sk(\ee)}{\sk(\ee')}  \right\}}_{ \eeq M } \\
        \eeqtag{ $M = \emptyset$ iff $\sk(\ee) \in \dom{\hh}$ }
        & \containsPointer{\ee}{-} \cdot \infty + (1-\containsPointer{\ee}{-}) \cdot \ff.
\end{align}
\end{proof}

\begin{lemma}\label{thm:plus-to-max}
    Let $\ff,\fg \in \E$ such that $\ff \cdot \fg = 0$. Then $\ff + \fg = \max\{\ff,\fg\}$.
\end{lemma}

\begin{proof}
    Let $(\sk,\hh) \in \States$. Then
    \begin{align}
              & (\ff \cdot \fg)(\sk,\hh) = 0 \\
    \iimpliestag{Definition of $\cdot$}
              & \ff(\sk,\hh) \cdot \fg(\sk,\hh) = 0 \\
    \iimpliestag{Algebra}
              & \ff(\sk,\hh) = 0 \qor \fg(\sk,\hh) = 0 \\
    \iimpliestag{Addition of both expectations}
              & (\ff+\fg)(\sk,\hh) = \begin{cases} \ff(\sk,\hh) & ~\text{if}~ \fg(\sk,\hh) = 0 \\ \fg(\sk,\hh) & ~\text{if}~ \ff(\sk,\hh) = 0 \end{cases} \\
    \iimpliestag{$0$ is least element of $\E$}
              & (\ff+\fg)(\sk,\hh) = \begin{cases} \ff(\sk,\hh) & ~\text{if}~ \ff(\sk,\hh) \geq \fg(\sk,\hh) \\ \fg(\sk,\hh) & ~\text{if}~ \fg(\sk,\hh) > \ff(\sk,\hh) \end{cases} \\
    \iimpliestag{Definition of $\max$}
              & (\ff+\fg)(\sk,\hh) = \max\{\ff,\fg\}(\sk,\hh)
    \end{align}
    Hence, $\ff + \fg = \max\{\ff,\fg\}$.
\end{proof}

\begin{lemma}\label{thm:single-pointer-wand-plus}
        Let $\ff,\fg \in \E$. Then 
        \begin{align*}
                \singleton{\ee}{\ee'} \sepimp (\ff + \fg) \eeq \singleton{\ee}{\ee'} \sepimp \ff + \singleton{\ee}{\ee'} \sepimp \fg.
        \end{align*}
\end{lemma}
\begin{proof}
\begin{align}
        & \singleton{\ee}{\ee'} \sepimp (\ff + \fg) \\
        \eeqtag{Definition of $\sepimp$}
        & \lambda(\sk,\hh)\mydot \inf_{\hh'} \left\{ \ff(\sk,\hh \sepcon \hh') + \fg(\sk,\hh \sepcon \hh') ~|~ \hh \disjoint \hh', (\sk,\hh) \models \singleton{\ee}{\ee'} \right\} \\
        \eeqtag{there ex. at most one $\hh'$ s.t. $\hh \disjoint \hh'$ and $(\sk,\hh') \models \singleton{\ee}{\ee'}$, namely $\{ \sk(\ee) \mapsto \sk(\ee') \}$}
        & \lambda(\sk,\hh)\mydot \inf_{\hh'} \left\{ \ff(\sk,\hh \sepcon \hh') ~|~ \hh \disjoint \hh', (\sk,\hh) \models \singleton{\ee}{\ee'} \right\} \\
        & \qquad + \inf_{\hh'} \left\{ \fg(\sk,\hh \sepcon \hh') ~|~ \hh \disjoint \hh', (\sk,\hh) \models \singleton{\ee}{\ee'} \right\} \notag \\
        \eeqtag{Definition of $\sepimp$}
        & \singleton{\ee}{\ee'} \sepimp \ff + \singleton{\ee}{\ee'} \sepimp \fg.
\end{align}
\end{proof}

\begin{lemma}\label{thm:single-pointer-wand:sepcon}
    Let $\ff,\fg \in \E$. 
    Moreover, assume $\containsPointer{\ee}{\ee'} \cdot \fg = 0$ or $\containsPointer{\neg\ee}{\ee'} \cdot \ff = 0$. 
    Then
    \begin{align*}
            \singleton{\ee}{\ee'} \sepimp (\ff \sepcon \fg) 
            \eeq \containsPointer{\ee}{\ee'} \cdot \infty + (1-\containsPointer{\ee}{\ee'}) \cdot \left(\left( \singleton{e}{e'} \sepimp \ff\right) \sepcon \fg\right).
    \end{align*}
\end{lemma}
\begin{proof}
\begin{align}
        & \singleton{\ee}{\ee'} \sepimp (\ff \sepcon \fg) \\
        \eeqtag{Definition of $\sepimp$}
        & \lambda(\sk,\hh)\mydot \inf_{\hh'} \setcomp{ (\ff \sepcon \fg)(\sk,\hh \sepcon \hh') }{ \hh \disjoint \hh', (\sk,\hh') \models \singleton{\ee}{\ee'} }  \\
        %
        \eeqtag{Definition of $\sepcon$}
        & \lambda(\sk,\hh)\mydot \inf_{\hh'} \setcomp{ \max_{\hh_1,\hh_2} \setcomp{ \ff(\sk,\hh_1) \cdot \fg(\sk,\hh_2) }{\hh \sepcon \hh' = \hh_1 \sepcon \hh_2} }{ \hh \disjoint \hh', (\sk,\hh') \models \singleton{\ee}{\ee'} }  \\
        %
        \eeqtag{$\hh' = \{ \sk(\ee) \mapsto \sk(\ee') \}$. By assumption the $\max$ is $0$ (if an $\hh'$ exists) or we have $\hh' \subseteq \hh_1$}
        & \lambda(\sk,\hh)\mydot \inf_{\hh'} \setcomp{ \max_{\hh_1,\hh_2} \setcomp{ \ff(\sk,\hh_1 \sepcon \hh') \cdot \fg(\sk,\hh_2) }{\hh = \hh_1 \sepcon \hh_2} }{ \hh \disjoint \hh', (\sk,\hh') \models \singleton{\ee}{\ee'} }  \\
        %
        \eeqtag{there is at most one choice for $\hh'$, namely $\hh' = \{ \sk(\ee) \mapsto \sk(\ee') \}$. Otherwise we get $\infty$.}
        & \containsPointer{\ee}{\ee'} \cdot \infty + (1-\containsPointer{\ee}{\ee'}) \cdot \lambda(\sk,\hh)\mydot \inf_{\hh'} \big\{ \notag \\
        & \qquad \max_{\hh_1,\hh_2} \setcomp{ \ff(\sk,\hh_1 \sepcon \hh') \cdot \fg(\sk,\hh_2) }{\hh = \hh_1 \sepcon \hh_2} \notag \\
        & ~\big|~ \hh \disjoint \hh', (\sk,\hh') \models \singleton{\ee}{\ee'} \big\} \notag \\
        \eeqtag{infimum over singleton set}
        & \containsPointer{\ee}{\ee'} \cdot \infty + (1-\containsPointer{\ee}{\ee'}) \cdot \lambda(\sk,\hh)\mydot \max_{\hh_1,\hh_2} \big\{ \notag \\
        & \qquad \inf_{\hh'} \setcomp{ \ff(\sk,\hh_1 \sepcon \hh') \cdot \fg(\sk,\hh_2) }{ \hh \disjoint \hh', (\sk,\hh') \models \singleton{\ee}{\ee'} } \notag \\
        & ~\big|~ \hh = \hh_1 \sepcon \hh_2 \big\} \notag \\
        \eeqtag{algebra}
        & \containsPointer{\ee}{\ee'} \cdot \infty + (1-\containsPointer{\ee}{\ee'}) \cdot \lambda(\sk,\hh)\mydot \max_{\hh_1,\hh_2} \big\{ \notag \\
        & \qquad \inf_{\hh'} \setcomp{ \ff(\sk,\hh_1 \sepcon \hh') }{ \hh \disjoint \hh', (\sk,\hh') \models \singleton{\ee}{\ee'} } \cdot \fg(\sk,\hh_2) \notag \\
        & ~\big|~ \hh = \hh_1 \sepcon \hh_2 \big\} \notag \\
        \eeqtag{Definition of $\sepimp$}
        & \containsPointer{\ee}{\ee'} \cdot \infty \\
        & + (1-\containsPointer{\ee}{\ee'}) \cdot \lambda(\sk,\hh)\mydot \max_{\hh_1,\hh_2} \big\{ (\singleton{\ee}{\ee'} \sepimp \ff)(\sk,\hh_1) \cdot \fg(\sk,\hh_2) ~\big|~ \hh = \hh_1 \sepcon \hh_2 \big\} \notag \\
        \eeqtag{Definition of $\sepcon$}
        & \containsPointer{\ee}{\ee'} \cdot \infty + (1-\containsPointer{\ee}{\ee'}) \cdot \left(\left( \singleton{e}{e'} \sepimp \ff\right) \sepcon \fg\right).
\end{align}
\end{proof}

\begin{lemma}\label{thm:single-pointer-wand:pure}
  Let $\preda \in \E$ be a pure predicate. Then
  \begin{align*}
          \preda \cdot (\predb \sepimp \ff) \eeq \preda \cdot (\predb \sepimp \ff \cdot \preda)
  \end{align*}
\end{lemma}
\begin{proof}
\begin{align*}
        & \preda \cdot (\predb \sepimp \ff) \\
        \eeqtag{algebra}
        & \lambda(\sk,\hh)\mydot \preda(\sk,\hh) \cdot (\predb \sepimp \ff)(\sk,\hh) \\
        \eeqtag{$\preda$ is a predicate}
        & \lambda(\sk,\hh)\mydot \preda(\sk,\hh) \cdot \preda(\sk,\hh) \cdot (\predb \sepimp \ff)(\sk,\hh) \\
        \eeqtag{Definition of $\sepimp$}
        & \lambda(\sk,\hh)\mydot \preda(\sk,\hh) \cdot \preda(\sk,\hh) \cdot (\inf_{\hh'} \left\{ \ff(\sk,\hh \sepcon \hh') ~|~ \hh \disjoint \hh', (\sk,\hh') \models \predb \right\}) \\
        \eeqtag{$\preda(\sk,\hh)$ is a constant w.r.t. $\inf_{\hh'}$, algebra}
        & \lambda(\sk,\hh)\mydot \preda(\sk,\hh) \cdot (\inf_{\hh'} \left\{ \preda(\sk,\hh) \cdot \ff(\sk,\hh \sepcon \hh') ~|~ \hh \disjoint \hh', (\sk,\hh') \models \predb \right\}) \\
        \eeqtag{$\preda$ is pure}
        & \lambda(\sk,\hh)\mydot \preda(\sk,\hh) \cdot (\inf_{\hh'} \left\{ \preda(\sk,\hh \sepcon \hh') \cdot \ff(\sk,\hh \sepcon \hh') ~|~ \hh \disjoint \hh', (\sk,\hh') \models \predb \right\}) \\
        \eeqtag{Definition of $\sepimp$}
        & \lambda(\sk,\hh)\mydot \preda(\sk,\hh) \cdot (\predb \sepimp (\preda \cdot \ff))(\sk,\hh) \\
        \eeqtag{algebra}
        & \preda \cdot (\predb \sepimp \ff \cdot \preda).
\end{align*}
\end{proof}

\begin{lemma}\label{thm:list-length:sepcon}
  \begin{align*}
          \singleton{\ee}{\ee'} \sepcon \Len{\ee'}{\ee''} \ppreceq \iverson{\ee \neq \ee''} \cdot (\Len{\ee}{\ee''} - \Ls{\ee}{\ee''}).
  \end{align*}
\end{lemma}
\begin{proof}
\begin{align}
        & \singleton{\ee}{\ee'} \sepcon \Len{\ee'}{\ee''} \\
        \eeqtag{Lemma~\ref{thm:ls-props}.\ref{thm:ls-props:char}}
        & \singleton{\ee}{\ee'} \sepcon (\Ls{\ee'}{\ee''} \cdot \heapSize) \\
        \eeqtag{Definition of $\sepcon$}
        & \lambda(\sk,\hh)\mydot \max_{\hh_1,\hh_2} \setcomp{ \singleton{\ee}{\ee'}(\sk,\hh_1) \cdot (\Ls{\ee'}{\ee''} \cdot \heapSize)(\sk,\hh_2)}{\hh = \hh_1 \sepcon \hh_2} \\
        \eeqtag{algebra, Definition of $\heapSize$}
        & \lambda(\sk,\hh)\mydot \max_{\hh_1,\hh_2} \{ \singleton{\ee}{\ee'}(\sk,\hh_1) \cdot (\Ls{\ee'}{\ee''}(\sk,\hh_2) \cdot \underbrace{|\dom{\hh_2}|}_{\eeq |\dom{h}| - |\dom{h_1}|} ~|~ \hh = \hh_1 \sepcon \hh_2 \} \\
        \eeqtag{$\hh_1 = \{ \sk(\ee) \mapsto \sk(\ee') \}$ or $\max_{\hh_1,\hh_2} \ldots = 0$.}
        %
        & \lambda(\sk,\hh)\mydot \max_{\hh_1,\hh_2} \setcomp{\singleton{\ee}{\ee'}(\sk,\hh_1) \cdot \Ls{\ee'}{\ee''}(\sk,\hh_2) \cdot (|\dom{h}|-1)}{\hh = \hh_1 \sepcon \hh_2} \\
        \eeqtag{$\singleton{\ee}{\ee'} \cdot \ldots = \singleton{\ee}{\ee'} \cdot \ldots \cdot \containsPointer{\ee}{\ee'}$, algebra}
        & \lambda(\sk,\hh)\mydot \max_{\hh_1,\hh_2} \big\{ \\
        & \qquad \singleton{\ee}{\ee'}(\sk,\hh_1) \cdot \Ls{\ee'}{\ee''}(\sk,\hh_2) \cdot (|\dom{h}|-\containsPointer{\ee}{\ee'}(\sk,\hh)) \notag \\
        & ~\big|~ \hh = \hh_1 \sepcon \hh_2\} \\
        %
        \eeqtag{algebra}
        & \lambda(\sk,\hh)\mydot  (|\dom{h}|-\containsPointer{\ee}{\ee'}(\sk,\hh)) \\
        & \qquad \cdot \max_{\hh_1,\hh_2} \setcomp{ \singleton{\ee}{\ee'}(\sk,\hh_1) \cdot \Ls{\ee'}{\ee''}(\sk,\hh_2) }{ \hh = \hh_1 \sepcon \hh_2 } \notag \\
        \eeqtag{Definition of $\sepcon$, $\heapSize$}
        & (\heapSize-\containsPointer{\ee}{\ee'}) \cdot \left( \singleton{\ee}{\ee'} \sepcon \Ls{\ee'}{\ee''} \right) \\
        \ppreceqtag{Definition of $\Ls{\ee}{\ee''}$}
        & (\heapSize-\containsPointer{\ee}{\ee'}) \cdot \Ls{\ee}{\ee''} \\
        \eeqtag{algebra}
        & \Ls{\ee}{\ee''} \cdot \heapSize - \Ls{\ee}{\ee''} \cdot \containsPointer{\ee}{\ee'} \\
        \eeqtag{Definition of $\Ls{\ee}{\ee''}$, $\containsPointer{\ee}{\ee'}$ implies $\ee \neq \ee''$}
        & \Ls{\ee}{\ee''} \cdot \heapSize - \Ls{\ee}{\ee''} \cdot \iverson{\ee \neq \ee''} \cdot \underbrace{\containsPointer{\ee}{\ee'}}_{\preceq 1} \\
        \ppreceqtag{algebra}
        & \Ls{\ee}{\ee''} \cdot \heapSize - \Ls{\ee}{\ee''} \cdot \iverson{\ee \neq \ee''} \\
        \eeqtag{Lemma~\ref{thm:ls-props}.\ref{thm:ls-props:char}, Definition of $\Len{\ee}{\ee''}$}
        & \iverson{\ee \neq \ee''} \cdot \Len{\ee}{\ee''} - \Ls{\ee}{\ee''} \cdot \iverson{\ee \neq \ee''} \\
        \eeqtag{algebra}
        & \iverson{\ee \neq \ee''} \cdot (\Len{\ee}{\ee''} - \Ls{\ee}{\ee''}).
\end{align}
\end{proof}
\begin{lemma}\label{thm:list-length:sepimp}
\begin{align*}
        & \singleton{\ee}{\ee'} \sepimp \Len{\ee}{\ee''} \\
        \eeq & \containsPointer{\ee}{\ee'} \cdot \infty + (1-\containsPointer{\ee}{\ee'}) \cdot \iverson{\ee \neq \ee''} \cdot \left(\Ls{\ee'}{\ee''} + \Len{\ee'}{\ee''} \right).
\end{align*}
\end{lemma}
\begin{proof}
\begin{align}
        & \singleton{\ee}{\ee'} \sepimp \Len{\ee}{\ee''} \\
        \eeqtag{Lemma~\ref{thm:ls-props}.\ref{thm:ls-props:char}}
        & \singleton{\ee}{\ee'} \sepimp (\Ls{\ee}{\ee''} \cdot \heapSize) \\
        \eeqtag{Lemma~\ref{thm:misc:sepimp-contains}}
        & \singleton{\ee}{\ee'} \sepimp (\Ls{\ee}{\ee''} \cdot \heapSize \cdot \containsPointer{\ee}{\ee'}) \\
        \eeqtag{Definition of $\Ls{\ee}{\ee''}$}
        & \singleton{\ee}{\ee'} \sepimp \left(\left(\iverson{\ee = \ee''} \cdot \emp + \iverson{\ee \neq \ee''} \cdot \sup_{\za \in \Ints} \singleton{\ee}{\za} \sepcon \Ls{\za}{\ee''} \right)\cdot \heapSize \cdot \containsPointer{\ee}{\ee'} \right) \\
        \eeqtag{algebra}
        & \singleton{\ee}{\ee'} \sepimp \left(\iverson{\ee \neq \ee''} \cdot \heapSize \cdot \left( \singleton{\ee}{\ee'} \sepcon \Ls{\ee'}{\ee''} \right) \right) \\
        \eeqtag{Theorem~\ref{thm:qsl:heap-size}.\ref{thm:qsl:heap-size:dist-full}}
        & \singleton{\ee}{\ee'} \sepimp \left(\iverson{\ee \neq \ee''} \cdot \left( (\singleton{\ee}{\ee'} \cdot \underbrace{\heapSize}_{\eeq 1}) \sepcon \Ls{\ee'}{\ee''} + \singleton{\ee}{\ee'} \sepcon \underbrace{(\Ls{\ee'}{\ee''} \cdot \heapSize)}_{\eeq \Len{\ee'}{\ee''}} \right) \right) \\
        \eeqtag{algebra, Definition of $\Len{\ee'}{\ee''}$}
        & \singleton{\ee}{\ee'} \sepimp \left(\iverson{\ee \neq \ee''} \cdot \left( \singleton{\ee}{\ee'} \sepcon \Ls{\ee'}{\ee''} + \singleton{\ee}{\ee'} \sepcon \Len{\ee'}{\ee''}  \right) \right) \\
        \eeqtag{Theorem~\ref{thm:sep-con-distrib}.\ref{thm:sep-con-distrib:sepcon-over-plus-full}, algebra}
        & \singleton{\ee}{\ee'} \sepimp \left(\singleton{\ee}{\ee'} \sepcon \left( \iverson{\ee \neq \ee''} \cdot \left(\Ls{\ee'}{\ee''} + \Len{\ee'}{\ee''} \right)  \right) \right) \\
        \eeqtag{Lemma~\ref{thm:misc:sepimp-sepcon}}
        & \containsPointer{\ee}{\ee'} \cdot \infty + (1-\containsPointer{\ee}{\ee'}) \cdot \iverson{\ee \neq \ee''} \cdot \left(\Ls{\ee'}{\ee''} + \Len{\ee'}{\ee''} \right).
\end{align}
\end{proof}
\begin{lemma}\label{thm:list-length:sepimp-simple}
\begin{align*}
        & \singleton{\ee}{\ee'''} \sepcon \left(\singleton{\ee}{\ee'} \sepimp \Len{\ee}{\ee''}\right) \\
        \eeq & \singleton{\ee}{\ee'''} \sepcon \left( \iverson{\ee \neq \ee''} \cdot \left( \Ls{\ee'}{\ee''} + \Len{\ee'}{\ee''} \right) \right).
\end{align*}
\end{lemma}
\begin{proof}
\begin{align}
        & \singleton{\ee}{\ee'''} \sepcon \left(\singleton{\ee}{\ee'} \sepimp \Len{\ee}{\ee''} \right) \\
        \eeqtag{Lemma~\ref{thm:list-length:sepimp}}
        & \singleton{\ee}{\ee'''} \sepcon \big( \containsPointer{\ee}{\ee'} \cdot \infty + (1-\containsPointer{\ee}{\ee'}) \cdot \iverson{\ee \neq \ee''} \cdot \left(\Ls{\ee'}{\ee''} + \Len{\ee'}{\ee''} \right) \big) \\
        \eeqtag{Theorem~\ref{thm:sep-con-distrib}.\ref{thm:sep-con-distrib:sepcon-over-plus-full}}
        & \singleton{\ee}{\ee'''} \sepcon (\containsPointer{\ee}{\ee'} \cdot \infty) \\
        & \qquad + \singleton{\ee}{\ee'''} \sepcon \left( (1-\containsPointer{\ee}{\ee'}) \cdot \iverson{\ee \neq \ee''} \cdot \left(\Ls{\ee'}{\ee''} + \Len{\ee'}{\ee''} \right) \right) \notag \\
        \eeqtag{Theorem~\ref{thm:sep-con-algebra-pure}}
        & \underbrace{(\singleton{\ee}{\ee'''} \sepcon \containsPointer{\ee}{\ee'})}_{\eeq 0} \cdot \infty \\
        & \qquad + \underbrace{\singleton{\ee}{\ee'''} \sepcon \big( (1-\containsPointer{\ee}{\ee'})}_{\eeq 0} \cdot \iverson{\ee \neq \ee''} \cdot \left(\Ls{\ee'}{\ee''} + \Len{\ee'}{\ee''} \right) \big) \notag \\
        \eeqtag{$\validpointer{\ee} \sepcon \containsPointer{\ee}{-} = 0$}
        & \singleton{\ee}{\ee'''} \sepcon \left( \iverson{\ee \neq \ee''} \cdot \left( \Ls{\ee'}{\ee''} + \Len{\ee'}{\ee''} \right) \right).
\end{align}
\end{proof}

\begin{definition}\label{def:domain-disjoint}
        Two predicates $\preda,\predb \in \E$ are \emph{domain-disjoint} iff 
        \begin{align*}
                \forall \sk \in \Stacks \forall \hh \in \Heaps \colon |\{ (\hh_1,\hh_2) ~|~ \hh = \hh_1 \sepcon \hh_2, (\sk,\hh_1) \models \preda, (\sk,\hh_2) \models \predb \}| \leq 1.
        \end{align*}
\end{definition}

\begin{lemma}\label{thm:ls:domain-disjoint}
  $\Ls{\ee}{0}$ and $\Ls{\ee'}{0}$ are domain-disjoint.
\end{lemma}

\begin{proof}
Let $\sk \in \Stacks$ be a stack.
For a given heap $\hh$, let us define
\begin{align}
        M(\sk,\hh) \eeq |\{ (\hh_1,\hh_2) ~|~ \hh = \hh_1 \sepcon \hh_2, (\sk,\hh_1) \models \preda, (\sk,\hh_2) \models \predb \}|.
\end{align}
We show for all heaps $\hh \in \Heaps$ that $M(\sk,\hh) \leq 1$ by induction on $n = |\dom{\hh}|$.

For $n = 0$, $\hh$ is the empty heap $\emptyheap$. Then $\hh = \emptyheap \sepcon \emptyheap$ is the only possible partitioning of $\hh$.
Hence, $M(\sk,\hh) \leq 1$.

For $n > 0$, we know that $\hh \neq \emptyheap$. 
Furthermore, assume $(\hh_1,\hh_2) \in M(\sk,\hh)$. If no such pair exists, we have $M(\sk,\hh) = 0$ and there is nothing to show.
We distinguish two cases: 

First, assume $\sk(\ee) = 0$. 
Since $(\sk,\hh_1) \models \Ls{\ee}{0}$, we have
\begin{align}
        1 \eeq \Ls{\ee}{0}(\sk,\hh_1) \eeq & 
        \underbrace{\iverson{\ee = 0}(\sk,\hh_1)}_{\eeq 1} \cdot \emp(\sk,\hh_1) \\
        & + \underbrace{\iverson{\ee \neq 0}(\sk,\hh_1)}_{\eeq 0} \cdot \left(\sup_{\za \in \Ints} \singleton{\ee}{\za} \sepcon \Ls{\za}{0}\right)(\sk,\hh_1).
\end{align}
Hence, $\hh_1 = \emptyheap$. Since $\hh = \hh_1 \sepcon \hh_2$, we then know that $\hh_2 = \hh$.
Consequently, $M(\sk,\hh) \leq 1$.

Now, assume $\sk(\ee) \neq 0$.
Since $(\sk,\hh_1) \models \Ls{\ee}{0}$, we have
\begin{align}
        & 1 \eeq \Ls{\ee}{0}(\sk,\hh_1) \\
        \eeqtag{Definition of $\Ls{\ee}{0}$}
        & \underbrace{\iverson{\ee = 1}(\sk,\hh_1)}_{\eeq 0} \cdot \emp(\sk,\hh_1) + \underbrace{\iverson{\ee \neq 0}(\sk,\hh_1)}_{\eeq 1} \cdot \left(\sup_{\za \in \Ints} \singleton{\ee}{\za} \sepcon \Ls{\za}{0}\right)(\sk,\hh_1) \\
        \eeqtag{by assumption}
        & \left(\sup_{\za \in \Ints} \singleton{\ee}{\za} \sepcon \Ls{\za}{0}\right)(\sk,\hh_1) \\
        \eeqtag{Definition of $\sepcon$}
        & \sup_{\za \in \Ints} \max_{\hh_3,\hh_4} \setcomp{ \singleton{\ee}{\za}(\sk,\hh_3) \cdot \Ls{\za}{0}(\sk,\hh_4) }{ \hh_1 = \hh_3 \sepcon \hh_4 } \\
\end{align}
Now, $\hh_3 = \{ \sk(\ee) \mapsto \za \}$ is the only possible choice such that $\singleton{\ee}{\za}(\sk,\hh_3) = 1$.
Hence, $|\dom{\hh_4}| = |\dom{\hh_1}| - 1 < |\dom{\hh}|$ 
We may thus apply the induction hypothesis, to conclude that $M(\sk,\hh_4 \sepcon \hh_2) \leq 1$ for expectations $\Ls{\za}{0}$ and $\Ls{\ee'}{0}$.
In other words, there is at most one possible choice for heap $\hh_4$ such that $\Ls{\za}{0}(\sk,\hh_4) = 1$.
Hence, $M(\sk,\hh) \leq 1$.
\end{proof}

\begin{lemma}\label{thm:sepcon-distrib-domain-disjoint}
  Let $\ff,\fg \in \E$.
  Moreover, let $\preda,\predb \in \E$ be domain-disjoint predicates. Then
  \begin{align*}
          \preda \sepcon (\predb \cdot \ff + \predb \cdot \fg) \eeq \preda \sepcon (\predb \cdot \ff) + \preda \sepcon (\predb \cdot \fg).
  \end{align*}
\end{lemma}

\begin{proof}
  \begin{align}
          & \preda \sepcon (\predb \cdot \ff + \predb \cdot \fg) \\
          \eeqtag{Definition of $\sepcon$}
          & \lambda(\sk,\hh)\mydot \max_{\hh_1,\hh_2} \setcomp{ \preda(\sk,\hh_1) \cdot (\predb(\sk,\hh_2) \cdot \ff(\sk,\hh_2) + \predb(\sk,\hh_2) \cdot \fg(\sk,\hh_2)) }{\hh = \hh_1 \sepcon \hh_2} \\
          \eeqtag{algebra}
          & \lambda(\sk,\hh)\mydot \max_{\hh_1,\hh_2} \big\{ \underbrace{\preda(\sk,\hh_1) \cdot \predb(\sk,\hh_2) \cdot \ff(\sk,\hh_2)}_{\geq 0~\text{for at most one choice of $\hh_1,\hh_2$}}
  + \underbrace{\preda(\sk,\hh_1) \cdot \predb(\sk,\hh_2) \cdot \fg(\sk,\hh_2))}_{\geq 0~\text{for at most one choice of $\hh_1,\hh_2$}} ~\big|~ \hh = \hh_1 \sepcon \hh_2 \big\} \\
          \eeqtag{$\preda,\predb$ domain-disjoint, $\max$ over singleton}
          & \lambda(\sk,\hh)\mydot \max_{\hh_1,\hh_2} \setcomp{ \preda(\sk,\hh_1) \cdot \predb(\sk,\hh_2) \cdot \ff(\sk,\hh_2) }{ \hh = \hh_1 \sepcon \hh_2 } \\
          & \qquad + \max_{\hh_1,\hh_2} \setcomp{ \preda(\sk,\hh_1) \cdot \predb(\sk,\hh_2) \cdot \fg(\sk,\hh_2) }{ \hh = \hh_1 \sepcon \hh_2 } \notag \\
          \eeqtag{algebra}
          & \lambda(\sk,\hh)\mydot \max_{\hh_1,\hh_2} \setcomp{ \preda(\sk,\hh_1) \cdot \predb(\sk,\hh_2) \cdot \ff(\sk,\hh_2) }{ \hh = \hh_1 \sepcon \hh_2 } \\
          & \qquad + \lambda(\sk,\hh)\mydot \max_{\hh_1,\hh_2} \setcomp{ \preda(\sk,\hh_1) \cdot \predb(\sk,\hh_2) \cdot \fg(\sk,\hh_2) }{\hh = \hh_1 \sepcon \hh_2 } \notag \\
          \eeqtag{Definition of $\sepcon$}
          & \preda \sepcon (\predb \cdot \ff) + \preda \sepcon (\predb \cdot \fg).
  \end{align}
\end{proof}
\begin{theorem}[Pure Frame Rule]\label{thm:batz}
      Let $\cc \in \hpgcl$ and $\ff,\fg \in \E$ such that $\fg$ is pure and $\Vars (\fg) \cap \Mod{\cc} = \emptyset$.
      Then $\wp{\cc}{\fg \cdot \ff} \eeq  \fg \cdot \wp{\cc}{\ff}$.
\end{theorem}
\begin{proof}
       Let $\cc \in \hpgcl$ and $\ff,\fg \in \E$ such that $\fg$ is pure and $\Vars (\fg) \cap \Mod{\cc} = \emptyset$.
       Then we have to show that $\wp{\cc}{\fg \cdot \ff} \eeq  \fg \cdot \wp{\cc}{\ff}$.

       By induction on the structure of $\cc$. \\ \\
       \textbf{The case} $\cc = \SKIP$. We have
       \begin{align}
              &\wp{\SKIP}{\fg \cdot \ff} \\
              \eeqtag{Table \ref{table:wp}}      
              &\fg \cdot \ff \\
              \eeqtag{Table \ref{table:wp}}            
              &\fg \cdot \wp{\SKIP}{\ff}~. 
       \end{align} \\
      \textbf{The case} $\cc = \ASSIGN{x}{\ee}$. We have
      \begin{align}
         &\wp{\ASSIGN{x}{\ee}}{\fg \cdot \ff} \\
         \eeqtag{Table \ref{table:wp}}
         & (\fg \cdot \ff) \subst{x}{\ee} \\
         \eeqtag{Substitution distributes}
         & \left( \fg \subst{x}{\ee} \right) \cdot \left( \ff \subst{x}{\ee} \right) \\
         \eeqtag{By assumption $\Vars (\fg) \cap \{ x \} = \emptyset$}
         & \fg \cdot \ff \subst{x}{\ee} \\
         \eeqtag{Table \ref{table:wp}}
         & \fg \cdot \wp{\ASSIGN{x}{\ee}}{\ff}~.
      \end{align}
      \textbf{The case} $\cc = \ALLOC{x}{\vec{\ee}}$.  In this case we make use of the fact that for $\emptyset \neq A \subseteq \PosReals$ and $\cc \in \PosReals$ it holds that
      \begin{align}
      \label{eqn:unaffected-inf}
          \inf \left\{\cc \cdot a ~\mid~ a \in A \right\} \eeq \cc \cdot \inf A~.
      \end{align}
      We then prove this case point-wise as follows:
      \begin{align}
         &\wp{\ALLOC{x}{\vec{\ee}}}{\fg \cdot \ff}(\sk,\hh) \\
         \eeqtag{Table \ref{table:wp}}
         &\big( \displaystyle\inf_{v \in \AVAILLOC{\vec{\ee}}} \singleton{v}{\vec{\ee}} \sepimp (\fg \cdot \ff)\subst{x}{v} \big) (\sk,\hh) \\
         \eeqtag{Definition of $\sepimp$}
         &\displaystyle\inf_{v \in \AVAILLOC{\vec{\ee}}} \displaystyle\inf_{\hh'}
         \left\{ \big( (\fg \cdot \ff)\subst{x}{v} \big) (\sk,\hh \sepcon \hh') ~\mid~ \hh' \disjoint \hh, \, \hh' \models \singleton{v}{\vec{\ee}}  \right\} \\
         \eeqtag{Substitution distributes and definition of $\cdot$ w.r.t.\ $\E$}
         &\displaystyle\inf_{v \in \AVAILLOC{\vec{\ee}}} \displaystyle\inf_{\hh'}
         \left\{ \left( \fg \subst{x}{v} \right)(\sk,\hh \sepcon \hh') \cdot \left( \ff \subst{x}{v} \right)(\sk,\hh \sepcon \hh') ~\mid~ \hh' \disjoint \hh, \, \hh' \models \singleton{v}{\vec{\ee}}  \right\} \\
         \eeqtag{By assumption $\Vars (\fg) \cap \{ x \} = \emptyset$} 
         &\displaystyle\inf_{v \in \AVAILLOC{\vec{\ee}}} \displaystyle\inf_{\hh'} \left\{ \fg(\sk,\hh \sepcon \hh') \cdot \left( \ff \subst{x}{v} \right)(\sk,\hh \sepcon \hh') ~\mid~ \hh' \disjoint \hh, \, \hh' \models \singleton{v}{\vec{\ee}}  \right\} \\
         \eeqtag{$\fg$ is pure} 
         &\displaystyle\inf_{v \in \AVAILLOC{\vec{\ee}}} \displaystyle\inf_{\hh'} \left\{ \fg(\sk,\hh) \cdot \left( \ff \subst{x}{v} \right)(\sk,\hh \sepcon \hh') ~\mid~ \hh' \disjoint \hh, \, \hh' \models \singleton{v}{\vec{\ee}}  \right\} \\
         \eeqtag{By Equation \ref{eqn:unaffected-inf}}
         &\displaystyle\inf_{v \in \AVAILLOC{\vec{\ee}}} \big\{ \fg(\sk,\hh) \cdot \displaystyle\inf_{\hh'} \left\{ \left( \ff \subst{x}{v} \right)(\sk,\hh \sepcon \hh') ~\mid~ \hh' \disjoint \hh, \, \hh' \models \singleton{v}{\vec{\ee}}  \right\} \big\} \\
         \eeqtag{By Equation \ref{eqn:unaffected-inf}}
         &\fg(\sk,\hh) \cdot \displaystyle\inf_{v \in \AVAILLOC{\vec{\ee}}}   \displaystyle\inf_{\hh'} \left\{ \left( \ff \subst{x}{v} \right)(\sk,\hh \sepcon \hh') ~\mid~ \hh' \disjoint \hh, \, \hh' \models \singleton{v}{\vec{\ee}}  \right\} \\
         \eeqtag{Definition of $\sepimp$}
         &\fg(\sk,\hh) \cdot \displaystyle\inf_{v \in \AVAILLOC{\vec{\ee}}} \singleton{v}{\vec{\ee}} \sepimp \ff\subst{x}{v} \\
         \eeqtag{Table \ref{table:wp}}
         &\fg(\sk,\hh) \cdot \wp{\ALLOC{x}{\vec{\ee}}}{\ff}(\sk,\hh)~.
      \end{align}
      \textbf{The case} $\ASSIGNH{x}{\ee}$. We have
      \begin{align}
         &\wp{\ASSIGNH{x}{\ee}}{\fg \cdot \ff} \\
         \eeqtag{Table \ref{table:wp}}
         &\displaystyle\sup_{v \in \Ints} \singleton{\ee}{v} \sepcon \bigl( \singleton{\ee}{v} \sepimp (\fg \cdot \ff)\subst{x}{v} \bigr) \\
         \eeqtag{Alternative version of the rule for heap lookup}
         &\displaystyle\sup_{v \in \Ints} \singleton{\ee}{v} \cdot (\fg \cdot \ff)\subst{x}{v} \\
         \eeqtag{Substitution distributes}
         &\displaystyle\sup_{v \in \Ints} \singleton{\ee}{v} \cdot \big( (\fg\subst{x}{v}) \cdot (\ff\subst{x}{v}) \big) \\
         \eeqtag{By assumption $\Vars (\fg) \cap \{ x \} = \emptyset$} 
         &\displaystyle\sup_{v \in \Ints} \singleton{\ee}{v} \cdot \big( \fg \cdot (\ff\subst{x}{v}) \big) \\
         \eeqtag{$\fg$ does not depend on $v$}
         & \fg \cdot \displaystyle\sup_{v \in \Ints} \singleton{\ee}{v} \cdot   (\ff\subst{x}{v}) \\
         \eeqtag{Alternative version of the rule for heap lookup}
         &\fg \cdot \displaystyle\sup_{v \in \Ints} \singleton{\ee}{v} \sepcon \bigl( \singleton{\ee}{v} \sepimp \ff\subst{x}{v} \bigr) \\
         &\eeqtag{Table \ref{table:wp}}
         &\fg \cdot \wp{\ASSIGNH{x}{\ee}}{\ff}~.
      \end{align}
\textbf{The case $\cc = \HASSIGN{\ee}{\ee'}$}. We prove the claim point-wise as follows:
Let $(\sk,\hh) \in \States$. We distinguish the cases $\sk(\ee) \in \dom{\hh}$ and $\sk(\ee) \not\in \dom{\hh}$. If $\sk(\ee) \not\in \dom{\hh}$, then
\begin{align}
     &\wp{\HASSIGN{\ee}{\ee'}}{\fg \cdot \ff}(\sk,\hh) \\
     \eeqtag{Table \ref{table:wp}}
     &\big( \validpointer{\ee} \sepcon \bigl(\singleton{\ee}{\ee'} \sepimp (\fg \cdot \ff) \bigr) \big) (\sk,\hh)  \\
     \eeqtag{$\sk(\ee) \not\in \dom{\hh}$}
     & 0 \\
     \eeqtag{$\sk(\ee) \not\in \dom{\hh}$}
     &\fg(\sk,\hh) \cdot \big( \validpointer{\ee} \sepcon \bigl(\singleton{\ee}{\ee'} \sepimp \ff \bigr) \big) (\sk,\hh)
      \\
     \eeqtag{Table \ref{table:wp}}
     &\fg(\sk,\hh) \cdot \wp{\HASSIGN{\ee}{\ee'}}{\ff}(\sk,\hh)~.
\end{align}
Now let $\sk(\ee) \in \dom{\hh}$. For two arithmetic expressions $\ee_1, \ee_2$, we denote by $\hh_{\ee_1,\ee_2}$ the heap 
with $\{ \sk(\ee_1) \} = \dom{\hh_{\ee_1,\ee_2}}$ and $\hh_{\ee_1,\ee_2}(\sk(\ee_1)) = \sk(\ee_2)$. The heap $\hh$ is thus of the form
$\hh = \hh' \sepcon \hh_{\ee,v}$ for some heap $\hh'$ and some $v \in \Ints$. We have
\begin{align}
     &\wp{\HASSIGN{\ee}{\ee'}}{\fg \cdot \ff }(\sk,\hh) \\
     \eeqtag{Table \ref{table:wp}}
     &\big( \validpointer{\ee} \sepcon \bigl(\singleton{\ee}{\ee'} \sepimp (\fg \cdot \ff) \bigr) \big) (\sk,\hh) \\
     \eeqtag{By assumption}
     &\big( \validpointer{\ee} \sepcon \bigl(\singleton{\ee}{\ee'} \sepimp (\fg \cdot \ff) \bigr) \big) (\sk,\hh' \sepcon \hh_{\ee,v}) \\
     \eeqtag{$(\validpointer{\ee} \sepcon u)(\sk,\hh' \sepcon \hh_{\ee,v}) = u (\sk,\hh')$ for all $u \in \E$}
     & \bigl(\singleton{\ee}{\ee'} \sepimp (\fg \cdot \ff) \bigr)  (\sk,\hh') \\
     \eeqtag{$\sk(\ee) \not\in \dom{\hh'}$}
     &(\fg \cdot \ff) (\sk,\hh' \sepcon \hh_{\ee,\ee'}) \\
     \eeqtag{Definition of $\cdot$ w.r.t.\ $\E$}
     &\fg(\sk,\hh' \sepcon \hh_{\ee,\ee'})  \cdot \ff(\sk,\hh' \sepcon \hh_{\ee,\ee'})  \\
     \eeqtag{$\fg$ is pure}
     &\fg(\sk,\hh)  \cdot \ff(\sk,\hh' \sepcon \hh_{\ee,\ee'}) \\
     \eeqtag{$\sk(\ee) \not\in \dom{\hh'}$}
     &\fg(\sk,\hh)  \cdot  \bigl(\singleton{\ee}{\ee'} \sepimp \ff \bigr)  (\sk,\hh') \\
     \eeqtag{$u (\sk,\hh') = (\validpointer{\ee} \sepcon u)(\sk,\hh' \sepcon \hh_{\ee,v})$ for all $u \in \E$}
     &\fg(\sk,\hh)  \cdot \big( \validpointer{\ee} \sepcon \bigl(\singleton{\ee}{\ee'} \sepimp \ff \bigr) \big) (\sk,\hh' \sepcon \hh_{\ee,v}) \\
     \eeqtag{By assumption}
     &\fg(\sk,\hh)  \cdot \big( \validpointer{\ee} \sepcon \bigl(\singleton{\ee}{\ee'} \sepimp \ff \bigr) \big) (\sk,\hh) \\
     \eeqtag{Table \ref{table:wp}}
     &\fg(\sk,\hh) \cdot \wp{\HASSIGN{\ee}{\ee'}}{\ff}(\sk,\hh)~.
\end{align}
\textbf{The case $\cc = \FREE{\ee}$}. We show the claim point-wise as follows:
We distinguish the cases $\sk(\ee) \in \dom{\hh}$ and $\sk(\ee) \not\in \dom{\hh}$. If $\sk(\ee) \not\in \dom{\hh}$, then
\begin{align}
    &\wp{\FREE{\ee}}{\fg \cdot \ff }(\sk,\hh) \\
    \eeqtag{Table \ref{table:wp}}
    &\big( \validpointer{\ee} \sepcon (\fg \cdot \ff) \big) (\sk,\hh) \\
    \eeqtag{$\sk(\ee) \not\in \dom{\hh}$}
    &0 \\
    \eeqtag{$\sk(\ee) \not\in \dom{\hh}$} 
    & \fg(\sk,\hh) \cdot \big( \validpointer{\ee} \sepcon \ff \big) (\sk,\hh)  \\
    \eeqtag{Table \ref{table:wp}}
    & a \cdot \wp{\FREE{\ee}}{\ff} (\sk,\hh) + \wp{\FREE{\ee}}{\fg}(\sk,\hh)~.
\end{align}
If $\sk(\ee) \in \dom{\hh}$, then the heap $\hh$ is of the form $\hh = \hh' \sepcon \hh_{\ee, v}$ for some heap $\hh'$ and some $v \in \Ints$.
We have
\begin{align}
     &\wp{\FREE{\ee}}{ \fg \cdot \ff}(\sk,\hh) \\
    \eeqtag{Table \ref{table:wp}}
    &\big( \validpointer{\ee} \sepcon ( \fg \cdot \ff) \big) (\sk,\hh) \\
    \eeqtag{By assumption}
    &\big( \validpointer{\ee} \sepcon ( \fg \cdot \ff ) \big) (\sk,\hh' \sepcon \hh_{\ee, v}) \\
    \eeqtag{$(\validpointer{\ee} \sepcon u)(\sk,\hh' \sepcon \hh_{\ee, v}) = u (\sk,\hh')$ for all $u \in \E$}
    &( \fg \cdot \ff) (\sk,\hh') \\
    \eeqtag{Definition of $\cdot$ w.r.t.\ $\E$}
    & \fg(\sk,\hh') \cdot \ff(\sk,\hh')  \\
    \eeqtag{$\fg$ is pure}
    &\fg(\sk,\hh) \cdot \ff(\sk,\hh') \\
    \eeqtag{$u (\sk,\hh')= (\validpointer{\ee} \sepcon u)(\sk,\hh' \sepcon \hh_{\ee, v})$ for all $u \in \E$}
    &\fg(\sk,\hh) \cdot (\validpointer{\ee} \sepcon \ff) (\sk,\hh' \sepcon \hh_{\ee, v}) \\
    \eeqtag{By assumption}
    &\fg(\sk,\hh) \cdot (\validpointer{\ee} \sepcon \ff) (\sk,\hh) \\
    \eeqtag{By Table \ref{table:wp}}
    &\fg(\sk,\hh) \cdot \wp{\FREE{\ee}}{\ff}(\sk,\hh)~.
\end{align}
As the induction hypothesis now suppose that for some arbitrary, but fixed, $\cc_1, \cc_2 \in \rhpgcl$, all $\ff \in \E$, all pure $\fg_1, \fg_2 \in \E$ with
$\Vars (\fg_1) \cap \Mod{\cc_1} = \emptyset$ and $\Vars (\fg_2) \cap \Mod{\cc_2} = \emptyset$, all variable environments $\varenv \in \VarEnv$, and 
all procedure environments $\procenv \in \ProcEnv$ satisfying the premise it holds that both
\begin{align}
   &\wp{\cc_1, \varenv, \procenv}{\fg_1 \cdot \ff} \eeq \fg_1 \cdot \wp{\cc_1}{\ff} \\
   \text{and} \quad 
   &\wp{\cc_2}{\fg_2 \cdot \ff} \eeq \fg_2 \cdot \wp{\cc_2}{\ff}~. 
\end{align}
\textbf{The case} $\cc = \COMPOSE{\cc_1}{\cc_2}$. We have
\begin{align}
   &\wp{\COMPOSE{\cc_1}{\cc_2}}{\fg \cdot \ff} \\
   \eeqtag{Table \ref{table:wp}}
   &\wp{\cc_1}{\wp{\cc_2}{\fg \cdot \ff}} \\
   \eeqtag{I.H.\ on $\cc_2$}
   &\wp{\cc_1}{\fg \cdot \wp{\cc_2}{\ff}} \\
   \eeqtag{I.H.\ on $\cc_1$}
   &\fg \cdot \wp{\cc_1}{\wp{\cc_2}{\ff}} \\
   \eeqtag{Table \ref{table:wp}}
   &\fg \cdot \wp{\COMPOSE{\cc_1}{\cc_2}}{\ff}~.
\end{align}
\textbf{The case} $\cc = \ITE{\guard}{\cc_1}{\cc_2}$. We have
\begin{align}
   &\wp{\ITE{\guard}{\cc_1}{\cc_2}}{\fg \cdot \ff} \\
   \eeqtag{Table \ref{table:wp}}
   &\iverson{\guard} \cdot \wp{\cc_1}{\fg \cdot \ff} + \iverson{\neg \guard} \cdot \wp{\cc_2}{\fg \cdot \ff} \\
   \eeqtag{I.H.\ on $\cc_1$ and I.H.\ on $\cc_2$}
   &\iverson{\guard} \cdot \fg \cdot \wp{\cc_1}{\ff} + \iverson{\neg \guard} \cdot \fg \cdot \wp{\cc_2}{\ff} \\
   \eeqtag{Algebra, $\iverson{\guard} + \iverson{\neg \guard} = 1$}
   &\fg \cdot \big( \iverson{\guard} \cdot \wp{\cc_1}{\ff} + \iverson{\neg \guard} \cdot \wp{\cc_2}{\ff} \big) \\
   \eeqtag{Table \ref{table:wp}}
   &\fg \cdot \wp{\ITE{\guard}{\cc_1}{\cc_2}}{\ff}~.
\end{align}
\textbf{The case} $\cc = \WHILEDO{\guard}{\cc_1}$. Due to the fact that there is an ordinal $\oa$ such that
\begin{align}
   &\wp{\WHILEDO{\guard}{\cc_1}}{\fg \cdot \ff} \eeq \charwpn{\guard}{\cc_1}{\fg \cdot \ff}{\oa}(0)~,
\end{align}
it suffices to show that 
\begin{align}
   \charwpn{\guard}{\cc_1}{\fg \cdot \ff}{\oc}(0) \eeq \fg \cdot \charwpn{\guard}{\cc_1}{\ff}{\oc}(0) \quad \forall ~ \text{ordinals} ~ \oc~.
\end{align}
We proceed by transfinite induction on $\oc$. \\ \\
\noindent
\emph{The case $\oc = 0$}. This case is trivial since
\begin{align}
   &\charwpn{\guard}{\cc_1}{\fg \cdot \ff}{0}(0) \\
   \eeqtag{By definition}
   & 0 \\
   \eeqtag{By definition}
   &\fg \cdot \charwpn{\guard}{\cc_1}{\ff}{0}(0)~.
\end{align}
\emph{The case $\oc$ successor ordinal.} We have
\begin{align}
   &\charwpn{\guard}{\cc_1}{\fg \cdot \ff}{\oc + 1}(0) \\
   \eeqtag{By definition}
   & \charwp{\guard}{\cc_1}{\fg \cdot \ff}(\charwpn{\guard}{\cc_1}{\fg \cdot \ff}{\oc}(0)) \\
   \eeqtag{I.H.\ on $\oc$}
   & \charwp{\guard}{\cc_1}{\fg \cdot \ff}(\fg \cdot \charwpn{\guard}{\cc_1}{\ff}{\oc}(0)) \\
   \eeqtag{By definition}
   &\iverson{\neg\guard} \cdot \fg \cdot \ff + \iverson{\guard} \cdot \wp{\cc_1}{\fg \cdot \charwpn{\guard}{\cc_1}{\ff}{\oc}(0)} \\
   \eeqtag{I.H.\ on $\cc_1$}
   &\iverson{\neg\guard} \cdot \fg \cdot \ff + \iverson{\guard} \cdot \fg \cdot \wp{\cc_1}{\charwpn{\guard}{\cc_1}{\ff}{\oc}(0)} \\
   \eeqtag{Algebra}
   &\fg \cdot \big( \iverson{\neg\guard} \cdot \ff + \iverson{\guard} \cdot \wp{\cc_1}{\charwpn{\guard}{\cc_1}{\ff}{\oc}(0)} \big) \\
   \eeqtag{By definition}
   &\fg \cdot \charwp{\guard}{\cc_1}{\ff} ( \charwpn{\guard}{\cc_1}{\ff}{\oc}(0) ) \\
   \eeqtag{By definition}
   &\fg \cdot \charwpn{\guard}{\cc_1}{\ff}{\oc + 1}(0)~.
\end{align}
\emph{The case $\oc$ limit ordinal.} We have
\begin{align}
   &\charwpn{\guard}{\cc_1}{\fg \cdot \ff}{\oc}(0) \\
   \eeqtag{By definition}
   &\sup_{\ob < \oc} \charwpn{\guard}{\cc_1}{\fg \cdot \ff}{\ob}(0) \\
   \eeqtag{I.H.\ on $\ob$}
   &\sup_{\ob < \oc} \, \big(  \fg \cdot \charwpn{\guard}{\cc_1}{\ff}{\ob}(0) \big) \\
   \eeqtag{$\fg$ does not depend on $\ob$}
   &\fg \cdot \sup_{\ob < \oc} \charwpn{\guard}{\cc_1}{\ff}{\ob}(0) \\
   \eeqtag{By definition}
   &\fg \cdot \charwpn{\guard}{\cc_1}{\ff}{\oc}(0)~.
\end{align}
\end{proof}
\begin{lemma}
   \label{lem:wand-pure-expectation-on-rhs}
   Let $\ff \in \E$ be pure and let $\ee,\ee'$ be arithmetic expressions. We have
   \begin{align*}
      \singleton{\ee}{\ee'} \sepimp \ff \eeq \containsPointer{\ee}{-} \cdot \infty + (1 - \containsPointer{\ee}{-}) \cdot \ff
   \end{align*}
\end{lemma}
\begin{proof}
   Let $(\sk,\hh ) \in \States$. We distinguish the cases $\sk (\ee) \in \dom{\hh}$ and $\sk (\ee) \not\in \dom{\hh}$.
   For the first case, we have
   \begin{align}
       &\big( \singleton{\ee}{\ee'} \sepimp \ff \big) (\sk,\hh ) \\
       \eeqtag{By assumption: $\sk (\ee) \in \dom{\hh}$}
       & \infty \\
       \eeqtag{$\containsPointer{\ee}{-}(\sk,\hh ) = 1$}
       & \big( \containsPointer{\ee}{-} \cdot \infty \big)(\sk,\hh ) \\
       \eeqtag{$(1 - \containsPointer{\ee}{-})(\sk,\hh ) = 0$}
       & \big( \containsPointer{\ee}{-} \cdot \infty + (1 - \containsPointer{\ee}{-}) \cdot \ff \big) (\sk,\hh ).
   \end{align}
   For the second case, i.e.\ $\sk (\ee) \not\in \dom{\hh}$, we get
   \begin{align}
      &\big( \singleton{\ee}{\ee'} \sepimp \ff \big) (\sk,\hh ) \\
      \eeqtag{Definition of $\sepimp$}
      & \inf_{\hh'}~ \setcomp{\ff(\sk, \hh \sepcon \hh')}{\hh' \disjoint \hh ~\textnormal{ and }~ (\sk, \hh') \models \singleton{\ee}{\ee'}} \\
      \eeqtag{$\ff$ is pure}
      & \inf_{\hh'}~ \setcomp{\ff(\sk, \hh )}{\hh' \disjoint \hh ~\textnormal{ and }~ (\sk, \hh') \models \singleton{\ee}{\ee'}} \\
      \eeqtag{$\sk (\ee) \not\in \dom{\hh}$ so there is a $\hh'$ with $\hh' \disjoint \hh ~\textnormal{ and }~ (\sk, \hh') \models \singleton{\ee}{\ee'}$}
      &\ff(\sk, \hh ) \\
      \eeqtag{$(1 - \containsPointer{\ee}{-})(\sk, \hh) = 1$}
      & \big( (1 - \containsPointer{\ee}{-}) \cdot \ff \big)(\sk, \hh) \\
      \eeqtag{$\containsPointer{\ee}{-}(\sk, \hh) = 0$}
      &\big( \containsPointer{\ee}{-} \cdot \infty + (1 - \containsPointer{\ee}{-}) \cdot \ff \big) (\sk, \hh).
   \end{align}
\end{proof}
\begin{lemma}
\label{lem:inf-over-addresses}
   Let $\ff \in \E$ and let $\ee$ be an arithmetic expression. We have
   \begin{align*}
      \displaystyle\inf_{v \in \AVAILLOC{\ee}} \containsPointer{v}{-} \cdot \infty + (1 - \containsPointer{v}{-}) \cdot \ff \eeq \ff
   \end{align*}
\end{lemma}
\begin{proof}
Since for every state $(\sk ,\hh )$, the domain of $\hh$ is finite, i.e.\ $|\dom{\hh}|< \infty$, there is an address $v$ such that $(1 - \containsPointer{v}{-})(\sk ,\hh ) = 1$.
Hence, it holds that
\begin{align}
    &\big( \displaystyle\inf_{v \in \AVAILLOC{\ee}} \containsPointer{v}{-} \cdot \infty + (1 - \containsPointer{v}{-}) \cdot \ff \big) (\sk, \hh) \\
    \eeqtag{Choose $v$ such that $(1 - \containsPointer{v}{-})(\sk ,\hh ) = 1$}
    &\big( (1 - \containsPointer{v}{-}) \cdot \ff \big) (\sk, \hh) \\
    \eeqtag{$(1 - \containsPointer{v}{-})(\sk ,\hh ) = 1$}
    &\ff(\sk, \hh).
\end{align}
\end{proof}

%

\end{document}